\documentclass[aps,pra,groupedaddress,nofootinbib,notitlepage,showpacs,floatfix]{revtex4-1}

\usepackage{graphicx,graphics,epsfig,subfigure,times,bm,bbm,amssymb,amsmath,amsfonts,amsthm,mathrsfs,MnSymbol}
\usepackage[matrix,frame,arrow]{xypic}
\usepackage[pdfstartview=FitH]{hyperref}
\usepackage[pdftex]{color}
\usepackage{hypernat}
\usepackage[pdfstartview=FitH]{hyperref}
\hypersetup{
    colorlinks=true,       
    linkcolor=cyan,          
    citecolor=magenta,        
    filecolor=magenta,      
    urlcolor=cyan,           
    runcolor=cyan
}
\usepackage{braket}
\usepackage{enumerate}
\input{Qcircuit.tex} 
\usepackage{multirow}

\newtheorem{thm}{Theorem}

\newtheorem{mytheorem}{Theorem}
\newtheorem{mylemma}{Lemma}

\newtheorem{myclaim}{Claim}
\newtheorem{mydefinition}{Definition}
\newtheorem{myassumption}{Assumption}

\newcommand{\mc}{\mathcal}
\newcommand{\tx}[1]{\text{#1}}
\newcommand{\derv}[3]{\frac{d^{#3}#1}{d#2^{#3}}}					
\newcommand{\half}{\tfrac{1}{2}}

\newcommand{\bes} {\begin{subequations}}
\newcommand{\ees} {\end{subequations}}
\newcommand{\bea} {\begin{eqnarray}}
\newcommand{\eea} {\end{eqnarray}}
\newcommand{\beq}{\begin{equation}}
\newcommand{\eeq}{\end{equation}}

\newcommand{\expv}[1]{\langle #1\rangle}							
\newcommand{\ph}{\ensuremath{\varphi}}
\newcommand{\eps}{\ensuremath{\varepsilon}} 
\newcommand{\R}{\ensuremath{\mathbb{R}}} 
\newcommand{\C}{\ensuremath{\mathbb{C}}} 
\newcommand{\e}{\ensuremath{{e}}} 
\newcommand{\ii}{\ensuremath{{i}}}
\newcommand{\abs}[1]{\ensuremath{\left|#1\right|}} 
\newcommand{\norm}[1]{\ensuremath{\left\|#1\right\|}} 
\newcommand{\opU}{\ensuremath{{{U}}}}
\newcommand{\opH}{\ensuremath{{{H}}}}

\newcommand{\ma}[1]{\mathcal{#1}}
\newcommand{\expect}[3]{\<{#1}|{#2}|{#3}\>}
\newcommand{\ave}[1]{\<{#1}\>}
\newcommand{\sinc}{\mathrm{sinc}}

\newcommand{\pen}[1]{\left(#1\right)}								
\newcommand{\ben}[1]{\left[#1\right]}								
\newcommand{\cen}[1]{\left\{#1\right\}}								

\newcommand{\ignore}[1]{}
\newcommand{\mcal}[1]{\mathcal{#1}}
\newcommand{\mcHS}{\mathcal{H}_S}
\newcommand{\mcHB}{\mathcal{H}_B}
\newcommand{\mcHR}{\mathcal{H}_R}
\newcommand{\mP}{\mathcal{P}}
\newcommand{\mQ}{\mathcal{Q}}
\newcommand{\mL}{\mathcal{L}}
\newcommand{\mG}{\mathcal{G}}
\newcommand{\mK}{\mathcal{K}}
\newcommand{\mJ}{\mathcal{J}}
\newcommand{\mU}{\mathcal{U}}
\newcommand{\tr}{\tilde{\rho}}
\newcommand{\pt}{\partial_t}

\def\a{\alpha}
\def\b{\beta}
\def\g{\gamma}
\def\d{\delta}
\def\e{\epsilon}

\def\h{\eta}
\def\t{\theta}
\def\m{\mu}
\def\n{\nu}
\def\x{\xi}
\def\p{\pi}
\def\r{\rho}
\def\s{\sigma}

\def\ph{\varphi}
\def\c{\chi}

\def\o{\omega}

\def\Ph{\Phi}

\def\O{\Omega}

\def\ox{\otimes}
\def\>{\rangle}
\def\<{\langle}
\def\Tr{\mathrm{Tr}}
\def\Pr{\mathrm{Pr}}
\newcommand{\ketbra}[1]{|{#1}\>\!\<#1|}
\newcommand{\bracket}[1]{\<{#1}|{#1}\>}
\newcommand{\bk}[2]{\<{#1}|{#2}\>}
\newcommand{\ketb}[2]{|{#1}\>\!\<#2|}
\newcommand{\ketbsub}[3]{|{#1}\>_{#3}\<#2|}

\newcommand{\on}[1]{\|#1\|_\infty}

\newcommand{\ee}{\ensuremath{{e}}} 

\def\lp{\left(}
\def\rp{\right)}
\def\ls{\left[}
\def\rs{\right]}
\def\lb{\left\{}
\def\rb{\right\}}

\def\dgr{\dagger}

\usepackage{tikz-cd}

\begin{document}

\title{Lecture Notes on the Theory of Open Quantum Systems}
\author{Daniel A. Lidar}
	\affiliation{Departments of Electrical Engineering, Chemistry, and Physics \& Astronomy\\
	Center for Quantum Information Science \&
		Technology\\ 
		University of Southern California, Los Angeles, California 90089, USA}

\begin{abstract}
This is a self-contained set of lecture notes covering various aspects of the theory of open quantum system, at a level appropriate for a one-semester graduate course. The main emphasis is on completely positive maps and master equations, both Markovian and non-Markovian.
\end{abstract}

\maketitle

\tableofcontents

\newpage

\section{Preface and Acknowledgments}

The theory of open quantum systems is the backbone of nearly all modern research in quantum mechanics and its applications. The reason is simple: the idealization of an isolated quantum system obeying perfectly unitary quantum dynamics is just that: an idealization. In reality every system is \emph{open}, meaning that it is coupled to an external environment. Sometimes these open system effects are small, but they can almost never be neglected. This is particularly relevant in the field of quantum information processing, where the existence of a quantum advantage over classical information processing is often derived first from the idealized, closed system perspective, and must then be re-scrutinized in the realistic, open system setting.\\

These lecture notes provide a fairly comprehensive and self-contained introduction to the theory of open quantum systems. They are based on lectures I gave at the University of Southern California as part of a one-semester graduate course on the topic taught in Fall 2006, Spring 2013, Spring 2017, and Fall 2018. There are several excellent textbooks and monographs either devoted to or containing the same subject, and these notes are in parts heavily influenced by these works, in particular the invaluable books by Heinz-Peter Breuer and Francesco Petruccione~\cite{Breuer:book} and by Robert Alicki and Karl Lendi~\cite{alicki_quantum_2007}. The notes do fill in many details not found in the original sources (at times tediously so!), and also draw on various articles and unpublished materials. I therefore hope that these notes will serve as a useful companion to the textbooks, and will help students and researchers interested in entering the field in a semester of dedicated study.\\

The notes were originally typeset by students serving as scribes during the lectures given in 2013 and 2017, and have undergone extensive editing and additions since then. I am extremely grateful to all the students who participated in this effort:
Chao Cao, Rajit Chatterjea, Yi-Hsiang Chen, Jan 
Florjanczyk, Jose Raul Gonzalez Alonso, Anastasia Gunina, Drew Henry, Kung-Chuan
Hsu, Zhihao Jiang, Joshua Job, Hannes Leipold, Milad Marvian,  Anurag Mishra, Nicolas
Moure Gomez, Siddharth Muthu Krishnan, Shayne Sorenson, Georgios Styliaris, Christopher Sutherland, Subhasish Sutradhar, Walter Unglaub, Ka Wa Yip, and Yicong
Zheng. I am also very grateful to the students in the 2018 course, who offered numerous additional feedback: 
Namit Anand, Mojgan Asadi, Brian Barch, Matthew Kowalsky, Lawrence Liu, Humberto
Munoz Bauza, Adam
Pearson, Bibek Pokharel, Evangelos Vlachos, Aaron Wirthwein, Haimeng Zhang, and Zihan Zhao. Finally, I wish to warmly thank Dr. Tameem Albash and Dr. Jenia Mozgunov, who filled in for me on various occasions, and whose notes I relied on as well. \\

Of course, all errors, typos, and omissions are mine. The reader is strongly encouraged to send me any corrections at 
\href{mailto:lidar@usc.edu}{lidar@usc.edu}.
The notes will be updated regularly to reflect these corrections, as well as new material of interest. I apologize in advance to all the numerous authors whose contributions I did not cite; the field is vast and the intent of these notes is \emph{not} to serve as a comprehensive review article. I have certainly not done justice to the literature.\\

The completion of this work was (partially) supported by the Office of
the Director of National Intelligence (ODNI), Intelligence Advanced
Research Projects Activity (IARPA), via the U.S. Army Research Office
contract W911NF-17-C-0050. 

\newpage



\section{Review of Quantum Mechanics}

The introductory material presented here is based on the approach of the excellent textbook by Nielsen \& Chuang \cite{nielsen2010quantum}. There are four main postulates on which Quantum Mechanics can be built. These four postulates are:
\begin{enumerate}
  \item \emph{Where things happen:} Hilbert space
  \item \emph{Combining Quantum Systems:} Tensor product of vectors, matrices and of Hilbert spaces
  \item \emph{Time Evolution (Dynamics):} Schr\"{o}dinger equation
  \item \emph{Information extraction:}  Measurements
\end{enumerate}

\subsection{Postulate 1}
\begin{quotation}\textit{
``To every Quantum system is associated a \emph{state space}, i.e, a Hilbert space $\mathcal{H}$.''}
\end{quotation}
A Hilbert space is a vector space equipped with an inner product. The vector spaces that we will be working with, $\ma{H}$, can be defined in the following way, in which $\C$ is the field of complex numbers:

\bes
\begin{align}
		\ma{H}&=\C^d \\
		&=\{\vec{v}=\lp\begin{array}{c} v_0\\  v_1 \\ \vdots\\ v_{d-1} \end{array}\rp | v_i\ \epsilon\ \C\}
		\label{eq:v}
\end{align}
\ees

Thus for our purposes a vector space can be defined as the set of $d$-dimensional vectors $\vec{v}$, each element of which, $v_i$, is a complex number. Recall that a vector space  has a couple of properties. First, for all vectors $\vec{v}\in \ma{H}$, $ a\vec{v}+b\vec{v'} \in \ma{H}$, with $a,b\in \C$, i.e., any linear combination of vectors $\vec{z}$ is also an element of the vector space $V$.  Second, the vector space must contain the zero vector, an element that satisfies the condition $\vec{0}+\vec{v}=\vec{v}\ \ \forall\vec{v}\ \in \C$.

The postulate means that physical states of a quantum system can be
associated to a vector $\vec{v}\in\mathcal{H}$. We shall use Dirac notation, in which column vectors are denoted by ``kets": $\vec{v} \mapsto \ket{v}$. In what follows, we
will usually assume that the dimension of $\mathcal{H}$ is finite, and find an orthonormal basis for it.
That is, if ${\dim}(\mathcal{H})=d$, then denote a such a basis
$\left\{\ket{k}\right\}_{k=0}^{d-1}$. A good (but obviously non-unique) choice is the standard basis
\beq
\ket{k}=
\begin{pmatrix}
0\\
\vdots\\
1\\
\vdots\\
0
\end{pmatrix}
\leftarrow\mbox{ $k+1$\textsuperscript{th} position}
\eeq

Any vector in the Hilbert space can be expanded in an orthonormal basis as a linear combination
\beq
\ket{v}=\sum_{k=0}^{d-1} v_k \ket{k},\; v_k\in\C,
\eeq
which quantum physicists often call a \emph{superposition}. The coefficients $v_k$ are called \emph{probability amplitudes}. The reason is that the probability of a quantum system ``being" in a specific state $\ket{k}$ is $\abs{v_k}^2$. This latter statement is part of the postulate. The different orthonormal basis vector $\ket{k}$ represent \emph{mutually exclusive} possibilities, such as the discrete positions of a particle on a line, or different spin configurations.

Of course, in order for the set $\{\abs{v_k}^2\}$ to be a proper probability distribution, the probabilities must sum to one. This is the reason that we need to endow the vector space $\ma{H}$ with an inner product, i.e., work with Hilbert spaces. To define the inner product function we first introduce the dual of a ket, called a ``bra". In Dirac notation, row vectors (or \emph{bras}) are written as $\bra{v}$, where by definition $\bra{v} = \ket{v}^\dagger$, where the dagger denotes Hermitian conjugation, i.e., transpose and complex conjugation. Thus if $\ket{v}$ is written as in Eq.~\eqref{eq:v} then $\bra{v} = \{v^*_1,v^*_2,\dots,v^*_{n}\}$. One reason that Dirac notation is
convenient because we can represent the inner product as a ``braket",
i.e.,
\beq
\braket{v|w} \equiv \begin{pmatrix}v_0^*, \hdots, v_{d-1}^*\end{pmatrix}
\begin{pmatrix}w_0\\\vdots\\w_{d-1}\end{pmatrix}=\sum_{k=0}^{d-1} v_k^*w_k.
\eeq
The normalization condition of the probability distribution can now be written as
\beq
1 = \sum_{k=0}^{d-1}|v_k|^2 = \sum_{k=0}^{d-1}v^*_k v_k = 
\braket{v|v} = \norm{\ket{v}}^2 ,
\eeq
which is to say that every vector $\ket{v}\in\mathcal{H}$ is  \emph{normalized}, i.e., $\norm{\ket{v}} = \sqrt{\braket{v|v}}=1$. 
Note that an overall phase does not affect normalization, i.e., $\ket{v}$ and $e^{i\theta}\ket{v}$ have the same norm. In fact we do not distinguish between states that differ only by an overall phase. Such states form a ``ray" in Hilbert space.

Using Dirac notation we can form the \emph{outer product} of two vectors in the same Hilbert space as follows
\beq
\ket{v}\bra{w} =
\begin{pmatrix}
v_0\\ \vdots \\ v_{d-1}
\end{pmatrix}
\begin{pmatrix}
w_0^*, \hdots, w_{d-1}^*
\end{pmatrix} =
\begin{pmatrix}
v_0w_0^* & \hdots & v_0w_{d-1}^*\\
\vdots & \ddots & \vdots\\
v_{d-1}w_0^* & \hdots & v_{d-1}w_{d-1}^*
\end{pmatrix}.
\eeq

Additional linear algebra and Dirac notation facts are collected in Appendix~\ref{app:A}.

\subsection{Postulate 2}
\begin{quotation}\textit{
``Given two quantum systems with respective Hilbert spaces $\mathcal{H}_1$ and $\mathcal{H}_2$ the
combined quantum system has associated with it a Hilbert space given by $\mathcal{H} = \mathcal{H}_1\otimes
\mathcal{H}_2$.''}
\end{quotation}
Let us define $\mathcal{H}_1$ to be the span of $\{\ket{{{v}_{i}}}\}_{i=0}^{{{d}_{1}}-1}$, and similarly $\mathcal{H}_2$ to be the span of $\{\ket{{{w}_{j}}}\}_{j=0}^{{{d}_{2}}-1}$. Then we have $\mathcal{H}$ defined as the span of
$\{\ket{{{v}_{i}}}\otimes \ket{{{w}_{j}}}\}_{i=0,j=0}^{{{d}_{1}}-1,{{d}_{2}}-1}$.
For two states $\ket{\psi}\in\mathcal{H}_1$, $\ket{\ph}\in\mathcal{H}_2$, the tensor product is given by
\beq
\ket{\psi}\otimes\ket{\ph} =
\begin{pmatrix}
\psi_0\\
\vdots\\
\psi_{d_1-1}
\end{pmatrix}
\otimes
\begin{pmatrix}
\ph_0\\
\vdots\\
\ph_{d_2-1}
\end{pmatrix}
=
\begin{pmatrix}
\psi_0\ph_0\\
\vdots\\
\psi_0\ph_{d_2-1}\\
\vdots\\
\psi_{d_1-1}\ph_{d_2-1}
\end{pmatrix}
\eeq
Note that the underlying Hilbert spaces could represent entirely different physical systems, e.g., the first could be the space of electron spins, whereas the second could be the space of photon polarizations.
We can also define the tensor product between matrices, i.e., if
\beq
{A} =
\begin{pmatrix}
a_{11} & \hdots & a_{1n}\\
\vdots & \ddots & \vdots\\
a_{m1} & \hdots & a_{mn}
\end{pmatrix}\;\ \ \ \
{B} =
\begin{pmatrix}
b_{11} & \hdots & b_{1q}\\
\vdots & \ddots & \vdots\\
b_{p1} & \hdots & b_{pq}
\end{pmatrix}
\eeq
then
\beq
{A}\otimes{B} =
\begin{pmatrix}
a_{11}b_{11} & \hdots & a_{1n}b_{1q}\\
\vdots & \ddots & \vdots\\
a_{m1}b_{p1} & \hdots & a_{mn}b_{pq}
\end{pmatrix},
\eeq
i.e., a matrix of dimension $mp\times nq$.

For example, let $\ma{H}_1 = \ma{H}_2 = \C^2$ and $\ket{\Psi }=\frac{1}{\sqrt{2}}(({{\left| 0 \right\rangle }_{1}}\otimes {{\left| 0 \right\rangle }_{2}})+({{\left| 1 \right\rangle }_{1}}\otimes {{\left| 1 \right\rangle }_{2}}))  = \frac{1}{\sqrt{2}}(1,0,0,1)^t\in \ma{H} =  \mathcal{H}_1\otimes\mathcal{H}_2$, where $\ket{0} = \left( \begin{matrix}
   1  \\
   0  \\
\end{matrix} \right)$ and  $
\ket{1} = \left( \begin{matrix}
   0  \\
   1  \\
\end{matrix} \right)$. This example is interesting and important since it represents an entangled state, i.e., a state which cannot be written as a tensor product in the same basis.

\subsection{Postulate 3}
\begin{quotation}\textit{
``$\exists$  a unitary operator $\opU(t)$ such that the time evolution of a state is given by
\beq
\ket{\psi(t)} = \opU(t)\ket{\psi(0)}.
\label{eq:U}
\eeq
Equivalently, the state vector of the system satisfies the
Schr\"{odinger} equation
\beq
\ket{\dot{\psi}(t)} = -\frac{\ii}{\hbar} \opH \ket{\psi(t)}
\label{eq:SE}
\eeq
 with $\opH$ being a Hermitian operator known as the
Hamiltonian.''}
\end{quotation}
The dot denotes $\partial/\partial t$, and in this course we will set $\hbar=1$, which means that the units of energy and frequency will be the
same. 

We shall show below that the equivalence holds provided $\opU(t) = \exp\left(-i\opH t\right)$ when $H$ is $t$-independent. In the time dependent case, the situation is more complicated, and we have 
\beq
\opU(t)={{T}_{+}}{{e}^{-i\int_{0}^{T}{H(t')dt'}}}
\eeq
where ${{T}_{+}}$ represents Dyson time ordering. This will be discussed later.

To prove the equivalence let us recall a bit of mathematical background. 
An operator $A$ is normal if $A^\dgr A = A A^\dgr$, it is Hermitian if $A^\dagger = A$, and it is unitary if $A^\dagger A = I$. Clearly, unitary operators and Hermitian operators are also normal.

\begin{mytheorem}[Spectral Theorem] A linear operator $A:V \rightarrow V$ obeys $A^\dgr A = A A^\dgr$ (i.e., it is a normal operator) if and only if $A=\sum_{a} \lambda_a \ket{a}\bra{a}$ for a set of orthonormal basis vectors $\{\ket{a}\}$ for $V$, which are also the eigenvectors of $A$ with respective eigenvalues $\{\lambda_a\}$. 
\end{mytheorem} 

Using this we can characterize the eigenvalues of Hermitian and unitary operators:
\begin{enumerate}
\item {Hermitian operators:} Applying the spectral theorem we get $\sum_{a} \lambda_a \ket{a}\bra{a} = \sum_{a} \lambda_a^* \ket{a}\bra{a}$, so that $\lambda_a=\lambda_a^*$. Thus the eigenvalues are real in this case.

\item {Unitary operators: } Applying the spectral theorem we get 
\beq
I = A^\dgr A=I=\left(\sum_a \lambda_a^* \ket{a}\bra{a}\right)\left(\sum_{a'} \lambda_{a'} \ket{a'}\bra{a'}\right) =\sum_a \lambda_a^* \lambda_a \ket{a}\bra{a} =\sum_a \abs{\lambda_a}^2 \ket{a}\bra{a}=\sum_a \ket{a}\bra{a},
\eeq 
where the last equality is the spectral representation of the identity operator $I$ (all its eigenvalues are $1$). Thus the eigenvalues of a unitary operator are all phases: $\lambda_a = e^{i\theta_a}$ where $\theta_a \in \mathbb{R}$.
\end{enumerate}

We now define functions of normal operators. If we have a function $f:\C\rightarrow\C$, then we can extend it to the case of
normal operators by defining
\beq
f(A)\equiv \sum_a f(\lambda_a) \ket{a}\bra{a} .
\label{eq:f-def}
\eeq
Note that the function operates only on the eigenvalues. 

Let us now prove the equivalence of the two evolution laws.
One direction is straightforward, namely assuming Eq.~\eqref{eq:U} we easily derive Eq.~\eqref{eq:SE}:
\beq
\frac{d}{dt}\left| \psi (t) \right\rangle =\frac{d}{dt}({{e}^{-iHt}}\ket{\psi (0)})=-iH{{e}^{-iHt}}\ket{\psi (0)}=-iHU(t)\left| \psi (0) \right\rangle =-iH\left| \psi (t) \right\rangle 
\eeq
Note that bringing the term involving the Hamiltonian down from the exponent is justified even for operators, as is easily verified using the spectral theorem (since $H$ is normal), or directly by differentiating the Taylor expansion of the matrix exponential (which applies even if $A$ is not normal):
\beq
\frac{d}{dt}({{e}^{At}})=\frac{d}{dt}(I+At+\frac{{{A}^{2}}{{t}^{2}}}{2!}+\frac{{{A}^{3}}{{t}^{3}}}{3!}+...)=A+\frac{{{A}^{2}}}{2!}(2t)+\frac{{{A}^{3}}}{3!}(3{{t}^{2}})+...=A(I+At+
\frac{{{A}^{2}}{{t}^{2}}}{2!}+...)=A{{e}^{At}}
\eeq
Now for the other direction, we start with writing the spectral decomposition of the Hamiltonian as $H=\sum_{a}{{{\lambda }_{a}}\ket{a} }\left\langle  a \right|$, and note also that from the definition \eqref{eq:f-def} we have:
\beq
U(t)={{e}^{-iHt}}\Rightarrow U(t)=\sum_{a}{{{e}^{-i{{\lambda }_{a}}t}}}\ket{a} \left\langle  a \right|
\eeq
Now, since the eigenvectors of $H$ are a basis (again, from the spectral theorem), we can decompose $\ket{\psi(t)}$ in this basis and write $\left| \psi (t) \right\rangle =\sum_{a}{{{\psi }_{a}}(t)\ket{a} }$, so that the left hand side of Eq.~\eqref{eq:SE} becomes
\beq
\frac{d}{dt}\left| \psi (t) \right\rangle =\sum_{a}{\frac{d}{dt}{{\psi }_{a}}(t)\ket{a} }.
\eeq
As for the right hand side of Eq.~\eqref{eq:SE},
\beq
-iH\left| \psi (t) \right\rangle =-i\sum_{a}{{{\lambda }_{a}}\ket{a} }\left\langle  a \right|\psi (t)\rangle=-i\sum_{aa'}{{{\lambda }_{a}}{{\psi }_{a'}}(t)\ket{a} }\underbrace{\left\langle a|a' \right\rangle }_{{{\delta }_{aa'}}}=-i\sum_{a}{{{\lambda }_{a}}{{\psi }_{a}}(t)\ket{a} }.
\eeq
For Eq.~\eqref{eq:SE} to hold, these two need to be equal term by term (from orthonormality of the basis), so that we find
\beq
\frac{d}{dt}{{\psi }_{a}}(t)=-i{{\lambda }_{a}}{{\psi }_{a}}(t)\Rightarrow {{\psi }_{a}}(t)={{e}^{-i{{\lambda }_{a}}t}}{{\psi }_{a}}(0)
\eeq
Plugging this result into $\left| \psi (t) \right\rangle =\sum_{a}{{{\psi }_{a}}(t)\ket{a} }$ and using orthonormality once more we now have:
\beq
\ket{\psi (t)}=\sum_{a}{{{e}^{-i{{\lambda }_{a}}t}}{{\psi }_{a}}(0)\ket{a} }=(\sum_{a}{{{e}^{-i{{\lambda }_{a}}t}}\ket{a} }\left\langle  a \right|)(\sum_{a'}{{{\psi }_{a'}}(0)\left| a' \right\rangle )={{e}^{-iHt}}\ket{\psi (0)}=U(t)\ket{\psi (0)}}.
\eeq
This completes the proof.

\subsection{Postulate 4}
This is the most controversial postulate, but we will not discuss those issues here and simply assume its validity.

This postulate has two parts: measuring states and measuring operators.
\begin{enumerate}
\item \emph{Measuring States:} Quantum measurements are described by a set $\left\{{M}_k\right\}_{k=1}^{N}$ of \emph{measurement operators} satisfying the constraint $\sum_k {M}_k^\dagger{M}_k = {I}$.

Given a state $\ket{\psi}\in\mathcal{H}$, instantaneously after the measurement it becomes,
\beq
\ket{\psi} \mapsto \frac{M_k\ket{\psi}}{\sqrt{p_k}} \equiv \ket{\psi_k},
\label{eq:22}
\eeq
with probability
\beq
p_k = \braket{\psi|{M}_k^\dagger{M}_k|\psi} = \|M_k \ket{\psi}\|^2 \geq 0.
\label{eq:23}
\eeq
The measurement outcome is the index $k$ of the state that resulted. The constraint listed in the postulate has the following origin. Notice that $\sum_k p_k = 1$ must be true since $p_k$ is a probability, which implies $\bra{\psi} \sum_k {M}_k^\dagger{M}_k \ket{\psi} = 1$. Since this is true for arbitrary $\ket{\psi}$  the sum rule follows
\beq
\sum_k {M}_k^\dagger{M}_k = {I}.
\eeq

\item To every physically measurable quantity is associated an \textnormal{observable}, i.e., a Hermitian operator $A$.
$A$ has a spectral decomposition (since it is Hermitian and hence normal),
\beq
A = \sum_{a} \lambda_a \ket{a}\bra{a},
\eeq
with $\lambda_a \in \mathbb{R}$ since $A$ is Hermitian. The $\lambda_a$'s, the eigenvalues, are the outcomes of the measurement (hence need to be real).\footnote{It is interesting to ask why physically measurable quantities should be associated with Hermitian operators.
Intuitively, since physical measurements produce real numbers we want to associate an operator with a real spectrum
as a physically observable quantity. Moreover, we would like states with different eigenvalues (or different results from the
measurement) to be orthogonal. A Hermitian operator satisfies both of these requirements. However, these justifications admittedly leave something to be desired. For more details see \url{https://physics.stackexchange.com/questions/39602/why-do-we-use-hermitian-operators-in-qm}.} 

The set of eigenvectors, $\{\ket{a}\}$ are an orthonormal set. Hence, $\{ P_a \equiv \ket{a}\bra{a} \}$ are projectors (defined below). These are the measurement operators corresponding to the measurement of this observable. Hence, if the system is in state $\ket{\psi}$ before the observable $A$ is measured, according to Eq.~\eqref{eq:23} the probability of outcome $\lambda_a$ is given by $p_a = \braket{\psi | P^\dagger_a P_a | \psi} =  \braket{\psi | P^2_a | \psi} = \braket{\psi | P_a | \psi} = |\bra{\psi}a\rangle|^2$. Moreover, according to Eq.~\eqref{eq:22} the state after the measurement is performed and outcome $\lambda_a$ is observed, becomes $\ket{\psi_a} = \frac{P_a\ket{\psi}}{\sqrt{p_a}} = \frac{\bra{a}\psi\rangle}{|\bra{a}\psi\rangle|}\ket{a} = e^{i\theta}\ket{a}$, where $e^{i\theta}$ is the phase associated with the complex number $\bra{a}\psi\rangle$.

\end{enumerate}

We next consider several important special cases of the \emph{generalized measurements} defined above.

\subsubsection{Projective (von Neumann) measurements}
\label{sec:PvNmeas}
Projective measurements are a special case of generalized measurements, in which the measurement operators, $M_k$ are Hermitian operators called \emph{projectors}. That is, $M_k=P_k$, where $P_k P_l = \delta_{k,l}P_k$ and $P_k^{\dagger}=P_k$. In particular, $P_k^2=P_k$. Using this, we can see that the probability of outcome $k$, $p_k = \braket{\psi | M_k ^{\dagger} M_k | \psi } = \braket{\psi | P_k | \psi }$.

\emph{\underline{Example:}} Let $\ket{\psi} = a\ket{0} + b\ket{1}$ where $\ket{0} = \begin{pmatrix} 1\\ 0 \end{pmatrix}$ , $\ket{1} = \begin{pmatrix} 0\\ 1 \end{pmatrix}$ and $a,b \in \C$. That is, $\ket{\psi} \in \C^2$ and $\{ \ket{0}, \ket{1} \}$ is the standard basis for the space. Such a $\ket{\psi}$ is called a \emph{qubit}.

Now, we define measurement operators, $M_0 = P_0 = \ketb{0}{0}$ and $M_1 = P_1 = \ketb{1}{1}$. We can see that $P_{0,1}^2 = P_{0,1}$ and $P_0 P_1 = 0$. Hence, this is a set of projective measurements. Thus, the probabilities of outcomes are,
\bes
\begin{eqnarray}
&p_0 = \braket{\psi|P_0|\psi}=\braket{\psi|0}\braket{0|\psi} = \abs{a}^2, \\
&p_1 = \braket{\psi|P_1|\psi}=\braket{\psi|1}\braket{1|\psi} = \abs{b}^2 .
\end{eqnarray}
\ees
This shows that the absolute value squared of the amplitudes of a wavefunction when expanding it in an orthonormal basis provide the probabilities of observing the outcomes corresponding to those basis states. This is sometimes called the Born rule in quantum mechanics.

Also, using Postulate 4, we can see that the state transformation in the above measurement would be:
\beq
\label{meastrans}
\ket{\psi} \mapsto
\begin{cases}
\frac{P_0 \ket{\psi}}{\abs{a}} \text{  with probability  } p_0=\abs{a}^2, \\
\frac{P_1 \ket{\psi}}{\abs{b}} \text{  with probability  } p_1=\abs{b}^2.
\end{cases}
\eeq
We can easily see that $P_0 \ket{\psi} = \ketb{0}{0}(a\ket{0} + b\ket{1})=a\braket{0|0}\ket{0}+b\braket{0|1}\ket{0}=a\ket{0}$, where in the last step we have used that $\{\ket{0},\ket{1}\}$ is an orthonormal set. Similarly, $P_1 \ket{\psi}=b\ket{1}$.
Hence the transformation \eqref{meastrans} becomes,
\beq
\label{phases}
\ket{\psi} \mapsto
\begin{cases}
\frac{a}{\abs{a}}\ket{0} = e^{\ii \theta_a}\ket{0}\text{  with probability  } p_0=\abs{a}^2, \\
\frac{b}{\abs{b}}\ket{1} = e^{\ii \theta_b}\ket{1}\text{  with probability  } p_1=\abs{b}^2.
\end{cases},
\eeq
where $\theta_a$ and $\theta_b$ are the arguments of the complex numbers $a$ and $b$ respectively. We can see that the phase factors $e^{\ii \theta_{a,b}}$ are completely arbitrary since they have no influence on the probabilities of the measurement outcomes.

Thus, quantum states are equivalent up to a global phase factor. Because of this, quantum states are \emph{\textit{rays}} in a Hilbert space, since they are not just one vector but an equivalence class of vectors: equivalent up to a global phase.


\subsubsection{Examples of measuring observables}

To  illustrate the concept of observables, let's consider a few examples.

\begin{itemize}

\item \emph{Pauli matrices:} The Pauli matrices and their properties are reviewed in Appendix~\ref{app:Pauli}. 
Consider, e.g., measuring the Pauli matrix $Z $ on a qubit $\ket{\psi} = a \ket{0} + b\ket{1}$. Writing the spectral decomposition of $Z$,
\beq
Z = (+1) \ketb{0}{0} + (-1)\ketb{1}{1},
\eeq
we can see that the set of measurement operators for this observable is  $\{ P_0 = \ketb{0}{0}, P_1 = \ketb{1}{1} \}$, with
outcomes as the corresponding eigenvalues $\{ \lambda_0 = +1, \lambda_1 =-1 \}$. Thus, we obtain $\lambda_0$ with
probability $p_0 = \braket{\psi|P_0|\psi}=\abs{a}^2$, and obtain $\lambda_1$ with probability $p_1 = \braket{\psi|P_0|\psi}=\abs{a}^2$.

Hence, the action of measuring $Z$ takes $\ket{\psi}$ to $e^{\ii \theta_a}\ket{0}$ if the outcome was $\lambda_0$; and
to $e^{\ii \theta_b} \ket{1}$  if the outcome was $\lambda_b$.
\item \emph{Measuring Energy}: When measuring energy, the observable we use is simply the Hamiltonian $H$ of  the system. Since $H$ is Hermitian it has a spectral decomposition,
We can write $H$ as,
\beq
H = \sum_{a} E_a \ket{a}\bra{a} ,
\eeq
where $E_a$ denotes the energy and $\ket{a}$ the associated energy eigenstate. So, in our experiment, we measure an energy of $E_a$ with probability $p_a = \braket{\psi | P_a | \psi}$, where $P_a = \ket{a}\bra{a}$. 
The post-measurement state is $\ket{\psi_a} = \frac{P_a \ket{\psi}}{\sqrt{p_a}} = \frac{\braket{a|\psi}}{\sqrt{p_a}}\ket{a}$, i.e., the new state is an eigenstate of the Hamiltonian, also sometimes called an energy eigenstate.

Consider the following single-qubit Hamiltonian:
\beq
H = \o_x \s_x + \o_z \s_z .
\eeq
What happens when we measure it in the state $\ket{\psi}$? To find out we need $H$'s spectral decomposition, i.e., we need to diagonalize $H$. The eigenvalues are easily found to be $E_\pm = \pm \sqrt{\o_x^2+\o_z^2}$, so that $H$ can be written in diagonal form as
\beq
H = E_-\ketbra{E_-} + E_+\ketbra{E_+} ,
\eeq
where $\ket{E_\pm}$ are the corresponding eigenvectors. When $H$ is measured, we find $E_\pm$ with probability $p_\pm = \<\psi\ketbra{E_\pm}\psi\> = |\bra{\psi}{E_\pm}\>|^2$. For example, if the system is prepared in the ground state $\ket{E_-}$ (the state with the lower energy), then $p_- = 1$ but $p_+ = 0$. Or, if the system is prepared in a uniform superposition of the ground state $\ket{E_-}$ and the excited state $\ket{E_+}$, i.e., $\ket{\psi} = 1/\sqrt{2}(\ket{E_-}+\ket{E_+})$, then $p_- =p_+ = 1/2$.

\end{itemize}

\subsubsection{Expectation value of an observable}

Given an observable $A = \sum_{a} \lambda_a \ket{a}\bra{a}$, since we obtain $\lambda_a$ with probability $p_a$, we can
naturally define an expectation value of this observable in the state $\ket{\psi}$ as
\bes
\label{eq:exp-val-pure}
\begin{align}
\langle A \rangle_\psi \equiv \sum_{a} \lambda_a p_a &= \sum_{a} \lambda_a \braket{\psi|P_a|\psi} \\
&=\braket{\psi| \left( \sum_{a} \lambda_a P_a \right) | \psi} \\
&=\braket{\psi| A | \psi} \\
&= \Tr (A \ket{\psi}\bra{\psi}).
\label{eq:39d}
\end{align}
\ees
The last equality can be proved as follows: first create an orthonormal basis for the Hilbert space with $\ket{\psi}$ as
one of the elements of the basis (say by using the Gram-Schmidt procedure \cite{nielsen2010quantum}[p.66]. That is,
\beq
\mathcal{H} = \text{Span} \{ \ket{\psi} = \ket{\phi_0}, \ket{\phi_1}, \ket{\phi_2}, ... \ket{\phi_{d-1}} \},
\eeq
where $d$ is the dimension of the Hilbert space and all vectors in the basis are orthonormal.
Now,
\bes
\begin{align}
\Tr( A \ket{\psi}\bra{\psi} ) &= \sum_{i=0}^{d-1} \braket{\phi_i |A| \psi }\braket{\psi|\phi_i} \\
\label{41b}
&=\sum_{i=0}^{d-1} \braket{\phi_i |A| \phi_0 }\braket{\phi_0|\phi_i} \\
&=\braket{\phi_0 |A| \phi_0} \\
&=\braket{\psi|A|\psi},
\end{align}
\ees
where in Eq.~\eqref{41b} we used the fact that $\ket{\psi}=\ket{\phi_0}$ is the first element in our basis. 

Likewise we can define the standard deviation as follows:
\beq
\Delta A=\sqrt{{{\left\langle {{(A-{{\left\langle A \right\rangle }_{\psi }})}^{2}} \right\rangle }_{\psi }}}
\eeq

\subsubsection{Heisenberg Uncertainty Principle}
In quantum mechanics, an important property of pairs of non-commuting observables is that they cannot be measured with
arbitrary precision simultaneously. What this means is that if we measure the Hermitian operators $C$ and $D$ on $\ket{\psi}$
then they obey the Heisenberg Uncertainty Principle:
\begin{align}
\lp \Delta C\rp \lp \Delta D\rp \ge \frac{1}{2}\abs{\braket{\psi|[C,D]|\psi}}.
\end{align}
Let us now show this. Define the Hermitian operators $A=C-\< C \>$, and $B=D-\< D \>$. We can always decompose the expectation value as a complex number:
\begin{align}
\braket{\psi|AB|\psi} = x + \ii y, \; x,y\in\R.
\end{align}
Note that $2AB=\{ A,B \} + [A,B]$, where $\{ A,B \}$ is Hermitian (purely real eigenvalues) and $[A,B]$ is anti-Hermitian (purely imaginary eigenvalues). Therefore 
\bes
\begin{align}
\braket{\psi|\{ A,B \}|\psi}^* &= \braket{\psi|\{ A,B \}^\dagger |\psi} = \braket{\psi|\{ A,B \}|\psi}\\
\braket{\psi|[ A,B ]|\psi}^* &= \braket{\psi|[ A,B ]^\dagger |\psi} = -\braket{\psi|[ A,B ]|\psi} ,
\end{align}
\ees
which means that $\braket{\psi|\{ A,B \}|\psi}$ is real while $\braket{\psi|[ A,B ]|\psi}$ is purely imaginary.
Hence the following must be true:
\bes
\begin{align}
\braket{\psi|\{ A,B \}|\psi} &= 2x\\
\braket{\psi|[ A,B ]|\psi} &= 2\ii y.
\end{align}
\ees
Therefore, by using the Cauchy-Schwarz inequality (see Appendix~\ref{app:A}) in the third line:
\bes
\begin{align}
4x^2 + 4y^2 &= \abs{\braket{\psi | \lbrace A,B\rbrace | \psi}}^2 + \abs{\braket{\psi | [A,B] | \psi}}^2\\
&= 4\abs{\braket{\psi |AB| \psi}}^2 \\
&\le 4 \braket{\psi| A^\dagger A |\psi}\braket{\psi| B^\dagger B |\psi}\\
&= 4 \braket{\psi| A^2 |\psi}\braket{\psi| B^2 |\psi}\\
&= 4 \braket{\psi| \lp C -\< C \> \rp^2 |\psi}\braket{\psi| \lp D -\< D \> \rp^2 |\psi} .
\end{align}
\ees
Obviously $\abs{\braket{\psi | \lbrace A,B\rbrace | \psi}}^2 \geq 0$, and hence:
\bes
\begin{align}
4 \braket{\psi| \lp C -\< C \> \rp^2 |\psi}\braket{\psi| \lp D -\< D \> \rp^2 |\psi} &\ge \abs{\braket{\psi | [A,B] | \psi}}^2\\
&= \abs{\braket{\psi | [C,D] | \psi}}^2
\end{align}
\ees
from which the Heisenberg uncertainty principle now follows.

\subsubsection{Positive Operator Valued Measures (POVMs)}
Given a generalized measurement with measurement operators $\{M_k\}$ we define the elements of a POVM via 
\beq
E_k=M_k^\dagger M_k.
\eeq
The normalization condition then becomes $\sum_k E_k = I$. Clearly, $E_k^\dagger=M_k^\dagger\left(M_k^\dagger\right)^\dagger=E_k$, so that the POVM elements are Hermitian. It is easy to show that the
$E_k$s are moreover \emph{positive} operators, i.e., that $\braket{\psi|E_k|\psi}\geq 0$
is true for every $\ket{\psi}$  (for more details on positive operators see Appendix~\ref{app:pos-ops}). Indeed, $\braket{\psi|E_k|\psi}=\braket{\psi|M_k^\dagger M_k|\psi}=\norm{M_k\ket{\psi}}^2\ge 0$. Note that the probability of outcome $k$ is simply $p_k = \braket{\psi|E_k|\psi}$. How about the effect of the measurement $E_k$ on a state $\ket{\psi}$? Suppose we are given an arbitrary set of positive operators
$\left\{E_k\right\}$ that satisfy $\sum_k  E_k = I$. How do we extend the measurement postulate in this case? The answer to
this is to use the so called \emph{polar decomposition} of the operator. It is true that for any operator $A$,
we can always find a unitary $U$ and a positive operator $P$ such that $A=UP$ with $P=\sqrt{A^\dagger A}$. If the operator
$A$ is invertible, then such a decomposition is unique and $U=AP^{-1}$. In our case, we could use the given POVMs and
define for every $k$
\begin{align}
M_k = U_k \sqrt{E_k},
\end{align}
where the $U_k$'s are just arbitrary unitaries. In other words, since only the $E_k$ are specified (by assumption), we are free to choose the $U_k$'s, and for every such choice we get a different set of $M_k$'s. Hence, we can now write the state after
the measurement as
\begin{align}
\ket{\psi}\mapsto\ket{\psi_k} = \frac{U_k\sqrt{E_k}\ket{\psi} }{\sqrt{p_k}}\text{ with probability }p_k=\braket{\psi | E_k|\psi}.
\end{align}
Since $U_k$'s are arbitrary (again, since only the $E_k$'s were specified), this unitary freedom is a generalization of the freedom to leave the overall phase of a state unspecified.

To see why POVMs are relevant let's consider the following example. Suppose we have to play a game. Alice always gives us
one of these two states:
\bes
\begin{align}
&\ket{\psi_1} = \ket{0}, \text{  or } \\
&\ket{\psi_2} = \frac{1}{\sqrt{2}} (\ket{0} + \ket{1}) \equiv \ket{+}.
\end{align}
\ees
We do not, \textit{a priori}, know which state has been handed to us. We do know that it is one of these two states. Our task is to
perform measurements and decide which of the two states we were given.  Also, we are not allowed to make an error in
identification, i.e., if we provide an answer, it has to be right. However, we are allowed to proclaim ignorance if we don't know
the answer. Moreover, we must treat both states equally, i.e., we cannot preferentially identify only one of the states and proclaim ignorance on the other. What is our strategy? Since these are non-orthogonal states, there is no way distinguish these two states with
complete certainty \cite{nielsen2010quantum}[Box 2.3, p.87]. Suppose we try to do it with projective measurements. Let's
take the measurement set to be $\{M_k\}=\{P_0,P_1\}$, where $P_i = \ket{i}\bra{i}, i=0,1$. Suppose that the outcome is the index $0$. This can happen in either of two ways: Alice prepared $\ket{\psi_1}$ or she prepared $\ket{\psi_2}$. The probability that the
outcome is $0$ given that she prepared $\ket{\psi_1}$ is $p(0|\psi_1)=\braket{\psi_1 |P_0|\psi_1}=\braket{0|0}\braket{0|0}=1$. And, the probability that the
outcome is $0$ given that she prepared $\ket{\psi_2}$ is $p(0|\psi_2)=\braket{\psi_2 |P_0|\psi_2}=1/2$. This means that if the outcome is $0$ then we cannot know for sure which of the two states Alice prepared, since both occur with non-vanishing probability. Therefore in this case we must proclaim ignorance. However, note that it also follows that  $p(1|\psi_1)=0$ and $p(1|\psi_2)=1/2$  which means that given outcome $1$ we know with certainty that Alice prepared $\ket{\psi_2}$. Thus we cannot satisfy the condition of treating the two states equally. As is easily checked, this will always be the case with a projective measurement. 

Now, let's try with an intelligent choice of POVMs. Define,
\bes
\begin{align}
&E_1 = \alpha \ketb{1}{1}, \\
&E_2 = \alpha \ket{-}\bra{-},\\
&E_3 = I - E_1 -E_2.
\end{align}
\ees
where, $\ket{-} \equiv \frac{1}{\sqrt{2}} (\ket{0} - \ket{1})$ and $\alpha>0$ is an arbitrary parameter which we can optimize
later, and which must be chosen so that $E_3 > 0$. If we do so then this clearly is a set of POVMs, since $\sum_{k} E_k = I$, and for suitable $\alpha$, all the $E_k$'s are positive. Let us now compute
the probabilities of the $3$ possible outcomes,
\bes
\begin{align}
&p(1|\psi_1)=\braket{\psi_1|E_1|\psi_1}=0, \\
&p(1|\psi_2)=\braket{\psi_2|E_1|\psi_2}=\frac{\alpha}{2}, \\
&p(2|\psi_1)=\braket{\psi_1|E_2|\psi_1}=\frac{\alpha}{2}, \\
&p(2|\psi_2)=\braket{\psi_2|E_2|\psi_2}=0, \\
&p(3|\psi_1)=\braket{\psi_1|E_3|\psi_1}=1-\frac{\alpha}{2},  \\
&p(3|\psi_2)=\braket{\psi_2|E_3|\psi_2}=1-\frac{\alpha}{2}.
\end{align}
\ees
So, if we get outcome $1$, we can say with certainty that the given state was $\ket{\psi_2}$ and if we get outcome $2$, we can
say with certainty that the given state was $\ket{\psi_1}$. With outcome $3$, we have no information about the state, i.e., we must proclaim ignorance. But in
two of the three outcomes we have been able to obtain an answer with certainty. So, in order to make the probability of
outcome $3$ as small as possible (since it yields no information), we have to increase $\alpha$ as much as possible while
keeping $E_3$ positive. If we write out $E_3$ as a matrix and place the constraint
of the eigenvalues of this matrix being positive, it easy to show that the maximal allowed value of $\alpha$ is $\frac{\sqrt{2}}{1+\sqrt{2}}$.

\section{Density Operators}
\label{sec:densityops}
We will motivate the study of density operators by considering ensembles of pure quantum states. Suppose, instead of having a single state vector, we only know that our system is in state $\ket{\psi_1}$ with probability $q_1$, or in state $\ket{\psi_2}$ with probability $q_2$, and so on. In other words, we have an \emph{pure state ensemble} $\{ q_i, \ket{\psi_i} \}_{i=1}^{N}$ describing our system.

Now, we would like to understand what happens when we make measurements on this quantum system. Suppose the state were $\ket{\psi_i}$ and we measure with a set of measurement operators  $\{ M_k\}$. The measurement transformation would be:
\beq
\ket{\psi_i} \mapsto \frac{M_k \ket{\psi_i} }{\sqrt{{p}_{k|i}}} = \ket{\psi_i^k}
\label{eq:44}
\eeq
with probability ${p}_{k|i}=\braket{\psi_i| M_k^{\dagger} M_k | \psi_i}$, which is the probability of outcome $k$, given a state $\ket{\psi_i}$.

Now, consider that we did not know what the state was but only that it came from the ensemble $\{ q_i, \ket{\psi_i} \}_{i=1}^{N}$. Then the probability of obtaining the outcome $k$ as a result of the measurement on the ensemble is:
\bes
\begin{align}
{p}_k &= \sum_{i} {p}_{k|i} q_i \\
&= \sum_{i} q_i \braket{\psi_i |M_k^{\dagger} M_k | \psi_i} \\
&= \Tr \left[ M_k^{\dagger} M_k \left( \sum_{i} q_i \ket{\psi_i}\bra{\psi_i} \right) \right]
\label{densmat}.
\end{align}
\ees

In Eq.~\eqref{densmat} we define the operator within the parentheses as,
\beq
\rho=\sum_{i} q_i \ket{\psi_i}\bra{\psi_i}.
\label{eq:rho}
\eeq
This is called the \emph{density matrix} or \emph{density operator} and is a central object in quantum mechanics. The density matrix is completely equivalent to the pure state ensemble $\{ q_i, \ket{\psi_i} \}_{i=1}^{N}$, but it has the advantage of being directly useful for calculations. Indeed,  using the density matrix, Eq. \eqref{densmat} becomes:
\beq
{p}_k = \Tr (E_k \rho),
\label{eq:p_k}
\eeq
where we have defined $E_k \equiv M_k^{\dagger} M_k$ as the element of a POVM.

What about the state that results after measurement result $k$ has been observed? Suppose that outcome $k$ is observed for a known initial state $\rho=\sum_{i} q_i \ket{\psi_i}\bra{\psi_i}$.  If we let $\ket{\psi_i^k} := \frac{M_k \ket{\psi_i}}{\sqrt{p_{k|i}}}$ [as in Eq.~\eqref{eq:44}], then $\{ p_{i|k}, \ket{\psi_i^k} \}_i$ is the resulting ensemble, where $p_{k|i}$ is the probability of outcome $k$ given state $\ket{\psi_i}$. On the other hand, if outcome $k$ was observed, and we {\em don't} know the initial state, then we should sum over all possible states compatible with outcome $k$ (the states $\ket{\psi_i^k}$) with their respective conditional probabilities $p_{i|k}$. Thus, the density operator for result $k$ becomes
 \bes
\begin{eqnarray}
	\r_k & = & \sum_i p_{i|k} \ketbra{\psi_i^k} \\
	& = & \sum_i p_{i|k} \frac{M_k \ketbra{\psi_i} M^{\dag}_k}{p_{k|i}} \\
	& = & \sum_i \frac{q_i}{p_k} M_k \ketbra{\psi_i} M^{\dag}_k \\
	& = & \frac{M_k \r M^{\dag}_k}{p_k} \\
	& = & \frac{M_k \r M^{\dag}_k}{{\Tr \ls \r M^{\dag}_k M_k \rs}},
\end{eqnarray}
\ees
where in the third line we used Bayes' rule ${\rm Pr}(i \& k) = {\rm Pr}(i|k) {\rm Pr}(k)={\rm Pr}(k|i) {\rm Pr}(i)$, where ${\rm Pr}(i)=q_i$ is the {\it a priori} probability of having state $\ket{\psi_i}$, and $p_k$ is the probability of measurement outcome $k$, as in Eq.~\eqref{eq:p_k}. Thus, comparing the pure state case to the generalized density operator case we observe
\beq
\ket{\psi} \mapsto \frac{M_k \ket{\psi}}{\sqrt{p_k}} \hspace{1in} \r \mapsto \frac{M_k \r M^{\dag}_k}{p_k}.
\eeq


\subsection{Properties of the density operator}
\begin{itemize}
\item \emph{Unit trace:} The trace operation is reviewed in Appendix~\ref{app:trace}. The density operator $\r$ has $\Tr \ls \r \rs = 1$. This property can easily be seen by the following calculation:
\beq\Tr \ls \r \rs = \sum_i q_i \Tr \ls \ketbra{\psi_i} \rs = \sum_i q_i = 1.
\eeq
\item \emph{Hermiticity:} The density operator $\r$ is Hermitian. The following line demonstrates this \beq \r^{\dag} = \sum_i q^*_i \lp \ketbra{\psi_i} \rp^{\dag} = \sum_i q_i \ketbra{\psi_i} = \r \eeq where we've used that probabilities $q_i$ are real and projectors formed from outer-products are Hermitian.
\item \emph{Positive definite:} For all vectors $\ket{\nu} \in \mc{H}$, the density operator $\r$ has $\braket{\nu | \r | \nu} \geq 0$:
\beq \braket{\nu | \r | \nu} = \sum_i q_i \left| \braket{\psi_i | v} \right|^2 \geq 0 ,
\eeq
since the $q_i$ are all non-negative by virtue of being probabilities. 
But since $\Tr \rho =1$ it clearly must have at least one eigenvalue that is non-zero. Therefore $\r$ must be positive, not just positive semi-definite (positive operators are defined in Appendix~\ref{app:pos-ops}).
\end{itemize}
Note that positivity implies Hermiticity, since an operator is Hermitian iff it has only real eigenvalues. Therefore we don't actually need to separately stipulate Hermiticity.
Also note that the density operator deserves to be called an operator: it acts as a transformation between two copies of the Hilbert space, i.e., $\r :\mc{H}\mapsto\mc{H}$. 

We define the space of positive, trace-one linear operators acting on $\mc{H}$ as $\mc{D}(\mc{H})$. Thus 
\beq
\r \in \mc{D}(\mc{H}).
\label{eq:D(H)}
\eeq

\subsection{Dynamics of the density operator}
Recall the two equivalent descriptions of dynamics of the pure quantum state
\beq \ket{\psi(t)} = U(t) \ket{\psi(0)} \hspace{0.5in} \Leftrightarrow \hspace{0.5in} \ket{\dot{\psi}} = -iH \ket{\psi} \eeq
where $U(t)$ and $H$ are related by $U(t) = e^{-iHt}$. Consider one of the pure states forming the ensemble $\{ q_i, \ket{\psi_i} \}_i$. This state will evolve as \beq
\ket{\psi_i(t)} = U(t) \ket{\psi_i(0)}
\eeq and the time-evolution of the density operator associated to the ensemble is
\bes
\begin{eqnarray}
	\r(t) & = & \sum_i q_i \ketbra{\psi_i(t)} \\
	& = & \sum_i q_i U(t) \ketbra{\psi_i(0)} U^{\dag}(t) \\
	& = & U(t) \r(0) U^\dag(t).
\end{eqnarray}
\ees
The Schr\"{o}dinger equation for the density operator takes a slightly different form however and we can derive it by taking the time-derivative of the first line above,
\bes
\begin{eqnarray}
	\frac{\partial}{\partial t} \r(t) & = & \frac{\partial}{\partial t} \sum_i q_i\ketbra{\psi_i(t)} \\
	\dot{\r}(t) & = & \sum_i q_i \ls \lp \frac{\partial}{\partial t}\ket{\psi_i(t)}\rp \bra{\psi_i(t)} + \ket{\psi_i(t)}\lp \frac{\partial}{\partial t} \bra{\psi_i(t)} \rp \rs.
\end{eqnarray}
\ees
At this point we invoke the Schr\"{o}dinger equation for pure states while making note that after Hermitian conjugation of the Schr\"{o}dinger equation we obtain $\bra{\dot{\psi}_j} = i\bra{\psi_j}H$. Thus:
\bes
\begin{eqnarray}
	\dot{\r}(t) & = & \sum_i q_i \lp -iH\ketbra{\psi_i(t)} + i\ketbra{\psi_i(t)}H \rp \\
	& = &  -i \ls H \lp \sum_i q_i \ketbra{\psi_i(t)} \rp - \lp \sum_i q_i \ketbra{\psi_i(t)} \rp H\rs \\
	& = &  -i \lp H \r - \r H\rp \\
	& = & -i \ls H, \r \rs
\end{eqnarray}
\ees
where $\ls \cdot , \cdot \rs$ represents the commutator of the two operators.

\subsection{Restatement of the postulates of quantum mechanics}
We can now summarize the four postulates in terms of the density operator.
\begin{enumerate}
\item The state space is the Hilbert-Schmidt space of linear operators $\r$ such that $\Tr \ls \r \rs =1$ and $\r > 0$. The inner product in the  Hilbert-Schmidt space is defined as $\Tr [A^\dagger B]$ for any two operators $A$ and $B$ acting on the same Hilbert space. This inner product defines a length in the usual way, i.e., $\|\rho\| = \sqrt{\<\rho,\rho\>}=\sqrt{P}$. The quantity 
\beq
P \equiv \Tr[\rho^2]
\label{eq:purity}
\eeq
is called the ``purity" of the state $\rho$. Thus a density matrix can have ``length" $\leq 1$. A state is called ``pure" if $P=1$ and ``mixed" if $P<1$.
\item State spaces are composed via the tensor product $\ox$.
\item Density operators evolve as $\dot{\r} = -i \ls H, \r \rs$ under a Hamiltonian $H$, or equivalently as $\r(t) = U(t) \r(0) U^\dagger(t)$ where the unitary $U (t)=e^{-itH}$.
\item A general measurement operation defined by elements $\{ M_k \}$ results with probability $p_k = \Tr \ls \r M^{\dag}_k M_k \rs$ in the state transformation $\r \mapsto \frac{M_k \r M^{\dag}_k}{{p_k}}$.
\end{enumerate}

Expectation values are now computable in terms of the $\r$ as well. Consider an observable $A$ measured for a system in the pure state ensemble $\{q_i,\ket{\psi_i}\}$. Previously, in Eq.~\eqref{eq:exp-val-pure}, we showed that the expectation value was 
$\langle A \rangle_\psi = \braket{\psi| A | \psi} = \Tr (A \ket{\psi}\bra{\psi})$.
We need to modify this by assigning each pure state $\ket{\psi_i}$ in the ensemble its weight $q_i$. Thus the new expression for the expectation value is:
\begin{align}
\label{eq:exp-val-rho}
\langle A \rangle_\r = \sum_i q_i \braket{\psi_i| A | \psi_i} = \sum_i q_i \Tr (A \ket{\psi_i}\bra{\psi_i}) = \Tr(A \sum_i q_i \ket{\psi_i}\bra{\psi_i}) = \Tr(A \r) = \Tr(\r A) .
\end{align}
Likewise, the standard deviation becomes:
\beq
\Delta A = \sqrt{\ave{(A-\ave{A}_\r)^2}_\r} ,
\eeq
and it is not hard to prove the associated uncertainty relation:
\beq
\Delta A \Delta B \geq \frac{1}{2}|\ave{[A,B]}_\r .
\eeq

To sum up, here is a comparison of the postulates for pure states and density operators:

\begin{center}
\begin{tabular}{|c|c|c|c|} \hline
& & \hspace{2.2cm}{\bf Pure States}\hspace{2.2cm} & \hspace{1.2cm}{\bf General States}\hspace{1.2cm} \\ \hline
\multirow{3}{*}{Postulate 1} & {\bf State space} & Hilbert space $\mathscr{H}$ & Trace-class operator space $\mathscr{D}$ \\ \cline{2-4}
& {\bf State} & ket vector $\ket{\psi}\in\mathscr{H}$ s.t. $\braket{\psi|\psi}=1$ & density operator $\rho$ s.t. $\left\{
\begin{array}{cl}
\Tr\left[ \rho \right]=1 \\
\rho > 0
\end{array}
\right.$ \\ \cline{2-4}
& {\bf Inner product} & $f(\ket{\mu},\ket{\omega})\equiv\braket{\mu|\omega}$, $\forall\ket{\mu},\ket{\omega}\in\mathscr{H}$ & $f(A,B)\equiv \Tr\left[ A^\dag B \right]$, $\forall A,B\in\mathscr{D}$\\&&& Hilbert-Schmidt inner product \\ \hline
Postulate 2 & {\bf Expansion} & tensor product $\otimes$ & tensor product $\otimes$ \\ \hline
\multirow{2}{*}{Postulate 3} & {\bf Dynamics} & Schr\"odinger equation: & Liouville-von Neumann equation: \\
& w/ Hamiltonian $H$ & $\frac{d\ket{\psi(t)}}{dt}=-iH\ket{\psi(t)}$ & $\frac{d\rho(t)}{dt}=-i\left[ H,\rho(t) \right]$ \\ \hline
\multirow{2}{*}{Postulate 4} & {\bf Measurement} & outcome $k\in K$ w.p. $p_k=\bra{\psi}M_k^\dag M_k\ket{\psi}$ & outcome $k\in K$ w.p. $p_k=\Tr\left[ M_k \rho M_k^\dag \right]$ \\
& w/ meas. ops. $\left\{ M_k \right\}_{k\in K}$ & $\ket{\psi}\mapsto \frac{M_k\ket{\psi}}{\sqrt{p_k}}$ & $\rho\mapsto \frac{M_k \rho M_k^\dag}{p_k}$ \\ \hline
\end{tabular}
\end{center}

\subsection{More on pure and mixed quantum states}
\label{sec:mixedpure}
We defined ``pure" and ``mixed" states above according to the value of the purity $P=\Tr[\r^2]$ being $1$ or $<1$.
Prior to introducing the density operator formalism, we had considered quantum states as vectors in the Hilbert space. This formalism is equivalent to pure state ensembles of the type $\{1,\ket{\psi}\}$, i.e., having only a single element. It is not hard to see that such special ensembles are ``pure" quantum states. The associated density operator is $\r = \ketbra{\psi}$. It is useful to think of pure states as ensembles with only one member and probability $1$.

Any state that is not not pure is by definition mixed. This means that they are described by ensembles of the form $\{p_i,\ket{\psi}_i\}$ where for all $i$, $0<p_i<1$. The density operator associated with a mixed ensemble is the \emph{mixture} of the pure states with their associated weights [as seen in Eq.~\eqref{eq:rho}]. 

Note that a pure state is a projector: $( \ketbra{\psi})( \ketbra{\psi})= \ketbra{\psi}$. Therefore, if a state $\rho$ is pure then $\r^2 = \r$. The converse is also true: $\r^2 = \r$ implies that $\rho$ is pure. It is easy to check that these conditions are equivalent to the definition in terms of purity $P$. 

We can also define a \emph{mixed state ensemble}, i.e., a collection of mixed states $\r_k$ with associated probabilities $p_k$, as
\beq
\{p_k,\r_k\} \Leftrightarrow \r = \sum_k p_k \r_k \ .
\label{eq:mixedens}
\eeq

\subsection{Unitary equivalence in ensembles}
\label{subsec:equiv}
When are two pure state ensembles equivalent? Consider for example the two ensembles
\bes
\bea
&& \{(3/4,1/4),(\ket{0},\ket{1})\}\\
&& \{(1/2,1/2),(\ket{a},\ket{b})\} ,
\eea
\ees
where
\bes
\bea
\ket{a} &=& \sqrt{3/4}\ket{0}+\sqrt{1/4}\ket{1} \\
\ket{b} &=& \sqrt{3/4}\ket{0}-\sqrt{1/4}\ket{1}.
\eea
\ees
On the face of it, the first of these ensembles represents a biased classical coin (``heads", or $0$, with probability $3/4$, tails, or $1$ with probability $1/4$), whereas the second is quantum in the sense that each state is a superposition state. But are they really different?
It is straightforward to check that in fact the two density matrices corresponding to these two ensembles are equal. This being the case, there is no measurement that can distinguish them, and that means we must consider them to be the same.

\begin{thm} \label{thm:relatingensembles} Two pure state ensembles with the same number of elements\footnote{If necessary pad the smaller set with zeroes to make it equal in length to the larger set.} $\{q_i, \ket{\psi_i}\}_i$ and $\{r_j, \ket{\phi_j} \}_j$ correspond to the same density operator if and only if there exists a unitary $U$ with entries $[U]_{ij}$ such that \beq \sqrt{q_i} \ket{\psi_i} = \sum_j [U]_{ij} \sqrt{r_j} \ket{\phi_j} \label{eqn:unitaryequivform} \eeq 
\end{thm}

\begin{proof} We show explicitly the ``if" direction of the proof. The complete proof is found in \cite{nielsen2010quantum}[p.104]. Consider the following mixture,
\bes
\begin{eqnarray}
	\sum_i q_i \ketbra{\psi_i} & = & \sum_i \lp \sqrt{q_i} \ket{\psi_i} \rp \lp \bra{\psi_i} \sqrt{q_i} \rp \\
	& = & \sum_i \lp \sum_j [U]_{ij} \sqrt{r_j} \ket{\phi_j} \rp \lp \sum_k \bra{\phi_k} \sqrt{r_k} [U^{\dag}]_{ik} \rp \\
	& = & \sum_{j,k} \sqrt{r_j r_k} \lp \sum_i [U]_{ij}[U^{\dag}]_{ik} \rp  \ketb{\phi_j}{\phi_k} \\
	& = & \sum_{j,k} \sqrt{r_j r_k} \lp \d_{jk} \rp  \ketb{\phi_j}{\phi_k} \label{eq:unitarity} \\
	& = & \sum_{j} r_j \ketbra{\phi_j},
\end{eqnarray}
\ees
where in Eq.~\eqref{eq:unitarity} we used the unitarity of $U$. Thus the two ensembles represent the same density operator.
\end{proof}

\subsection{Visualizing the density matrix of a qubit: the Bloch sphere}

A qubit is a quantum state $\ket{\psi_i}$ in a two-dimensional Hilbert space $\mc{H} = \C^2 = \mathrm{span} \{ \ket{0}, \ket{1} \}$ where $\ket{0}$ and $\ket{1}$ form an orthonormal basis for $\mc{H}$. The density operator for any state in this space is thus of the form $\sum_i q_i \ketbra{\psi_i}$ and can hence be represented by a $2 \times 2$ complex matrix of the form
\beq 
\r = \ls \begin{array}{cc} a&b\\c&d \end{array} \rs. 
\eeq
However, applying the properties of density operators can reduce this to an expression of only two variables. First, the unit trace reduces to the condition $d=1-a$ and Hermiticity reduces to the condition that $c=b^*$ and that $a$ be real. Thus, the density matrix is completely parametrized by the complex number $b$ and the real number $a$ and takes the form
\beq 
\r = \ls \begin{array}{cc} a&b\\b^{*}&1-a \end{array} \rs. 
\eeq
Positivity is the statement that the eigenvalues $\lambda_\pm$ are non-negative:
\beq
 \left| \r - \lambda I \right| = 0 \Rightarrow \lambda^2 - (\Tr\r)\lambda + |\r | = 0 ,
 \eeq
 i.e., using $\Tr\r=1$:
 \beq
 \lambda_\pm = \frac{1}{2} (1\pm \sqrt{1-4|\r |}) \geq 0 .
 \eeq
This parametrization requires only three parameters and we can thus embed it naturally in three dimensions. Before we proceed to do this we will decompose the density operator one more time but in a more useful basis.

Recall the Pauli matrices $\s_x, \s_y, \s_z, \s_0$. Any qubit density matrix can represented by 
\beq 
\r = \frac{1}{2} \left( I +\sum_i v_i\s_i \right) = \frac{1}{2} \left( I + \vec{v} \cdot \vec{\s} \right)\ ,
\label{eq:Bloch-vec}
\eeq 
where $\vec{v} = (v_x, v_y, v_z)$ and $\vec{\s} = (\s_x, \s_y, \s_z)$. In terms of the elements of $\vec{v}$, $\r$ appears as \beq \r = \frac{1}{2} \ls \begin{array}{cc} 1+v_z & v_x - i v_y \\ v_x + i v_y & 1 - v_z \end{array} \rs. \eeq To relate this to our previous analysis simply let $b = 1/2 (v_x - i v_y)$ and $a = 1/2 (1+ v_z)$. We call $\vec{v}$ the \emph{Bloch vector}. The $2 \times 2$ matrix we have constructed using the Bloch vector is not, however, necessarily a valid quantum state. Unit trace is guaranteed by the construction, and positivity can now be made explicit by noting that
\beq
|\r | = \frac{1}{4} \left(1-v_z^2-(v_x^2+v_y^2)\right) = \frac{1}{4} \left(1-\|\vec{v}\|^2\right),
\eeq
so that 
\beq
 \lambda_\pm = \frac{1}{2} \lp 1 \pm \left\| \vec{v} \right\| \rp .
\eeq
The two solutions are $\left\| \vec{v} \right\| \leq 1$ and $\left\| \vec{v} \right\| \geq -1$, which is trivially satisfied. Thus if we require positivity, the relevant constraint is 
\beq
\boxed{\left\| \vec{v} \right\| \leq 1} .
\eeq 
Let us also relate the magnitude of the Bloch vector to the purity of the quantum state. Recall that a pure quantum state is a projector and thus $\r^2=\r$ for pure states. If we calculate the density operator $\r^2$ we find
\beq
	\r^2 = \frac{1}{4} \lp I+ \vec{v} \cdot \vec{\s} \rp \lp I+ \vec{v} \cdot \vec{\s} \rp = \frac{1}{4} \lp I+ 2 \vec{v} \cdot \vec{\s}  + \lp \vec{v} \cdot \vec{\s} \rp^2 \rp.
	\label{eq:76}
\eeq
The term $(\vec{v} \cdot \vec{\s})^2$ becomes \beq \sum_{k,l \in \{x,y,z\}} v_kv_l\s_k\s_l. 
\eeq 
Recall Eq.~\eqref{eq:Pauli-mult}. 
Taking the trace and noting that the Pauli matrices are traceless
only the $\d_{kl}$ term remains. Thus $\Tr(\vec{v} \cdot \vec{\s})^2 = \left\| \vec{v} \right\|^2 \Tr I$, with $\Tr I=2$, and Eq.~\eqref{eq:76} yields:
\begin{eqnarray}
	\Tr \r^2  =  \frac{1}{2} \left(1+\|\vec{v} \|^2 \right).
\end{eqnarray}
From this form it is clear that any unit Bloch vector will make $\Tr\r^2=1$, i.e., a pure state, and Bloch vectors of length less than $1$ yield mixed states.

Having gathered the requisite intuition for the geometry at hand, we call the set of all valid Bloch vectors $\vec{v}$ the \emph{Bloch sphere}, also known as the \emph{Poincar\'{e} sphere} in optics.

\begin{figure}[t]
	\begin{center}
	\label{fig:blochsphere}
	\includegraphics[width=0.3\linewidth]{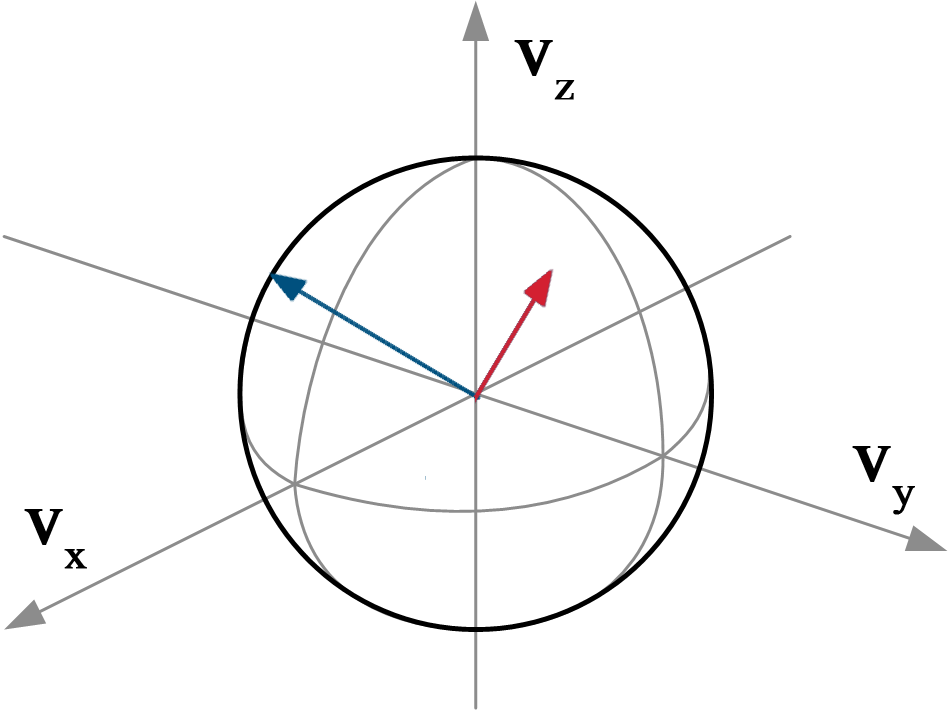}
	\caption{The Bloch sphere is a geometric representation of the collection of all Bloch vectors $\vec{v}$ which describe valid qubit density operators. Thus, the sphere is of radius $1$, its surface represents all pure states, and its interior represents all mixed states. In this diagram the blue vector lies on the surface of the sphere indicating a pure state, whereas the red vector lies in its interior indicating a mixed state.}
	\end{center}
\end{figure}
Since the Bloch sphere can describe all qubit states and can be embedded in three dimensions it is a useful tool for illustrating various common qubit states.
\begin{itemize}
\item {$Z$ poles ($\vec{v} = (0,0,\pm 1)$)}:  The density matrix takes the form
\begin{eqnarray*}
	\r & = & \frac{I \pm \s_z}{2} \\
	& = & \frac{1 \pm 1}{2} \ketbra{0} + \frac{1 \mp 1}{2} \ketbra{1}
\end{eqnarray*}
which yields $\ketbra{0}$ for $v_z=1$ and $\ketbra{1}$ for $v_z=-1$.
\item {$X$ poles ($\vec{v} = (\pm1,0,0)$)}: The density matrix takes the form
\begin{eqnarray*}
	\r & = & \frac{I \pm \s_x}{2} \\
	& = & \frac{1}{2}\lp \ketbra{0} + \ketbra{1} \pm (\ketb{0}{1} +\ketb{1}{0}) \rp \\
	& = & \frac{1}{2} \lp \ket{0} \pm \ket{1} \rp \lp \bra{0} \pm \bra{1} \rp
\end{eqnarray*}
which yields $\ketbra{+}$ for $v_x=1$ and $\ketbra{-}$ for $v_x=-1$.
\item {$Y$ poles ($\vec{v} = (0,\pm1,0)$)}: The density matrix takes the form
\begin{eqnarray*}
	\r & = & \frac{I \pm \s_y}{2} \\
	& = & \frac{1}{2}\lp \ketbra{0} + \ketbra{1} \pm (-i\ketb{0}{1} + i\ketb{1}{0}) \rp \\
	& = & \frac{1}{2} \lp \ket{0} \pm i\ket{1} \rp \lp \bra{0} \pm (-i)\bra{1} \rp
\end{eqnarray*}
which yields $(\ket{0} + i \ket{1})/\sqrt{2}$ for $v_y=1$ and $(\ket{0} - i \ket{1})/\sqrt{2}$ for $v_y=-1$.
\item {Center ($\vec{v} = (0,0,0)$)}: The density matrix takes the form $\r = I/2$, the maximally mixed state.
\end{itemize}

Since the dimensionality of this geometric representation goes as $d^{2}-1$ for a $d$-level system (the density matrix becomes a $d\times d$ matrix, and the trace constraints removes one matrix element), the Bloch sphere is typically only used to represent two-level systems.
As we shall see later on, the Bloch sphere plays an important visualization role in understanding the dynamics of open quantum systems.


\section{Composite Systems}
\subsection{Combining a system and a bath}
Now that we have discussed in detail the dynamics of a single system, let us consider more complex systems. Consider a two component system, where we have a subsystem of our interest,  $A$ (often we'll just call it ``system"), and the other subsystem is the bath, $B$. Together, the system and the bath comprise the lab, or even the entire universe. We can think of $A$ as a quantum computer, or a molecule, or any other system we're interested in studying. We shall assume that the total system evolves according to the Schr\"odinger equation and that it is described by a density matrix $\r(t)$. Further, let the subsystem Hilbert spaces be
\begin{eqnarray}
\mathcal{H}_A &=& {\rm span}\{\ket{i}_A\} \\
\mcHB &=& {\rm span}\{\ket{\mu}_B\}
\end{eqnarray}
Here, $i$ goes from $0$ to $d_A-1$, the dimension of the Hilbert space of $A$, and $\mu$ goes from $0$ to $d_B-1$, the dimension of the Hilbert space of $B$.
Usually, the dimension of the bath, $d_B \rightarrow \infty $, while $d_A$ is finite. By the second postulate, the Hilbert space of the two system combined is the tensor product of the individual spaces:
\bes
\begin{eqnarray}
\mathcal{H} &=& \mathcal{H}_A \otimes \mcHB \\
				&=& {\rm span}\{\ket{i}_A \otimes \ket{\mu}_B\}
\end{eqnarray}
\ees
Let us figure out the structure of a density matrix in this combined Hilbert space. We can define a pure state ensemble $\{\ket{\Psi_a},q_a\}$ for a set of pure states $\ket{\Psi_a}\in\mc{H}$. Each of these states can be expanded in the basis above, i.e.,
\beq
\ket{\Psi_{a}} = \sum_{i,\mu} c_{a;i\mu}\ket{i}_A \otimes \ket{\mu}_B .
\eeq
Thus, the associated density matrix is:
\beq
\r = \sum_a q_a \ketbra{\Psi_a} = \sum_a q_a (\sum_{i,\mu} c_{a;i\mu}\ket{i}_A \otimes \ket{\mu}_B)(\sum_{j,\nu} c^*_{a;j\nu}\bra{j}_A \otimes \bra{\nu}_B) .
\eeq
Therefore any density matrix in the combined Hilbert space can be written down as
\beq
\r = \sum_{ij\mu\nu}\lambda_{ij\mu\nu} \ket{i}_A\!\bra{j} \otimes \ket{\mu}_B\!\bra{\nu} ,
\eeq
where $\lambda_{ij\mu\nu} = \sum_a q_a c_{a;i\mu}c^*_{a;j\nu}$.

Note that if $\lambda_{ij\mu\nu} = \lambda^A_{ij}\lambda^B_{\mu\nu}$ then $\r = \r_A\otimes \r_B$, where $\r_A = \sum_{i\mu}\lambda^A_{ij} \ket{i}_A\!\bra{j}$ and $\r_B = \sum_{\mu\nu}\lambda^B_{\mu\nu} \ket{\mu}_B\!\bra{\nu}$. In this case $\r$ is called a ``factorized" state. Such states exhibit no correlations at all between the $A$ and $B$ subsystems. Clearly, however, 
this is a special case and in general, $\r $ cannot be factored in this manner. When it cannot, the subsystems are correlated. These correlations can be quantum (due to entanglement), classical, or both.

We are primarily interested in the system $A$. We thus need to find a way to remove the bath $B$ from our description. To do, we now define a new operation called partial trace, which effectively averages out of the components of $B$ from the combined density matrix. The resultant density matrix then describes only $A$.

\subsection{Partial Trace}

\subsubsection{Definition}
The partial trace is a linear operator that maps from the total Hilbert space to the Hilbert space of $A$, i.e., $\mathcal{H} \mapsto \mathcal{H}_A$, defined as follows.
Consider an operator $O = M_A \otimes N_B$ such that $O$ acts on $\mathcal{H}=\mathcal{H}_A \otimes \mcHB$. Then
\bes
\begin{align} 
\Tr_B (M_A \otimes N_B) &\equiv M_A \; \Tr(N_B) \\
&= M_A \; \sum_\mu \braket{\mu |N_B| \mu} \\
&= \sum_\mu \braket{\mu | [ M_A \otimes N_B ] | \mu }
\label{eq:OpPartialTrace}
\end{align}
\ees
It is understood in the last line that the basis vectors $\{\ket{\mu}\}$, which span the space $\mcHB$, act only on the second Hilbert space. In other words, the expression $\braket{\mu | [ M_A \otimes N_B ] | \mu }$ is a \emph{partial matrix element}, where the matrix element is taken only over the second factor, and the result is an operator acting on $\mathcal{H}_A$. Thus, if $O = \sum_{ij}\ell_{ij} M_A^i \otimes N_B^j$, then by linearity:
\bes
\begin{align} 
\Tr_B [O] &= \sum_{ij} \ell_{ij} \Tr_B (M_A^i \otimes N_B^j) \\
&= \sum_{ij} \ell_{ij} M_A^i \sum_\mu \braket{\mu |N_B^j| \mu} \\
&=\sum_\mu  \sum_{ij} \ell_{ij} \braket{\mu | [ M_A^i \otimes N_B^j ] | \mu }\\
& = \sum_\mu \braket{\mu | O | \mu }
\label{eq:OpPartialTrace2}
\end{align}
\ees

For example, when applied to a summand in the expression for $\r$:
\beq
\Tr_B [ \ket{i}_A\!\bra{j} \otimes \ket{\mu}_B\!\bra{\nu}] \equiv \ket{i}_A\!\bra{j} \braket{\nu|\mu}_B .
\eeq
By linearity,
\bes
\begin{align}
\Tr_B [ \sum_{i j \mu \nu} \lambda_{i j \mu \nu} \ket{i}_A\!\bra{j} \otimes \ket{\mu}_B\!\bra{\nu}] &= \sum_{i j\mu\nu} \lambda_{i j \mu\nu} \ket{i}_A\!\bra{j} \braket{\nu|\mu}_B \\
& = \sum_{i j\mu} \lambda_{i j \mu\mu} \ket{i}_A\!\bra{j}  = \sum_{i j} \bar{\lambda}_{i j} \ket{i}_A\!\bra{j} ,
\end{align}
\ees
where in the second line we assumed that $\{\ket{\mu}_B\}$ forms an orthonormal basis, and we defined $\bar{\lambda}_{i j} = \sum_{\mu} \lambda_{i j \mu\mu} $. This shows that taking the partial trace leads to a form that looks like a density matrix for the $A$ subsystem. Of course, we'll have to verify that it satisfies the properties of a density matrix (unit trace and positivity). Positivity is more challenging, but unit trace is obvious if we assume (as we should) that $\Tr\r=1$. For, it is then easy to check that this implies $\sum_{i\mu}\lambda_{ii\mu\mu}=1$. On the other hand, if we are to interpret $\Tr_B[\r]$ as a valid density matrix then $\Tr(\Tr_B[\r]) = \sum_i\bar{\lambda}_{ii}$ should be $1$, which it is, since it equals $\sum_{i\mu}\bar{\lambda}_{ii\mu\mu}$.

\subsubsection{State of a quantum subsystem}

Crucially, we now claim that the density matrix of the subsystem $A$ is given by taking the partial trace of the combined density matrix with respect to $B$.
 \beq 
 \r_A = \Tr_B \left[\r\right]. 
 \eeq
This is called the reduced density matrix.

To justify this intuitively, we consider the cases which lie on the two extreme ends of combination of bath and system, viz. the simplest case of a separable density operator, and the case where system and bath are maximally entangled.

\begin{enumerate}

\item \emph{Case 1}: Consider a case where the states of the bath and system are completely separate, and hence form a tensor product. In such a case, we expect that the density operator of $A$ obtained by partial trace should be the same as the component of $A$ contributed in the tensor product. And indeed, clearly, if $\r = \r_A \otimes \r_B$ where both terms in the product are properly normalized states, then 
\beq \Tr_B \left[\r\right] = \r_A  \Tr_B\left[\r_B\right] = \r_A \eeq
\item \emph{Case 2}: Consider two qubits that are maximally entangled, that is
\beq
\ket\psi_{AB} = \frac{1}{\sqrt{2}}(\ket{0}_A \ket{0}_B + \ket{1}_A \ket{1}_B) .
\label{eq:MaxEntangled}
\eeq
This means that the state state $\ket\psi_{AB}$ contains no separate information about $A$ or $B$'s state. The reason is that if we measure, say, $B$ using the measurement operators $\{M_0=\ketbra{0},M_1=\ketbra{1}\}$, then we find the outcomes $0$ and $1$ with equal probability $1/2$, and at the same time the state of $A$ becomes either $\ket{0}$ or $\ket{1}$, respectively. It is easy to check that this random outcome remains true for any other choice of measurement operators. This means we gain no knowledge at all about $A$ or $B$ since the measurement outcome is perfectly random. In terms of the partial trace we find:
\bes
\begin{eqnarray}
\r_A &=& \Tr_B \left[\r\right] = \Tr_B[\ket{\psi}_{AB}\bra{\psi}] \\
  &=&  \frac{1}{2} \Tr_B [\ketbsub{0}{0}{A} \otimes \ketbsub{0}{0}{B} + \ketbsub{0}{1}{A} \otimes \ketbsub{0}{1}{B} + \ketbsub{1}{0}{A} \otimes \ketbsub{1}{0}{B} + \ketbsub{1}{1}{A} \otimes \ketbsub{1}{1}{B}] \notag \\
  \\
  &=& \frac{1}{2}[ \ketb{0}{0} \times 1 + \ketb{0}{1} \times 0 + \ketb{1}{0} \times 0 + \ketb{1}{1} \times 1] \\
  &=& \frac{1}{2} [\ketbsub{0}{0}{A} + \ketbsub{1}{1}{A}] = I_A/2
\end{eqnarray}
\ees
Therefore, the state of $A$ is an equal probabilistic mixture of the $\ket{0}$ and $\ket{1}$ states, as  expected. \\
\end{enumerate}

Next, we provide a formal justification. 

\subsubsection{Formal justification of using the partial trace to define a subsystem state}

Consider a composite system with the Hilbert space $\mathcal{H} = \mathcal{H}_A \otimes \mcHB$. If we had an observable $M_A$ on subsystem $A$, then, the expectation value of that operator would be given by
\beq
\braket{M_A}_{\rho_A} = \Tr[ \rho_A M_A] ,
\eeq
where we used Eq.~\eqref{eq:exp-val-rho}.

However, in the case of this composite system, this measurement is actually of the observable $\widetilde{M} = M_A \otimes I_B$ on the entire system $\r $ in $\mathcal{H}_A \otimes \mcHB$, where we do nothing (the identity operation) to $B$. Thus,

\beq
\braket{\widetilde{M}} = \Tr [\rho \; \widetilde{M} ]
\eeq

The key idea is that these two operations should correspond to the same physical observation and they should produce the same number. For the theory to be consistent, we demand that

\beq \<M_A\> \equiv \<\widetilde{M}\> ,\eeq
i.e.,
\beq \Tr[M_A\r_A] = \Tr[\widetilde{M}\r] \eeq

It can be shown that this condition is satisfied iff we define $\r_A \equiv \Tr_B(\r)$. We shall prove the theorem in one direction, that is, if $\r_A =  \Tr_B(\r)$, then $\<M\> = \<\widetilde{M}\>$. \\

\textit{\emph{Proof}}. 
Let $\mathcal{H} = \mathcal{H}_A \otimes \mcHB  = {\rm span}\{\ket{i}_A \otimes \ket{\mu}_B\}$. Then

\bes
\begin{align}
	\<M_A\> &= \sum_i \left. _A\langle i\right|\r_A M_A | i\rangle_A \\
		   &= \sum_i \left. _A\langle i\right|\Tr_B[ \r] M_A| i\rangle_A \label{88b} \\
		   &= \sum_i \left. _A\langle i\right| \sum_\mu \left. _B\langle \mu\right| \r | \mu\rangle_B  M_A | i\rangle_A \label{88c} 
\end{align}
\ees
In going from Eq.~\eqref{88b} to \eqref{88c}, we used the expression for the partial trace over operators as given in Eq.~\eqref{eq:OpPartialTrace2}. But note that $\rho$ is an operator acting on the composite system, not just on $A$, since $\left. _B\langle \mu\right| \r | \mu\rangle_B $ is a partial matrix element. If we wish to likewise consider $ M_A$ as an operator acting on the composite system, then we should extend it to $M_A \otimes I_B$. Also, the correct order for the product $\ket{\mu}_B\ket{i}_A$, including the tensor product symbol explicitly, is: $\ket{i}_A\otimes\ket{\mu}_B$. Thus:
\bes
\begin{align}
	\<M_A\> &= \sum_{i,\mu} \bra{i} \otimes \bra{\mu} [\r (M_A \otimes I_B)] \ket{i}\otimes\ket{\mu} \\
&= \Tr [\r (M_A \otimes I_B)] \\ 
&=\<\widetilde{M}\> \ ,
\end{align}
\ees
which shows the desired equality  $\<M_A\> = \<\widetilde{M}\>$.

\section{Open System Dynamics}

In this section we shall find the dynamical evolution of an open quantum system.

\subsection{Kraus Operator Representation}
\label{sec:Kraus}

Consider a system $S$ and bath $B$, such that they have a joint unitary evolution given by $U(t) = e^{-\ii H t}$. 
The initial joint state is $\r(0)$. Then, by Schr\"odinger's equation,
\beq \r(t) = U(t) \r (0) U^\dagger (t) \eeq
As the density operator of the bath is positive and normalized, it has a spectral decomposition in an orthonormal basis with non-negative eigenvalues. Hence
\beq 
\r_B(0) = \sum_\nu \lambda_\nu \ketbra{\nu} 
\eeq
where $\lambda_\nu$ are the eigenvalues (probabilities) and $\{\ket{\nu}\}$ are the corresponding orthonormal eigenvectors.

The state of the system is then found by performing a partial trace over the bath, i.e., 
\beq
\r_S(t) = \Tr_B [\r(t)].
\eeq
We can perform the partial trace in the orthonormal basis of bath eigenstates, i.e., 
\bes
\begin{align}
\r_S(t) &= \Tr_B [U(t) \r (0) U^\dagger (t)] \\
        &= \sum_\mu \braket{\mu| U(t) \r(0) U^\dagger(t)|\mu}
\end{align}
\ees
Let us now assume that the initial state is completely decoupled, that is 
\beq 
\r (0) = \r_S (0) \otimes \r_B (0).
\label{eq:decoupled}
\eeq 
Then
\bes
\begin{align}
\r_S(t) 	  &= \sum_\mu \braket{\mu|[ U(t) \r_S(0) \otimes  \sum_\nu \lambda_\nu \ketbra{\nu}U^\dagger(t)]           |\mu} \\
		  &= \sum_{\mu \nu} \sqrt{\lambda_\nu} \braket{\mu|U(t)|\nu}_B \r_S(0) \sqrt{\lambda_\nu} \braket{\nu|U^\dagger(t)|\mu}_B \\
        &= \sum_{\mu \nu } K_{\mu \nu}(t) \r_S(0) K_{\mu \nu}^\dagger (t) .
\end{align}
\ees

The system-only operators $\{K_{\mu \nu}\}$ are called the Kraus operators and are given by
\beq
K_{\mu \nu}(t) = \sqrt{\lambda_\nu} \braket{\mu|U(t)|\nu}
\label{eq:K}
\eeq
(note the \emph{partial} matrix element, leaving us with an operator acting on the system), and the equation defining the evolution of the system in terms of Kraus operator is called the Kraus Operator Sum Representation (OSR)
\beq
\boxed{\r_S(t) = \sum_{\mu \nu } K_{\mu \nu}(t) \r_S(0) K_{\mu \nu}^\dagger (t)}
\label{eq:Kraus OSR}
\eeq
This is a pivotal result; as we shall see it includes the Schr\"{o}dinger equation as a special case.

\subsection{Normalization and the special case of a single Kraus operator}
The system state should be normalized at all times, so we demand
\bes
\begin{eqnarray}
\Tr [\r_S(t)] &=& 1 \\
				&=& \Tr [\sum K_{\mu \nu}(t) \r_S(0) K_{\mu \nu}^\dagger (t)] \\
				&=& \sum \Tr [ K_{\mu \nu}(t) \r_S(0) K_{\mu \nu}^\dagger (t)] \\
				&=& \sum \Tr [ K_{\mu \nu}^\dagger (t) K_{\mu \nu}(t)  \r_S(0) ] \\
            &=& \Tr [\sum  K_{\mu \nu}^\dagger (t) K_{\mu \nu}(t)  \r_S(0) ]
\end{eqnarray}
\ees
It is easy to check that the equation is satisfied if $\sum  K_{\mu \nu}^\dagger(t) K_{\mu \nu} (t) = I$. However, this condition is not necessary. Thus the system state is guaranteed to be normalized provided the Kraus operators satisfy the following identity,  
\beq 
\sum_{\mu \nu}  K_{\mu \nu}^\dagger (t)K_{\mu \nu} (t) = I \ .
\label{eq:sum-Kraus}
\eeq
This criterion can be verified for our definition of Kraus operators, given by Eq.~\eqref{eq:K}.
\bes
\begin{eqnarray}
\sum_{\mu \nu} K_{\mu \nu}^\dagger K_{\mu \nu} &=& \sum_{\mu \nu} \lambda_\nu \braket{\mu|U(t)|\nu}\braket{\nu|U^\dagger(t)|\mu} \\
&=& \sum_\nu \lambda_\nu \braket{\nu|U^\dagger(t) \left( \sum_\mu \ketbra{\mu}\right)U(t)|\nu} \\
&=& \sum_\nu \lambda_\nu \braket{\nu|\nu} \\
&=& \sum{\lambda_\nu} = 1
\end{eqnarray}
\ees
Thus, such a set of Kraus operators preserves normalization.

Note that when there is just a single Kraus operator, the normalization condition~\eqref{eq:sum-Kraus} forces it to be unitary, which is just the case of closed system evolution! We can see more explicitly how this comes about, as follows.

\subsection{The Schr\"{o}dinger equation as a special case}
Assume that $U = U_S\otimes U_B$. In this special case the Kraus operators become $K_{\mu \nu} =  U_S \sqrt{\lambda_\nu} \braket{\mu|U_B|\nu} \equiv c_{\mu\nu} U_S$. It's easy to see that the sum rule normalization condition implies $\sum_{\mu\nu} c^*_{\mu\nu}c_{\mu\nu} = 1$, since now $\sum_{\mu \nu} K_{\mu \nu}^\dagger K_{\mu \nu}  = \sum_{\mu \nu} c_{\mu \nu}^* c_{\mu \nu} U_S^\dagger U_S = I$. Thus:
\beq
\r_S(t) = \sum_{\mu \nu } c_{\mu \nu} U_S(t) \r_S(0) c^*_{\mu \nu} U_S^\dagger (t) = U_S(t) \r_S(0) U_S^\dagger (t)\ ,
\eeq
which is unitary, Schr\"{od}inger-like dynamics.
Hence, the Kraus operator sum representation is more general than the Schr\"odinger equation, because it contains the latter as a special case.


\section{Complete Positivity and Quantum Maps}

We have seen [Eq.~\eqref{eq:Kraus OSR}] that the evolution of the state $\r_S$ of an open quantum system can be expressed as unitary evolution of the composite system+bath, followed by a partial trace, which leads to the Kraus operator sum representation (Kraus OSR):
\bea
\r_S (t)=\Tr_B [U(t)(\r_S \otimes \r_B)U^\dagger(t)] =\sum_{\a} K_{\a} (t) \r_S (0) K^\dagger_{\a} (t)\ ,
\eea
where we have collected the earlier $\m\n$ indices into a single index: $\a = (\m\n)$. From now on let us drop the $S$ subscript since we'll be focusing on the system alone. We'll reintroduce it as necessary.

\subsection{Non-selective measurements}
\label{sec:non-selective}
Let us observe that the OSR represents more than dynamics. It can also capture measurements. Specifically, consider measurement operators $\{ M_k \}$  with $\sum_k M^{\dag}_k M_k = I$. Recall that a state subjected to this measurement maps to
\beq 
\rho \mapsto \rho_k = \frac{M_k \rho M_k^{\dag}}{\Tr \ls M_k \rho M_k^{\dag}\rs}
\eeq
with probability $p_k = \Tr \ls M_k \rho M_k^{\dag}\rs$. Consider the case where we perform this measurement but do not learn the outcome $k$. What happens to $\rho$ after this measurement? In this case
\beq 
\rho \longmapsto  \langle \rho \rangle = \sum_k p_k \rho_k = \sum_k M_k \rho M^{\dag}_k  
\eeq
which we recognize as a \emph{non-selective measurement}. This last form is in the Kraus operator-sum representation with the Kraus operators $M_k$. Thus, we can encapsulate the non-selective measurement postulate in the operator-sum formalism. 

Since both dynamics and measurements are captured by the OSR, and there are no other quantum processes according to our postulates, this suggests that the OSR is truly fundamental. It thus deserves further scrutiny.

\subsection{The OSR as a map}
It is useful to think of the OSR as a \emph{map} (or synonymously a \emph{process} or \emph{channel}) $\Ph$ from the initial to the final system state, i.e., 
\beq
\r(t) = \Phi[\r(0)] \ \qquad \leftrightarrow \qquad \Ph:\r (0)\mapsto \r (t)\ ,
\eeq
where $\Phi[X] \equiv  \sum_{\a } K_{\a}X K_{\a}^\dagger $. Note that $\Ph$ is an operator acting on operators, sometimes called a superoperator. While we started with vectors $\ket{v}$ in a Hilbert space $\mc{H}_S$, and moved the density operators $\r : \mc{H}_S \mapsto \mc{H}_S$ belonging to the space of positive trace-class operators $\mc{D}(\mc{H}_S)$, the map $\Ph : \mc{D}\mapsto\mc{D}$ belongs to $\mc{D}[\mc{D}(\mc{H}_S)]$, as we shall see shortly. In terms of dimensions, if $\dim(\mc{H}_S) =d$, then $\dim(\mc{D}(\mc{H}_S))=d^2$, and $\dim(\mc{D}[\mc{D}(\mc{H}_S)]) = d^4$, reflecting the fact that vectors are of dimension $d\times 1$, density matrices of dimension $d\times d$, and quantum maps of dimension $d^2\times d^2$.

It will prove to be profitable to adopt an even more abstract point of view, and seek to determine the key properties that any such map possesses. We can easily identify three properties by inspection:

\begin{enumerate}

\item Trace Preserving:
\bea
\Tr[\Ph(\r)]=\sum_\a\Tr(K_\a \r K_\a^\dagger)
=\sum_\a \Tr( K_\a^\dgr K_\a \r)
=\Tr(\sum_\a K_\a^\dagger K_\a \r)
=\Tr(\r)\ ,
\eea
where we used the fact that $\sum_\a K_\a^\dagger K_\a=I$. Thus the map $\Ph$ is trace-preserving.

\item Linear:

By direct substitution we find:
\beq
\Ph(a\r_1+b\r_2) = \sum_\a\Tr(K_\a a\r_1 K_\a^\dagger) + \sum_\a\Tr(K_\a b\r_2 K_\a^\dagger) = a\sum_\a\Tr(K_\a \r_1 K_\a^\dagger) + b \sum_\a\Tr(K_\a \r_2 K_\a^\dagger) = a\Ph(\r_1)+b\Ph(\r_2)
\eeq
for any scalars $a$ and $b$. Thus the map $\Ph$ is linear.

\item[3a.] Positivity:

This property means that $\Ph$ maps positive operators to positive operators. Assume the operator $A>0$, i.e., it has only non-negative eigenvalues, not all zero. Note that any density matrix $\r$ must be positive, and we can write $A = \sum_i \lambda_i \ketbra{i}$ where all $\lambda_i \geq 0$ (the spectral decomposition of $A$).

In order to demonstrate that $\Ph(A)>0$ it is sufficient show that $\expect{\n}{\Ph(A)}{\n}\geq 0$ for all $\ket{\n}\in\mc{H}_S$, since this means in particular that the eigenvalues of $\Ph(A)$ are all non-negative. Let $\ket{w_a}=K^\dgr_\a\ket{\n}$. Then:
\beq
\expect{\n}{\Ph(A)}{\n} = \sum_\a \expect{\n}{K_\a A K^\dgr_\a}{\n}=\sum_\a \expect{w_a}{A}{w_a}\ = \sum_{ai} \lambda_i |\bk{w_a}{i}|^2\ .
\label{eq:137}
\eeq
On the right hand side it is clear that each term in the sum is positive. Therefore $\Ph(A)>0$, and $\Ph$ itself is a positive map.

\end{enumerate}

The Kraus OSR satisfies these three properties, but does every map that satisfy the same properties have a Kraus OSR? The answer is negative. It turns out that we must modify and strengthen the positivity property into ``complete positivity".

\subsection{Complete Positivity}

The map $\Phi$ is a \emph{completely positive} (CP) map. It maps positive operators to positive operators (is ``positivity preserving"), and moreover, it can be shown that even $\Phi\otimes \mc{I}_R^{(k)}$ is positive for all $k$, where $k$ is the dimension of an ancillary Hilbert space $\mc{H}_R$, and $\mc{I}_R$ denotes the identity (super-)operator on $\mc{H}_R$. Conversely, every CP map can be represented as a Kraus OSR.

More formally, let $\mathcal{B}(\mathcal{H})$ denote the space of linear operators acting on the Hilbert space $\mathcal{H}$, i.e., $X:\mathcal{H}\mapsto \mathcal{H}$ is equivalent to $X\in \mathcal{B}(\mathcal{H})$. 
Let $A\in \mathcal{B} (\mcHS \ox \mcHR)$, where $\mcHS$ denotes the system space and $\mcHR$ is some auxiliary space with dimension $k$. Assume that $A>0$. Denote by $\mathcal{I}_R$ the identity map on $\mathcal{B} (\mcHR)$ [i.e., $\mathcal{I}_R(V)=I V I$ for all $V\in \mathcal{B} (\mcHR)$]. Also, let $\Ph\in \mathcal{B}[\mathcal{B} (\mcHS)]$, i.e., $\Phi:\mathcal{B} (\mcHS)\mapsto \mathcal{B} (\mcHS)$.

\begin{enumerate}

\item[3b.] Complete Positivity

If $(\Ph \ox \mathcal{I}_R)(A)>0$ $\forall k$, then $\Ph$ is called a completely positive (CP) map. If in addition $\Tr[\Ph(X)]=\Tr(X)$ $\forall X\in \mathcal{B} (\mcHS)$ then $\Ph$ is called a completely positive trace preserving (CPTP) map.

\end{enumerate}

Note that when $k=1$, complete positivity reduces to ordinary positivity.

It turns out that conditions $1,2,3$b are necessary and sufficient for the Kraus OSR. That is:
\begin{mytheorem}
A map $\Ph$ has a Kraus operator sum representation [i.e., $\Ph(X) = \sum_{\a} K_{\a} X K^\dagger_{\a}$ with $\sum_{\a} K^\dagger_{\a}K_{\a}=I$] iff it is trace preserving, linear, and completely positive.
\end{mytheorem}

Let us prove one direction of this theorem: that the Kraus OSR is completely positive (we already showed trace preservation and linearity). To this end, note that if $\Ph$ has a Kraus OSR then
\beq
(\Ph \ox \mathcal{I}_R)(A)=\sum_\a(K_\a \ox {I}_R)(A)(K^\dgr_\a \ox {I}_R)\ .
\eeq
To prove that $\Ph$ is CP we need to show that $\Ph \ox \mathcal{I}_R$ is positive for all $d=\dim(\mcHR)$. Indeed:
\bea
\bra{\n} (\Ph \ox \mathcal{I}_R)(A) \ket{\n} = \sum_\a \bra{\n} (K_\a \ox {I}_R)A(K^\dgr_\a \ox {I}_R) \ket{\n} =\sum_\a \expect{w_\a}{A}{w_\a}>0 \ ,
\label{eq:139}
\eea
where we defined $\ket{w_\a}=(K_\a^\dgr \ox \mathcal{I}_R) \ket{\n}$, where now $\ket{\n}\in\mc{H}_S\ox\mc{H}_R$, and we drew upon the fact that $A>0$, as in Eq.~\eqref{eq:137}.  

The key feature of the Kraus OSR that makes it a completely positive map is having the same operator ($K_\a$) on both sides. For example, something like $\sum_{\a\b} K_\a X K^\dagger_\b$ is not a CP map, and the proof of positivity as in Eq.~\eqref{eq:139} would clearly not have worked.
 
To prove the reverse direction, that all maps that satisfy conditions $1,2,3$b have a Kraus OSR, is more challenging and requires a tool known as the Choi decomposition \cite{Choi:75}.

From now on we \emph{define} a \emph{quantum map} (or quantum channel) as a map that is (1) trace preserving, (2) linear, (3) completely positive. This definition is motivated by the fact that we know that such maps have a Kraus OSR, and that the Kraus OSR arises both from the physical prescription of unitary evolution followed by partial trace, and from (non-selective) measurements.

\subsection{Positive but not Completely Positive: Transpose}

Do maps that are positive but not completely positive exist?  The answer is affirmative. The canonical example is the elementary transpose map $T$. 

Given a real basis $\{\ket{i}\}$ for $\mcHS$, the action of the transpose on the basis elements is: $T (\ketb{i}{j})=\ketb{j}{i}$ (for a real basis this is the same as Hermitian conjugation). For example, for a $2\times 2$ matrix:
\bea
&T:\lp \begin{array}{cc}
a & b\\
c & d \end{array} \rp = a \ketb{0}{0} + b \ketb{0}{1} + c \ketb{1}{0} + d \ketb{1}{1} \longmapsto a \ketb{0}{0} + b \ketb{1}{0} + c \ketb{0}{1} + d \ketb{1}{1} = 
\lp \begin{array}{cc}
a & c\\
b & d \end{array} \rp \ .
\eea

\begin{myclaim}
$T$ is a positive map.
\end{myclaim}

\begin{proof}
To prove the claim it suffices to show that the eigenvalues of $X$ and $T(X)$ are the same for any $X\in\mcHS$ [since then in particular their sign is preserved, so if $X>0$ then also $T(X)>0$]. The eigenvalues of $X$ are found by solving for the roots of its characteristic polynomial: $p(X) = \det(X-\lambda I)$. Now, since the determinant is invariant under elementary row and column operations, it is invariant under transposition. Therefore $\det(X-\lambda I) = \det[T(X-\lambda I)] = \det[T(X)-\lambda T(I)] = \det[T(X)-\lambda I]$, i.e., $p(X) = p[T(X)]$.
\end{proof}

Is $T$ also completely positive? To test this we need to check if any extension $T^\text{p} \equiv T \ox \mathcal{I}_R$ of $T$ is also positive. This extension is called the \emph{partial transpose}, and its action on any basis element of $\mc{B}(\mcHS\ox \mcHR)$ is as follows:
\beq
T^\text{p} (\ketb{i}{j}\ox \ketb{\m}{\n})=\ketb{j}{i}\ox \ketb{\m}{\n}\ .
\eeq
To prove that $T$ is \emph{not} a CP map, it suffices to find a counterexample. Indeed, consider the pure state $\r=\ketb{\psi}{\psi}$, where $\ket{\psi}=\frac{1}{\sqrt{2}}(\ket{0}_S \ket{0}_R + \ket{1}_S \ket{1}_R)$.  Then:
\bes
\begin{align}
T^\text{p}(\r)&=\frac{1}{2} (T \ox \mathcal{I})[\ketbra{0_S 0_R}+\ketb{00}{11}+\ketb{11}{00}+\ketbra{11}]\\
&=\frac{1}{2}(\ketbra{00}+\ketb{10}{01}+\ketb{01}{10}+\ketbra{11}) = 
\frac{1}{2} \lp \begin{array}{cccc}
1&0&0&0\\
0&0&1&0\\
0&1&0&0\\
0&0&0&1 \end{array} \rp \ .
\end{align}
\ees
The eigenvalues of this matrix are $(\frac{1}{2},\frac{1}{2},\frac{1}{2},-\frac{1}{2})$, and the existence of a negative eigenvalue shows that $T$ is not a CP map, since $T^\text{p}\not >0$.  Therefore $T$ does not have a Kraus OSR, and is not a quantum map. Note furthermore that this means that a maximally entangled two-qubit state has a negative partial transpose. This observation motivates the study of the partial transpose as a tool for entanglement testing.

\subsection{Partial Transpose as a Test for Separability/Entanglement: the PPT criterion}
  
Consider a separable (thus by definition unentangled) state $\r=\sum_i p_i \r_i^A \ox \r_i^B$, where the $p_i$ are probabilities and the $\r_i^A$ and $\r_i^B$ are quantum states (positive, normalized). The state $\r$ obviously arises from the mixed state ensemble $\{\r_i^A \ox \r_i^B, p_i\}$, in which every element is a tensor product state. Mixing such states classically does not generate any entanglement between $A$ and $B$, hence the definition. 

Applying the partial transpose yields:
\bea
T^\text{p}(\r) = (T \ox \mathcal{I})(\r)=\sum_i p_i T(\r_i^A) \ox \r_i^B
=\sum_i p_i \s_i^A \ox \r_i^B\ .
\eea
Since the transpose does not change the eigenvalues, $\s_i^A\equiv T(\r_i^A)$ is also a valid quantum state, and hence $T^\text{p}(\r)$ is another separable quantum state. In particular, this shows that every separable state has a positive partial transpose (PPT). In other words, separability implies PPT. Conversely, a negative partial transpose (NPT) implies entanglement. This means that PPT is a \emph{necessary} condition for separability. 

Is PPT also sufficient for separability? It turns out that this is the case only for the $2 \times 2$ (two qubits) or $2 \times 3$ (qubit and qutrit) cases. I.e., only in these cases a state is separable iff it has a positive partial transpose (PPT) (conversely, is entangled iff it has a NPT) \cite{Peres:1996aa,Horodecki:1996aa}. Indeed, we saw in the previous subsection that a (maximally) entangled state has NPT.

In higher dimensions the PPT criterion it is still necessary but no longer sufficient. In such higher dimensions there are examples of so-called ``bound-entangled" states that have PPT but are not separable \cite{Horodecki:1999aa}.

As an example of the use of the PPT criterion consider the Werner states:
\beq
\rho = p |\Psi^-\rangle \langle \Psi^-| + (1-p) \frac{I}{4} \, 
\eeq
where $|\Psi^-\rangle$ is a maximally entangled singlet state: $\ket{\Psi^-} = (\ket{01}-\ket{10})/\sqrt{2}$. This represents a family of quantum states parametrized by the probability $p$ of being in the singlet state as opposed to the maximally mixed state.

Its density matrix in the standard basis is
\beq
\rho = \frac{1}{4}\begin{pmatrix}
1-p & 0 & 0 & 0\\
0 & p+1 & -2p & 0\\
0 & -2p & p+1 & 0 \\
0 & 0 & 0 & 1-p\end{pmatrix}\ ,
\eeq
and the partial transpose
\beq
T^\text{p}(\rho) = \frac{1}{4}\begin{pmatrix}
1-p & 0 & 0 & -2p\\
0 & p+1 & 0 & 0\\
0 & 0 & p+1 & 0 \\
-2p & 0 & 0 & 1-p\end{pmatrix} \ .
\eeq
The eigenvalues of this matrix are $(1-3p)/4$ and (threefold) $(1+p)/4$. Therefore, the state is entangled for $p > 1/3$ and separable for $p\leq 1/3$ (for $p=1/3$ all eigenvalues are non-negative so PPT).

\subsection{Kraus OSR as a composition of CP maps}

The Kraus OSR is a actually a composition of three other maps:
\beq
\Phi = \Tr_B \circ \mathcal{U} \circ \mathcal{A},
\eeq
where (i) $\mathcal{A}$ is the ``assignment map" which associates to every initial system state $\r_S(0)$ a fixed bath state $\r_B(0)$, i.e., $\mathcal{A}[\r_S(0)] = \r_S(0)\otimes\r_B(0)$; (ii) $\mathcal{U}$ is the unitary evolution superoperator, i.e., $\mathcal{U}[X] = UXU^\dagger$; (iii) $\Tr_B$ is the usual partial trace operator. This is depicted in Fig.~\ref{fig:commutative-diagram}.

\begin{figure}[h]
\begin{equation*}
\begin{tikzcd}[row sep=huge]
\rho_S(0) \otimes \rho_B(0) \arrow[r,"\mathcal{U}"] &
U \left[ \rho_S(0) \otimes \rho_B(0) \right] U^\dagger \arrow[d,"\text{Tr}_B"]\\
\rho_S (0) \arrow[u,"\mathcal{A}"] \arrow[r,dashed,"\Phi"] & \rho_S (t)
\end{tikzcd}
\end{equation*}
	\caption{A commutative diagram showing that the quantum map $\Ph$ can be viewed as a composition of three maps.}
	\label{fig:commutative-diagram}
\end{figure}

Let us show that each of these three maps is, in turn, CP.

\subsubsection{The assignment map is CP}

The map $\mc{A}$ is from $\mc{D}(\mc{H}_S)$ to $\mc{D}(\mc{H}_S\ox\mc{H}_B)$. To prove that it is CP we need to consider positive operators $A \in \mc{D}(\mc{H}_S\ox\mc{H}_R)$. Thus, writing $A = \sum_{ir} \lambda_{ir}\ketbsub{i}{i}{S} \ox \ketbsub{r}{r}{R}$ and $\r_B = \sum_{\mu} \lambda_{\mu} \ketbsub{\mu}{\mu}{B}$, with $\lambda_{ir},\lambda_{\mu} \geq 0$:
\bes
\begin{align}
\bra{v}(\mc{A}\ox\mc{I}_R)(A)\ket{v} &= \sum_{ir} \lambda_{ir}\bra{v}[\ketbsub{i}{i}{S} \ox \r_B\ox \ketbsub{r}{r}{R}]\ket{v}  \\
&= \sum_{ir\mu} \lambda_{ir}\lambda_{\mu} \bra{v}[\ketbsub{i}{i}{S} \ox \ketbsub{\mu}{\mu}{B}\ox \ketbsub{r}{r}{R}]\ket{v}  \\
&= \sum_{ir\mu} \lambda_{ir}\lambda_{\mu} |\bk{v}{i\mu r}|^2 \geq 0 .
\end{align}
\ees

\subsubsection{The unitary map is CP}
This is obvious since the unitary map is a special case of a Kraus OSR having such a single Kraus operator $U$.

\subsubsection{The partial trace is CP}

To demonstrate that the partial trace, $\Tr_B:\r_{SB} \mapsto \r_S$, is CP, we can perform a direct calculation like we did for the assignment map. However, instead we can also directly demonstrate that it has a Kraus OSR (since this is a sufficient condition for CPness).

Consider the following explicit Kraus operators for the partial trace map: 
\beq
K_\a = I_S\ox \bra{\a}\ ,
\label{eq:K-Tr_B}
\eeq
where $\{\ket{\a}\}$ denotes the elements of some chosen basis for the bath Hilbert space. This choice is motivated by the fact that the partial trace leaves the system alone but sandwiches the bath between basis states.

Applying the map $\Ph = \{K_\a\}$ to an arbitrary system-bath state $\r_{SB}= \sum_{ij\m\n} \lambda_{ij\m\n} \ketb{i}{j} \ox \ketb{\m}{\n}$ written in the same basis for the bath, and noting that $\Tr_B(\r_{SB}) = \sum_{ij\m\n\a} \lambda_{ij\m\n} \ketb{i}{j} \bk{\a}{\m}\bk{\n}{\a} = \sum_{ij\a} \lambda_{ij\a\a}\ketb{i}{j} $, we find the following:
\beq
\Ph(\r_{SB}) = \sum_\a K_\a\r_{SB} K_\a^\dgr =\sum_\a I_S\ox \bra{\a}\left( \sum_{ij\m\n} \lambda_{ij\m\n} \ketb{i}{j} \ox \ketb{\m}{\n}\right) I_S\ox \ket{\a} = \sum_{ij\a} \lambda_{ij\a\a}\ketb{i}{j} = \Tr_B(\r_{SB})\ , 
\eeq
as desired. Thus, the partial trace has Kraus elements as given in Eq.~\eqref{eq:K-Tr_B}, and is CP.



\subsection{OSR for a general initial condition?}
What would happen if we were to relax the initial condition? Will we still get a CP map?

\subsubsection{General initial states}

Using a general orthonormal basis for the joint Hilbert space we can always write
\beq 
\r (0) = \sum_{ij \a\b} \lambda_{ij \a\b} \ketb{i}{j}\otimes \ketb{\a}{\b}.
\eeq
The corresponding initial state of the system is 
\beq
\r_S(0) = \Tr_B[\r(0)] = \sum_{ij \a} \lambda_{ij \a\a} \ketb{i}{j}\ .
\label{eq:rhoS0}
\eeq

If we go through the same steps as in the derivation of the Kraus OSR, we have, with $\{\ket{\mu}\}$ now representing the same bath basis as $\{\ket{\a}\}$:
\bes
\begin{align}
\r_S (t) &= \sum_\mu \braket{\mu|U(t) \r(0) U^\dagger (t)|\mu} \\
         &= \sum_{\mu ij\a \beta}\lambda_{ij \a\beta} \braket{\mu|U(t)  \ketb{i}{j}\otimes \ketb{\a}{\beta}U^\dagger(t)|\mu} \\
         &= \sum_{\mu ij\a \beta} \; \bra{\mu}U(t)\ket{\a}\lambda_{ij\a \beta} \ketb{i}{j} \; \braket{\beta|U^\dagger (t)|\mu} \\
         \label{108d}
         &=  \sum_{\mu,\a} \; \bra{\mu}U(t)\ket{\a}\left(\sum_{ij} \lambda_{ij\a \a} \ketb{i}{j}\right) \; \braket{\a|U^\dagger (t)|\mu}  + \sum_{\mu ij,\a\neq\beta} \; \bra{\mu}U(t)\ket{\a}\left(\sum_{ij} \lambda_{ij\a \b} \ketb{i}{j}\right) \; \braket{\beta|U^\dagger (t)|\mu} \ .
         \end{align}
\ees
The first summand in Eq.~\eqref{108d} has Kraus operators $\bra{\mu}U(t)\ket{\a}$
and may look fine. However, because of the sum over $\alpha$ we cannot factor out $\r_S(0)$.
Moreove, the second summand in Eq.~\eqref{108d} in addition involves off-diagonal terms $\lambda_{ij\a \b}$ that do not appear in $\r_S(0)$ [Eq.~\eqref{eq:rhoS0}]. Clearly, we cannot factor out $\r_S(0)$, so we do not even get a map from $\r_S(0)$ to $\r_S(t)$.

\subsubsection{Separable states}
What if we consider {separable} states, 
\beq
\r (0) = \sum_{i} p_{i} \rho_S^i\otimes \rho_B^i
\eeq
where $\r_S^i$ and $\r_B^i$ are themselves states of the system and bath? For such a state the initial system state is $\r_S (0) = \Tr_B[\r(0)] = \sum_{i} p_{i} \r_S^i$. Let's decompose each bath state as 
\beq
\r_B^i = \sum_{\nu_i} \lambda_{\nu_i} \ketbra{\nu_i} 
\eeq
and try again:
\bes
\begin{align}
\r_S (t) &=  \sum_{\mu}\sum_i\sum_{\nu_i}p_i\lambda_{\nu_i} \braket{\mu|U(t)   \rho_S^i\otimes \ketbra{\nu_i} U^\dagger(t)|\mu} \\
         &= \sum_{\mu}\sum_i\sum_{\nu_i}\lambda_{\nu_i} \braket{\mu|U(t)\ket{\nu_i} p_i  \rho_S^i  \bra{\nu_i}U^\dagger(t)|\mu} .
\end{align}
\ees
We can move the sum over $i$ inside if we first assume that all $\rho_B^i$ commute, i.e., are diagonal in the same basis so that ${\nu_i} = {\nu}$ $\forall i$, for then $\r_B^i(0) = \sum_\nu \lambda_{\nu}^i \ketbra{\nu}$ and hence 
\begin{align}
\r_S (t) = \sum_{\mu,\nu} \braket{\mu|U(t)\ket{\nu} \sum_i  \lambda^i_{\nu} p_i \rho_S^i  \bra{\nu}U^\dagger(t)|\mu} \ ,
\end{align}
but this still doesn't allow us to extract the initial system state $\sum_{i} p_{i} \r_S^i$. To accomplish this we may moreover assume that eigenvalues are the same, i.e., $\lambda^i_{\nu} = \lambda_{\nu}$ $\forall i$. If we do so we find $\r_S (t) = \sum_{\mu,\nu} \lambda_{\nu}\braket{\mu|U(t)\ket{\nu} \sum_i p_i \rho_S^i  \bra{\nu}U^\dagger(t)|\mu}$, and this involves a map acting on $\r_S(0) = \sum_{i} p_{i} \r_S^i$ as desired, but we haven't gained anything: this is the case if $\rho_B^i = \rho_B$ $\forall i$, i.e., we're back to Eq.~\eqref{eq:decoupled} again.

\subsection{The quantum discord perspective}
\label{sec:discord}

\subsubsection{Quantum Discord}
In classical information theory there are two equivalent ways to define the mutual information between two random variables $X$ and $Y$:
\bes
\begin{align}
I(Y:X) &= H(Y) + H(X) - H(X,Y) \\
J(Y:X) &= H(Y) - H(Y|X) \ ,
\end{align}
\ees
where $H(X) = - \sum_i p_i \log(p_i)$ is the Shannon entropy associated with $X$, with $p_i = \Pr(x_i)$ being the probability of $X$ assuming the value $x_i$. The quantity $H(X,Y)$ is the entropy of the joint distribution, and $H(Y|X)$ is the entropy of $Y$ conditioned on $X$. The equivalence follows directly from Bayes' rule [the joint probability satisfies $p(y,x) = p(y|x)p(x)$, where $p(y|x)$ is the conditional probability], which implies that $H(X,Y) = H(Y|X) + H(X)$, and hence that $I(Y:X) =J(Y:X)$.

In the quantum case, measuring system $X$ generally affects system $Y$ if the joint state $\rho_{XY}$ is correlated, so the asymmetry inherent in the second expression $J(Y|X)$ means that there is the potential for a different outcome from the symmetric first expression $I(Y|X)$. This observation forms the basis for the definition of the quantum discord, $I_Q(Y:X) - J_Q(Y:X)$.
Let us thus define the quantum mutual information expressions $I_Q(Y:X)$ and $J_Q(Y:X)$.

First, we need the quantum von Neumann entropy associated with a state $\r$:
\beq
S(\r) = -\Tr [\r \log(\r)] \ .
\eeq
Then
\beq
I_Q(Y:X) = S(\r_Y) + S(\r_X) - S(\r_{XY}) \ ,
\eeq
where $\rho_{XY}$ is the total state of systems $X$ and $Y$, $\r_Y = \Tr_X\r_{XY}$, and $\r_X = \Tr_Y\r_{XY}$.
The second mutual information $J_Q$ arises from first measuring $X$. Assume that this is done using a projective measurement with projectors $\{\Pi_i\}$, acting only on $X$. Then the post-measurement state obtained in case $i$ is $\r_{Y|\Pi_i} \equiv \Pi_i \r_{XY}\Pi_i/p_i$, where $p_i = \Tr[\Pi_i \r_{XY}] $ is the probability of case $i$. Let us associate an entropy to this state: $S(\r_{Y|\Pi_i})$. The entropy conditioned non-selectively on the entire measurement is $S(Y|\{\Pi_i\}) = \sum_i p_i S(\r_{Y|\Pi_i})$, and the conditional entropy is the minimum over all possible measurements, since we're interested in maximizing the mutual information: $S(Y|X) = \min_{\{\Pi_i\}} S(Y|\{\Pi_i\})$. Explicitly:
\beq
S(Y|X) = \min_{\{\Pi_i\}} \sum_i p_i S(\Pi_i \r_{XY}\Pi_i/p_i)\ , \quad p_i = \Tr[\Pi_i \r_{XY}] \ .
\eeq
With this, we are ready to define the second quantum mutual information:
\beq
J_Q(Y:X) = S(Y) - S(Y|X) \ .
\eeq
Generally, $J_Q(Y:X) \neq I_Q(Y:X)$. We thus define the quantum discord \cite{Ollivier:01} as
\beq
D(\r_{XY}) = I_Q(Y:X) - J_Q(Y:X)\ .
\eeq 
$D(\r_{XY})=0$ only for zero-discord states (by definition), which are states that have no quantum correlations at all. Note that separable states can have non-zero discord [i.e., $J_Q(Y:X) \neq I_Q(Y:X)$], which means that they have some quantum correlations despite being a convex combination of product states. However, it is not hard to show that 
a special class of separable states does have zero discord. Such states are known as zero-discord states, and they are of the form
\beq
\r_{SB}(0) = \sum_i p_i \Pi_i \otimes \rho_B^i\ ,
\label{eq:ZD}
\eeq
where the $\Pi_i$ are projectors, i.e., $\Pi_i \Pi_j = \delta_{ij}\Pi_i$. This initial state would be the result of a non-selective projective measurement of the system with measurement operators $\{\Pi_i\}$ (you can easily check that the state is invariant under a non-selective projective measurement with the same set of measurement operators, which is the property we expect from the state after a first projective measurement; see subsection~\ref{sec:non-selective} for non-selective measurements).

\subsubsection{Zero discord initial states and CP maps}
It turns out that zero-discord states do allow us to generalize the assumption of a factorizable initial state [Eq.~\eqref{eq:decoupled}] in the derivation of the Kraus OSR~\cite{Rodriguez:08}. Let's assume that the initial state is of the form given in Eq.~\eqref{eq:ZD}. Thus the system state becomes
\bes
\begin{align}
\r_S(t) &= \sum_{\mu} \bra{\mu} U(t) \sum_i p_i \Pi_i \otimes \rho_B^i U^\dagger (t)\ket{\mu} \\
&=  \sum_{\mu,i} \bra{\mu} U(t) \sqrt{\rho_B^i}  p_i \Pi_i \otimes  \sqrt{\rho_B^i} U^\dagger (t)\ket{\mu} \\
& \sum_{\mu\nu,i} \bra{\mu} U(t) \sqrt{\rho_B^i}\ket{\nu}  p_i \Pi_i \otimes  \bra{\nu}\sqrt{\rho_B^i} U^\dagger (t)\ket{\mu} \ ,
\end{align}
\ees
where we used the fact that ${\rho_B^i}$ is a positive operator to take its square root, and inserted a bath identity operator $\sum_\nu \ketbra{\nu}$ in the last line. Let's define 
\beq
D_{i\mu\nu} \equiv \bra{\mu} U(t) \sqrt{\rho_B^i}\ket{\nu}\ ,
\eeq
and note that this is a system-only operator. Now, we can always write  $D_{i\mu\nu} = \sum_{m} D_{m\mu\nu} \delta_{im}$. Inserting this into the last equation we have
\bes
\begin{align}
\r_S(t) &= \sum_{\mu\nu,i} p_i   \left(\sum_{m} D_{m\mu\nu} \delta_{im}\right)  \Pi_i   \left(\sum_{n} D^\dagger_{n\mu\nu} \delta_{in}\right) \\
&= \sum_{\mu\nu,i} p_i   \left(\sum_{m} D_{m\mu\nu} \delta_{im} \Pi_i\right)     \left(\sum_{n} \Pi_i \delta_{in} D^\dagger_{n\mu\nu} \right) \ ,
\end{align}
\ees
where we used $\Pi_i^2 = \Pi_i$.
Next, note that $\delta_{im} \Pi_i = \Pi_m \Pi_i$ and $\delta_{in} \Pi_i = \Pi_i \Pi_n$, which allows us to replace the $\delta$'s by $\Pi$'s:
\begin{align}
\r_S(t) = \sum_{\mu\nu,i,m,n} p_i  D_{m\mu\nu} \Pi_m   \Pi_i   \Pi_n D^\dagger_{n\mu\nu} \ .
\end{align}
We can now move the sum over $i$ inside so it is performed first. Thus, we have 
\begin{align}
\r_S(t) = \sum_{\mu\nu,m,n} D_{m\mu\nu} \Pi_m \left(\sum_i p_i    \Pi_i  \right) \Pi_n D^\dagger_{n\mu\nu} \ ,
\end{align}
and using Eq.~\eqref{eq:ZD} we recognize the middle term as the initial system state: $\r_S(0) = \sum_i p_i    \Pi_i$. We can also define new Kraus operators as $K_{\mu\nu} = \sum_m D_{m\mu\nu}\Pi_m$. This then gives us a proper Kraus OSR:
\begin{align}
\r_S(t) = \sum_{\mu\nu} K_{\mu\nu} \r_S(0) K_{\mu\nu}^\dagger \ .
\end{align}

It turns out that there are also discordant states that give rise to CP maps, and even entangled states. Read about generalizations in Refs.~\cite{Buscemi:2013,Dominy:14,Dominy:2016xy}.

\subsection{Equivalence of Quantum Maps}

Given two quantum maps, a natural question is under which conditions they are equivalent.
As an example, consider the two single-qubit quantum maps defined by the following two sets of Kraus operators: ${\Ph}=\{K_{0}=\frac{1}{\sqrt{2}}I, K_{1}=\frac{1}{\sqrt{2}}\s^{z}\}$ and ${\Psi}=\{L_{0}=\ketbra{0}, L_{1}=\ketbra{1}\}$. Note that $\Ph$ can be interpreted as the map the flips the phase or leaves the state alone with equal probability, while ${\Psi}$ can be interpreted as a non-selective measurement in the $\s^z$ basis. Thus, a priori it seems that the two maps describe very different physical processes. Nevertheless, it's easy to show that the two maps are identical,\footnote{Simply write $\r$ as a general $2\times 2$ matrix and note that $Z\r Z$ flips the sign of the off-diagonal elements, so that both $\Ph$ and $\Psi$ erase $\r$'s off-diagonal elements.} i.e., $\forall \r$
\bea
\frac{I}{\sqrt{2}}\r \frac{I}{\sqrt{2}}+ \frac{\s^{z}}{\sqrt{2}} \r \frac{\s^{z}}{\sqrt{2}}=
\ketbra{0} \r \ketbra{0}+ \ketbra{1} \r \ketbra{1}.
\eea

\subsubsection{General conditions for equivalence}

What is the general condition such that two maps are equivalent? The following theorem provides the answer:

\begin{thm}
\label{th:Kraus-equiv}
Consider the maps produced by the following two sets of Kraus operators ${\Ph}=\{K_\a\}$, ${\Psi}=\{L_\b\}$:  $\r'=\sum{K_{\a} \r K_{\a}^{\dgr}}$ and $\r''=\sum{L_{\b} \r L_{\b}^{\dgr}}$. Then
\bea
\forall \r:\r'=\r'' 	\Longleftrightarrow \exists \text{ a unitary operator},  u: \textrm{s.t. } K_{\a}=\sum_\b {u_{\a \b} L_{\b}}.
\label{eq:Kraus-equiv}
\eea
\end{thm}
\begin{proof}
Here we prove the ``if" direction, i.e., assume that such a unitary exists; then
\bes
\begin{align}
\r'&=\sum_{\a}{(\sum_{\b}{u_{\a \b} L_{\b}})\r(\sum_{\b'}{u_{\a\b'}^{*} L^{\dgr}_{\b'}})}\\
&=\sum_{\b\b'}{ L_{\b}\r L^{\dgr}_{\b'}\sum_{\a}[u^\dgr]_{\b' \a}{[u]_{\a \b} }} =\sum_{\b\b'}L_{\b}\r L^{\dgr}_{\b'}[u^\dgr u]_{\b'\b} \\
&=\sum_{\b\b'}{ L_{\b}\r L^{\dgr}_{\b'} \d_{\b \b'}}\\
&=\sum_{\b}{ L_{\b}\r L^{\dgr}_{\b}}=\r'' .
\end{align} 
\ees
\end{proof}

In the example above the relation between the operators is:
\begin{align}
K_{0}=\frac{1}{\sqrt{2}}(L_{0}+L_{1})\ , \qquad K_{1}=\frac{1}{\sqrt{2}}(L_{0}-L_{1}),
\end{align}
so the unitary is $u= \frac{1}{\sqrt{2}}\left( \begin{array}{cc}
1 & 1\\
1 & -1 \end{array} \right)$.

\subsubsection{Physical origin of the equivalence}

Where does this unitary equivalence between Kraus operators come from? To see this intuitively, note that in deriving the Kraus operators, after evolving with a unitary operator acting on both the system and the bath, we trace out the bath, so the Kraus operators should remain equivalent under the change of basis of the bath. 
Let us show that this ``gauge freedom" gives rise to the unitary equivalence between different sets of Kraus operators. As we shall see, we need to be a bit careful in accounting for the presence of the square-root of the eigenvalue of the bath density matrix in the definition of the Kraus operators.

Let us write Eq.~\eqref{eq:Kraus-equiv} as
\beq
K_{\m\n}(t) = \sum_{\h\x} u_{\m\n\h\x}L_{\h\x}(t) ,
\eeq
where we have let $\a =(\m\n)$ and $\b = (\h\x)$.
In terms of the explicit form of the Kraus operators this becomes
\beq
\sqrt{\lambda_\n}\bra{\m}U\ket{\n} = \sum_{\h\x}u_{\m\n\h\x}\sqrt{\lambda_\x}\bra{\h}U\ket{\x} .
\label{eq:168}
\eeq
Let us now assume that
\beq
u_{\m\n\h\x} \equiv v_{\m\h}w_{\x\n} = \bra{\m}v\ket{\h}\bra{\x}w\ket{\n} ,
\label{u-4index}
\eeq 
where $v$ and $w$ are both unitary. We can then show that the matrix $u$ is unitary:
\bes
\begin{align}
[u^\dagger u]_{\a\b} &= \sum_{\g} [u^\dagger]_{\a\g}[u]_{\g\b} = \sum_{\g}u^*_{\g\a}u_{\g\b} = \sum_{\m'\n'} u^*_{\m'\n'\m\n}u_{\m'\n'\h\x} = \sum_{\m'\n'} (v_{\m'\m}w_{\n\n'})^* v_{\m'\h}w_{\x\n'}\\
& = \sum_{\m'} v^*_{\m'\m}v_{\m'\h} \sum_{\n'}w^*_{\n\n'}w_{\x\n'} = \delta_{\m\h}\delta_{\x\n} = \delta_{\a\b} ,
\end{align}
\label{eq:u-unitary}
\ees
where we used the unitarity of $v$ and $w$ in the penultimate equality.

Plugging this expression for $u_{\m\n\h\x}$ into Eq.~\eqref{eq:168} gives:
\bes
\begin{align}
\bra{\m}U\sqrt{\lambda_\n}\ket{\n} &= \sum_{\h\x}\sqrt{\lambda_\x}\bra{\m}v\ketbra{\h}U\ketbra{\x}w\ket{\n} \\
&= \bra{\m}v[\sum_\h \ketbra{\h}] U [\sum_\x \sqrt{\lambda_\x}\ketbra{\x}] w\ket{\n}\\
&= (\bra{\m}v) U (\sqrt{\r_B} w\ket{\n}) ,
\end{align}
\ees
i.e., the gauge freedom giving rise to the unitary equivalence between to sets of Kraus operators is:
\bes
\begin{align}
\label{eq:171a}
\bra{\m} &\mapsto \bra{\m}v \\
\label{eq:171b}
\sqrt{\lambda_\n}\ket{\n} &\mapsto \sqrt{\r_B} w\ket{\n} .
\end{align}
\ees
Eq.~\eqref{eq:171a} simply expresses the freedom to apply a unitary transformation on the bath basis vectors before taking the partial trace (which we did by sandwiching inside $\bra{\m} \cdots \ket{\m}$). Eq.~\eqref{eq:171b} tells us that we can also apply a second unitary transformation on the eigenstates of $\r_B$ (i.e., $\ket{\n} \mapsto w\ket{\n}$), but that in general we should also replace the eigenvalue term $\sqrt{\lambda_\n}$ by $\sqrt{\r_B}$. 
To understand the latter, note that in deriving the Kraus OSR we can also proceed as follows:
\bes
\label{eq:169}
\begin{align}
\r'_S &= \Tr_B[U\r_S\ox\r_B U^\dagger] = \sum_{\mu} \bra{\mu} \left[ U \r_S\ox\left(\sqrt{\r_B} \sum_\nu \ketbra{\nu}\sqrt{\r_B}\right)U^\dag \right] \ket{\mu} \\
&= \sum_{\mu\nu}\left(\bra{\mu}U\sqrt{\r_B}\ket{\nu}\right)\r_S\left(\bra{\nu}\sqrt{\r_B}U^\dag\ket{\mu}\right) ,
\end{align}
\ees
which means that the Kraus operators we derived originally by using $\r_B$'s spectral decomposition, $\bra{\m}U\sqrt{\lambda_\n}\ket{\n}$, are equivalent to Kraus operators of the form $\bra{\mu}U\sqrt{\r_B}\ket{\nu}$. In other words, the spectral decomposition was just one of infinitely many equivalent ways to decompose $\r_B$. We recover the spectral decomposition if we choose the basis $\{\ket{\nu}\}$ in Eq.~\eqref{eq:169} as the eigenbasis of $\r_B$.



\section{Quantum Maps of a Qubit}

In this section, by focusing on the case of one qubit, we will develop a geometric picture of the action of quantum maps. The main tool that will allow us to do this is the Bloch sphere representation.  

Recall that the density matrix of a qubit may be written as $\r=\frac{1}{2}(I+\vec{v}\cdot\vec{\s})$ where $\vec{\s}=(\s_x,\s_y,\s_z)$ and $\vec{v}=(v_x,v_y,v_z)\in\R^3$ is the Bloch vector. In this way, a single-qubit state may be thought of as a point in or on the unit sphere in $\R^3$---the Bloch sphere.  States with $\|\vec{v}\|=1$ lie on the surface of the sphere and correspond to pure states of the form $\r=\ketbra{\psi}$.  Points on the interior of the sphere correspond to mixed states with purity $P=\Tr[\r^2]<1$.

\subsection{Transformation of the Bloch Vector}

What happens when a quantum map acts on a single qubit? As a map of the density matrix, $\Ph:\r\mapsto\r'$. At the same time $\r'$ must be expressible in terms of a new Bloch vector $\vec{v}'$, where $\r'=\frac{1}{2}(I+\vec{v}'\cdot\vec{\s})$. We shall show that $\r\mapsto\r'$ is equivalent to mapping the Bloch vector
\beq
\vec{v}\mapsto\vec{v}'=M\vec{v}+\vec{c}
\label{eq:Phi(v)}
\eeq
for some real $3\times3$ matrix $M$ and a vector $\vec{c}\in\R^3$. This is an \emph{affine transformation}. Before proving Eq.~\eqref{eq:Phi(v)}, let us decompose $M$ in a way that will reveal more of the geometric aspects of the transformation.

Recall the polar decomposition, which allows us to write any square matrix $A$ as $A=U|A|$, where $U$ is a unitary matrix and $|A|\equiv\sqrt{A^\dagger A}$ is Hermitian (since clearly its eigenvalues are real), a generalization of the polar representation of a complex number $z=e^{i\theta}|z|$. If $A$ is a real matrix, $U$ becomes real-unitary, i.e., orthogonal, and $|A|$ becomes real-Hermitian, i.e., symmetric. So, for our $3\times3$ real matrix $M$ we can write $M=OS$, for orthogonal $O$ and symmetric $S=\sqrt{M^\dgr M}$. $S$ causes deformation by scaling along the directions of the eigenvectors by a factor of the corresponding eigenvalues. $O$ is a rotation matrix.  Now we may interpret the action of a quantum map on a qubit state as mapping the Bloch vector according to
\beq
\vec{v}\mapsto\vec{v}'=OS\vec{v} + \vec{c},
\eeq
as a shift by $\vec{c}$, a deformation by $S$ and a rotation by $O$.  Because the Bloch sphere represents the set of possible Bloch vectors, we may view the Kraus map acting on a qubit as a transformation of the Bloch sphere that displaces its center by $\vec{c}$ and turns the sphere into an angled ellipsoid.

To prove Eq.~\eqref{eq:Phi(v)}, we plug the Bloch vector representation of $\r$ into the quantum map:
\beq
\r' =  \sum_\a K_\a \r K_\a^\dgr = \frac{1}{2}\sum_\a K_\a(I+\vec{v}\cdot\vec{\s})K_\a^\dgr = \frac{1}{2}(\sum_\a K_\a K_\a^\dgr +\sum_{\a j} v_j K_\a \s_j K_\a^\dgr)\ .
\label{eq:160}
\eeq
To isolate the components of $\vec{v}'$ we multiply both sides by $\s_i$ and take the trace, while remembering that the Pauli matrices are all traceless and satisfy Eq.~\eqref{eq:tracePaulis}. Thus, Eq.~\eqref{eq:160} becomes
\beq
\Tr(\r' \s_i) =  \frac{1}{2}[\sum_\a \Tr(K_\a K_\a^\dgr \s_i) +\sum_{\a j} v_j \Tr(K_\a \s_j K_\a^\dgr \s_i)]\ ,
\label{eq:160a}
\eeq
On the other hand, using $\r'=\frac{1}{2}(I+\vec{v}'\cdot\vec{\s})$ and Eq.~\eqref{eq:tracePaulis} again:
\beq
\Tr(\r' \s_i) = \frac{1}{2}[\Tr(\s_i)+\sum_j{v}'_j \Tr(\s_j{\s_i})] = 0 + v'_i\ .
\label{eq:161}
\eeq
Equating Eqs.~\eqref{eq:160a} and~\eqref{eq:161} we thus have 
\beq
v'_i = c_i + \sum_j M_{ij} v_j\ ,
\eeq 
where
\bes
\label{eq:M-c}
\begin{align}
M_{ij} &= \frac{1}{2}\sum_\a \Tr( \s_i K_\a \s_j K_\a^\dgr) \\
c_i &= \frac{1}{2}\sum_\a \Tr( \s_i K_\a K_\a^\dgr)\ .
\end{align}
\ees
Moreover, using the Hermiticity of the Pauli matrices and properties of the trace [Eq.~\eqref{eq:trace-equalities}]:
\beq
M_{ij}^* = \frac{1}{2}\sum_\a \Tr( \s_i K_\a \s_j K_\a^\dgr)^\dgr = \frac{1}{2}\sum_\a \Tr( K_\a \s_j K_\a^\dgr \s_i) = \frac{1}{2}\sum_\a \Tr( \s_i K_\a \s_j K_\a^\dgr) = M_{ij}\ , 
\eeq
i.e., $M$ is real. Likewise, 
\beq
c_i^* = \frac{1}{2}\sum_\a \Tr(\s_i K_\a K_\a^\dgr)^\dgr = \frac{1}{2}\sum_\a \Tr(K_\a K_\a^\dgr \s_i ) = \frac{1}{2}\sum_\a \Tr(\s_i K_\a K_\a^\dgr) = c_i\ ,
\eeq 
so $\vec{c}\in\mathbb{R}^3$. This proves Eq.~\eqref{eq:Phi(v)}.

\subsection{Unital Quantum Maps}

Returning temporarily to the general (beyond a single qubit) case, a quantum map is said to be unital if it maps the identity operator to itself, i.e.:
\begin{mydefinition}
$\Ph$ is unital if $\Ph(I) = I$. Otherwise it is non-unital.
\end{mydefinition}

Since a quantum map always has a Kraus OSR, we find that unital quantum maps satisfy
\beq
\sum_\a K_\a K_\a^\dgr = I \ ,
\label{eq:unital}
\eeq
in addition to the trace-preservation constraint $\sum_\a K_\a^\dgr K_\a = I$.

Note that if $\Ph$ is unital, so that Eq.~\eqref{eq:unital} holds, then
\beq
c_i = \frac{1}{2} \Tr( \s_i \sum_\a K_\a K_\a^\dgr) = \frac{1}{2} \Tr( \s_i ) = 0 \qquad (\text{unital case})\ .
\label{eq:c-unital}
\eeq
Conversely, if $\Ph$ is non-unital, then $\vec{c}\neq\vec{0}$. 

Note that, as is clear from Eq.~\eqref{eq:161}, $M$ is associated purely with the transformation of $\vec{v}\cdot\vec{\s}$ under the map, while $\vec{c}$ is associated purely with the transformation of $I$ under the map. This observation will help us read off $M$ and $\vec{c}$ in the examples we study below.

\subsection{The Phase Damping Map}
\label{sec:PD}

The phase damping map is:
\beq
\Ph(\r') = p\r+(1-p)Z\r Z\ ,
\eeq
where $Z\equiv\s_z$, so the Kraus operators are $K_0=\sqrt{p}I$ and $K_1=\sqrt{1-p}Z$. This map can be understood as 
\beq
\r\mapsto\r'=\begin{cases}
\r &\mbox{w/ prob. } p\\
Z\r Z &\mbox{w/ prob. } 1-p
\end{cases}
\eeq

Using our general result, Eq.~\eqref{eq:M-c} we have in this case:
\beq
c_i = \frac{1}{2}\sum_\a \Tr( \s_i K_\a K_\a^\dgr) = \frac{1}{2}[ p\Tr(\s_i) + (1-p)\Tr(\s_i)] = 0
\eeq
[in agreement with the fact that the phase damping map is unital; recall Eq.~\eqref{eq:c-unital}], and:
\beq
M_{ij} = \frac{1}{2}\sum_\a \Tr(\s_i K_\a \s_j K_\a^\dgr) =  \frac{1}{2}[p \Tr(\s_i \s_j) + (1-p)\Tr(\s_i Z \s_j Z)] = p\d_{ij} + \frac{1}{2}(1-p)J_{ij}\ ,
\eeq
where $J_{ij}\equiv \Tr(\s_i Z \s_j Z)$.
Written explicitly the matrix $J$ is:
\beq
J = \left(
\begin{array}{ccc}
\Tr (XZXZ)  & \Tr (XZYZ)  & \Tr (XZZZ)   \\
\Tr (YZXZ)  & \Tr (YZYZ)  & \Tr (YZZZ)   \\
\Tr (ZZXZ)  & \Tr (ZZYZ)  & \Tr (ZZZZ)    
\end{array}
\right) 
=
\left(
\begin{array}{ccc}
\Tr (-I)  & \Tr (\s)  & \Tr (\s)   \\
\Tr (\s)  & \Tr (-I)  & \Tr (\s)   \\
\Tr (\s)  & \Tr (\s)  & \Tr (I)    
\end{array}
\right) 
= \text{diag}(-2,-2,2)\ ,
\eeq
where $\s$ denotes a Pauli matrix. Thus, 
\beq
M = \text{diag}[p-(1-p),p-(1-p),p+(1-p)] = 
\begin{pmatrix}
2p-1&0&0\\
0&2p-1&0\\
0&0&1
\end{pmatrix}\ ,
\eeq
and 
\beq
\vec{v}' = M\vec{v} = \ls(2p-1)v_x,(2p-1)v_y,v_z\rs^t\ .
\eeq
The corresponding transformation of the Bloch sphere is shown in Fig.~\ref{fig:Bloch-PD}.
There is no shift of the Bloch sphere, while there is a rescaling along the $v_x$ and $v_y$ directions by a factor of $(2p-1)$, and all points on the $v_z$ axis are fixed.  The map has two fixed pure states, the north and south poles of the Bloch sphere, $\ketbra{0}$ and $\ketbra{1}$.  For $p=1$, the Bloch sphere remains unchanged.

\begin{figure}
\begin{center}
\includegraphics[width=0.75\textwidth]{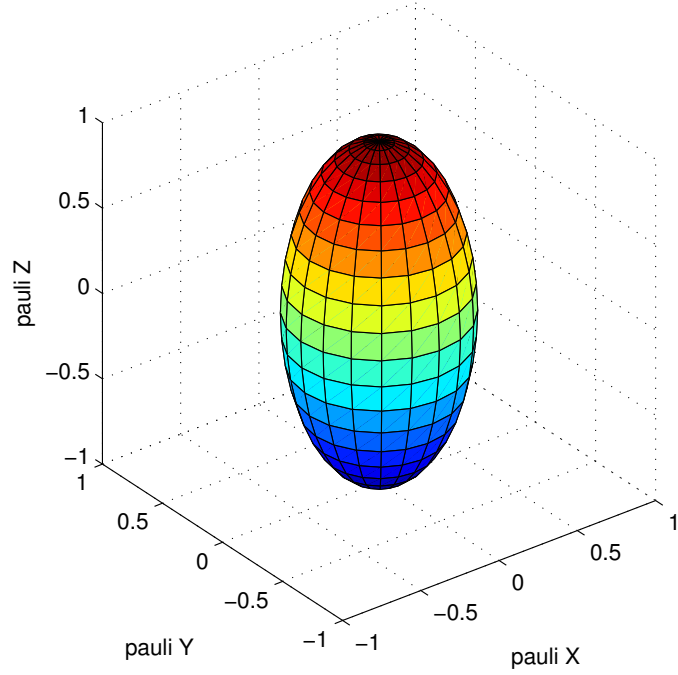}
\caption{The Bloch sphere become an ellipsoid after transformation by the phase damping channel.  The invariant states are those on the $\s_z$ axis. The major axis has length $2$, the minor axis has length $2(2p-1)$.}
\label{fig:Bloch-PD}
\end{center}
\end{figure}

Because $p$ is a probability, $-1\leq2p-1\leq1$. Hence the scaling factor can take negative values, corresponding to a rotation by $\pi$ about the $v_z$ axis.  To see why, let us use the polar decomposition to write $M=OS$, where $S=\sqrt{M^\dgr M} = \text{diag}(|2p-1|,|2p-1|,1)$. Therefore the rotation matrix must be $O=\text{diag}(\text{sign}(2p-1),\text{sign}(2p-1),1) = (\pm 1,\pm 1,1)$. When $2p-1<0$, $O$ is a rotation by $\pi$ about the $v_z$ axis.

The purity [Eq.~\eqref{eq:purity}] of the transformed state is
\beq
P'=\Tr[(\r')^2]=\frac{1}{2}(1+\|\vec{v}'\|^2) = \frac{1}{2}[1+(2p-1)^2(v_x^2+v_y^2)+v_z^2] \leq P\ .
\eeq
Thus the purity always decreases under the phase damping channel, except for the states on the $v_z$ axis (with $v_x=v_y=0$), whose purity is invariant.

\subsection{The Bit Flip Map}
\label{sec:bitflipmap}

The bit flip map is:
\beq\label{eq:bitflip}
\r\mapsto\r'=\begin{cases}
\r &\mbox{w/ prob. } p\\
X\r X &\mbox{w/ prob. } 1-p
\end{cases}
\eeq
In the computational basis, the bit flip map acts like a classical error channel, flipping bits at random. The phase damping map is purely quantum in the same basis, since of course the notion of a phase is not classical. However, mathematically the two maps are essentially identical. We can guess that since the phase flip map leaves the $v_z$ axis alone and shrinks the Bloch sphere in the $(v_x,v_y)$ plane, the bit flip map will leave $v_x$ axis alone and shrinks the Bloch sphere in the $(v_y,v_z)$ plane. To confirm this, let us use a more direct approach than the one we used for the phase flip map.

Using $\r =\frac{1}{2}\left(I+\vec{v}\cdot\vec{\s}\right)$, we have:
\begin{align}
\r \mapsto \r' = p\r + (1-p)X\r X = \frac{1}{2}\left(I+p\vec{v}\cdot\vec{\s}+(1-p)X\vec{v}\cdot\vec{\s}X\right) \ .
\end{align}
The key point is now that 
\beq
X(\vec{v}\cdot\vec{\s})X = X(v_x X +v_y Y +v_z Z)X = v_x X - v_y Y - v_z Z\ .
\label{eq:X-conj}
\eeq 
This shows that $v_x$ is unchanged, but the sign of both $v_y$ and $v_z$ is flipped. Had we studied the phase damping map instead, we would have seen that $v_z$ is unchanged, but the sign of both $v_x$ and $v_y$ is flipped.
We now have:
\beq
\r' = \frac{1}{2}\left(I+v_x X + (2p-1) v_y Y + (2p-1) v_z Z\right) = \frac{1}{2}\left(I+\vec{v}'\cdot\vec{\s}\right) \ .
\eeq
Thus, we find that the bit flip channel transforms $\vec{v}$ as:
\beq
\vec{v}\mapsto \vec{v}'=\ls v_x,(2p-1)v_y,(2p-1)v_z\rs =M\vec{v}+\vec{c}\ ,
\eeq
where
\bes
\begin{align}
M &=\begin{pmatrix}
1&0&0\\
0&2p-1&0\\
0&0&2p-1
\end{pmatrix}\\
\vec{c} &=\bar{0} \ .
\end{align}
\ees
Geometrically, this corresponds to the exact same deformation of the Bloch sphere as depicted in Fig.~\ref{fig:Bloch-PD}, but with the $v_x$ and $v_z$ axes interchanged. If we replace $X$ with $Y$ in Eq.~\eqref{eq:bitflip} we have the ``bit-phase flip channel," where the roles of the $v_x$ and $v_y$ axes is interchanged.


\subsection{The Depolarizing Map}

The depolarizing map acting on a qubit either takes the state to the maximally mixed state with probability $p$, or leaves the state unchanged with probability $1-p$:
\beq
\r\mapsto\r'=\begin{cases}
\frac{1}{2}I &\mbox{w/ prob. } p\\
\r &\mbox{w/ prob. } 1-p
\end{cases}\ .
\eeq
Thus, with probability $p$, all the information held in the state is erased.  Equivalently,
\beq
\r'=p\frac{I}{2}+(1-p)\r\ .
\eeq
Clearly, this is also a unital map. However, note that it is not in Kraus OSR form. To put it in Kraus OSR form, note that
\beq
{\r + X\r X+Y\r Y+Z\r Z} = 2I\ ,
\eeq
which we can prove easily using the same idea as in Eq.~\eqref{eq:X-conj}:
\bes
\begin{align}
Y(\vec{v}\cdot\vec{\s})Y &= - v_x X + v_y Y - v_z Z\\
Z(\vec{v}\cdot\vec{\s})Z &= - v_x X - v_y Y + v_z Z\ ,
\label{eq:YZ-conj}
\end{align}
\ees
so that 
\beq
\sum_i \s_i (\vec{v}\cdot\vec{\s}) \s_i = 0 \ .
\eeq
Thus we may write the map as:
\beq
\r\mapsto\r' = p\frac{1}{4}({\r + X\r X+Y\r Y+Z\r Z})+(1-p)\r = (1-\frac{3}{4}p)\r+\frac{p}{4}(X\r X+Y\r Y+X\r Z)\ ,
\eeq
from which we see that the Kraus operators are 
\beq
K_0=\sqrt{1-\frac{3}{4}p}I\ , \qquad K_i=\sqrt{\frac{p}{4}}\s_i \, \text{  for  } i=1,2,3 \ .
\eeq  

The analysis is particularly straightforward in terms of the Bloch vector:
\begin{eqnarray}
\r' = p\frac{I}{2}+\frac{1-p}{2}(I+\vec{v}\cdot\vec{\s}) = \frac{I}{2}+\frac{1-p}{2}\vec{v}\cdot\vec{\s} = \frac{1}{2}(I+\vec{v}'\cdot\vec{\s})\ ,
\end{eqnarray}
which implies that $\vec{v}'=(1-p)\vec{v}$, so that 
\bes
\begin{align}
M&=(1-p)I\\
\vec{c}&=\bar{0} \ .
\end{align}
\ees
This corresponds to the Bloch sphere shrinking uniformly to a radius of $1-p$, as illustrated in Fig.~\ref{fig:depol}.  The only invariant state is the fully mixed state (the origin, $\vec{v}=\bar{0}$).  Every other state loses purity as it becomes more mixed.

\begin{figure}[t]
\begin{center}
\includegraphics[width=0.75\textwidth]{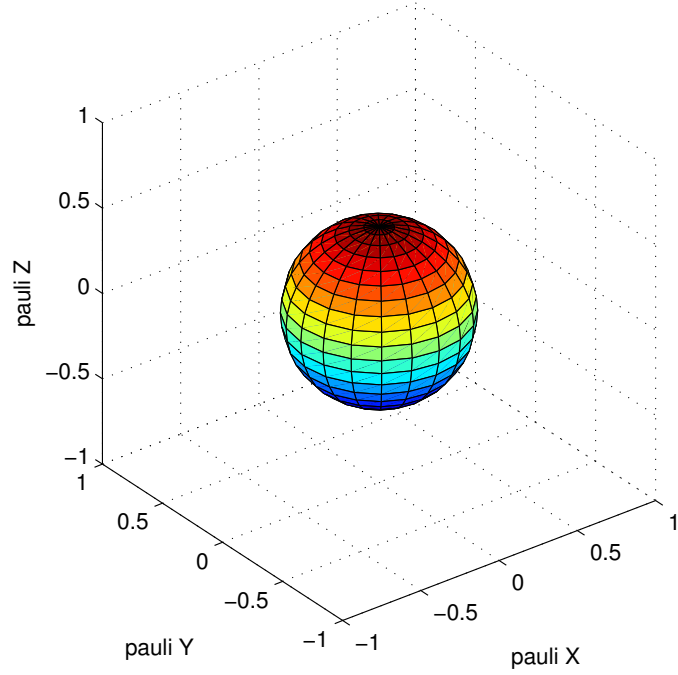}
\caption{The Bloch sphere transformed by the depolarizing channel.  As $p\rightarrow1$, all states converge to the fully mixed state at the origin.}
\label{fig:depol}
\end{center}
\end{figure}


\subsection{Amplitude Damping / Spontaneous Emission}

Spontaneous emission (SE) is the process by which an atom, nucleus, etc., undergoes a transition from a higher state of energy to a lower state of energy, thus releasing energy to the bath (relaxation). This could through the release of a photon, a phonon, or some other elementary excitation. If the bath is at temperature $T=0$, as we assume in this subsection, then the system cannot absorb energy, so the reverse process of excitation does not occur. We shall deal with it in the next subsection.

We consider a single qubits, with a ground state $\ket{0}$ and an excited state $\ket{1}$. Thus the map $\Ph$ is:
\bes
\begin{align}
\ket{0} &\mapsto \ket{0} \quad {\mbox {with probability}}\, 1 \\
\ket{1} &\mapsto \ket{0} \quad {\mbox {with probability}}\, p
\end{align}
\ees
Let us find the Kraus operators for this process. One Kraus operator is obvious: the transition from the excited state to the ground state is given by 
\beq
K_1 = \sqrt{p} \ketb{0}{1}\ .
\eeq
The second Kraus operator should keep the ground state in place, i.e., contains $\ketb{0}{0}$. But this isn't enough, since the normalization condition must be satisfied, and it's easy to check that it isn't if these are our Kraus operators. Instead, let us add an unspecified matrix $A$ and find out its form from the normalization condition. Thus:
\beq
K_0 = \ketb{0}{0} + A = \begin{pmatrix}
1 & a \\
b & c
\end{pmatrix}\ ,
\eeq
and the normalization condition ${K^{\dag}_0}K_0 +{K^{\dag}_1}K_1 = I$ becomes:
\begin{align}
\begin{pmatrix}
1+|b|^2 & a+b^* c \\
a^* +bc^* & |a|^2+|c|^2
\end{pmatrix} + p |1 \rangle \langle 0 | 0 \rangle \langle 1 | = \begin{pmatrix}
1 & 0 \\
0 & 1
\end{pmatrix}\ .
\end{align}
On equating the upper left entries we get $b=0$, which implies from the off-diagonal entries that $a=0$. Equating the bottom right entries then yields $c=\sqrt{1-p}$. Thus:
\beq
K_0 =\begin{pmatrix}
1 & 0 \\
0 & \sqrt{1-p}
\end{pmatrix}\ .
\eeq
The (perhaps curious) $\sqrt{1-p}$ component expresses the fact that \emph{not} observing an emission event (imagine a detector for the emitted photons) increases the likelihood that the system is in its ground state, but we cannot know this with certainty since the emission event might yet arrive in the future. We will see this more clearly later when we discuss quantum trajectories in Sec.~\ref{sec:trajectories}.

We can now directly derive $M$ and $\vec{c}$. Since $\rho ' = \sum_{\alpha=0}^1
K_{\alpha}\left[ \frac{1}{2} (I+\vec{v}\cdot
\vec{\sigma})\right] K^{\dag}_{\alpha}$, the most direct way to do this is to map $I$ and $\vec{v}\cdot \vec{\sigma}$ via the Kraus OSR and read off $M$ and $\vec{c}$. Starting with $I$, we have:
\beq
I \mapsto K_0 K^{\dag}_0 + K_1 K^{\dag}_1 = \begin{pmatrix}
1 & 0 \\
0 & 1-p
\end{pmatrix} + p | 0 \rangle \langle 1 | 1 \rangle \langle 0 | =\begin{pmatrix}
1+p & 0 \\
0 & 1-p
\end{pmatrix} = I + p Z.
\eeq
Thus SE is not a unital map. Since $\vec{c}$ captures the mapping of $I$, we see that 
\beq
\vec{c} = (0,0,p)\ .
\eeq

Next, $\vec{v}\cdot \vec{\sigma} \mapsto K_0 (\vec{v}\cdot \vec{\sigma}) K^{\dag}_0 + K_1 (\vec{v}\cdot \vec{\sigma})K^{\dag}_1$. It is simple to check by explicit matrix multiplication that 
\bes
\begin{align}
&K_0 X K^{\dag}_0 + K_1 X K^{\dag}_1 = \sqrt{1-p}X\\
&K_0 Y K^{\dag}_0 + K_1 Y K^{\dag}_1 = \sqrt{1-p}Y\\
&K_0 Z K^{\dag}_0 + K_1 Z K^{\dag}_1 = {(1-p)}Z
\end{align}
\ees
We thus arrive at the following $M$ matrix:
\beq
M=\begin{pmatrix}
\sqrt{1-p} & 0 & 0\\
0 & \sqrt{1-p} & 0\\
0 & 0 & {1-p}
\end{pmatrix}\ .
\label{eq:M-SE}
\eeq

The geometric meaning of the spontaneous emission map is now clear. The center  $(0,0,0) \mapsto (0,0,p)$, and the Bloch sphere is compressed more along the $v_z$-axis than along the $v_x$ and $v_y$-axes. In other words, all points on the Bloch sphere move closer to its north pole, which is the ground state. If $p=1$ then the entire Bloch sphere is compressed to a single point, the north pole. The latter is a fixed point of the map. To see this, note that
\beq
\Phi ( | 0 \rangle \langle 0 |) = K_0 | 0 \rangle \langle 0 |
K^{\dag}_0 + K_1 | 0 \rangle \langle 0 | K^{\dag}_1 = |0 \rangle
\langle 0 | + 0 = |0 \rangle \langle 0 |
\eeq

\subsection{Generalized (finite temperature) Amplitude Damping/Spontaneous Emission}
\label{sec:gen-AD}

If the qubit is able to absorb energy from the bath (since the latter is at a temperature $T>0$), then the reverse process, of excitation from the ground state to the excited state, is also possible. To account for this let us assume that the spontaneous emission process of the previous subsection occurs with probability $q$, while the reverse process occurs with probability $1-q$. Then the Kraus operators for the SE event become
\bes
\begin{align}
K_0 &= \sqrt{q}\begin{pmatrix}
1 & 0 \\
0 & \sqrt{1-p}
\end{pmatrix}
\\
K_1 &= \sqrt{qp} \ketb{0}{1}\ .
\end{align}
\ees
The Kraus operators for the reverse process are simply:
\bes
\begin{align}
K_2 &= \sqrt{1-q} \begin{pmatrix}
\sqrt{1-p} & 0 \\
0 & 1
\end{pmatrix} \\
K_3 &= \sqrt{(1-q)p} \ketb{1}{0}\ .
\end{align}
\ees

Thus:
\bes
\begin{align}
I &\mapsto K_0 K^{\dag}_0 + K_1 K^{\dag}_1 + K_2 K^{\dag}_2 + K_3 K^{\dag}_3 = q\begin{pmatrix}
1+p & 0 \\
0 & 1-p
\end{pmatrix}  + (1-q)\left[ \begin{pmatrix}
{1-p} & 0 \\
0 & 1
\end{pmatrix} + p| 1 \rangle \langle 0 | 0 \rangle \langle 1 | \right] \\
&= q(I+pZ) + (1-q)(I-pZ) =  I + (2q-1)pZ 
\ ,
\end{align}
\ees
which shows that 
\beq
\vec{c}= (0,0,(2q-1)p) \ .
\eeq
As for the $M$ matrix it is again simple to check by explicit matrix multiplication that 
\bes
\begin{align}
&\sum_{i=0}^3 K_i X K_i^{\dag} = \sqrt{1-p}X\\
&\sum_{i=0}^3 K_i Y K_i^{\dag} = \sqrt{1-p}Y\\
&\sum_{i=0}^3 K_i Z K_i^{\dag} = ({1-p})Z\ ,
\end{align}
\ees
i.e., $M$ is unchanged and is still given by Eq.~\eqref{eq:M-SE}.

Thus the only effect of allowing relaxation is to modify the center of the deformed Bloch sphere, which is now positioned at $(0,0,(2q-1)p)$. This corresponds to a new fixed point, $\r_{\text{eq}} = \text{diag}(q,1-q)$:
\beq
\Phi ( \r_{\text{eq}}) = \sum_{i=0}^3 K_i \r_{\text{eq}} K_i^{\dag} = \r_{\text{eq}}\ .
\eeq

Note that the case $q=1/2$ is unital (it corresponds to $\vec{c}=\vec{0}$) and has a fixed point the fully mixed state. Also note that when $q<1/2$ the new center is at $(0,0,-|2q-1|p)$, which corresponds to a preference for the excited state rather than the ground state. 



\section{Quantum Maps from First Principles}

So far we postulated the form of certain quantum maps. Let us now consider examples where we can analytically derive the Kraus operators from first principles.

\subsection{A qubit coupled to a single-qubit bath}

Consider a system of two qubits, such that the first qubit is the system ($\mathscr{H}_S$) and the second is the bath ($\mathscr{H}_B$). Consider also the interaction Hamiltonian $H_{SB} = \lambda \s_S^{\a} \ox \s_B^{\beta}$ where $\a, \b \in \{ x,y,z \}$. The system qubit is initially in the pure state $\rho_S(0) = \ketbra{\psi}$, $\ket{\psi} = a \ket{0} + b \ket{1}$, written in the computational basis (eigenbasis of $\s^z$). The initial state of the bath is mixed:
\beq 
\rho_B(0) = \lambda_0 \ketbra{0} + \lambda_1 \ketbra{1} = \ls \begin{array}{cc} \lambda_0 & 0 \\ 0 & \lambda_1 \end{array} \rs \ ,
\eeq
where $\lambda_1 = 1 - \lambda_0$. There are $4$ Kraus operators:
\bes
\begin{eqnarray}
	K_{00} & = & \sqrt{\lambda_0} \bra{0} e^{-i \lambda t\s^{\a}_S \ox \s^{\beta}_B} \ket{0} \\
	K_{01} & = & \sqrt{\lambda_1} \bra{0} e^{-i \lambda t\s^{\a}_S \ox \s^{\beta}_B} \ket{1} \\
	K_{10} & = & \sqrt{\lambda_0} \bra{1} e^{-i \lambda t\s^{\a}_S \ox \s^{\beta}_B} \ket{0} \\
	K_{11} & = & \sqrt{\lambda_1} \bra{1} e^{-i \lambda t\s^{\a}_S \ox \s^{\beta}_B} \ket{1}
\end{eqnarray}
\ees
Let $\theta \equiv \lambda t$. Recall now that if $A^2 = I$ then $e^{i \theta A} = \cos \theta I + i \sin \theta A$ (which can be easily checked by Taylor expansion). Therefore 
\beq 
e^{i \theta \s^{\a}_S \ox \s^{\beta}_B} = \cos \theta \cdot I_S \ox I_B + i \sin \theta \cdot \s^{\a}_S \ox \s^{\beta}_B\ ,
\eeq
and hence (for a general $\mu, \nu \in \{ 0 ,1 \}$)
\beq
	K_{\mu \nu} = \sqrt{\lambda_{\nu}} \lb \cos \theta \d_{\mu \nu} \cdot I_S - i \sin \theta \bra{\mu} \s^{\beta}_B \ket{\nu} \cdot \s^{\a}_S \rb \ .
	\label{eqn:general2qubitKraus}
\eeq
The system then evolves according to the Kraus map
\beq 
\rho_S(t) = \sum_{\mu \nu} K_{\mu \nu}(t) \ketbra{\psi} K^{\dag}_{\mu \nu} (t)\ .
\eeq


\subsubsection{$Z\otimes X$ coupling}

Consider first $H_{SB} = \lambda Z_S \ox X_B$. In this case, we can use Eq.~\eqref{eqn:general2qubitKraus} to find
\bes
\begin{align}
	K_{00} & =  \sqrt{\lambda_0} \cos \theta \cdot I \\
	K_{11} & =  \sqrt{\lambda_1} \cos \theta \cdot I \\
	K_{01} & =  -i \sqrt{\lambda_1}  \sin \theta \cdot \s^z \\
	K_{10} & =  -i \sqrt{\lambda_0}  \sin \theta \cdot \s^z
\end{align}
\ees
The density matrix for this map evolves under the action of these $4$ Kraus operators:
\bes
\begin{align}
	\rho_S(t) & =  \sum_{\mu \nu} K_{\mu \nu}(t) \rho_S(0) K^{\dag}_{\mu \nu}(t)  \\
	& =  \lp \sqrt{\lambda_0} \cos \theta \rp ^2 \ketbra{\psi} + \lp \sqrt{\lambda_1} \cos \theta \rp ^2 \ketbra{\psi} + \lp \sqrt{\lambda_0} \sin \theta \rp ^2 \s^z \ketbra{\psi} \s^z + \lp \sqrt{\lambda_1} \sin \theta \rp ^2 \s^z \ketbra{\psi} \s^z  \\
	& =  \cos^2 \theta \ketbra{\psi} + \sin^2 \theta \cdot \s^z \ketbra{\psi} \s^z  \\
	&= \ls \begin{array}{cc} |a|^2 & ab^* \cos(2\theta) \\ a^*b \cos(2\theta) & |b|^2 \end{array}\rs  \ ,
	\label{eq:224d} 
\end{align}
\ees
where we used the fact that $\lambda_0 + \lambda_1 = 1$.

Can we relate this result to the phase damping map discussed in Sec.~\ref{sec:PD}? This seems plausible since in both cases the system is affected by a $Z$ operator. In the phase damping case we have 
\beq
\r_S(t) = \Ph[\r_S(0)] = p\ketb{\psi}{\psi} + (1-p)Z\ketb{\psi}{\psi}Z = \ls \begin{array}{cc} |a|^2 & (2p-1) ab^* \\ (2p-1) a^*b  & |b|^2 \end{array}\rs\ ,
\eeq
which we would like to equate with Eq.~\eqref{eq:224d}. Clearly, this requires $2p-1 = f(\theta)$, so that
\beq
p = \frac{1+f(\theta)}{2} \ ,
\eeq
and the phase damping map has as a physical origin the model given by $H_{SB} = \lambda Z_S \ox X_B$. 

Why did  $\lambda_0$ and  $\lambda_1$ drop out? The intuitive reason is that by having the bath qubit subject to $\s^x$, its $\ket{0}$ and $\ket{1}$ state are constantly flipped, which also interchanges  $\lambda_0$ and  $\lambda_1$, so it is as if they are averaged to $1/2$.

There is much more to say about this result, but first let us consider another case, which will turn out to subsume this one.

\subsubsection{$Z\otimes Z$ coupling}
Consider the interaction Hamiltonian $H = \lambda \s^z_S \ox \s^z_B$. For this choice, since $\s^z$ is diagonal, only the $K_{00}$ and $K_{11}$ Kraus operators are non-zero and have the form
\bes
\begin{eqnarray}
	K_{00} & = & \sqrt{\lambda_0} \lp \cos \theta \cdot I_S  - i \sin \theta \cdot \s^z \rp = \sqrt{\lambda_0}  \ls \begin{array}{cc} e^{-i \theta} & 0 \\ 0 & e^{i \theta} \end{array} \rs \\
	K_{11} & = & \sqrt{\lambda_1} \lp \cos \theta \cdot I_S  + i \sin \theta \cdot \s^z \rp = \sqrt{\lambda_1}  \ls \begin{array}{cc} e^{i \theta} & 0 \\ 0 & e^{-i \theta} \end{array} \rs
\end{eqnarray}
\ees
Altogether, the pure state $\ket{\psi}$ under each of these operators becomes
\bes
\begin{eqnarray}
	K_{00} \ket{\psi} & = & \sqrt{\lambda_0} \ls \begin{array}{cc} e^{-i \theta}&0 \\ 0&e^{i \theta} \end{array} \rs \ls \begin{array}{c} a\\ b \end{array}\rs = \sqrt{\lambda_0} \ls \begin{array}{c} ae^{-i \theta} \\ be^{i \theta} \end{array}\rs \\
	K_{11} \ket{\psi} & = & \sqrt{\lambda_1} \ls \begin{array}{c} ae^{i \theta} \\ be^{-i \theta} \end{array}\rs
\end{eqnarray}
\ees
Therefore:
\begin{eqnarray}
	\rho_S(t) & = & \lambda_0 \ls \begin{array}{c} ae^{-i \theta} \\ be^{i \theta} \end{array}\rs \ls \begin{array}{cc} a^*e^{i \theta} & b^*e^{-i \theta} \end{array}\rs + \lambda_1 \ls \begin{array}{c} ae^{i \theta} \\ be^{-i \theta} \end{array}\rs \ls \begin{array}{cc} a^*e^{-i \theta} & b^*e^{i \theta} \end{array}\rs \nonumber \\
	& = & \ls \begin{array}{cc}|a|^2 & ab^* \lp \lambda_0 e^{-2i \theta} + \lambda_1 e^{2 i \theta} \rp \\ a^*b \lp \lambda_0 e^{2i\theta} + \lambda_1 e^{-2i \theta} \rp & |b|^2 \end{array}\rs
\end{eqnarray}
where the diagonal elements have again been simplified with the use of the fact that $\lambda_0 + \lambda_1 = 1$. Defining
\beq
f(\theta) = \lambda_0 e^{-2i \theta} + \lambda_1 e^{2i \theta}\ ,
\label{eq:f-dephasing}
\eeq 
yields 
\beq 
\rho_S(t) = \ls \begin{array}{cc} |a|^2 & ab^* f(\theta) \\ a^*b f^*(\theta) & |b|^2 \end{array}\rs \ .
\label{eq:rho-ZZ-model}
\eeq
The previous example, $H_{SB} = Z \ox X$, is now seen to be a special case of this one, where $\lambda_0 = \lambda_1 = 1/2$ [for then $f(\theta) = \cos(2\theta)$], so everything we discuss next applies to it as well.

Note that the diagonal elements (``population") haven't changed under time evolution and yet the off-diagonal elements (``coherence") are modulated by the periodic function $f$.  This is like elastic scattering where no energy is exchanged and only relative phases are impacted. More precisely, this is a \emph{dephasing} process, although in our case, the phase coherence recurs periodically. The period of $f$ is $\tau = \pi / \lambda$. 

Consider the {purity} $P = \Tr(\r^2)$:%
\footnote{We can obtain the same result using the formula $P = \frac{1}{2}(1+\|\vec{v}\|^2)$, as follows: $\frac{1}{2}(v_x-iv_y)=ab^*f$ and $\frac{1}{2}(v_x+iv_y)=a^*bf^*$, so that $v_x^2+v_y^2 = 4|a|^2|b|^2|f|^2$. Also, $v_z = |a|^2-|b|^2$, and $1=(|a|^2+|b|^2)^2$. Adding all this up gives Eq.~\eqref{eq:232c}.}
\bes
\begin{eqnarray}
	P & = & \Tr \ls \lp \begin{array}{cc} |a|^2 & ab^* f(\theta) \\ a^*b f^*(\theta) & |b|^2 \end{array}\rp^2 \rs  \\
	& = & \Tr \lp \begin{array}{cc} |a|^4 + |a|^2|b|^2|f|^2 & \dots \\ \dots & |b|^4 + |a|^2|b|^2|f|^2 \end{array} \rp  \\
	& = & |a|^4 + |b|^4 + 2 |a|^2|b|^2|f|^2\ .
	\label{eq:232c}
\end{eqnarray}
\ees
Thus, this function is periodic with period $\tau_P = \pi / (2 \lambda)$ since $f$ appears squared in the expression. Since $|f|^2 = \lambda_0^2 + \lambda_1^2 + 2\lambda_0\lambda_1\cos(4\theta)$ we have
\bes
\begin{align}
\label{eq:fmin}
\min_{\theta} |f|^2 &= \lambda_0^2 + \lambda_1^2 - 2 \lambda_0 \lambda_1 = \lp \lambda_0 - \lambda_1 \rp^2 \\ 
\max_{\theta} |f|^2 &= \lambda_0^2 + \lambda_1^2 - 2 \lambda_0 \lambda_1 = \lp \lambda_0 + \lambda_1 \rp^2 = 1,
\end{align}
\ees
so that the minimum and maximum values of the purity are 
\bes
\begin{align}
\min_{\theta} P &= |a|^4 + |b|^4 + 2 |ab|^2 |\lambda_0 - \lambda_1|^2 \\
\max_{\theta} P &= |a|^4 + |b|^4 + 2 |ab|^2 = (|a|^2+|b|^2)^2=1 \ . 
\end{align}
\ees
The purity achieves a minimum of $1/2$ when the bath qubit is in a maximally mixed state, $\rho_B = I_B / 2$ (so that $\lambda_0 = \lambda_1=1/2$), and when the system qubit is an equal superposition, $|a|=|b|=1/\sqrt{2}$.\footnote{To see this note that $|a|^4 + |b|^4 = |a|^4 + (1-|a|^2)^2 = 2|a|^4-2|a|^2+1 = 2x^2-2x+1$ with $x=|a|^2$; this is minimized at $4x-2=0$, i.e., $x=1/2$, or $|a|=1/\sqrt{2}$.} 

For short times $t\ll\tau_P$ the purity decays quadratically. This is typical of non-Markovian decay, as we will see later (in contrast, Markovian decay is always exponential, i.e., it starts out linearly). One might also write this inequality as a weak coupling limit $\lambda \ll \pi/t$, which suggests that in this limit the purity appears to be only decaying (i.e., there is no time for a recurrence). However, if the coupling between the system and the bath is strong, that is $\lambda \gg 1$, then we may not necessarily resolve the oscillations in purity and instead measure an average purity significantly lower than $1$. In both these limits an observer would conclude that the state of the system is mixed, even though it started out pure.


Can we relate this model to the phase damping map discussed in Sec.~\ref{sec:PD}? Clearly, this requires $2p-1 = f(\theta) = f^*(\theta)$. Thus, equality only holds subject to the additional constraint that $f(\theta)$ is real. The constraint that $f(\theta) = f^*(\theta)$ requires that $\lambda_1=\lambda_2 = 1/2$, i.e., the initial bath state is $I/2$. Therefore this Hamiltonian model is more general than the phase damping map. This is because the bath operator $Z_B$ in the former does not affect the bath state $\r_B(0) = \text{diag}(\lambda_0,\lambda_1)$, with which it commutes. This keeps $\lambda_0$ and $\lambda_1$ in play, unlike the previous case where they were averaged out.


\subsection{Irreversible open system dynamics: infinite dimensional bath with a continuous density of states}
\label{sec:irr-dyn}

Our previous example involved a finite-dimensional bath, and we saw that the purity in this case is periodic. This reflect reversibility, which is a general characteristic of the finite dimensional case. To exhibit irreversibility we shall investigate an infinite-dimensional bath, but as we shall see, one additional ingredient (a continuum) will be needed as well.

Assume the system is either a qubit or a quantum harmonic oscillator (QHO). We will work out both cases. The system-bath interaction Hamiltonian has one of the following forms:
\bes
\label{eq:138}
\begin{align}
H_{SB}&=\lambda  \s_{S}^{z} \ox \hat{n}_{B} \\
H_{SB}&=\lambda  \hat{n}_{S} \ox \hat{n}_{B}\ ,
\end{align}
\ees
where $\hat{n}_S$ is the number operator satisfying $\hat{n}_S\ket{n}=n\ket{n}$ for $n=0,1,...,\infty$. The total Hamiltonian is $H=H_{SB}+H_B$, where we have set $H_S=0$ for simplicity.
We assume that the bath is itself a QHO with Hamiltonian 
\beq
H_B = \sum_{\n=0}^{\infty} E_{\n} \ketbra{\n} \ ,
\eeq
where $\hat{n}_B\ket{\n}=\n\ket{\n}$ and $\hat{n}_B$ is the number operator, and $E_{\n}$ are QHO energies: $E_{\n}=\o(\n+\frac{1}{2})$ (where as before we set $\hbar \equiv 1$). 
We assume that the initial state of the bath is a Gibbs state:
\bea
\r_{B}(0)=\frac{1}{Z}e^{-\b H_B}=\frac{1}{Z}\sum_{\n=0}^{\infty} e^{-\b {E_{\n}}}\ketbra{\n} \equiv \sum_{\n=0}^{\infty} \lambda_{\n}\ketbra{\n} \ ,
\label{eq:Gibbs1}
\eea
where $\b=\frac{1}{k_{B}T}$ and $\lambda_{\n}=\frac{1}{Z}e^{-\b E_{\n}}$ are the eigenvalues of the bath density matrix.
The denominator is the partition function: $Z=\Tr[e^{-\b H_B} ]=\sum_{\n=0}^{\infty}{e^{-\b E_{\n}}}$. 

Using the Hamiltonians in Eq.~\eqref{eq:138} and again defining $\t\equiv \lambda t$, the joint unitary evolution operator becomes
\beq
U(t)=e^{-itH}=e^{-itH_{SB}}e^{-itH_B} \ ,
\eeq
where we have used the fact that $[H_{SB},H_B]=0$.
Thus:
\bes
\begin{align}
U(t)&=\exp\left[{-i\t\left\{\begin{array}{l}
\s^z\\
\hat{n}_S \end{array}\right\}\ox \hat{n}_B}\right] \exp\left[-it I_S\ox\sum_{\n=0}^{\infty} E_{\n} \ketbra{\n}\right]\\
&= \sum_{\n=0}^{\infty} \exp\left[{-i\t\left\{\begin{array}{l}
\s^z\\
\hat{n}_S \end{array}\right\}\ox \hat{n}_B}\right] \exp\left(-it E_{\n}  I_S\right) \ox \ketbra{\n}\\
&= \sum_{\n'=0}^{\infty} \exp\left(-it E_{\n'}I_S  \right) \exp\left[{-i\t\left\{\begin{array}{l}
\s^z\\
\hat{n}_S \end{array}\right\} \n' }\right] \ox \ketbra{\n'}\ ,
\end{align}
\ees
where in the last equality we used $\hat{n}_B\ket{\n}=\n\ket{\n}$.

Taking the partial matrix element with respect to the bath, we find:
\bes
\begin{align}
\bra{\m}U(t)\ket{\n}&=\sum_{\n'=0}^{\infty} \exp\left(-it E_{\n'} I_S \right) \exp\left[{-i\t\left\{\begin{array}{l}
\s^z\\
\hat{n}_S \end{array}\right\} \n' }\right] \langle{\m}\ketbra{\n'} \n\rangle\\
&=
\sum_{\n'=0}^{\infty} \exp\left(-it E_{\n'} I_S \right) \exp\left[{-i\t\left\{\begin{array}{l}
\s^z\\
\hat{n}_S \end{array}\right\} \n' }\right] \delta_{\m\n'}\delta_{\n'\n}\\
&=\exp\left[-it E_{\n} I_S \right] \exp\left[{-i\t\left\{\begin{array}{l}
\s^z_S\\
\hat{n}_S \end{array}\right\} \m}\right] \d_{\m\n}\ .
\end{align}
\ees
Thus, the Kraus operators $K_{\m\n}(t)=\sqrt{\lambda_{\n}}\bra{\m}U(t)\ket{\n}$ can be written as
\bea
K_{\m\n}(t)=\sqrt{\lambda_{\n}}\exp\left[{-i\t\n\left\{\begin{array}{l}
\s^z\\
\hat{n} \end{array}\right\}}\right]\d_{\m\n} \ ,
\eea
where we dropped the $S$ subscripts since it is now clear that the remaining operators act only on the system, and also dropped the term $\exp\left[-it E_{\n} I \right]$ (whose origin was $H_B$), since it acts as an overall phase and will drop out once we apply $K_{\m\n}(t)[\cdot]K_{\m\n}^\dgr(t)$.


Let us write the initial system density matrix as:
\bea
\r_S(0)=\sum_{m,n=0}^{1\ \text{or}\ \infty}{r_{mn}\ketb{m}{n}},
\eea
where we expanded the density matrix in the eigenvectors of the $\s^z$ or $\hat{n}$ operator, with the upper limits being $1$ or $\infty$, respectively.  

In the case where the system is a qubit, we have, using Eq.~\eqref{eq:Kraus OSR}:
\bes
\begin{align}
\r_S(t)&=\sum_{m,n=0}^1{r_{mn}\sum_{\n}{\lambda_{\n} e^{-i\t \n \s^z}\ketb{m}{n}e^{i\t \n \s^z}}}\\
&=\sum_{m,n=0}^1{r_{mn}\sum_{\n}{\lambda_{\n} e^{-i\t \n (-1)^m}\ketb{m}{n}e^{i\t \n (-1)^n}}}.
\end{align}
\ees
Let us rewrite this as
\bea
\r_S(t)=\sum_{m,n=0}^1{r_{mn}\ketb{m}{n}g_{n,m}(\t)},
\eea
where 
\beq
g_{n,m}(\t)\equiv\sum_{\n=0}^\infty{\lambda_{\n} e^{i [(-1)^n-(-1)^m] \n \t}}.
\eeq 
The diagonal terms $g_{m,m}=\sum_{\n=0}^\infty{\lambda_{\n}}=1$ are constant, and therefore they do not evolve in this case.
Let us focus next on the case in which both system and bath are QHO's. We then have, using Eq.~\eqref{eq:Kraus OSR}:
\bes
\begin{align}
\r_S(t)&=\sum_{m,n}{r_{mn}\sum_{\n}{\lambda_{\n} e^{-i\t \n \hat{n}}\ketb{m}{n}e^{i\t \n \hat{n}}}}\\
&=\sum_{m,n}{r_{mn}\sum_{\n}{\lambda_{\n} e^{-i\t \n m}\ketb{m}{n}e^{i\t \n n}}}.
\end{align}
\ees
Let us rewrite this as
\bea
\r_S(t)=\sum_{m,n}{r_{mn}\ketb{m}{n}f_{n-m}(\t)},
\eea
where 
\beq
f_{x}(\t)\equiv\sum_{\n}{\lambda_{\n} e^{i x \n \t}}.
\eeq 

Note that $f_{0}(\t)=\sum_{\n}{\lambda_{\n} }=1$, 
so the state of the system at time $t$ can be split into diagonal (population) and off-diagonal (coherence) terms:
\bea
\r_{S}(t)=\sum_{n}{r_{nn}\ketb{n}{n}}+\sum_{m\neq n}{r_{mn}\ketb{m}{n}f_{n-m}(\lambda t)}
\label{eq:rhoHO}
\eea
The population term is time-independent, i.e., is the same as in $\r_S(0)$. The coherence term is time-dependent and is affected by the coupling to the bath. Its behavior is completely determined by the modulation function $f$, which can be computed explicitly by performing the geometric sum:
\bes
\begin{align}
\label{eq:fxtheta}
f_{x}(\t)&=\frac{1}{Z}\sum_{\n=0}^\infty e^{-\b \o(\n+\frac{1}{2})} e^{i x \n \t}=\frac{e^{-\frac{1}{2}\b \o}}{Z}\sum_{\n=0}^\infty e^{-\b \o\n} e^{i x \n \t}\\
& = \frac{e^{-\frac{1}{2}\b \o}}{Z}\sum_{\n=0}^\infty q^\n \  , \ \ \ q\equiv e^{-(\b \o -i x \t)} \\
=& \frac{e^{-\frac{1}{2}\b \o}}{Z} \frac{1}{1-q} ,
\end{align}
\ees
where convergence of the infinite series is guaranteed since $|q|=e^{-\b\o}<1$ due to $\b\o>0$. 

Note that $f_x(\t) = f_x(\t+2\p/x)$, i.e., $f$ is periodic, with period $T(x) = 2\p/(\lambda x)$. Each off-diagonal element $\ketb{m}{n}$ in Eq.~\eqref{eq:rhoHO} thus has a different period $\tau_{mn} = 2\p/(\lambda |m-n|)$. This suggests that we might have an example of irreversible decoherence [decay of the off-diagonal elements of $\r_S(t)$], if $\r_S(t)$ isn't periodic. But is it? Periodicity requires there to be a time $\tilde{\tau}$ that is simultaneously divisible by all periods $\tau_{mn}$ (i.e., all such periods fit an integer number of times into $\tilde{\tau}$). Clearly, $\tilde{\tau} = 2\p/\lambda$ is just such a time: $\tilde{\tau}/\tau_{mn} = |m-n|$. Thus $\r_S(t)$ is periodic after all, with a period of $2\p/\lambda$, and we do not have irreversibility.

Note that the qubit-system case is just a special case of the QHO-system. To see this observe that $g_{00}=g_{11}=\sum_\nu \lambda_\nu  = 1$, and $g_{01}=g_{10}^*=\sum_\nu \lambda_\nu e^{2i\theta \nu} = f_2(\theta)$.

To better understand the emergence of irreversibility, we thus consider a modified model, where we introduce a mode density $\Omega(\n)$ (a standard trick in condensed matter physics; consider, e.g., the Debye model). We thus replace the sum by an integral over $\n$, and write 
\bea
f_{x}(\t)=\frac{1}{Z}\int_{0}^{\infty}{d\n e^{-\b \o(\n+\frac{1}{2})} e^{i x \t \n} \Omega(\n)} .
\eea
If $\Omega(\n) = \sum_{\n'=0}^\infty \delta(\n-\n')$ then we recover Eq.~\eqref{eq:fxtheta}. The modified model has the following mode density:
\bea
\Omega(\n)= \left\{ \begin{array}{ll}
         \Omega_{0} & \mbox{if $\n_{c}\geq \n \geq 0$};\\
        0 & \mbox{otherwise}.\end{array} \right. ,
        \label{eq:Om}
\eea
i.e., it has a continuous set of modes with a high-mode cutoff of $\n_c$. The cutoff is physically well-motivated: it reflects the fact that any physical model must have a highest but finite accessible energy. Then:
\bes
\begin{align}
f_{x}(\t)&=\frac{\Omega_{0}}{Z}\int^{\n_c}_{0}{e^{-\b \o (\n+\frac{1}{2})}e^{ix\t \n}} d\n\\
&=
\frac{\Omega_{0}e^{-\frac{1}{2}\b\o}}{Z}\int^{\n_c}_{0}{e^{-(\b \o -ix\t) \n }d\n}\\
&= \frac{\Omega_{0}e^{-\frac{1}{2}\b\o}}{Z}\frac{e^{-(\b\o-i x \t) \n_c} -1}{-\b \o +i x \t}
\end{align}
\ees
The numerator is periodic just like in the previous case, so the same comments apply. However, the denominator contains a $(n-m) \lambda t$ dependence (the $x\t$ term), which shows that the coherences decay irreversibly as $1/t$, with the decay being faster for off-diagonal elements that are farther apart. 

We have thus seen how an infinite-dimensional bath with a continuous mode density can result in a decay which is truly irreversible. The decay of the off-diagonal elements is often called \emph{decoherence}, since it refers to the gradual disappearance of coherence, the name given to the off-diagonal elements. This is not an entirely satisfactory definition of decoherence, since it is obviously basis dependent. We shall give a more careful definition later.


\section{Derivation of the Lindblad equation from a short time expansion of quantum maps}
\label{sec:LE-deriv1}

Just as the Hamiltonian is the generator of unitary evolution, we may ask if there is a generator for open system dynamics. By this we mean that the solution of the differential equation $\dot{\r} = \mc{L}\r$ is a quantum map, and $\mc{L}$ plays the role of a generator. In this section we will see how to find such a generator for very short evolution times using just a short time expansion of the Kraus OSR. We will then postulate that the same generator applies for all times (a type of Markovian approximation), and thus arrive at a ``master equation" of the form $\dot{\r} = \mc{L}\r$ that generates a quantum map. The generator $\mc{L}$ is called the Lindbladian, and the master equation is the Lindblad equation, whose special form guarantees complete positivity (i.e., that the evolution it generates is a quantum map).

\subsection{Derivation}
\label{sec:LE-derivation}

By Taylor expansion around $t=0$ we have:
\beq
\label{eq:deriv}
\rho (dt) = \rho (0)+\dot{\rho}|_{0} dt + O(dt^2)\ .
\eeq
On the other hand, the Kraus OSR tells us that:
\beq
\rho (dt) =\sum_{\alpha} K_{\alpha}(dt)\rho (0)
K^{\dag}_{\alpha}(dt)\ .
\label{eq:245}
\eeq
Let's try to find the Kraus operators that make these two equations agree up to $O(dt)$. Clearly, to get the $\r(0)$ term in Eq.~\eqref{eq:deriv} one of the Kraus operators must contain the identity operator. Thus, let us write 
\beq
K_0 = I+L_0 dt \ ,
\label{eq:K_0}
\eeq
so that
\beq
K_{0}\rho (0) K^{\dag}_{0} = \rho(0)+[L_0 \rho (0) +\rho (0)L^{\dag}_{0}]dt+O(dt^2)\ .
\eeq
This contributes one term of order $dt$, but there must be more (since as we know a Kraus OSR with a single Kraus operator is equivalent to unitary evolution). Thus, we can pick all other Kraus operators as
\beq
K_{\alpha} = \sqrt{dt}L_{\alpha} \ , \quad \alpha \geq 1\ ,
\eeq
so that
\beq
K_{\alpha}\rho (0) K^{\dag}_{\alpha} = L_{\alpha}\rho (0) L^{\dag}_{\alpha} dt \ .
\eeq
Let us now enforce the normalization condition $\sum_{\alpha=0} K^{\dag}_{\alpha} K_{\alpha} = I$, up to $O(dt)$:
\beq
I = K^{\dag}_{0} K_{0}+ \sum_{\alpha\geq 1} K^{\dag}_{\alpha} K_{\alpha} =
I+dt\left(L_0+L_0^\dgr  + \sum_{\a\geq 1} L^{\dag}_{\alpha}L_{\alpha} \right) + O(dt^2)\ .
\label{eq:250}
\eeq
Without loss of generality we can decompose the general operator $L_0$ into a Hermitian and anti-Hermitian part: $L_0=A-iH$, with $A=A^{\dag}$ and $H=H^{\dag}$. Thus, Eq.~\eqref{eq:250} tells us that to $O(dt)$: 
\beq
A = -\frac{1}{2} \sum_{\a\geq 1} L^{\dag}_{\alpha}L_{\alpha}\ .
\eeq
Plugging all this back into the Kraus OSR, Eq.~\eqref{eq:245}, we find:
\bes
\begin{align}
\rho (dt) &= K_0 \rho (0)K^{\dag}_0 + \sum_{\alpha \geq 1} K_{\alpha} \rho (0) K^{\dag}_{\alpha}\\
& = \rho (0) + (A-iH)dt \rho (0) + \rho (0)(A+ iH)dt + \sum_{\alpha \geq 1}L_{\alpha} \rho (0)
L^{\dag}_{\alpha} dt + O(dt^2)\\
&= \rho (0) -i[H,\r(0)]dt + \{A,\rho(0)\}dt + \sum_{\alpha \geq 1}L_{\alpha} \rho (0)
L^{\dag}_{\alpha} dt + O(dt^2)\\
&=\rho (0) -i[H,\r(0)]dt + \sum_{\alpha \geq 1}\left(L_{\alpha} \rho (0)
L^{\dag}_{\alpha} -\frac{1}{2} \left\{L^{\dag}_{\alpha}L_{\alpha},\r(0)\right\} \right) dt + O(dt^2)\ .
\end{align}
\ees
Therefore:
\beq
\dot{\rho}(t)|_{0} =\lim_{dt \rightarrow 0} \frac{\rho(dt)-\rho(0)}{dt} = 
 -i[H,\rho(0)]+\sum_{\alpha \geq 1}\left(L_{\alpha} \rho (0)
L^{\dag}_{\alpha} -\frac{1}{2} \left\{L^{\dag}_{\alpha}L_{\alpha},\r(0)\right\} \right)\ .
\label{eq:253}
\eeq
This is almost the form of the master equation we are after. Note that the operators $L_{\a}$ are not dimensionless, but must have units of $1/\sqrt{\text{time}}$. To make them dimensionless, let us replace them by $\sqrt{\g'_\a} L_\a$, where $\g'_\a$ has units of $1/\text{time}$, so that the new $L_\a$ are dimensionless. Substituting this into Eq.~\eqref{eq:253} only generates the combinations $\sqrt{\g'_\a}\sqrt{\g'^*_\a} = |\g'_\a| \equiv \g_\a \geq 0$. Thus:
\beq
\dot{\rho}(t)|_{0}  = 
 -i[H,\rho(0)]+\sum_{\alpha \geq 1}\g_\a\left(L_{\alpha} \rho (0)
L^{\dag}_{\alpha} -\frac{1}{2} \left\{L^{\dag}_{\alpha}L_{\alpha},\r(0)\right\} \right)\ .
\label{eq:254}
\eeq
This result is valid as a short time expansion near $t=0$. We now make an extra, very significant assumption:
\begin{myassumption}
Eq.~\eqref{eq:254} is valid for \emph{all} times $t>0$. 
\end{myassumption}
This is essentially the Markovian limit, which states (informally) that there is no memory in the evolution, as manifested by the fact that the evolution ``resets" every $dt$. It is motivated in part by the observation that if we limit our attention just to $\dot{\rho}(t)|_{0}  = 
 -i[H,\rho(0)]$, then we already know that this replacement is valid, i.e., that we can indeed replace this with $\dot{\rho}(t)  = 
 -i[H,\rho(t)]$ for all $t$, since this is just the Schr\"odinger equation. With this we finally arrive at the \emph{Lindblad equation}:
\beq
\frac{d\rho}{dt} = -i[H,\rho(t)] + \sum_{\alpha}\g_\a\left(L_{\alpha} \rho (t)
L^{\dag}_{\alpha} -\frac{1}{2} \left\{L^{\dag}_{\alpha}L_{\alpha},\r(t)\right\} \right)
 \equiv {\cal L}\rho\ .
 \label{eq:Lindblad-eq}
\eeq
The generator of the evolution, $\mc{L}$, is called the Lindbladian. The $L_{\alpha}$  are called the Lindblad operators. The operator $H$ is Hermitian and will be interpreted later as the Hamiltonian of the system (plus a correction called the Lamb shift). The form of the dissipative part of the Lindbladian,  also known as the dissipator, is:
\beq
\mc{L}_D[\cdot] = \sum_\a \g_\a\left( L_{\alpha} \cdot L^{\dag}_{\alpha} -\frac{1}{2}\{L^{\dag}_{\alpha}L_{\alpha},\cdot\}\right)\ , \quad \g_\a \geq 0 \ .
\label{eq:L_D}
\eeq
We can now define decoherence in a basis-independent manner. Decoherence is what happens when $\mc{L}_D \neq 0$. In this case the evolution of the density matrix is governed not only by the Schr\"odinger component $-i[H,\cdot]$ (responsible for unitary evolution), but also by the dissipator, which gives rise to non-unitary evolution.

The positivity of the Lindblad rates (they have units of $1/\text{time}$) is a direct consequence of complete positivity. Conversely, it guarantees that the map generated by the Lindblad equation~\eqref{eq:Lindblad-eq} is CP, as we will show in Sec.~\ref{sec:LEtoCP}. As our derivation shows, this map has Kraus operators given by
\bes
\begin{align}
K_0 &= I + (-iH+A)dt \\
K_\a &= \sqrt{\g_\a} L_\a \sqrt{dt} \ , \quad \a\geq 1 \ .
\end{align}
\ees

\subsection{The Markovian evolution operator as a one-parameter semigroup}

The formal solution of the (Lindblad) equation $\dot{\rho}(t) = {\mathcal{L}} \rho$ is 
\beq
\rho(t) = e^{{\mathcal{L}} t} \rho(0) \equiv \Lambda_t \rho(0) \ ,
\eeq
where $\Lambda_t$ is called the Markovian evolution operator (it is also a quantum map). The set $\{\Lambda_t\}_{t\geq 0}$ forms a one-parameter semigroup. The one-parameter part is clear: the set depends only on the time $t$, once the Lindblad generator $\mc{L}$ is fixed. The reason this is a semi-group is that the superoperators $\Lambda_t$ only satisfy three of the four properties of a group: 
\begin{enumerate}
\item Identity operator: $\Lambda_0 = \mathcal{I}$.
\item Closed under multiplication: $\Lambda_t \Lambda_s = e^{\mathcal{L} t} e^{\mathcal{L} s} = e^{\mathcal{L} (t+s)} = \Lambda_{t+s}$.
\item Associative: $(\Lambda_t \Lambda_s)\Lambda_r = \Lambda_t (\Lambda_s \Lambda_r)$.
\end{enumerate}
However, not every element has an inverse: as we shall see, complete positivity forces all the eigenvalues of $\mc{L}$ to be non-positive, so that the map $\Lambda_t$ is contractive, corresponding to exponential decay. This means that $\Lambda_\infty$ has at least one zero eigenvalue, so it does not possess an inverse. We shall shortly see this in examples.

\subsection{Proof that the solution of the Lindblad Equation is a CP map}
\label{sec:LEtoCP}

The argument we use to prove that the solution of the Lindblad Equation is a CP map is essentially the reverse of that presented in Sec.~\ref{sec:LE-derivation}, plus a proof that the concatenation of CP maps (and in particular of a CP map with itself) is still a CP map.

Let us start from the Lindblad equation and let $A \equiv -\frac{1}{2} \sum_{\a\geq 1} L^{\dag}_{\alpha}L_{\alpha}$:
\bes
\begin{align}
\label{eq:365a}
\dot{\rho}(t) &= \mc{L}\r(t) = 
 -i[H,\rho(t)]+\sum_{\alpha \geq 1}\left(L_{\alpha} \rho (t)
L^{\dag}_{\alpha} -\frac{1}{2} \left\{L^{\dag}_{\alpha}L_{\alpha},\r(t)\right\} \right)\\
&=-i[H,\r(t)] + \sum_{\alpha \geq 1}\left(L_{\alpha} \rho (t)
L^{\dag}_{\alpha} -\frac{1}{2} \left\{L^{\dag}_{\alpha}L_{\alpha},\r(t)\right\} \right) \\
&= -i[H,\r(t)] + \{A,\rho(t)\} + \sum_{\alpha \geq 1}L_{\alpha} \rho (t)
L^{\dag}_{\alpha} \\
& =  (A-iH) \rho (t) + \rho (t)(A+ iH)+ \sum_{\alpha \geq 1}L_{\alpha} \rho (t) L^{\dag}_{\alpha} \ .
\label{eq:253'}
\end{align}
\ees
Now define $K_0 \equiv  [I+(A-iH)dt] + O[(dt)^2]$,
where in the end we will take the $dt\to 0$ limit to remove any residual $O[(dt)^2]$ terms.
Then:
\bes
\begin{align}
K_0\r(t)K_0^\dgr &= 
[I+(A-iH)dt]\r(t)[I+(A+iH)dt] + O[(dt)^2] \\
&= \r(t) + 
[(A-iH)dt] \rho (t) + \rho (t)[(A+ iH)dt] + O[(dt)^2]
 \ .
 \label{eq:367c}
\end{align}
\ees

%
Thus, using Eq.~\eqref{eq:253'}, Eq.~\eqref{eq:367c}, and defining $K_\a \equiv  L_{\alpha}\sqrt{dt}$:
\bes
\begin{align}
\rho(t+dt) &= \rho(t)+ \dot{\rho}(t)dt + O[(dt)^2]  \\
&= \rho(t) + [(A-iH)dt] \rho (t) + \rho (t)[(A+ iH)dt] + O[(dt)^2] + \sum_{\alpha \geq 1} (L_{\alpha}\sqrt{dt} ) \rho (t) (L^{\dag}_{\alpha}\sqrt{dt})  \\
&= \sum_{\alpha\geq 0} K_{\alpha} \r(t) K^{\dag}_{\alpha} + O[(dt)^2] \\
&\equiv \Phi[\r(t)] \ ,
\label{eq:367d}
\end{align}
\ees
which is in Kraus OSR form. However, to prove that this is a valid quantum map we still need to show that the set $\{K_\a\}_{\a\geq 0}$ satisfies the normalization condition. Indeed, we have:
\beq
 \sum_{\alpha\geq 0} K^{\dag}_{\alpha} K_{\alpha} = 
I+dt\left(2A  + \sum_{\a\geq 1} L^{\dag}_{\alpha}L_{\alpha} \right) + O(dt^2)
= I  + O(dt^2) \ ,
\label{eq:250'}
\eeq
where in the first equality we used the Hermiticity of $A$ and $H$, and in the second equality we used the definition of $A$.

Thus, we have shown that in the $dt\to 0$ limit the map $\Phi$ [Eq.~\eqref{eq:367d}] is a quantum map from $\r(t)$ to $\r(t+dt)$. Let $dt = \lim_{n\to\infty} t/n$, and consider the concatenated sequence of maps $\lim_{n\to\infty} \Phi^{\circ n}[\r(0)] = \Phi[\Phi[\cdots \Phi[\r(0)]]] = \Lambda_t\r(0)$, which is clearly equivalent to the solution of the Lindblad equation [since it maps $\r(0) \to \r(dt) \to \r(2dt)\to  \cdots\to \r(t)$], i.e., if $\dot{\r} = \mc{L}\r(t)$, with $\mc{L}$ the Lindbladian of Eq.~\eqref{eq:365a}, then $ \Lambda_t = e^{\mc{L}t}$. Since we have shown that $\Phi$ is a CP map, it remains to be shown that a concatenation of quantum maps is still a quantum map. This is true, since if $\Phi_1$ and $\Phi_2$ are quantum maps then
\beq
\Phi_2 \circ \Phi_1(\r) = \Phi_2[\Phi_1(\r)] = \sum_{\b} K'_\b \Phi_1(\r) {K'_\b}^\dgr = \sum_{\a\b} K'_\b K_\a(\r) K_\a^\dgr {K'_\b}^\dgr = \sum_\g K''_\g \r {K''_\g}^\dgr\ ,
\eeq  
where $K''_\g = K'_\b K_\a$, and $\sum_\g {K''_\g}^\dgr K_\g = \sum_{\a} K_\a^\dgr (\sum_\b {K'_\b}^\dgr K_\b) K_\a =\sum_{\a} K_\a^\dgr  K_\a = I$ as required. 



\subsection{Examples}

\subsubsection{Just $H$ for a single qubit: the Bloch equations}
\label{sec:Blocheq}

Consider the Lindblad equation with all $\g_\a = 0$, i.e., $\dot{\r} = -i[H,\r]$. This is just the Schr\"{o}dinger equation written for density matrices (also known as the Liouville-von Neumann equation). Let us solve it for the case of a single qubit. We can always write $H = h_0 I+ \sum_{i\in\{x,y,z\}}h_i \s^i$, with $\vec{h} = (h_x,h_y,h_z)\in\mathbb{R}^3$ (since the Pauli matrices with identity form a basis over $\mathbb{R}^4$ for all $2\times 2$ matrices). Thus, using $\rho=\frac{1}{2}(I+\vec{v}\cdot\vec{\sigma})$:
\beq
-i[H,\rho] = -\frac{i}{2}\sum_{i\in\{x,y,z\}}h_i [\s^i,\vec{v}\cdot\vec{\s}] = -\frac{i}{2}\sum_{i,j\in\{x,y,z\}}h_i v_j [\s^i,\s^j] = \sum_{i,j,k\in\{x,y,z\}}\varepsilon_{ijk} h_i v_j\s^k = (\vec{h}\times\vec{v})\cdot\vec{\s} \ .
\eeq
Since $\dot{\r} = \frac{1}{2}(\dot{\vec{v}}\cdot\vec{\sigma})$, we find
\beq
 \dot{v}\cdot\vec{\s} = 2(\vec{h}\times\vec{v})\cdot\vec{\s}\ ,
\label{eq:Blocheq}
\eeq
which are three coupled first order differential equations for the components of $\vec{v}$. These are known as the Bloch equations, and their solution has the Bloch vector $\vec{v}$ rotating around the vector $\vec{h}$ with a frequency equal to ${2}\|\vec{h}\|$, as is easily checked. For example, consider a rotation about the $v_x$ axis, i.e., let $\vec{h}=(h,0,0)$. Then Eq.~\eqref{eq:Blocheq} becomes:
$\dot{v}_x = 0$, $\dot{v}_y = -2hv_z$, and $\dot{v}_z = 2hv_y$. Differentiating again gives $\ddot{v}_y = -4h^2 v_y$. The solution of these equations is 
\bes
\begin{align}
v_x(t)&=v_x(0)\\
v_y(t) &=v_y(0)\cos(2ht)-v_z(0)\sin(2h t) \\
v_z(t) &=v_z(0)\cos(2h t)+v_y(0)\sin(2ht)\ .
\end{align}
\ees
The general case follows from this one by a reorientation of the axes to align with what we called the $v_x$ axis in the solution above.

\subsubsection{Phase Damping for a single qubit}
\label{sec:PD-Lind}

We already encountered the phase damping model in the Kraus OSR setting in Sec.~\ref{sec:PD}. Let us now study a Lindblad equation model that generates the same map.

Let $L_1 = \sigma^z = Z$, $\gamma_1 = \gamma$, $\g_{\a\geq 2}=0$, and $H = 0$. Thus,
\beq
\dot{\rho}(t) = \gamma(Z\rho Z^{\dagger} -\frac{1}{2}\lbrace Z^{\dagger} Z,\rho \rbrace) = \gamma(Z\rho Z - \rho)\ .
\label{eq:267}
\eeq
Using $\rho=\frac{1}{2}(I+\vec{v}\cdot\vec{\sigma})$, the left-hand side evaluates to $\dot{\rho} = \frac{1}{2} \dot{\vec{v}}\cdot\vec{\sigma}$. For the right hand side $Z\rho Z = \frac{1}{2}(I-v_x X -v_y Y +v_z Z)$, and we thus arrive at:
\beq
\frac{1}{2}\left(\dot{v_x}X +\dot{v_y}Y+\dot{v_z}Z \right)= -\g \left( v_x X + v_y Y\right) \ .
\eeq
Equating the two sides componentwise (multiply both sides by $X$, $Y$, or $Z$, and take the trace) gives:
\bes
\begin{align}
\dot{v}_x &= -2\gamma v_x \implies v_x(t) = v_x(0) e^{-2\gamma t} \\
\dot{v}_y &= -2\gamma v_y \implies v_y(t) = v_y(0) e^{-2\gamma t} \\
\dot{v}_z &= 0 \implies v_z(t) = v_z(0) \ .
\end{align}
\ees
We can see that, since $\gamma \geq 0$, the map is contractive, and the Bloch sphere collapses to the $v_z$-axis exponentially fast with time. In the limit $t\to\infty$, this simply projects every state directly to the $v_z$-axis, which is manifestly uninvertible.

We can now match the Lindblad equation solution to the Kraus OSR result from Sec.~\ref{sec:PD}, where we found the Kraus operators $K_0 = \sqrt{p} I$ and $K_1 = \sqrt{1-p} Z$, and found that the Bloch vector is mapped to
\beq
\vec{v}' = ((2p-1)v_x(0),(2p-1)v_y(0),v_z(0)) \equiv (v_x(t),v_y(t),v_z(t)) \ .
\eeq
The Lindblad phase damping result and the Kraus OSR thus have exactly the same effect provided we identify
\beq
2p-1 = e^{-2\gamma t} \implies p(t) = \frac{1}{2}(1+e^{-2\gamma t}) \ .
\eeq
The probability in this model approaches $1/2$ in the limit $t \rightarrow \infty$.

If we now allow $H\neq0$, i.e., solve the full Lindblad equation $\dot{\r} = -i[H,\r]+\gamma(Z\rho Z^{\dagger} - \rho)$, then the result in Sec.~\ref{sec:Blocheq} shows that this gives rise to a rotating Bloch ellipsoid that is simultaneously shrinking exponentially along its principal axis.

\subsubsection{Amplitude damping / Spontaneous Emission for a single qubit}
\label{sec:Lind-AD}

Likewise, we can construct a Lindblad equation for amplitude damping, which we encountered as a quantum map in Sec.~\ref{sec:bitflipmap}.

Let $L_1 = \sigma^- = \ketb{0}{1} = (\s^+)^\dgr$, $\gamma_1 = \gamma$, $\g_{\a\geq 2}=0$, and $H = 0$. Plugging these into the Lindblad equation we get:
\beq
\dot{\rho}(t) = \gamma \left(\sigma^- \rho \sigma^+ - \frac{1}{2}\lbrace \sigma^+ \sigma^-,\rho \rbrace \right) 
= \gamma \left(\ketb{0}{1}\rho\ketb{1}{0} -\frac{1}{2} \lbrace \ket{1}\bracket{0}\bra{1},\rho \rbrace\right)
\label{eq:282}
\eeq
Using $\rho=\frac{1}{2}(I+\vec{v}\cdot\vec{\sigma})$ we find, for the right-hand side:
\bes
\begin{align}
& \ketb{0}{1}\left[\frac{1}{2}(I+v_x X +v_y Y+v_z Z) \right]\ketb{1}{0} = \frac{1}{2}\left(\ketb{0}{0}-v_z \ketb{0}{0} \right) \\
& -\frac{1}{4} \ketb{1}{1}(I+v_x X +v_y Y+v_z Z)  = -\frac{1}{4}\left(\ketb{1}{1} +(v_x+iv_y)\ketb{1}{0} - v_z\ketb{1}{1}\right)\\
& -\frac{1}{4} (I+v_x X +v_y Y+v_z Z)\ketb{1}{1}  = -\frac{1}{4}\left(\ketb{1}{1} +(v_x-iv_y)\ketb{1}{0} - v_z\ketb{1}{1}\right) .
\end{align}
\ees
Adding up all these terms gives:
\beq
\frac{1}{2}\left(Z -\frac{1}{2}v_x X-\frac{1}{2}v_y Y-v_z Z\right) ,
\eeq
which we need to equate with $\frac{1}{2}\dot{\vec{v}}\cdot\vec{\sigma}$. Therefore:
\beq
\dot{v}_x = -\frac{1}{2}\g v_x\ , \quad \dot{v}_y = -\frac{1}{2}\g v_y\ , \quad \dot{v}_z = -\g (v_z-1) \ .
\eeq
The last of these is solved by writing $dv_z/(v_z-1) = -\g dt$ and integrating, to give $\ln(v_z-1) = -\g t + c$, i.e., $v_z(t) = c' e^{-\g t}+1$, so that $c'=v_z(0)-1$. Thus: 
\bes
\label{eq:270}
\begin{align}
v_x (t) &= v_x(0)e^{-\gamma t/2} \\
v_y (t) &= v_y(0)e^{-\gamma t/2}\\
v_z (t) &= 1+[v_z(0)-1]e^{-\gamma t} \ .
\end{align}
\ees
As $t \rightarrow \infty$, $v_x,v_y \rightarrow 0$ and $v_z \rightarrow 1$. This represents a contraction of the Bloch sphere to the north pole state $\ketb{0}{0}$. Eq.~\eqref{eq:270} also show that the contraction rate is twice as high along the $v_z$ axis than the $v_x$ and $v_y$ axes.

Now recall that in our Kraus OSR treatment of amplitude damping (Sec.~\ref{sec:bitflipmap}) we   had the Kraus operators $K_0 = \ketb{0}{0} + \sqrt{1-p}\ketb{1}{1}$ and $K_1=\sqrt{p}\ketb{0}{1}$, and found that the Bloch vector was mapped to $\vec{v}'=(\sqrt{1-p}v_x(0),\sqrt{1-p}v_y(0),(1-p)v_z(0)+p)$. The Lindblad amplitude damping result and the Kraus OSR thus have exactly the same effect provided we identify $p = 1-e^{-\gamma t}$. Thus, the probability of a transition from the excited state to the ground state increases exponentially with time, and in the limit $t \rightarrow \infty$ we have $p\to 1$.

Note that this dynamical description is not unique, as the Kraus map only fixes the discrete mapping from the initial to the final state, and there are many dynamical descriptions which will recreate the mapping. Markovian dynamics is only one of the possible evolutions.

\section{The Lindblad equation via coarse graining}
\label{sec:LE-CG}

In this section we provide an alternative analysis leading to the Lindblad equation. The derivation is longer than the one we saw in Sec.~\ref{sec:LE-deriv1}, but provides additional insight and generalizability. Our analysis follows Ref.~\cite{Lidar200135}, with some changes of notation as well as clarifications and minor corrections.

\subsection{Derivation}
\label{sec:deriv-LE-CG}

Let us start again with the Kraus OSR, and recall that the Kraus operators act on $\mc{H}_S$, i.e., $K_\a \in \mc{B}(\mc{H}_S)$. Let us introduce a fixed (time-independent) operator basis $\{F_i\}_{i=0}^{d_S^2-1}$ for $\mc{B}(\mc{H}_S)$, where $d_S = \dim(\mc{H}_S)$, such that $F_0=I$. We can then expand the Kraus operators in this basis:
\beq
K_\a(t) = \sum_{i=0}^{d_S^2-1} b_{i\a}(t) F_i \ ,
\eeq
where $b_{i\a}$ are the time-dependent elements of a (rectangular) $d_S^2\times d_B^2$-dimensional matrix $b$, and $d_B = \dim(\mc{H}_B)$. Then the Kraus OSR becomes:
\bes
\begin{align}
\r(t) &= \sum_{\a} K_\a(t) \r(0) K^\dgr_\a(t) = \sum_{ij} \chi_{ij}(t) F_i\r(0)F_j^\dgr \\
&= \chi_{00}(t)\r(0) + \sum_{i>0}[\chi_{0i}(t)\r(0)F_i^\dgr + \chi_{i0}(t)F_i\r(0)] + \sum_{i,j>0} \chi_{ij}(t) F_i\r(0)F_j^\dgr \ ,
\label{eq:272b}
\end{align}
\ees
where
\beq
\chi_{ij}(t) = \sum_\a b_{i\a}(t) b_{j\a}^*(t)\ ,
\eeq
i.e., $\chi = b b^\dgr$. It follows immediately that $\chi$ is positive semidefinite: $\bra{v}\chi\ket{v} = \| b^\dgr\ket{v}\|^2 \geq 0$. Note that $\chi$ is a $d_S^2\times d_S^2$ matrix.

Now consider the normalization condition:
\bes
\begin{align}
I &= \sum_\a K^\dgr_\a(t) K_\a(t) = \sum_{ij}\chi_{ij}(t)F_j^\dgr F_i \\
&= \chi_{00}(t)I + \sum_{i>0}(\chi_{0i}(t)F_i^\dgr + \chi_{i0}(t)F_i) + \sum_{i,j>0} \chi_{ij}(t)F_j^\dgr F_i \ .
\label{eq:274b}
\end{align}
\ees
We can use this to eliminate the $\chi_{00}\r(0)$ term from Eq.~\eqref{eq:272b}. Multiply Eq.~\eqref{eq:274b} first from the right by $\frac{1}{2}\r(0)$, then from the left, and add the resulting two equations:
\beq
\r(0) =  \chi_{00}(t)\r(0) + 
\frac{1}{2}\sum_{i>0} \left[\chi_{0i}(t)\left(F_i^\dgr\r(0)+\r(0)F_i^\dgr\right) + \chi_{i0}(t)\left(F_i\r(0)+\r(0)F_i\right)\right] 
+ \frac{1}{2}\sum_{i,j>0}\chi_{ij}(t)\left\{F_j^\dgr F_i,\r(0)\right\}\ .
\eeq
Subtracting this from Eq.~\eqref{eq:272b} yields:
\beq
\r(t)-\r(0) = \frac{1}{2}\sum_{i>0} \left[\chi_{i0}(t)\left(F_i\r(0)-\r(0)F_i\right) - \chi_{0i}(t)\left(F_i^\dgr\r(0)-\r(0)F_i^\dgr\right) \right] + \sum_{i,j>0}\chi_{ij}(t)\left(F_i\r(0)F_j^\dgr - \frac{1}{2}\left\{F_j^\dgr F_i,\r(0)\right\}\right) \ .
\label{eq:276}
\eeq
Let us now define
\beq 
Q(t) \equiv \frac{i}{2} \sum_{j>0} \chi_{j0}(t) F_j-\chi_{0j}(t)F_j^\dgr \ , 
\label{eq:Q}
\eeq
and note that $Q = Q^\dgr$, i.e., $Q$ is Hermitian.
Then we can rewrite Eq.~\eqref{eq:276} as:
\beq
\r(t)-\r(0) = -i [Q(t),\r(0)] + \sum_{i,j>0}\chi_{ij}(t)\left(F_i\r(0)F_j^\dgr - \frac{1}{2}\left\{F_j^\dgr F_i,\r(0)\right\}\right)\ .
\label{eq:278}
\eeq
This obviously resembles the Lindblad equation, but it relates the state at $t=0$ to the state at some arbitrary later time $t$, i.e., it still represents a quantum map. Indeed, everything we have done so far is exact and we have simply rewritten the Kraus OSR in a fixed operator basis. As a first step towards getting this closer to standard Lindblad form, let us diagonalize the $\chi$ matrix, which will allow us to rewrite the double sum in Eq.~\eqref{eq:278} as a single sum. We have already noted that $\chi\geq 0$, so that it can be diagonalized via some unitary matrix $u$: $\tilde{\g} = u\chi u^\dgr$,
where $\tilde{\g}$ is diagonal and positive semidefinite. Define $L_k = \sum_{j>0}u^*_{kj} F_j$, so that, using the unitarity of $u$:
\beq
F_i = \sum_{k>0}u_{ki} L_k\ ,
\label{eq:285}
\eeq
where the sum over $k>0$ excludes $L_0 = I$.
Thus, again using the unitarity of $u$:
\bes
\begin{align}
& \sum_{i,j>0}\chi_{ij} F_i\r(0)F_j^\dgr = \sum_{k,l>0} L_k \r(0) L_l^\dgr \sum_{i,j>0} u_{ki} \chi_{ij} (u^\dgr)_{jl} =
\sum_{k>0} {\g}_k L_k\r(0)L_k^\dgr \\
& \sum_{i,j>0}\chi_{ij} F_j^\dgr F_i = \sum_{k,l>0} L_l^\dgr L_k \sum_{i,j>0} u_{ki} \chi_{ij} (u^\dgr)_{jl} = \sum_{k>0} {\g}_k L_k^\dgr L_k \ ,
\end{align}
\ees
where ${\g}_k \geq 0$ are the eigenvalues of $\chi$. We can now rewrite Eq.~\eqref{eq:278} as:
\beq
\r(t)-\r(0) = -i [Q(t),\r(0)] + \sum_{k>0}{\g}_k(t)\left(L_k\r(0)L_k^\dgr - \frac{1}{2}\left\{L_k^\dgr L_k,\r(0)\right\}\right)\ .
\label{eq:281}
\eeq
This is as far as we can go towards the Lindblad equation without introducing an approximation.

Let us now take a step back and introduce a generator for the exact quantum map. I.e., let us write $\r(t) = \Lambda(t,0)[\r(0)]$, where
\beq
\Lambda(t,0) = T_+ e^{\int_0^t \mc{L}(s)ds}\ .
\eeq
Let $\tau$ denote a short time interval, where the meaning of short will become clear momentarily. We define a ``coarse-grained" generator $\mc{L}_j$ as follows:
\beq
\mc{L}_j = \frac{1}{\tau} \int_{j\tau}^{(j+1)\tau} \mc{L}(s)ds \ .
\eeq
Then $\frac{1}{\tau}\int_0^t \mc{L}(s)ds = \sum_{j=0}^{n-1} \mc{L}_j$ provided $t=n\tau$, so that
\beq
\Lambda(t,0) = T_+ e^{ \tau\sum_{j=0}^{n-1} \mc{L}_j}\ .
\eeq

We now make a (strong) assumption:
\begin{myassumption}
The coarse-grained generators belonging to different time intervals commute:
\beq
[\mc{L}_j,\mc{L}_k] = 0 \quad \forall j,k\ .
\label{eq:286}
\eeq
\end{myassumption}
This assumption amounts to there being no memory of the evolution from one interval to the next.\footnote{It is an interesting open problem to derive rigorous conditions for this to hold from first principles.} Under this assumption, which we can also understand as a Markovian approximation, the time-ordered exponential becomes a product of exponentials:
\beq
\Lambda(t,0) = \prod_{j=0}^{n-1} e^{ \tau\sum \mc{L}_j} \equiv \prod_{j=0}^{n-1} \Lambda_j\ .
\eeq
Thus, $\r_{j+1} = \Lambda_j[\r_j]$, where $\r_j \equiv \r(j\tau)$, or, after Taylor expansion:
\beq
\r_{j+1} = \left(I+\tau \mc{L}_j + O(\tau^2)\right)\r_j \implies \frac{\r_{j+1}-\r_j}{\tau} = \mc{L}_j\r_j \ 
\label{eq:289a}
\eeq
where we dropped the higher order corrections subject to the following, additional assumption:
\begin{myassumption}
\beq
\label{eq:tau-cond}
\tau\| \mc{L}_j\| \ll 1 \quad \forall j \ .
\eeq
\end{myassumption}
Note that Eq.~\eqref{eq:tau-cond} sets an upper bound on $\tau$ in terms of the largest eigenvalue of the coarse-grained Lindblad generator. This eigenvalue determines the fastest timescale for the system evolution (we'll see later that these eigenvalues are all possible differences of energies, i.e., they correspond to transition frequencies). Thus, Eq.~\eqref{eq:tau-cond} can also be interpreted as stating that the coarse-graining timescale should be small compared to the timescale over which $\r_j$ changes.

Eq.~\eqref{eq:289a} implies that, in particular, for $j=0$:
\beq
\frac{\r(\tau)-\r(0)}{\tau} = \mc{L}_0[\r(0)] \ .
\label{eq:289}
\eeq

\begin{mylemma}
\beq
\chi_{ij}(0) = \d_{i0}\d_{j0} \ .
\eeq
\end{mylemma}
\begin{proof}
Using $U(t)=e^{-iHt}$, we have for the Kraus operators:
\bes
\begin{align}
K_\a(0) &= b_{0\a}(0)I + \sum_{i>0}b_{i\a}(0)F_i \\
&= \sqrt{\lambda_\nu}\bra{\mu}U(0)\ket{\nu} =  \sqrt{\lambda_\nu}\delta_{\mu\nu} I \ ,
\end{align}
\ees
so that [recall that $\a = (\m\n)$]
\beq
b_{i\a}(0) = \sqrt{\lambda_\nu}\delta_{\mu\nu}\d_{i0}\ .
\eeq
Therefore
\beq
\chi_{ij}(0) = \sum_\a b_{i\a}(0)b_{j\a}^*(0) = \sum_\n \lambda_\n \d_{i0}\d_{j0}  \ ,
\eeq
which proves the lemma, since $\sum_\n \lambda_\n = 1$.
\end{proof}
It follows immediately that $\chi(0)$ is already diagonal, and its eigenvalues are ${\g}_0(0)=1$ and ${\g}_{k>0}(0)=0$. It also follows immediately from Eq.~\eqref{eq:Q} that $Q(0)=0$.

Now define
\beq
\ave{X}_j \equiv \frac{1}{\tau}\int_{j\tau}^{(j+1)\tau} X(s) ds \ .
\label{eq:aveX}
\eeq
Then 
\bes
\begin{align}
\ave{\dot{\rho}}_0 &= \frac{\r(\tau)-\r(0)}{\tau} \\
\ave{\dot{Q}}_0 &= \frac{Q(\tau)-Q(0)}{\tau} = \frac{Q(\tau)}{\tau} \\
\ave{\dot{\g}_k}_0 &= \frac{\g_k(\tau)-\g_k(0)}{\tau} = \frac{\g_k(\tau)-\d_{k0}}{\tau} \ .
\end{align}
\ees
We can therefore rewrite Eq.~\eqref{eq:281} as:
\beq
\frac{\r(\tau)-\r(0)}{\tau} = -i [\frac{Q(\tau)}{\tau},\r(0)] + \sum_{k>0}\frac{{\g}_k(\tau)-\d_{k0}}{\tau}\left(L_k\r(0)L_k^\dgr - \frac{1}{2}\left\{L_k^\dgr L_k,\r(0)\right\}\right)\ ,
\label{eq:281-mod}
\eeq
which must equal $\mc{L}_0[\r(0)]$ by Eq.~\eqref{eq:289}.
Hence, we can read off $\mc{L}_0$: 
\beq
\mc{L}_0[X] = -i [\ave{\dot{Q}}_0,X] + \sum_{k>0}\ave{\dot{\g}_k}_0\left(L_k X L_k^\dgr - \frac{1}{2}\left\{L_k^\dgr L_k,X\right\}\right)\ ,
\eeq
This generator is precisely in Lindblad form. However, it only connects $\r(0)$ to $\r(\tau)$. In order to connect $\r(j\tau)$ to $\r((j+1)\tau)$ we may now postulate that the same generator form remains valid, i.e., that 
\beq
\mc{L}_j[X] = -i [\ave{\dot{Q}}_j,X] + \sum_{k>0}\ave{\dot{\g}_k}_j\left(L_k X L_k^\dgr - \frac{1}{2}\left\{L_k^\dgr L_k,X\right\}\right)\quad \forall j \ ,
\label{eq:292}
\eeq
which we can do as long as Eq.~\eqref{eq:286} is satisfied. The simplest way to ensure this is to demand that in fact 
\beq
\mc{L}_j = \mc{L}_0\quad  \forall j\ .
\label{eq:293}
\eeq
This is again the Markovian limit, where there is no memory of the previous evolution segment. If, instead, we keep the more general form of Eq.~\eqref{eq:292} [again, subject to Eq.~\eqref{eq:286}], then we have a \emph{time-dependent Markovian process}, where the generator is allowed to change over time, as long as these changes are uncorrelated between different time-segments.

Retaining the time-independent Markovian form of Eq.~\eqref{eq:293}, 
and further replacing $\ave{\dot{\r}}_j$ by $\dot{\r}$ (another approximation, that becomes exact in the limit $\tau\to 0$),
we finally have the following result for the coarse-grained Lindblad equation, representing a time-independent Markovian limit:
\beq
\boxed{
\dot{\rho}(t) = -i [\ave{\dot{Q}}_0,\rho(t)] + \sum_{k>0}\ave{\dot{\g}_k}_0\left(L_k \rho(t) L_k^\dgr - \frac{1}{2}\left\{L_k^\dgr L_k,\rho(t)\right\}\right)
}
\label{eq:295}
\eeq

One point remains, which is to show that the coefficients $\ave{\dot{\g}_k}_0$ are non-negative, which is a requirement for complete positivity of the map generated by the Lindblad equation. To show this, note that
\beq
\ave{\dot{\g}_k}_0 = \frac{1}{\tau}\int_0^\tau \dot{\g}_k(t)dt = \frac{1}{\tau}\left({\g}_k(\tau)-{\g}_k(0)\right) \ .
\label{eq:296}
\eeq
We already know that ${\g}_k(t)\geq 0$ $\forall t$ (recall that these are the eigenvalues of $\chi$), so we need to show that nothing is spoiled by subtracting ${\g}_k(0)$. But, this is true since we already showed above that ${\g}_{k>0}(0)=0$. Thus, Eq.~\eqref{eq:296} shows that the coefficients are all non-negative, as required for the Lindblad equation.



\subsection{Interaction picture}
\label{sec:IP}

As a brief digression, let us review the interaction picture, in preparation for the example we shall study in the next subsection.

Consider a (time-dependent) Hamiltonian $H$ of the form: 
\begin{equation}
H(t)=H_{0}(t)+V(t)\ .  
\label{eq:ht}
\end{equation}%
The unitary evolution operators satisfy:%
\bes
\begin{align}
\frac{dU(t)}{dt}&=-iH(t)U(t)  \label{eq:origU}\\
\frac{dU_{0}(t)}{dt}&=-iH_{0}(t)U_{0}(t)\ .  \label{eq:U0}
\end{align}
\ees
Define the interaction picture propagator with respect to $H_{0}$ via:%
\begin{equation}
\tilde{U}(t)=U^{\dag }_{0}(t)U(t,0)\ .  \label{eq:int1}
\end{equation}%
\begin{myclaim}
$\tilde{U}$ satisfies the Schr\"{o}dinger equation 
\begin{equation}
\frac{d\tilde{U}(t)}{dt}=-i\tilde{H}(t)\tilde{U}(t),  \label{eq:int2}
\end{equation}%
with the interaction picture Hamiltonian%
\begin{equation}
\tilde{H}(t)=U_{0}^{\dagger }(t)V(t)U_{0}(t).  \label{eq:int3}
\end{equation}%
\end{myclaim}

\begin{proof}
Differentiate both sides of Eq.~\eqref{eq:int1}, while making use of Eqs.~\eqref{eq:ht},~\eqref{eq:int2} and \eqref{eq:int3}: 
\begin{eqnarray}
\frac{d\tilde U(t)}{dt} &=&\frac{d\left[ U^\dag_{0}(t){U}(t)\right] }{dt} = \dot U_0^\dagger U + U_0^\dag \dot U = iU_0 H_0 U + U_0^\dag (-i H U) \notag \\
&=& iU_0 H_0 U -i U_0^\dag (H+V) U_0 \tilde U = -i  U_0^\dag V U_0 \tilde U= -i \tilde H \tilde U\ .
\end{eqnarray}%
The initial
conditions of the equations are also the same [$U(0)=I$], thus Eqs.~\eqref
{eq:int1}-\eqref{eq:int3} describe the propagator generated by $H(t)$.
\end{proof}

To make contact with open quantum systems, let $V = H_{SB}$ and $H_0 = H_S+H_B$. Then $U_0 = e^{-itH_S}\ox e^{-itH_B}$. We can now transform the Schr\"{o}dinger picture density matrix to the interaction picture via $\tilde{\r}_{SB}(t) = U_0^\dgr(t)\r_{SB}(t)U_0(t)$, and if we write $H_{SB} = \sum_a \lambda_a S_a\ox B_a$ ($S_a$ and $B_a$ are system-only and bath-only operators, respectively), then $\tilde{H}_{SB}(t) = \sum_a \lambda_a S_a(t)\ox B_a(t)$, where $S_a(t) = e^{itH_S}S_a e^{-itH_S}$ and $B_a(t) = e^{itH_B}B_a e^{-itH_B}$. This interaction picture density matrix satisfies 
\beq
\tilde{\r}_{SB}(t) = \tilde{U}(t){\r}_{SB}(0)\tilde{U}^\dgr(t)
\eeq
(note that the Schr\"{o}dinger picture and the interaction picture coincide at $t=0$).

At this point everything we've shown for quantum maps and the Lindblad equation carries through with appropriate modifications. The Kraus OSR in the interaction picture becomes
\beq
\tilde{\r}(t) = \Tr_B [\tilde{\r}_{SB}(t)] = \sum_\a \tilde{K}_\a (t) \r(0)\tilde{K}^\dgr_\a (t) \, 
\eeq
where the interaction picture Kraus operators are
\beq
\tilde{K}_\a (t) = \sqrt{\lambda_\n} \bra{\m}\tilde{U}(t)\ket{\n} \ .
\eeq
The interaction picture Lindblad equation, replacing Eq.~\eqref{eq:295}, becomes:
\beq
\dot{\tilde{\rho}}(t) = -i [\ave{\dot{\tilde{Q}}}_0,\tilde{\rho}(t)] + \sum_{k>0}\ave{\dot{\tilde{\g}}_k}_0\left(L_k \tilde{\rho}(t) L_k^\dgr - \frac{1}{2}\left\{L_k^\dgr L_k,\tilde{\rho}(t)\right\}\right)\ ,
\label{eq:LE-IP}
\eeq
where $\tilde{Q} = Q-H_S$ and $\tilde{\g}_k$ are the eigenvalues of the interaction picture $\chi$-matrix $\tilde{\chi} = \tilde{b}\tilde{b}^\dgr$, with $\tilde{b}$ the expansion matrix of the interaction picture Kraus operators: $\tilde{K}_\a (t) = \sum_i \tilde{b}_{i\a}(t) F_i$.

\subsection{Example: the spin-boson model for phase damping}
\label{sec:spin-boson-1q}

To illustrate the predictions of the coarse-grained Lindblad equation, we consider the spin-boson model for phase damping of a single qubit, described by the Hamiltonian
\bes
\label{eq:spin-boson-1q}
\begin{align}
H &= H_S + H_B + H_{SB} \\
H_S &= -\frac{1}{2}g Z \ , \quad H_B = \sum_k \o_k (n_k+1/2) \ , \quad H_{SB} = Z\ox \left(\sum_k \lambda_k b_k + \lambda^*_k b^\dgr_k\right) \ ,
\end{align}
\ees
where $n_k=b_k^\dgr b_k$ and $b_k$ are the bosonic number and annihilation operator for mode $k$, respectively ($[b_k,b_l^\dgr]=\d_{kl}I$). Here $H_{SB}$ describes coupling of the qubit phase to the position $x$ of each oscillator; recall that quantization means replacing $x$ by $\left(b + b^\dgr\right)/\sqrt{2m\o} $ (where $m$ is the oscillator mass), so that 
\beq
\lambda_k \propto 1/\sqrt{\o_k} \ ,
\label{eq:311}
\eeq
a relation we will need later. In the interaction picture, it is easy to show that:\footnote{Some basic quantum mechanics would make this process very simple. Note that $[b,n]=b$ gives $bn=(n+1)b$. And therefore we would have $be^{n}=e^{n+1}b$.}
\beq
\tilde{H}_{SB}(t) = Z\ox \left(\sum_k \lambda_k e^{-i\o_k t} b_k + \lambda^*_k e^{i\o_k t} b^\dgr_k\right)\ .
\label{eq:HSB-sb}
\eeq

Assume that the bath is initially in a thermal Gibbs state at inverse temperature $\b=1/T$: $\r_B(0) = e^{-\b H_B}/Z$ [Eq.~\eqref{eq:Gibbs1}], and let $\ave{X}_B \equiv \Tr(X\r_B)$. It is then a standard exercise to show that 
\beq
\ave{b_k^\dgr b_l}_B = \d_{kl} \frac{1}{e^{\b\o_k}-1}\ , \quad 
\ave{b^\dgr_k}_B = \ave{b_k}_B = \ave{b_k b_l}_B = \ave{b^\dgr_k b^\dgr_l}_B = 0 \ .
\eeq

Using this, it can be shown that the coarse-grained, interaction picture Lindblad equation Eq.~\eqref{eq:LE-IP} becomes \cite{Lidar200135}:
\beq
\dot{\tilde{\rho}}(t) =  \g(\tau) \left(Z \tilde{\rho}(t) Z - \tilde{\rho}(t)\right)\ ,
\label{eq:LE-IP2}
\eeq
i.e., $\ave{\dot{\tilde{Q}}}_0=0$, $L_1=Z$, and there are no other Lindblad operators (as should be obvious from the form of $H_{SB}$ above), and where
\beq
\g(\tau) = \pi \sum_k|\lambda_k|^2 \coth(\b\o_k/2)\bar{\d}(\o_k,\tau)
\label{eq:g-CGLE}
\eeq
is the dephasing rate, where we have defined
\beq
\bar{\d}(\o,\tau)\equiv \frac{1}{\pi} \tau \sinc^2(\o\tau/2)\ .
\eeq

\begin{figure}[t]
	\includegraphics[width=0.5\linewidth]{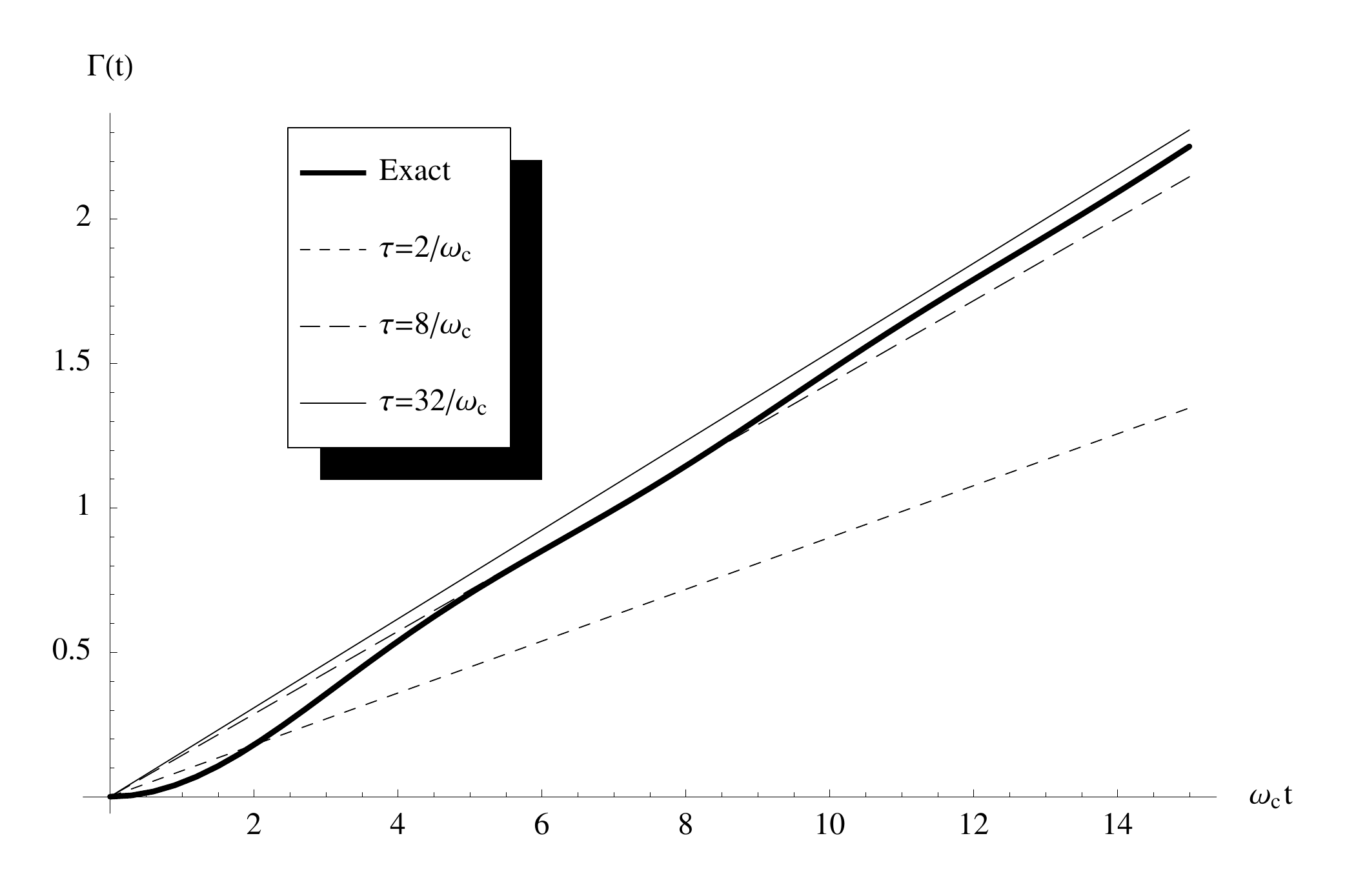}
\caption{Comparison of the exact solution of the spin-boson model for
single-qubit phase damping to the result obtained from the coarse-grained Markovian master 
equation. Plotted are the arguments $\Gamma(t)$ of the exponentials in Eq.~\eqref{eq:324}. Straight
lines correspond to the Markovian solution, which intersects the exact
solution (thick line) at $t=\tau$, as seen from
Eqs.~\eqref{eq:SME-Debye} and \eqref{eq:exact-Debye}. The bosonic bath density of states is represented by the Debye model [Eq.~\eqref{eq:Debye}]. The results shown correspond to $C=0.05$ and $\omega_c=1$. Reproduced from Ref.~\cite{Lidar200135}.}
\label{fig:cpfig}
\end{figure}

We already encountered Eq.~\eqref{eq:LE-IP2} in Sec.~\ref{sec:PD-Lind}, and as we saw there its solution for the coherence (off-diagonal elements) is
\beq
\tilde{\r}_{01}(t)  = e^{-2\g(\tau)t} {\r}_{01}(0)\ .
\eeq
As we shall see in Sec.~\ref{sec:SBmodel}, the spin-boson model we are considering here has an exact analytical solution. The exact solution for the coherence is:
\beq
\tilde{\r}_{01}^{(e)}(t)  = e^{-2\g(t)t} {\r}_{01}(0)\ .
\eeq
This allows us to compare the Markovian result to the exact one, and better understand the condition the coarse-graining timescale $\tau$ must satisfy. The only difference between the two is the argument of $\g$: $\tau$ versus $t$. However, this is a very significant difference, since while the Markovian solution represents irreversible exponential decay, the exact solution is oscillatory: $\g(t)t \sim \sum_k \sin^2(\o_k t)$. In order to observe closer agreement, we must once again invoke a continuous density of states $\g(\o)$, as we did in Sec.~\ref{sec:irr-dyn} [recall Eq.~\eqref{eq:Om}], which results in irreversible decay also in the case of the exact solution. Doing so replaces Eq.~\eqref{eq:g-CGLE} by
\beq
\g(\tau) = \pi \int_0^{\o_c} \Omega(\o) |\lambda(\o)|^2 \coth(\b\o/2)\bar{\d}(\o,\tau) d\o\ ,
\label{eq:g-CGLE2}
\eeq
where we assumed that $\Omega(\o)$ has a high-frequency cutoff at $\o_c$. Now note that $\bar{\d}$ behaves similarly to the Dirac-$\delta$ function:
\beq
\int_0^\infty \bar{\d}(\o,\tau)d\o = 1\ , \qquad \lim_{\tau\to\infty}\bar{\d}(\o,\tau) = \d(\o) \ ,
\eeq
i.e., it is sharply peaked at $\o=0$, and the peak becomes sharper as $\tau$ grows. The peak width is $\sim 1/\tau$. This suggests under what condition $\g(t)\approx\g(\tau)$, such that the exact and Markovian solutions agree: $\tau\gg1/\o_c$. The reason is that then $\int_0^{\o_c} $ captures nearly all the area under $\bar{\d}(\o,\tau)$, whereas in the opposite case ($\tau\lesssim1/\o_c$), most of the area under $\bar{\d}(\o,\tau)$ is not captured by the same integral. Thus, assuming $\tau\gg1/\o_c$, $\bar{\d}(\o,\tau)$ effectively behaves as a Dirac-delta function, and if we assume in addition that $t>\tau$, then certainly also $\bar{\d}(\o,t)$ behaves as a Dirac-$\delta$ function. Thus, assuming
\beq
t>\tau\gg1/\o_c\ ,
\eeq
we have 
\beq
\g(\tau) \approx \g(t) \approx \pi \int_0^{\o_c} \Omega(\o) |\lambda(\o)|^2 \coth(\b\o/2){\d}(\o) d\o\ ,
\eeq
so that the exact and Markovian cases agree. This is borne out numerically as well. Assume a Debye model, so that 
\beq
\label{eq:Debye}
\Omega(\omega )\propto \left\{ 
\begin{array}{c}
\omega^{2}\text{ for }\omega <\omega _{c} \\ 
0\text{ for }\omega \geq \omega _{c}
\end{array}
\right. \ , 
\eeq
and that $|\lambda (\omega )|^{2}\propto \omega ^{-1}$, in accordance with Eq.~\eqref{eq:311}. In
the high-temperature limit $\coth ({\beta\omega }/{2})\propto
\omega ^{-1}$, so that in all we have
\bes
\label{eq:324}
\begin{align}
\tilde{\rho}_{01}(t) &\propto \exp \left( -Ct\tau \int_{0}^{\omega
_{c}}d\omega {\sinc}^{2}\left( \omega \tau /2\right) \right)
\label{eq:SME-Debye} \\
\tilde{\rho}_{01}^{(e)}(t) &\propto \exp \left( -Ct^{2}\int_{0}^{\omega
_{c}}d\omega {\sinc}^{2}\left( \omega t/2\right) \right)\ ,
\label{eq:exact-Debye}
\end{align}
\ees
where $C$ is the temperature-dependent coupling-strength, with dimensions of
frequency. Figure~\ref{fig:cpfig} shows the argument of the exponentials in Eq.~\eqref{eq:324}, $\Gamma (t)$, for
the exact solution and for the coarse-grained Lindblad equation, corresponding to different
values of the course-graining time-scale, $\tau $. The curves corresponding
to the Markovian solutions are just straight lines, as they all describe
simple exponential decays. It is clear that the Markovian solutions cannot account
for the initial transition period, but for sufficiently large $\tau $
(in units of the bath cutoff time $1/\omega_c$) the
Lindblad result approximates the exact solution very well at large times.

To summarize, the Markovian approximation gives reliable
results for times greater than the coarse-graining time-scale, which in turn
must be greater than the inverse of the bath high-frequency cut-off. It does not account for the initial (Zeno-like) time evolution.


\section{Analytical solution of the spin-boson model for phase damping}
\label{sec:SBmodel}

We present the analytical solution of the spin-boson model for pure
dephasing. The derivation is based on \cite{Duan:97,Lidar200135}.

The model is the same as the one we considered in Sec.~\ref{sec:spin-boson-1q}, except that we will consider a system of multiple qubits (indexed by $i$). Starting from the interaction picture system-bath Hamiltonian [generalizing Eq.~\eqref{eq:HSB-sb}]:
\begin{equation}
\tilde{H}_{SB}(t)=\sum_{i,k}Z_i\otimes \left[ \lambda_{k}^{i}e^{-i\omega _{k}t}{a}_{k}+\left( \lambda _{k}^{i}\right) ^{\ast }e^{i\omega _{k}t}{a}_{k}^{\dagger }\right]\ ,
\end{equation}
we want to find the system density matrix 
\begin{equation}
\tilde{\rho}(t)={\rm Tr}_{B}\left[ \tilde{\rho}_{{\rm tot}}(t)\right] ={\rm Tr}_{B} 
\left[ \tilde{U}(t)\rho (0)\otimes \rho _{B}(0)\tilde{U}^{\dagger }(t)\right]\ ,
\end{equation}
where 
\beq
\tilde{U}(t)={T}_+\exp \left[ -i\int_{0}^{t}\tilde{H}(\tau )d\tau \right] . 
\eeq

\subsection{Calculation of the Evolution Operator}

Note that $\tilde{H}(t)$ does not commute with itself at different times,
which is why we need the time-ordered product: 
\bes
\begin{align}
\left[ \tilde{H}(t),\tilde{H}(t^{\prime })\right] &=  \sum_{i,i';k,k'} Z_i Z_{i'}\ox \lambda_k^i(\lambda_{k'}^{i'})^* e^{-i(\o_k t-\o_{k'}t')}[a_k,a_{k'}^\dgr] + Z_i Z_{i'}\ox (\lambda_k^i)^*\lambda_{k'}^{i'} e^{i(\o_k t-\o_{k'}t')}[a_k^\dgr,a_{k'}] \\
&= 2i\sum_{i,i^{\prime}}Z_iZ_{i'}\sum_k \Im\left[\lambda _{k}^{i}\left( \lambda _{k}^{i^{\prime }}\right) ^{\ast }e^{-i\o_k( t-t')}\right]\otimes  I_B
\label{eq:373}
\end{align}
\ees
where we used the canonical bosonic commutation relations $\left[ {a}_{k},{a}_{k'}^{\dagger }\right] = -\left[ {a}_{k}^{\dagger },{a}_{k'}\right] =I \delta _{kk'}$, $\left[ {a}_{k},{a}_{l}\right] =\left[ {a}^\dgr_{k},{a}^\dgr_{l}\right] =0$. Note that further, 
\begin{equation}
\left[ \left[ \tilde{H}(t),\tilde{H}(t^{\prime })\right] ,\tilde{H}(t^{\prime \prime })\right] =0.
\end{equation}
This means that we can use the Baker-Hausdorf formula $\exp (A+B)=\exp
(-[A,B]/2)\exp (A)\exp (B)$ (valid if $[[A,B],A]=[[A,B],B]=0$) to calculate $
{U}(t)$. To do so note the generalization 
\begin{equation}
\exp \left( \sum_{n}A_{n}\right) =\left( \prod_{n<n^{\prime }}\exp \left( -
\frac{1}{2}[A_{n},A_{n^{\prime }}]\right) \right) \left( \prod_{n}\exp
(A_{n})\right) \ ,
\end{equation}
which is valid if every second-order commutator vanishes. To apply this result for our case let us formally discretize the integrals and denote ${\cal H}_{n}\equiv -i\tilde{H}(n\Delta t)$. We let $\Delta t = t/N$ and take the limit $N\to\infty$. Then: 
\bes
\begin{align}
{U}(t) &=T_+\exp \left[ -i\int_{0}^{t}\tilde{H}(\tau )d\tau \right] =T_+\lim_{\Delta t\rightarrow 0}\exp \left[
\sum_{n=0}^{N}{\cal H}_{n}\Delta t\right]  \\
&=\lim_{\Delta t\rightarrow 0}\prod_{n<n^{\prime }}\exp \left( -\frac{1}{2}[
{\cal H}_{n},{\cal H}_{n^{\prime }}]\left( \Delta t\right) ^{2}\right)
\prod_{n}\exp ({\cal H}_{n}\Delta t) \\
&=\lim_{\Delta t\rightarrow 0}\prod_{n<n^{\prime }}\left( 1-\frac{1}{2}[
{\cal H}_{n},{\cal H}_{n^{\prime }}]\left( \Delta t\right) ^{2}\right)
\prod_{n}\left( 1-{\cal H}_{n}\Delta t\right)  \\
&=\lim_{\Delta t\rightarrow 0}\left[ 1-\frac{1}{2}\sum_{n<n^{\prime }}[
{\cal H}_{n},{\cal H}_{n^{\prime }}]\left( \Delta t\right) ^{2}\right] \left[
1-\sum_{n}{\cal H}_{n}\Delta t\right]  \\
&=\lim_{\Delta t\rightarrow 0}\exp \left( -\frac{1}{2}\sum_{n<n^{\prime }}[
{\cal H}_{n},{\cal H}_{n^{\prime }}]\left( \Delta t\right) ^{2}\right) \exp
\left( \sum_{n}{\cal H}_{n}\Delta t \right)  \\
&=\exp \left[ \frac{1}{2}\int_{0}^{t}dt_{1}\int_{0}^{t_{1}}dt_{2}\left[ \tilde{H}(t_{2}),\tilde{H}(t_{1})\right] \right] \exp \left[ -i\int_{0}^{t}\tilde{H}(\tau )d\tau \right] \ .
\label{eq:375f}
\end{align}
\ees
Note that in the second line we enforced time-ordering by keeping $n<n'$. To go from the third to the fourth line we kept the lowest relevant order in each term, inherited from the second line. Note how in the last line
the time-ordering is implemented via $t_2\leq t_1$. We find: 
\begin{equation}
-i\int_{0}^{t}\tilde{H}(\tau )d\tau =\sum_{i}Z_{i}\otimes\sum_k \left( (\alpha _{k}^{i})^{\ast }{a}_{k}-\alpha _{k}^{i}{a}_{k}^{\dagger}\right) \ ,
\end{equation}
where
\begin{equation}
\alpha _{k}^{i}(t)=\frac{\left( \lambda _{k}^{i}\right) ^{\ast }(e^{i\omega
_{k}t}-1)}{ \omega _{k}}\ .
\end{equation}

Now, since 
\begin{equation}
\int_{0}^{t}dt_{1}\int_{0}^{t_{1}}dt_{2}e^{-i\o_k (t_2-t_1)} = \int_{0}^{t} dt_1 e^{i\o_{k}t_1} \frac{e^{-i\o_k t_1}-1}{-i\o_k} = \frac{1-e^{i\o_k t}+i\o_k t}{\o_k^2}\ ,
\end{equation}
we have, using Eq.~\eqref{eq:373}:
\bes
\begin{align}
&-\frac{i}{2}\int_{0}^{t}dt_{1}\int_{0}^{t_{1}}dt_{2}\left[ \tilde{H}(t_{2}),\tilde{H}(t_{1})\right]   
=  \sum_{jj'}Z_jZ_{j'}\sum_k \Im\left[\lambda _{k}^{j}\left( \lambda _{k}^{j^{\prime }}\right) ^{\ast }\int_{0}^{t}dt_{1}\int_{0}^{t_{1}}dt_{2}e^{-i\o_k( t_2-t_1)}\right]\otimes  I_B\\
&=  \sum_{jj'}Z_jZ_{j'}\sum_k \Im\left[\lambda _{k}^{j}\left( \lambda _{k}^{j^{\prime }}\right) ^{\ast }\frac{1-e^{i\o_k t}+i\o_k t}{\o_k^2}\right]\otimes  I_B\ .
\end{align}
\ees
Therefore, defining  
\beq
f_{jj'}(t)\equiv \sum_k \Im\left[\lambda _{k}^{j}\left( \lambda _{k}^{j^{\prime }}\right) ^{\ast }\frac{e^{i\o_k t}-i\o_k t-1}{\o_k^2}\right]\ ,
\eeq
we can write the first term in Eq.~\eqref{eq:375f} as follows:
\beq
 \exp \left[ \frac{1}{2}\int_{0}^{t}dt_{1}\int_{0}^{t_{1}}dt_{2}\left[ \tilde{H}(t_{2}),\tilde{H}(t_{1})\right] \right] = e^{i\sum_{jj'} f_{jj'}(t)Z_jZ_{j'}}\ox I_B \ .
\eeq
Note that this is an operator acting non-trivially just on the system, and is a global phase for the case of a single qubit. Its action is, however, non-trivial for multiple qubits (it represents a Lamb shift).

Since the ${a}_k$ operators
commute for different modes we have as our final simplified result for the
evolution operator: 
\begin{equation}
\tilde{U}(t)=e^{if(t)}\prod_{i,k}\exp \left[ Z_i\otimes \left(
\alpha _{k}^{i}(t){a}_{k}-\alpha _{k}^{i}(t)^{\ast }{a}_{k}^{\dagger
}\right) \right] .
\end{equation}

\subsection{Calculation of the Density Matrix}

Now recall the definition of the coherent states. These are eigenstates of
the annihilation operator: 
\begin{equation}
{a}|\alpha \rangle =\alpha |\alpha \rangle \ .
\end{equation}
They are minimum-uncertainty states in a harmonic potential, and can be expanded as  
\begin{equation}
|\alpha \rangle =e^{-|\alpha |^{2}/2}\sum_{n=0}^{\infty }\frac{\alpha ^{n}}{
\sqrt{n!}}|n\rangle
\end{equation}
where $|n\rangle $ are number (Fock) states. The completeness relation for coherent states is: 
\begin{equation}
\frac{1}{\pi }\int d^{2}\alpha \,|\alpha \rangle \langle \alpha |=1
\end{equation}
where the integration is over the entire complex plane. They are useful in
our context since they are created by the displacement operator 
\begin{equation}
D\left( \alpha \right) \equiv \exp \left( \alpha {a}^{\dagger }-\alpha
^{\ast }{a}\right) =D(-\alpha )^{\dagger }
\end{equation}
acting on the vacuum state: 
\begin{equation}
D\left( \alpha \right) |{0}\rangle =|\alpha \rangle ,
\end{equation}
which is clearly related to ${U}(t)$. We will need the result: 
\begin{equation}
D\left( \alpha \right) D\left( \beta \right) =\exp \frac{\alpha \beta ^{\ast
}-\alpha ^{\ast }\beta }{2}D(\alpha +\beta ),
\end{equation}
which is easily derived from $D\left( \alpha \right) =\exp \left( \alpha 
{a}^{\dagger }-\alpha ^{\ast }{a}\right) $, $[{a,a}^{\dagger
}]=1 $, and the Baker-Hausdorf formula $\exp (A+B)=\exp (-[A,B]/2)\exp
(A)\exp (B) $ (again, valid if $[[A,B],A]=[[A,B],B]=0$).

Now let $R_{ik}(t)\equiv \alpha _{k}^{i}(t){a}_{k}^{\dagger }-\alpha
_{k}^{i}(t)^{\ast }{a}_{k}$ and consider $\exp \left[ Z_{i}\otimes R_{ik}(t)\right] $: 
\bes
\begin{align}
\exp \left[ Z\otimes R\right]  &=I_{S}\otimes
\sum_{n=0}^{\infty }\frac{R^{2n}}{(2n)!}+Z\otimes
\sum_{n=0}^{\infty }\frac{R^{2n+1}}{(2n+!)!}  =I_{S}\otimes \cosh R+Z\otimes \sinh R   \\
&=I_{S}\otimes \frac{1}{2}[D\left( \alpha \right) +D\left( -\alpha \right) ]+Z\otimes \frac{1}{2}[D\left( \alpha \right) -D\left( -\alpha
\right) ]  =|0\rangle \langle 0|\otimes D\left( \alpha \right) +|1\rangle \langle
1|\otimes D\left( -\alpha \right) \ .
\end{align}
\ees
This shows that depending on whether the
field is coupled to the qubit $|0\rangle $ or $|1\rangle $ state, the field
acquires a different displacement.\footnote{Note that this is the source of the dephasing the
qubits undergo, since when acting on a superposition state of a qubit, the
qubit and field become entangled: 
\[
\exp \left[ {\sigma }_{z}\otimes R\right] (a|0\rangle +b|1\rangle )|\beta
\rangle  =a|0\rangle \otimes D\left( \alpha \right) |\beta \rangle  
+b|1\rangle \otimes D\left( -\alpha \right) |\beta \rangle  =e^{(\alpha \beta ^{\ast }-\alpha ^{\ast }\beta )/2}a|0\rangle \otimes
|\alpha +\beta \rangle +e^{-(\alpha \beta ^{\ast }-\alpha ^{\ast }\beta )/2}b|1\rangle \otimes
|\beta -\alpha \rangle \ .
\]
}
The evolution operator can thus be written as: 
\begin{equation}
{U}(t)=e^{i\sum_{jj'} f_{jj'}(t)Z_jZ_{j'}}\prod_{i,k}\left[ |0\rangle _{i}\langle 0|\otimes
D\left( \alpha _{k}^{i}\right) +|1\rangle _{i}\langle 1|\otimes D\left(
-\alpha _{k}^{i}\right) \right] .
\end{equation}
Now assume that the bosonic bath is in thermal equilibrium: 
\bes
\begin{align}
\rho _{B} &=\frac{1}{Z}e^{-\beta {H}_{B}}  =\left[ \prod_{k}\frac{e^{-\beta  \omega _{k}/2}}{1-e^{-\beta 
\omega _{k}}}\right] ^{-1}\exp \left( -\beta \sum_{k} \omega _{k}\left( 
{N}_{k}+\frac{1}{2}\right) \right)  = \prod_{k}\rho _{B,k}\ ,
\end{align}
\ees
where 
\beq
\rho _{B,k}=\frac{1}{\langle {N}_{k}\rangle }\exp \left( -\beta \omega _{k}{N}_{k}\right)\ ,
\eeq
and the mean boson occupation number is given by the Bose-Einstein distribution: 
\begin{equation}
\langle {N}_{k}\rangle =\frac{1}{e^{\beta  \omega _{k}}-1}.
\end{equation}
As shown in \cite{Gardiner:book}, p.122-3, this can be transformed into the
coherent-state representation, with the result: 
\begin{equation}
\rho _{B,k}=\frac{1}{\pi \langle {N}_{k}\rangle }\int d^{2}\alpha
_{k}\,\exp \left( -\frac{|\alpha _{k}|^{2}}{\langle {N}_{k}\rangle }
\right) |\alpha _{k}\rangle \langle \alpha _{k}|\ .
\end{equation}


For simplicity let us from now on consider the case of a single qubit. It suffices to
calculate the evolution of each of the four pure states $|x\rangle \langle
y| $, where $x,y\in\{0,1\}$, separately. Thus
\begin{eqnarray*}
\rho _{x,y}(t) &=&{\rm Tr}_{B}\left[ {U}(t)|x\rangle \langle y|\otimes
\rho _{B}(0){U}^{\dagger }(t)\right]  \\
&=&{\rm Tr}_{B}\left[ \prod_{k}\left[ |0\rangle \langle 0|\otimes
D\left( \alpha _{k}\right) +|1\rangle \langle 1|\otimes D\left( -\alpha
_{k}\right) \right]  |x\rangle \langle y|\otimes \prod_{m}\rho _{B,m}\prod_{l}\left[ |0\rangle \langle 0|\otimes D^{\dagger }\left(
\alpha _{l}\right) +|1\rangle \langle 1|\otimes D^{\dagger }\left( -\alpha
_{l}\right) \right] \right] .
\end{eqnarray*}
The terms in the three products match one-to-one for equal indices, so we
can write everything as a product over a single index $k$. Using ${\rm Tr}
(A\otimes B)={\rm Tr}A\times {\rm Tr}B$ to rearrange the order of the trace
and the products, and $D^{\dagger }\left( -\alpha \right) =D\left( \alpha
\right) $, we have:
\bes
\begin{align}
\rho _{x,y}(t) &=\delta _{x,0}\delta _{y,0}|0\rangle \langle
0| \otimes \prod_{k}{\rm Tr}\left[ D\left( \alpha _{k}\right) \rho
_{B,k}D\left( -\alpha _{k}\right) \right]  \\
&+\delta _{x,0}\delta _{y,1}|0\rangle \langle 1| \otimes \prod_{k}{\rm Tr}\left[ D\left( \alpha _{k}\right) \rho
_{B,k}D\left( \alpha _{k}\right) \right]  \\
&+\delta _{x,1}\delta _{y,0}|1\rangle \langle 0| \otimes \prod_{k}{\rm Tr}\left[ D\left( -\alpha _{k}\right) \rho
_{B,k}D\left( -\alpha _{k}\right) \right]  \\
&+\delta _{x,1}\delta _{y,1}|1\rangle \langle 1| \otimes \prod_{k}{\rm Tr}\left[ D\left( -\alpha _{k}\right) \rho
_{B,k}D\left( \alpha _{k}\right) \right] .
\end{align}
\ees
Consider the ${\rm Tr}$ terms:\ for $|0\rangle \langle 0|$ and $
|1\rangle \langle 1|$ by cycling in the trace the displacement operators
cancel and ${\rm Tr}\left[ \rho _{B,k}\right] =1$. Thus, as expected the
\ diagonal terms do not change:
\beq
\r_{0,0}(t) = \r_{0,0}(0)\ , \qquad \r_{1,1}(t) = \r_{1,1}(0)\ .
\eeq 

As for the off-diagonal terms: 
\begin{align}
{\rm Tr}\left[ D\left( \pm 2\alpha _{k}\right) \rho_{B,k}\right] &=\frac{1}{\pi \langle {N}_{k}\rangle }\int d^{2}\beta _{k}\,\exp
\left( -\frac{|\beta _{k}|^{2}}{\langle {N}_{k}\rangle }\right) 
\langle \beta _{k}|D\left( \pm 2\alpha _{k}\right) |\beta _{k}\rangle\ .
\end{align}
Now: 
\bes
\begin{align}
\langle \beta |D\left( \pm 2\alpha \right) |\beta \rangle  &= \exp \left[
\pm \left( \alpha \beta ^{\ast }-\alpha ^{\ast }\beta \right) \right]
\langle \beta |\pm 2\alpha +\beta \rangle  \\
&=\exp \left[ \pm \left( \alpha \beta ^{\ast }-\alpha ^{\ast }\beta \right) 
\right] \exp \left[ \beta ^{\ast }\left( \pm 2\alpha +\beta \right) -\frac{1}{2}
\left( |\beta |^{2}+|\pm 2\alpha +\beta |^{2}\right) \right]  \\
&=\exp \left( -2|\alpha |^{2}\pm 2\left( \alpha \beta ^{\ast }-\alpha
^{\ast }\beta \right) \right)\ .
\end{align}
\ees
Thus: 
\bes
\begin{align}
{\rm Tr}\left[ D \left( \pm 2\alpha _{k}\right) \rho _{B,k}\right] 
&=\exp \left( -2|\alpha _{k}|^{2}\right) \frac{1}{\pi \langle {N}
_{k}\rangle }\int d^{2}\beta _{k}\,\exp \left( -\frac{|\beta _{k}|^{2}}{\langle {N}
_{k}\rangle }\pm 2\left( \alpha _{k}\beta _{k}^{\ast }-\alpha _{k}^{\ast
}\beta _{k}\right) \right)  \\
&= \frac{\exp \left( -2|\alpha _{k}|^{2}\right)}{\pi \langle {N}
_{k}\rangle }\left[ \pi \langle {N}_{k}\rangle \exp \left( -4|\alpha
_{k}|^{2}\langle {N}_{k}\rangle \right) \right]  \\
&=\exp \left[ -4|\alpha _{k}|^{2}\left( \langle {N}_{k}\rangle +\frac{1
}{2}\right) \right]  \\
&=\exp \left[ -4\left| \frac{\lambda _{k}^{\ast }(e^{i\omega _{k}t}-1)}{
 \omega _{k}}\right| ^{2}\left( \frac{1}{e^{\beta  \omega _{k}}-1}+
\frac{1}{2}\right) \right]  \\
&=\exp \left[ -4|\lambda _{k}|^{2}\frac{1-\cos (\omega _{k}t)}{ 
\omega _{k} ^{2}}\coth \frac{\beta  \omega _{k}}{2}\right]\ .
\end{align}
\ees
Thus decay of the off-diagonal terms goes as $e^{-2\gamma (t)t}$, with 
\begin{equation}
\gamma(t) =2\sum_{k}|\lambda _{k}|^{2}\coth \frac{\beta \omega _{k}}{2}\frac{1-\cos( \omega _{k}t)}{\omega _{k}^{2}t} = \sum_{k}|\lambda _{k}|^{2}\coth \frac{\beta \omega _{k}}{2} t \text{sinc}^2\frac{\o_k t}{2}\ ,
\label{eq:401}
\end{equation}
which coincides with the exact result quoted in Sec.~\ref{sec:spin-boson-1q}, specifically Eq.~\eqref{eq:g-CGLE} with $\tau$ replaced by $t$.



\section{Quantum trajectories and unravelling the Lindblad equation}
\label{sec:trajectories}
 
Solving the Lindblad equation numerically is demanding. For a $d$-dimensional system Hilbert space, the density matrix is $d\times d$, involving $d^2-1$ real numbers that one must store and update at each time-step. Is there a more space-efficient alternative? It turns out that instead one can propagate a wavefunction (only $2d-1$ real numbers, so a quadratic savings), at the expense of introducing statistical averaging over many runs. A very interesting side-benefit of this so-called unravelling procedure is that each wavefunction undergoes a ``quantum trajectory", that can be correlated to an individual sequence of quantum events, whereas the density matrix instead corresponds to an ensemble of such events.
 
Let us write down the Lindblad equation [Eq.~\eqref{eq:Lindblad-eq}] in the following form:
\begin{equation}
     \dot\rho = -i[H,\rho] + \sum_{k=1}^{d^2}\gamma_k \left(L_k\rho L_k^\dag - \frac{1}{2}\{L_k^\dag L_k,\rho\}\right)
     \label{eq:318}
 \end{equation}
Here $L_k$ are the Lindblad operators and $\gamma_k$ are scalars. As is clear from the derivation presented in Sec.~\ref{sec:LE-CG}, the number of non-zero terms in the sum is at most $d^2$. If one sets $\|L_k\|=1$ then the scalars $\gamma_k$ can be understood as rates of the corresponding relaxation process.%
\footnote{ 
Here we use the operator norm $\|O\|$:
\[
     \|O\| = \text{max}_{|v\rangle: \langle v| v\rangle =1}\sqrt{\langle v| O^\dag O|v\rangle}
\]
This norm is the largest eigenvalue of $\sqrt{O^\dag O}$. For Hermitian $O$, it reduces to the largest absolute value of eigenvalues of $O$.}
 
 There are multiple ways we can proceed to study this equation:
 
 \begin{enumerate}
     \item Derive $\gamma_k, L_k$ given the description of open system;
     \item Find equivalent dynamics of the wavefunction $|\psi(t)\rangle$ 
      (in the closed system case the wavefunction is a $d$-dimensional vector over $\mathbb{C}$ such that $\langle \psi|\psi\rangle =1$; this time we will let its norm be arbitrary);
      \item Suppose that measurements are performed repeatedly on the system, and derive the equation for dynamics given a string of measurement outcomes.
 \end{enumerate}
Here we will address points 2 and 3. In a very narrow sense we will address 1, if the closed system $+$ measurement apparatus are thought of as an open system.

\subsection{Method summary}
\label{sec:traj-summary}

To begin, we rewrite the Lindblad equation, Eq.~\eqref{eq:Lindblad-eq}, in the form
\beq
\dot{\r} = -i\left(H_{\text C}\rho(t)-\r(t)H_{\text C}^\dgr\right) + \sum_{\alpha}\g_\a L_{\alpha} \rho (t)
L^{\dag}_{\alpha}\ ,
 \label{eq:Lindblad-eq2}
\eeq
where
\beq
H_{\text C} = H-\frac{i}{2}\sum_\a\g_\a L_\a^\dgr L_\a \ 
\eeq
is called the ``conditional Hamiltonian". Note that it is non-Hermitian. Consider the evolution of a pure state $\ket{\psi(0)}$ subject to $H_{\text C}$:
\beq
\ket{\psi(0)} \overset{H_{\text C}}{\longmapsto} e^{-iH_{\text C} t} \ket{\psi(0)} = \ket{\tilde{\psi}(t)} \ .
\eeq
Since $H_{\text C}$ is non-Hermitian, the norm of $\ket{\tilde{\psi}(t)}$ decreases over time (hence the tilde):
\bes
\begin{align}
\frac{d}{dt} \| \ket{\tilde{\psi}(t)}\|^2 &= \bra{\psi(0)} e^{iH_{\text C}^\dgr t}(iH_{\text C}^\dgr) e^{-iH_{\text C} t} + e^{iH_{\text C}^\dgr t}(-iH_{\text C}) e^{-iH_{\text C} t}\ket{\psi(0)} \\
&= i \bra{\psi(0)} e^{iH_{\text C}^\dgr t} (H_{\text C}^\dgr-H_{\text C})e^{-iH_{\text C} t} \ket{\psi(0)} \\
&= -\sum_\a \g_\a \bra{\psi(0)} e^{iH_{\text C}^\dgr t} L_\a^\dgr L_\a e^{-iH_{\text C} t} \ket{\psi(0)} \\
&= -\sum_\a \g_\a \|L_\a e^{-iH_{\text C} t}  \ket{\psi(0)}\|^2 \leq 0 \ .
\end{align}
\ees

The action of the other term in Eq.~\eqref{eq:Lindblad-eq2} can be viewed as inducing a ``quantum jump":
\beq
 \ket{\tilde{\psi}(t)} {\longmapsto} \frac{L_\a \ket{\tilde{\psi}(t)}}{\| L_\a \ket{\tilde{\psi}(t)} \|} = \ket{\psi(t)}
 \qquad \text{with probability}\ p_\a = \frac{\g_\a \|L_\a \ket{\tilde{\psi}(t)}\|^2}{\sum_\a \g_\a \|L_\a \ket{\tilde{\psi}(t)}\|^2}\ ,
\eeq
where the post-jump state $\ket{\psi(t)}$ is normalized, and $p_\a$ tells us the probability that the particular jump $L_\a$ was realized.

If a jump took place at time $t$, then the probability that the next jump takes place in the interval $(t,t+\tau]$ is given by 
\beq
\Pr(\text{jump in}\ (t,t+\tau]\ | \text{ jump at}\ t) = 1- \|e^{-iH_{\text C}\tau}\ket{\tilde{\psi}(t)}\|^2 \ .
\label{eq:jump-prob}
\eeq 
In this way, the probability of a second jump at $\tau=0$ is zero, but the probability increases exponentially as $\tau$ grows.

Putting these steps together one arrives at the following algorithm for evolution from $t=0$ to $t_f$:
\begin{enumerate}
\item Initialize the state as $\ket{\psi(0)}$, set $j=1$
\item Evolve under the conditional Hamiltonian: $\ket{\tilde{\psi}(t_{j})} = e^{-iH_{\text C} t} \ket{\psi(t_{j})}$
\item Perform a jump at $t_j+\tau$ with probability given by Eq.~\eqref{eq:jump-prob}: $\ket{\tilde{\psi}(t_j+\tau)} \longmapsto \frac{L_\a \ket{\tilde{\psi}(t_j+\tau)}}{\| L_\a \ket{\tilde{\psi}(t_j+\tau)} \|}$, with the index $\a$ chosen with probability $p_\a = \frac{\g_\a \|L_\a \ket{\tilde{\psi}(t_j+\tau)}\|^2}{\sum_\a \g_\a \|L_\a \ket{\tilde{\psi}(t_j+\tau)}\|^2}$ 
\item If a jump took place, advance $j$ to $j+1$: call the new (normalized state) $\ket{\psi(t_{j+1})}$ and set $t_{j+1} = t_j+\tau$
\item Return to step 2, unless $t_{j+1} \geq t_f$
\item Repeat $K$ times from step 1, calling the output from the $k^{\text{th}}$ round $\psi_k(t_f)$, and construct $\r(t_f) = \frac{1}{K}\sum_{k=1}^K \ketbra{\psi_k(t_f)}$, stop when $\r(t_f)$ has converged
\end{enumerate}

It turns out that this algorithm converges to the solution $\r(t_f)$ of the Lindblad equation at $t=t_f$ (see, e.g., Section 7.1 of Ref.~\cite{Breuer:book}, and also the proof below). Its major advantage is that, as mentioned above, it propagates wavefunctions rather than density matrices, thus resulting in a quadratic space savings. The error in the approximation of $\r(t_f)$ decreases as $1/\sqrt{K}$. By the ``no-free lunch theorem" it should be the case that it is sufficient to use $K$ on the order of the Hilbert space dimension, so that the total cost is conserved. However, in practice fewer repetitions may suffice, so that the quantum trajectories algorithm may in fact be more efficient than brute force solution of the Lindblad equation. 

Each sequence $\{\psi_k(0),\tilde{\psi}_k(t_1),\psi_k(t_1),\tilde{\psi}_k(t_2),\psi_k(t_2),\dots,\psi_k(t_j),\tilde{\psi}_k(t_j),\dots\}$ is a ``quantum trajectory". It describes a series of norm-decreasing evolutions interrupted by quantum jumps. This provides an interesting and insightful interpretation of what actually takes place during open quantum system evolution. Consider, e.g., generalized amplitude damping (Sec.~\ref{sec:gen-AD}). An atom undergoes spontaneous emission to its ground state, but due to thermal excitation it can repopulate its excited state. As we saw in Sec.~\ref{sec:Lind-AD}, the probability of a transition from the excited state to the ground state increases exponentially with time, which is in accordance with Eq.~\eqref{eq:jump-prob}. But now we see that the actual emission event is a ``jump", whereby the atom suddenly and discontinuously finds itself in the ground state. The process can also work in the opposite direction, and by absorbing energy from the bath, the atom can find itself in an excited state, etc.  The downward transition event is accompanied by the emission of a photon (by energy conservation), or phonon, or some other elementary excitation, which can be detected. And indeed, such quantum trajectories have been measured in quantum optics experiments (see, e.g., Ref.~\cite{brun:719} and references therein).

We now proceed to give a more careful and detailed description and analysis.

\subsection{Equivalent dynamics of the wavefunction}

\subsubsection{Naive form}

Starting over, we note that we can rewrite the Lindblad equation as follows, in the limit $dt\to 0$:
\bes
\begin{align}
     \rho(t+dt) &= \rho(t)-i[H,\rho(t)]dt - \sum_{k=1} \gamma_k   \frac{1}{2}\{L_k^\dag L_k,\rho(t)\}dt+ \sum_{k=1} \gamma_k L_k\rho(t) L_k^\dag dt \\
     &=
     e^{-iH_{\text{C}} dt}\rho(t) e^{iH_{\text{C}}^\dag dt} + \sum_{k=1} M_k \rho(t) M_k^\dag \quad \text{where} \\
     H_{\text{C}} &= H - \frac{i}{2}\sum_{k=1} \gamma_k L_k^\dag L_k, \quad M_k = \sqrt{\gamma_k dt}L_k \ .
 \end{align}
 \ees
Here again $H_{\text{C}}$ is the non-Hermitian conditional Hamiltonian. If we define $M_0 = e^{iH_{\text{C}}^\dag dt}$ then this is the standard channel decomposition that we started with:
\begin{equation}
    \rho(t+dt)  =\sum_{k=0} M_k\rho(t) M_k^\dag
    \label{eq:326}
\end{equation}
We note that instead of using a differential equation solver to obtain $\rho(t+dt)$, using the non-selective measurement formalism of Sec.~\ref{sec:non-selective} we can instead mathematically ``simulate" the above formula in the following way:
\begin{enumerate}
    \item choose $k\geq 0$ with probability $p_k = \Tr [M_k\rho(t) M_k^\dag]$;
    \item set $\rho(t+dt) = \frac{1}{\mathcal{N}} M_k\rho(t) M_k^\dag$;
    \item repeat for the next time step $dt$.
\end{enumerate}
This simulation uses random numbers $\{k\}$. It is easy to see that the expectation value of the density matrix at some later time $T$ is exactly the same as the solution of the master equation:
\begin{equation}
    \text{lim}_{dt \to 0}\text{Av}_{\{k\}}\rho(T,\{k\}) = \rho(T)
\end{equation}
Now we note that the whole process was linear with respect to $\rho(t) = \sum_i p_i |\psi_i(t)\rangle \langle \psi_i(t)|$. So we can work with the states instead! Generate a random number $i$ with probability $p_i$ given by initial conditions, so as to choose $\ket{\psi_i(0)}$ as the initial state (a pure state).  Then follow these instructions with normalized $|\psi(t)\rangle$ at each step to produce $|\psi(t+dt)\rangle$:
\begin{enumerate}
    \item choose $k\geq 0$ with probability $p_k=\langle \psi(t)|M_k^\dag M_k|\psi(t)\rangle $
    \item set $|\psi(t+dt)\rangle = \frac{1}{\mathcal{N}} M_k|\psi(t)\rangle$
    \item repeat for the next time step $dt$
\end{enumerate}
We have derived the equivalence, so we know that
\begin{equation}
    \text{lim}_{dt \to 0}\text{Av}_{\{k\},i}|\psi(T,\{k\},i)\rangle \langle \psi(T,\{k\},i)| = \rho(T) \ .
\end{equation}
Here we average over random numbers $\{k\},i$ to obtain the same density matrix as the solution of the master equation. For a small range of $k$ and simple operators $L_k$ this method already leads to substantial savings, as one never needs to store $d\times d$ matrices during the simulation, only $d$-dimensional vectors. However, note that to obtain the average in practice one needs to sample from $\{k\},i$ some number of times $K$, repeating the whole simulation. In principle $K$ can be as large as $d$, thus defeating the purpose of the method, but in practice one can observe convergence of the average with increasing $K$, e.g. by studying the dispersion of some observable
\begin{equation}
   D(O) =  \text{A}v_{\{k\},i}(\langle \psi(T,\{k\},i)| O|\psi(T,\{k\},i)\rangle)^2  \ .
\end{equation}
Convergence is often achieved for $K\ll d$.

\subsubsection{Telegraph noise form}
We consider a slightly different perspective that is essentially the same as above, but we note that $k=0$ corresponding to $M_0 = e^{iH_{\text{C}}^\dag dt}$ dominates the probability distribution for $k$ in the limit $dt \to 0$:
\begin{equation}
    \langle \psi(t)|M_0^\dag M_0|\psi(t)\rangle = 1-O(dt) =1 - \sum_{k=1} \gamma_k   \langle \psi(t)|L_k^\dag L_k|\psi(t)\rangle dt +O(dt^2)\label{prob}
\end{equation}
this means that one does not need to calculate $\langle \psi(t)|M_k^\dag M_k|\psi(t)\rangle$ every $dt$. One only calculates 
\begin{equation}
    p_{\text{no-jump}}=\langle \psi(t)|M_0^\dag M_0|\psi(t)\rangle\ ,
\end{equation}
and generates an auxilliary random variable {\sc{jump}}$=0,1$ with probability $p_{\text{no-jump}},1-p_{\text{no-jump}}$ respectively.  Only if {\sc{jump}}$=1$ we ask which $k$ actually happened.

Looking at Eq.~\eqref{prob} we see that at first the probability of a jump happening within an interval $[t,t+\tau]$ increases from $0$ linearly with $\tau$, and at large $\tau$ it approaches $1$ exponentially. The coefficient in front of the linear dependence is $|\psi\rangle$-dependent, but weakly so. There is a well-known stochastic process given by
\begin{equation}
  p_{\text{jump}} = r dt \ . 
\end{equation}
In other words, independent jumps occur with rate $r$ per unit of time. This process is called \emph{telegraph noise}. The simulation method described above is a quantum evolution interrupted by essentially independent jumps following a telegraph noise distribution. Below we will study a different type of noise.

\subsubsection{Stochastic Schr\"odinger equation approach}
What we did above was produce a map from a wavefunction $|\psi(t)\rangle$ plus a random variable $\xi$ to the wavefunction at the next time step $|\psi(t+dt)\rangle$. The way we proved that this map is equivalent to the original master equation is by observing that
 \begin{equation}
   \text{Av}_{\xi}|\psi(t+\Delta t,\xi)\rangle \langle \psi(t+\Delta t,\xi)| =|\psi(t)\rangle \langle \psi(t)| + \mathcal{L}(|\psi(t)\rangle \langle \psi(t)|) \Delta t + O(\Delta t) \ .
\end{equation}
Here $\mathcal{L}$ is the generator of the Lindblad equation we are trying to simulate.

Let us now demonstrate that the Lindblad equation can also be derived from a stochastic Schr\"odinger equation approach. For simplicity, let us consider a generator with just one Hermitian term:
\begin{equation}
    \mathcal{L} = A\rho A - \frac{1}{2}A^2 \rho - \frac{1}{2}\rho A^2 \ .
    \label{eq:334}
\end{equation}
Let the random variable $\xi$ actually be a stochastic function of time $\xi(t)$ on the interval $[t,t+\Delta t]$. Define the time-step for our trajectory as:
\begin{equation}
    |\psi(t +\Delta t )\rangle =e^{ iA\int_t^{t+\Delta t}\xi(\tau) d\tau} |\psi(t   )\rangle  \ .
    \label{eq:335}
\end{equation}
This is the solution of the differential equation:
\begin{equation}
    \frac{d}{dt}|\psi(t  )\rangle =iA\xi(t) |\psi(t   )\rangle  \ .
\end{equation}
We can do a second order Taylor series expansion of Eq.~\eqref{eq:335}. The average of the density matrix after our time-step is then given by:
\bes
 \begin{align}
   &\text{Av}_{\xi}|\psi(t+\Delta t,\xi)\rangle \langle \psi(t+\Delta t,\xi)| =|\psi(t)\rangle \langle \psi(t)| \\
   &\qquad +  \text{Av}_{\xi}\left( iA\int_t^{t+\Delta t}\xi(\tau) d\tau|\psi(t)\rangle \langle \psi(t)|   -i|\psi(t)\rangle \langle \psi(t)|A\int_t^{t+\Delta t}\xi(\tau) d\tau \right) \\
   &\qquad +\text{Av}_{\xi} \left[ \int_t^{t+\Delta t}\int_t^{t+\Delta t}\xi(\tau)\xi(\tau') d\tau d\tau'\left(A|\psi(t)\rangle \langle \psi(t)|A -  \frac{A^2}{2}|\psi(t)\rangle \langle \psi(t)| -|\psi(t)\rangle \langle \psi(t)| \frac{A^2}{2}\right) \right]+O(\Delta t^2) \ .
\end{align}
\ees
Note that before the choice of $\xi$ is made, we don't really know what the smallness of the next order in Taylor series is. First of all we set 
\beq\text{Av}_\xi \xi(\tau) =0\ ,
\eeq 
to get rid of the first order in $A$. We also define the correlation function 
\beq
C(\tau,\tau') = \text{Av}_\xi \xi(\tau)\xi(\tau') \equiv C(\tau-\tau')
\eeq 
to be translation-invariant in time (i.e., to depend only the difference $\tau-\tau'$). Together these two conditions define the first two moments of Gaussian stochastic random variable. We then have:
\bes
\begin{align}
&   \text{Av}_{\xi}|\psi(t+\Delta t,\xi)\rangle \langle \psi(t+\Delta t,\xi)| =|\psi(t)\rangle \langle \psi(t)| \\
 &\qquad  + \int_t^{t+\Delta t}\int_t^{t+\Delta t}C(\tau-\tau') d\tau d\tau'\left(A|\psi(t)\rangle \langle \psi(t)|A -  \frac{A^2}{2}|\psi(t)\rangle \langle \psi(t)| -|\psi(t)\rangle \langle \psi(t)| \frac{A^2}{2}\right) +O(\Delta t^2) \ .
\end{align}
\ees
We would like
\begin{equation}
    \int_t^{t+\Delta t}\int_t^{t+\Delta t}C(\tau-\tau') d\tau d\tau'\sim \Delta t \ .
\end{equation}
We note that this will be the case if $C(t)$ is peaked at $0$ with width $w\ll \Delta t$ and height $C_0$:
\begin{equation}
    \int_t^{t+\Delta t}\int_t^{t+\Delta t}C(\tau-\tau') d\tau d\tau'\approx wC_0\Delta t
\end{equation}
Setting $wC_0=1$ will recover the desired Lindblad generator $\mathcal{L}$ given in Eq.~\eqref{eq:334}. Since $w$ is the smallest timescale in the problem we can just choose
\begin{equation}
    C(t) = \delta(t)
\end{equation}
where $\delta(t)$ is the Dirac delta function.
We have proven:
\bes
\begin{align}
&   \text{Av}_{\xi}|\psi(t+\Delta t,\xi)\rangle \langle \psi(t+\Delta t,\xi)| =|\psi(t)\rangle \langle \psi(t)| \\
&\qquad   + \left(A|\psi(t)\rangle \langle \psi(t)|A -  \frac{A^2}{2}|\psi(t)\rangle \langle \psi(t)| -|\psi(t)\rangle \langle \psi(t)| \frac{A^2}{2}\right) \Delta t +O(\Delta t^2)\ .
\end{align}
\ees
Now the smallness of the remaining terms can be guaranteed as $O(\Delta t^2)$, and we have indeed recovered the Lindblad generator $\mathcal{L}$ given in Eq.~\eqref{eq:334}.

Using the same idea for the derivation, we can prove the equivalence between the original Lindblad equation~\eqref{eq:318}
and the following differential equation on $|\psi(t)\rangle$:
 \begin{equation}
    \frac{d}{dt}|\psi(t  )\rangle = -i(H -\sum_k L_k\xi_k(t)) |\psi(t   )\rangle , \quad \text{Av}_\xi \xi_k(t) \xi_m(t') =  \delta_{km}\delta(t-t')\sqrt{\gamma_k} \ .
\end{equation}
Here $\delta_{km}$ is the Kronecker delta function. The equivalence states that
\begin{equation}
     \text{Av}_{\xi}|\psi(T,\xi)\rangle \langle \psi(T,\xi)| = \rho(T)
\end{equation}
the limit is included in $\delta$-function and the definition of the differential equation, so no additional limit needs to be taken here. In practice, though, some discretization scheme needs to be applied and the numerical simulation uses $\int_t^{t+\Delta t}\xi(\tau) d\tau$ instead of the raw $\xi(t)$.

\subsubsection{Comparison between the telegraph noise and Stochastic Schr\"odinger equation approaches}
If we compare the Stochastic Schr\"odinger equation approach to the telegraph noise method, we find that $\xi^{\text{tel}}(t)$ is a sequence of randomly spaced peaks with $0$ in between. It is possible to arrange for the correlation function of that signal to be $C(t) = \delta(t)$, however the higher order correlation functions will be vastly different from the Gaussian noise that is usually used for stochastic differential equations. The defining characteristic of the Gaussian noise is that higher order correlations (or moments) are expressed via $C(t)$ according to Wick's theorem. Another way to think about it is that the Fourier transforms $\int_0^T \xi(t)e^{ikt} dt$ are i.i.d. random variables for each $k$ for Gaussian noise, but not for telegraph noise.

Let us discuss the properties of individual $|\psi(t,\xi)\rangle $ or $|\psi(t,\{k\})\rangle $ for a given realization of random variables, under the two approaches. One way to look at this is to take an observable $O$ (s.t. $\|O\|=1$) and follow its average:
\begin{equation}
  \langle\psi(t,\xi)|O |\psi(t,\xi)\rangle \quad \text{or} \quad \langle\psi(t,\{k\})|O |\psi(t,\{k\})\rangle\ .
\end{equation}
If the closed system evolution of the observable $O(t)$ has a characteristic frequency $\omega \sim \|[O,H]\|$ and the relaxation has the characteristic rates $r \sim $max$_k \gamma_k$, then there are two possible regimes: $\omega \ll r$ and $\omega \gg r$. The qualitative picture that we will see is as follows:
\begin{center}
\begin{tabular}{ |c| c| c| }
\hline
 ~ & $\omega \ll r$ & $\omega \gg r$ \\ 
 \hline
 telegraph & smooth curves interrupted  & rapid sine wave interrupted  \\  
 ~ & by discontinuities & by discontinuities \\
 \hline
 stochastic & noisy diffusive behaviur & noisy almost periodic behavior   \\
 \hline
\end{tabular}
\end{center}
Even though the two methods are both equivalent to the same master equation, other characteristics such as the dispersion $D(O)$ or the diffusion coefficient of individual trajectories vary between the two methods. Thus, we find very different visual behavior of individual trajectories. It is possible to interpolate between the two by chossing a non-Gaussian $\xi(t)$. We note that the results for a single trajectory are reminiscent of experimental measurements. We next make this analogy more precise.

\subsection{Weak measurements}
One way is to choose the distribution of the random process $\xi(t)$ in such a way that an individual trajectory $\langle\psi(t,\xi)|O |\psi(t,\xi)\rangle $ matches the measurement output $M(t)$ of some repeated measurement. However this is an unphysical approach. What we should be doing is to come up with a mapping $M(\xi)$ since $\xi$ contain the information about random choices made outside of the system, while $\langle\psi(t,\xi)|O |\psi(t,\xi)\rangle $ contains information ``private" to the system, something that has not been measured yet.

We note that the first method with the decomposition given in Eq.~\eqref{eq:326}
can be directly interpreted as a measurement where $k$ is an answer. The stochastic one requires some transformations before this can be done, as the width of the $\delta$-function is the smallest time-scale that is faster than the supposed data collection timescale. We do not know of any research that makes this connection. There is a lot of research connecting weak measurements with trajectories, which could be seen as such an interpretation of stochastic equations. The difference with telegraph noise is that every $M_k$ is close to identity $I$ with a small probability in front.


\section{Analytical solution of the general Lindblad equation}

In this section we discuss the analytical solution of the Lindblad equation in arbitrary dimensional (but finite) Hilbert spaces.

\subsection{The coherence vector}

Let us first introduce a ``nice" operator basis for $\mc{B}(\mc{H}_S)$, where $d=\dim(\mc{H}_S)$. Let $F_0 = I_S$ and choose $M$ other traceless, Hermitian operators $\{F_j\}_{j=1}^M$, where $M=d^2-1$, such that
\beq
\label{eq:F-reqs}
\Tr(F_j) = 0 \ , \quad \Tr(F_j F_k)=\d_{jk}\ , \quad F_j^\dgr = F_j \ .
\eeq
A common choice is the generators of $\text{su}(d)$ (just as in the single-qubit case we chose the Pauli matrices), but for our purposes the explicit form of the operator basis won't matter. Note that this is similar to what we did in Sec.~\ref{sec:deriv-LE-CG}, except that for later convenience we make our basis choice somewhat more explicit here.

We can now expand any operator in this basis, including the density matrix:
\beq
\r = \frac{1}{d}F_0 + \sum_{j=1}^M v_j F_j = \frac{1}{d}I + \vec{F}\cdot\vec{v} \ ,
\label{eq:333}
\eeq
where $\vec{v} = (v_1,\dots,v_M)^T \in \mathbb{R}^M$ is called the ``coherence vector" (a generalized Bloch vector), and $\vec{F} = (F_1,\dots,F_M)$ collects the operator basis into a vector. Thus the components of the coherence vector are
\beq
v_j = \Tr(\r F_j) \ .
\label{eq:v_j}
\eeq

In analogy to Eq.~\eqref{eq:Phi(v)} for the single qubit case, we shall see that as a consequence of the Lindblad equation $\dot{\r} = \mc{L}\r$, the coherence vector satisfies the first order, inhomogeneous differential equation
\beq
\dot{\vec{v}} = G\vec{v} + \vec{c} \ .
\label{eq:335}
\eeq
Moreover, the decomposition of $\mc{L}$ as 
\beq
\mc{L} = \mc{L}_H+\mc{L}_D\ ,
\eeq 
with
\bes
\begin{align}
\mc{L}_H[\cdot] &= -i[H,\cdot] \\
\mc{L}_D[\cdot] &= \sum_{ij} a_{ij}\left(F_i \cdot F_j -\frac{1}{2}\{F_j F_i,\cdot\}\right)\ ,
\label{eq:337b}
\end{align}
\ees
induces the decomposition of $G$ into $G=Q+R$, where $\mc{L}_H[\r] \leadsto Q\vec{v}$ and $\mc{L}_D[\r] \leadsto R\vec{v} + \vec{c}$.

To explain the form of the dissipative term given in Eq.~\eqref{eq:337b}, recall the original form given in Eq.~\eqref{eq:L_D}. Combine this with the unitary transformation between the operator basis and the Lindblad operators given in Eq.~\eqref{eq:285}, to see that we can always transform between the non-diagonal and diagonal forms of the Lindblad equation. This transformation preserves positivity, i.e., we know that the coefficient matrix $a\equiv(a_{ij})$ is positive semi-definite.

Note that the normalization convention we have chosen for the coherence vector is slightly different from the Bloch vector, since we did not divide $\vec{v}\cdot\vec{F}$ by $d$ in Eq.~\eqref{eq:333}. As a result, the coherence vector is confined to a sphere with a radius less than one. Recall that the purity $P = \Tr(\r^2)$ [Eq.~\eqref{eq:purity}] satisfies $P\leq 1$. Thus
\beq
1 \geq  \Tr(\r^2) = \Tr\left[\left(\frac{1}{d}I + \vec{F}\cdot\vec{v}\right)^2\right] = \frac{1}{d} + \sum_{ij} \Tr(F_iF_j)v_iv_j = \frac{1}{d}+\|\vec{v}\|^2\ ,
\eeq
i.e.,
\beq
0\leq \|\vec{v}\| \leq \left(1-\frac{1}{d}\right)^{1/2} \ .
\label{eq:v-bounded}
\eeq
The upper bound is saturated for pure states, which thus live on the surface of an $d^2-1$-dimensional sphere with radius $\left(1-\frac{1}{d}\right)^{1/2}$.

\subsection{Just the non-dissipative part}
Let us assume that $\mc{L}_D = 0$. In this case we have, starting from Eq.~\eqref{eq:v_j}:
\bes
\label{eq:338}
\begin{align}
\dot{v}_k &= \Tr\left(\dot{\r} F_k\right) = -i\Tr\left([H,\vec{F}\cdot\vec{v}]F_k\right) \\
&= -i\Tr\left(\sum_j (HF_jF_k-F_jH F_k) v_j \right) = i\sum_j \Tr\left(H[F_k,F_j]\right)v_j\\
&= (Q \vec{v})_k\ ,
\end{align}
\ees
i.e.,
\beq
\dot{\vec{v}} = Q\vec{v} \ ,
\label{eq:339}
\eeq
where
\beq
\label{eq:Qjk}
Q_{jk} \equiv i \Tr\left(H[F_j,F_k]\right) \ .
\eeq
Note that the appearance of the commutator $[F_j,F_k]$ is a good reason to use as an operator basis the generators of a Lie algebra, for which the commutator can be expressed in terms of the algebra's structure constants.

The matrix $M\times M$ dimensional $Q$ is clearly skew symmetric: $Q_{jk} = -Q_{kj}$, i.e., 
\beq
Q = -Q^T \ .
\eeq
The solution of Eq.~\eqref{eq:339} is
\beq
\vec{v}(t) = e^{Qt} \vec{v}(0) \equiv \Omega(t)\vec{v}(0) \ .
\eeq
The evolution operator $\Omega$ is orthogonal:
\beq
\Omega^T\Omega= e^{Q^T t}e^{Qt} = e^{-Q t}e^{Qt} = I \ ,
\eeq
where we used the skew-symmetry of $Q$. This immediately implies that the norm of the coherence vector is preserved: $\|\vec{v}(t)\|^2 = \vec{v}^T(0) \Omega^T\Omega\vec{v}(0) = \|\vec{v}(0)\|^2$.

Thus, the evolution of the coherence vector in the absence of the dissipative part $\mc{L}_D = 0$ is a rotation in $\mathbb{R}^M$, generated by $Q$. 

\subsection{Full Lindblad equation for the coherence vector}
Let us now assume that both $\mc{L}_H,\mc{L}_D \neq 0$. Starting again from Eq.~\eqref{eq:v_j}, and using Eq.~\eqref{eq:337b}, we have:
\bes
\label{eq:338}
\begin{align}
\dot{v}_k &=  (Q \vec{v})_k +\sum_{ij} a_{ij}\Tr\left[F_i  \left(\frac{1}{d}I + \vec{F}\cdot\vec{v}\right)F_j F_k -\frac{1}{2}\{F_j F_i,\frac{1}{d}I + \vec{F}\cdot\vec{v}\}F_k\right]  \\
&= (Q \vec{v})_k + \sum_l \sum_{ij} a_{ij}\Tr\left[\left(F_j F_k F_i  -  \frac{1}{2} (F_k F_j F_i + F_j F_i F_k)\right)F_l \right] v_l 
+ \frac{1}{d}\sum_{ij} a_{ij}\Tr(\left[F_i , F_j \right]F_k) \\
& =  [(Q + R)\vec{v}]_k + c_k\ ,
\end{align}
\ees
where
\bes
\label{eq:Randc}
\begin{align}
R_{kl} &\equiv \sum_{ij} a_{ij}\Tr\left[\left(F_i F_l F_j   -  \frac{1}{2} \left\{F_j F_i, F_l\right\} \right)F_k \right]\\
c_k &\equiv \frac{1}{d}\sum_{ij} a_{ij}\Tr(\left[F_i , F_j \right]F_k) \ .
\end{align}
\ees
Thus, we have established that Eq.~\eqref{eq:335} holds, with $G=Q+R$, and with $Q$, $R$, and $\vec{c}$ as given in Eqs.~\eqref{eq:Qjk} and \eqref{eq:Randc}, respectively.

\subsection{Solution for diagonalizable and invertible $G$}
\label{sec:diag-inv-G}

Equation~\eqref{eq:335} is a linear, first order, inhomogeneous differential equation. Solving it is a standard exercise in linear algebra. For simplicity, let us assume that $G$ is diagonalizable over $\mathbb{R}^M$ and also invertible. Neither of these assumptions holds in general, and we deal with the general case in the next subsection.

We look for a solution in the form
\beq
\vec{v}(t) = \vec{v}^{(0)}(t) + \vec{v}^{(\infty)} \ ,
\eeq
where $\vec{v}^{(0)}(t)$ is the homogeneous part and $\vec{v}^{(\infty)}$ is the inhomogeneous part. Let $\vec{x}^{(k)}$ and $\lambda_k$ represent the eigenvectors and (possibly degenerate and complex) eigenvalues of $G$, i.e.,
\beq
G \vec{x}^{(k)} = \lambda_k \vec{x}^{(k)}\ , \qquad k=1,\dots,M \ .
\eeq
It is then straightforward to check by direct differentiation and substitution that 
\bes
\begin{align}
\vec{v}^{(0)}(t) &= \sum_{k=1}^M s_k e^{\lambda_k t}  \vec{x}^{(k)} \\
\vec{v}^{(\infty)} &= -G^{-1}\vec{c} 
\end{align}
\ees
in the solution of Eq.~\eqref{eq:335}. Indeed:
\beq
\dot{\vec{v}} = G\vec{v}(t)+\vec{c} = G \vec{v}^{(0)}(t) + G \vec{v}^{(\infty)} + \vec{c}
= \sum_{k=1}^M s_k e^{\lambda_k t}  \lambda_k \vec{x}^{(k)}  -G G^{-1}\vec{c} + \vec{c} = \dot{\vec{v}}^{(0)}\ ,
\eeq
as required. The coefficients $s_k$ are determined by the initial condition $\vec{v}(0)$:
\beq
\vec{v}^{(0)}(0) = \sum_{k=1}^M s_k  \vec{x}^{(k)} = X \vec{s} \ ,\qquad \text{col}_k(X) = \vec{x}^{(k)}\ ,
\eeq
i.e., $X$ is the matrix whose columns are the eigenvectors of $G$. Also, $\vec{v}^{(0)}(0) = \vec{v}(0) - \vec{v}^{(\infty)}$. Thus
\beq
\vec{s} = X^{-1} (\vec{v}(0)+G^{-1}\vec{c}) \ .
\eeq

Now, since the eigenvalues are in general complex numbers, they can be decomposed as 
$\lambda_k = \Re(\lambda_k ) + i \Im(\lambda_k )$. The imaginary part describes a rotation of the coherence vector (though we can be sure that since this vector lives in $\mathbb{R}^M$, such rotations are ultimately described by an orthogonal (purely real) matrix). The real part is constrained by complete positivity and trace preservation to be non-positive, or else the norm of the coherence vector would not be bounded [recall Eq.~\eqref{eq:v-bounded}]. Thus, the overall behavior of the coherence vector is described by rotations at frequencies given by $\{\Im(\lambda_k )\}$, some of which are exponentially damped on a timescale given by the set of non-zero $\{\Re(\lambda_k )\}$.

\subsection{Solution for general $G$}
\label{sec:gen-G-sol}

The general case is where $G$ is not diagonalizable over $\mathbb{R}^M$, and may not be invertible. In this case we can still use a similarity transformation $S$ to transform $G$ into Jordan canonical form:
\beq
G_J = SGS^{-1} = 
\left(
\begin{array}{ccc}
 J_1 &   &   \\
   & \ddots &   \\
   &   & J_q \\
\end{array}
\right)\ ,
\eeq
where the $q$ Jordan blocks have the form
\beq
J_j = 
\left(
\begin{array}{ccccc}
 \mu _j & 1 &   &   &   \\
   & \mu _j & \ddots &   &   \\
   &   & \ddots & \ddots &   \\
   &   &   & \mu _j & 1 \\
   &   &   &   & \mu _j \\
\end{array}
\right) = \m_j I + K_j \ .
\eeq
The $\m_j$'s are the (possibly degenerate, complex) eigenvalues and $K_j$ are nilpotent matrices: $K_j^{d_j} = 0$, where $d_j$ is the dimension of $J_j$. When all $d_j=1$, $G$ is diagonalizable and $G_J$ reduces to the diagonalized form of $G$. 

Applying $S$ from the left to Eq.~\eqref{eq:335} yields
\beq
S\dot{\vec{v}} = SGS^{-1}S\vec{v} + S\vec{c} \quad \implies \quad \dot{\vec{w}} = G_J \vec{w} + \vec{c}'\ ,
\label{eq:354}
\eeq
where $\vec{w} = S\vec{v}$ and we defined $\vec{c}' = S\vec{c}$.
This is still a linear, first order, inhomogeneous differential equation. The different Jordan blocks don't couple, so we can solve this as a set of $q$ independent problems, and take the direct sum of all the sub-solutions. 

Consider first the case of a $2\times 2$ Jordan block, i.e., $d_j=2$. The homogeneous part becomes:
\bes
\begin{align}
\dot{\vec{w}}_j^{(0)} = \left(
\begin{array}{cc}
 \mu _j & 1 \\
   & \mu _j \\
\end{array}
\right)\vec{w}_j^{(0)} \ ,
\end{align}
\ees
i.e.,
\bes
\begin{align}
\dot{w}_{j,1}^{(0)} & = \m_j {w}_{j,1}^{(0)} + {w}_{j,2}^{(0)} \\
\dot{w}_{j,2}^{(0)} & = \m_j {w}_{j,2}^{(0)}  \ .
\end{align}
\ees
Solving the second of these yields ${w}_{j,2}^{(0)}(t) = e^{\m_j t} {w}_{j,2}^{(0)}(0) $, which can be substituted into the first, and solved to yield ${w}_{j,1}^{(0)}(t) = e^{\m_j t} ({w}_{j,1}^{(0)}(0) + {w}_{j,2}^{(0)}(0)t)$.

Similarly, the $d_j=3$ case yields:
\bes
\begin{align}
\dot{w}_{j,1}^{(0)} & = \m_j {w}_{j,1}^{(0)} + {w}_{j,2}^{(0)} \\
\dot{w}_{j,2}^{(0)} & = \m_j {w}_{j,2}^{(0)} + {w}_{j,3}^{(0)} \\
\dot{w}_{j,3}^{(0)} & = \m_j {w}_{j,3}^{(0)}  \ ,
\end{align}
\ees 
which is easily solved in the same manner, and gives:
\bes
\begin{align}
{w}_{j,3}^{(0)}(t) &= e^{\m_j t} {w}_{j,3}^{(0)}(0)\\
{w}_{j,2}^{(0)}(t) &= e^{\m_j t} ({w}_{j,2}^{(0)}(0) + {w}_{j,3}^{(0)}(0)t)\\
{w}_{j,1}^{(0)}(t) &= e^{\m_j t} ({w}_{j,1}^{(0)}(0) + {w}_{j,2}^{(0)}(0)t + {w}_{j,3}^{(0)}(0)\frac{t^2}{2!})\ .
\end{align}
\ees
The general pattern can now be inferred. The solution for a general $d_j$ dimensional Jordan block is a vector $\vec{w}_j^{(0)} = (\vec{w}_{j,1}^{(0)},\dots,\vec{w}_{j,d_j}^{(0)})^T$ with components:
\beq
\vec{w}_{j,k}^{(0)}(t) = e^{\m_j t} \sum_{n=k}^{d_j} \vec{w}_{j,n}^{(0)}(0) \frac{t^{n-k}}{(n-k)!}\ , \quad k=1,\dots,d_j \ .
\label{eq:jordan-sol}
\eeq

The general solution of the homogenous part is then
\beq
\vec{w}^{(0)}(t) = \bigoplus_{j=1}^q \vec{w}_{j}^{(0)}(t) \ ,
\eeq
where the direct sum notation means that the summands need to be joined into a single column vector. The new aspect of the general $G$ case is thus the appearance of the degree $d_j-1$ polynomials in $t$. These polynomials induce an additional non-trivial time-dependence in addition to the rotations and exponential decay we found for the case of diagonalizable $G$. Note that we can be certain that for all $d_j>1$ [when the degree of the polynomial in Eq.~\eqref{eq:jordan-sol} is $\geq 1$], the corresponding $\Re(\m_j)< 0$, since a positive or zero real part would violate the general norm upper bound~\eqref{eq:v-bounded}.

As for the inhomogeneous part, we can write the solution of Eq.~\eqref{eq:354} as
\beq
\vec{w}(t) = \vec{w}^{(0)}(t) + \vec{w}^{(\infty)}\ ,
\eeq
and find the particular solution that satisfies
\beq
G_J\vec{w}^{(\infty)} = -\vec{c}' \ .
\eeq
Depending on the rank $r(G)$ of $G$, this equation has either zero [$r(G)=1$], one [$r(G)=M$], or infinitely many [$0<r(G)<M$] solutions. The first case is unphysical, the second is unproblematic, and for the third every initial condition still determines a corresponding final state in a unique way.

\subsection{Phase Damping Example}

As a simple example meant to illustrate how we construct and solve the differential equation for the coherence vector, assume that a single qubit is subject to a magnetic field along the $z$ direction along with dephasing:
\beq
\dot{\r} = -i[\o Z,\r] + \g(Z\r Z-\r) \ .
\eeq
As a fixed operator basis satisfying the conditions in Eq.~\eqref{eq:F-reqs}, we choose the Pauli matrices:
\beq
F_j = \sigma_j/\sqrt{2}\ , 
\eeq
with the normalization due to the requirement that $\Tr(F_i F_j) = \d_{ij}$.
The $Q$ matrix elements [Eq.~\eqref{eq:Qjk}] are then
\beq
Q_{jk} = i\o \frac{1}{2}\Tr\left(Z[\s_j,\s_k]\right) \ ,
\eeq
and are non-vanishing only when $[\s_j,\s_k] \propto Z$, i.e., $[X,Y]=2iZ$ and $[Y,X]=-2iZ$. Therefore $Q_{12} = 2\o = -Q_{12}$, and all other $Q$ matrix elements are zero. 

Next, we need to calculate the $R$ matrix and the $\vec{c}$ vector, using Eq.~\eqref{eq:Randc}. Note that, in this case, only $a_{33}=2\gamma$ is non-zero in the $a$-matrix of the Lindblad equation (the factor of $2$ is due to the normalization of the $F$'s). Therefore $\sum_{ij}$ reduces to just the term with $i=j=3$:
\bes
\begin{align}
c_k &= \frac{1}{{2}^{3/2}d}\g\Tr(\left[Z , Z \right]\s_k) = 0\\
R_{kl} &=\frac{1}{4}\Tr\left(Z \s_l Z \s_k   -  \frac{1}{2} \left\{Z^2, \s_l\right\}\s_k \right)  = \frac{1}{4}\g\Tr\left(Z \s_l Z \s_k   -    \s_l \s_k \right)    \ .
\end{align}
\ees
Clearly, $\s_k$ must equal $\s_l$ in order for the trace to be non-zero. When $\s_k=\s_l=X$, or when $\s_k=\s_l=Y$, we get $\frac{1}{4}\g\Tr(-I -I) = -\g$, whereas when $\s_k=\s_l=Z$ we get $0$. Thus $R = \text{diag}(-\g,-\g,0)$. Combining with the result for $Q$, we have:
\beq
G = \left(
\begin{array}{ccc}
-\g  & 2\o  &  0 \\
-2\o  &  -\g  &   0\\
 0 & 0  &   0 
\end{array}
\right)\ .
\eeq
This $G$ matrix is diagonalizable but not invertible (its rank is $2$), so we are in a scenario that is in between that of Secs.~\ref{sec:diag-inv-G} and~\ref{sec:gen-G-sol}. Non-invertibility only affects the existence of the limit of $\vec{v}(t)$ as $t\to\infty$. Since $G$ is diagonalizable, all its Jordan blocks have dimension $d_j=1$, i.e., they are simply the eigenvalues. The eigenvalues are $-\g\pm 2i\o$ and $0$. This corresponds to a coherence vector rotating at angular frequency $2\o$ in the $X-Y$ plane, while exponentially decaying towards the $Z$ axis with rate $\gamma$. This means that the entire $Z$ axis is the limit as $t\to\infty$, hence there is no unique final state. However, every initial state decays to a unique final state (its projection onto the $Z$ axis).


\section{Derivation of the Lindblad equation from the cumulant expansion and coarse graining}
\label{sec:LE-CCG}
%
We now present a derivation of the Lindblad equation (LE) from first principles, following Ref.~\cite{Majenz:2013qw}. This derivation avoids the so-called rotating wave approximation (RWA), which is the most commonly used approach to deriving the LE. We shall return to an RWA-based approach later. 

\subsection{Cumulant expansion}
\label{sec:cumulant-exp}
Let $\lambda$ be a small, dimensionless parameter, and consider the Hamiltonian
\begin{equation}
H = H_S + H_B + \lambda H_{SB}
\label{eq:402H}
\end{equation}
with
\begin{equation}
H_{SB} = A \otimes B
\end{equation}
where $A$ is a Hermitian system operator and $B$ is a Hermitian bath operator.  We have restricted ourself to a single term to simplify the notation, but the more general case with multiple terms follows in an analogous fashion.  

Define:
\bes
\begin{align}
H_0 &\equiv H_S \otimes I_B + I_S\otimes H_B \ , \\
U_0(t) &\equiv  \exp \left( - i t H_0\right) = U_S(t)\ox U_B(t) = e^{-itH_S}\ox e^{-itH_B}\ , \\
\tilde{\rho}_{SB}(t) &\equiv U_0^\dagger(t) \rho_{SB}(0) U_0(t) \ ,
\end{align}
\ees
where $\tilde{\rho}_{SB}(t)$ is the state in the interaction picture (recall Sec.~\ref{sec:IP}).  We have the interaction picture Hamiltonian
\begin{equation}
\tilde{H}(t) = U_0^\dgr(t)H_{SB} U_0(t) = U_{S}^\dagger(t) A U_{S}(t) \otimes U_B^\dagger(t) B U_{B}(t) \equiv A(t)\ox B(t)  \ .
\end{equation}
The density matrix in the interaction picture satisfies
\begin{equation}
\frac{d}{dt} \tilde{\rho}_{SB}(t) =  - i \left[ \lambda \tilde{H}(t), \tilde{\rho}_{SB}(t) \right] \ ,
\label{eq:408}
\end{equation}
which we can solve formally by integration followed by substitution and iteration:
\bes
\label{eqt:formal}
\begin{align}
\label{eqt:formal-a}
\tilde{\rho}_{SB}(t) &= \rho_{SB}(0) - i \int_0^t ds \left[ \lambda \tilde{H}(s), \tilde{\rho}_{SB}(s) \right] \\
\label{eqt:formal-b}
&= \rho_{SB}(0) - i \lambda \int_0^t ds \left[  \tilde{H}(s), \rho_{SB}(0) \right] +(-i \lambda)^2 \int_0^t d s \int_0^s ds' \left[\tilde{H}(s), \left[ \tilde{H}(s') , \rho_{SB}(0) \right] \right] + \cdots \ ,
\end{align}
\ees
and it is clear how this continues.
A simple norm estimate (see Sec.~\ref{sec:norms}) shows that the norm of the $n$th order term is $O[(\|H_{SB}\|t)^n]$. Therefore a sufficient convergence condition is  $\lambda \|H_{SB}\| t < 1$. Terms of third order and above can be neglected provided $\lambda \|H_{SB}\| t \ll 1$. This is known as the \emph{Born approximation}.

We are interested in the reduced density matrix:
\begin{equation}
\tilde{\rho}(t) = \Tr_B \left[ \tilde{\rho}_{SB}(t) \right] \equiv \Lambda_\lambda(t) {\rho}(0) \ .
\end{equation}
The \emph{cumulant expansion} is given by introducing unknown, to be determined operators $K^{(n)}$ in the exponent:
\bes
\begin{align}
\Lambda_\lambda(t) &= \exp \left( \sum_{n=1}^{\infty} \lambda^n K^{(n)}(t) \right) \\
&= I + \lambda K^{(1)}(t) + \lambda^2 \left( K^{(2)}(t) + \frac{1}{2} \left( K^{(1)} (t) \right)^2 \right) + O(\lambda^3)  \ ,
\end{align}
\ees
where in the second line we used a Taylor expansion of the exponential. We solve for $K^{(n)}$ by matching powers of $\lambda$ with Eq.~\eqref{eqt:formal}.  We get:
\begin{equation}
K^{(1)} (t) \rho(0) = - i \int_0^t ds \ \Tr_B \left( \left[ \tilde{H}(s) , \rho_{SB}(0) \right] \right)  \ .
\end{equation}
We will see later that, without loss of generality, this can always be made to vanish (for a stationary bath) by shifting the operator $B$, i.e:
\begin{equation}
K^{(1)} (t) \rho(0) = 0  \ .
\label{eq:K1=0}
\end{equation}
The next order in $\lambda$ gives:
\begin{equation}
K^{(2)} (t) \tilde{\rho} (0) =  - \int_0^t d s \int_0^s ds' \ \Tr_B \left( \left[ \tilde{H}(s), \left[ \tilde{H}(s') , \rho_{SB}(0) \right] \right] \right) \ .
\end{equation}
Expanding the double commutator gives:
\bes
\begin{align}
\label{eq:AArho}
 \Tr_B \left( \left[ \tilde{H}(s), \left[ \tilde{H}(s') , \rho_{SB}(0) \right] \right] \right) &= \left[ A(s) A(s') \rho(0) - A(s') \rho(0) A(s) \right]  \Tr \left[ B(s) B(s') \rho_B \right]  + \mathrm{h.c.} \\
 &=  \left[ A^\dagger(s) A(s') \rho(0) - A(s') \rho(0) A^\dagger(s) \right]  \Tr \left[ B^\dagger(s) B(s') \rho_B \right]  + \mathrm{h.c.}  \\
 &= \left[ A^\dagger(s) A(s') \rho(0) - A(s') \rho(0) A^\dagger(s) \right]  \mathcal{B}(s,s')  + \mathrm{h.c.}
\end{align}
\ees
where 
\beq
\mathcal{B}(s,s') \equiv  \langle B^\dagger(s) B(s') \rangle = \mathcal{B}(s',s)^\ast \ ,
\label{eq:B-corr-func}
\eeq
and
\beq
\label{eq:404}
\langle X \rangle_B \equiv {\rm Tr}[\rho_B X] \ ,
\eeq
and $\rho_B$ is, e.g., the thermal (Gibbs) state of the bath [Eq.~\eqref{eq:Gibbs1}]. Equation~\eqref{eq:B-corr-func} holds since:
\beq
\label{eq:421}
\langle B^\dagger(s) B(s') \rangle = \Tr[\r_B(0) B^\dagger(s) B(s')] = (\Tr[B(s')^\dgr B(s) \r_B(0) ]^\dgr)^* = (\Tr[\r_B(0)B(s')^\dgr B(s)]^\dgr)^* = \langle B^\dagger(s') B(s) \rangle^*\ .
\eeq

\subsection{The second order cumulant}
\label{sec:cumulant-2}

It turns out to be convenient to express the interaction picture system operator $A(t)$ in the frequency domain. To do so, let us first expand $H_S$ in its eigenbasis:
\beq
H_S = \sum_a \varepsilon_a \ketb{\varepsilon_a}{\varepsilon_a} \ ,
\eeq
where $\{\varepsilon_a\}$ are the eigenenergies of $H_S$.
Thus
\begin{equation}
A(t) = U_S^\dagger(t) A U_S(t) = \sum_{a, b} e^{-i \left( \varepsilon_b - \varepsilon_a \right) t } |\varepsilon_a \rangle \langle \varepsilon_a | A | \varepsilon_b \rangle \langle \varepsilon_b| =  \sum_{\omega} A_{\omega} e^{-i \omega t } \ ,
\label{eq:417A}
\end{equation}
where $\o \equiv \varepsilon_b - \varepsilon_a$ is a Bohr frequency, and  
\beq
A_{\omega} \equiv \sum_{\varepsilon_b - \varepsilon_a=\o} \langle \varepsilon_a | A | \varepsilon_b \rangle |\varepsilon_a \rangle  \langle \varepsilon_b|\ .
\label{eq:A_om}
\eeq
To clarify, the sum over $\varepsilon_b - \varepsilon_a=\o$ in Eq.~\eqref{eq:A_om} is over all pairs of eigenenergies $\{\varepsilon_b,\varepsilon_a\}$ whose difference gives the same Bohr frequency $\o$. The sum over $\o$ in Eq.~\eqref{eq:417A} is a sum over all Bohr frequencies (negative, zero, and positive).
%
%
This then gives the following map from time $0$ to $t$:
\beq
\label{eqt:TimeIndependent1}
K^{(2)} (t) \rho(0) =   \sum_{\omega, \omega'} \mathcal{B}_{\omega \omega'}(t) \left( A_\omega \rho(0)  A^\dagger_{\omega'}  -A^\dagger_{\omega'} A_\omega \rho(0) \right)  + \mathrm{h.c.} ,
\eeq
where
\begin{align}
\label{eq:uneq}
\mathcal{B}_{\omega \omega'}(t) &\equiv \int_0^t d s \int_0^s ds' e^{i ( \omega' s - \omega s' )}  \mathcal{B}(s,s') \ .
\end{align}

We will see that Eq.~\eqref{eqt:TimeIndependent1} can be rewritten in the form of a Lindblad generator:
\begin{eqnarray}
\label{eq:K2-gen}
K^{(2)}(t) {\rho}(0) = - i \left[\mathcal{Q}(t), {\rho}(0) \right] + \sum_{\omega, \omega'} b_{\omega \omega'}(t) \left[ A_{\omega} {\rho}(0) A_{\omega'}^\dagger - \frac{1}{2} \left\{A_{\omega'}^{\dagger} A_\omega , {\rho}(0) \right\} \right] \ ,
\end{eqnarray}
where the elements of the matrix $b(t)$ are given by
\beq
\label{eq:bomompr}
b_{\omega \omega'} (t) \equiv \int_0^t d s \int_0^t ds'  e^{i ( \omega' s - \omega s' )}  \mathcal{B}(s,s') = b^{\ast}_{\omega' \omega}(t)\ ,
\eeq
and we will show that $b(t)$ is positive semi-definite.

The ``Lamb shift" term is
\begin{equation}
\label{eq:Q-LS}
\mathcal{Q}(t) =  \sum_{\omega, \omega'} Q_{\omega \omega'}(t) A_{\omega'}^{\dagger} A_{\omega} ,
\end{equation}
where

\bes
\label{eq:Lambshift}
\begin{align}
Q_{\omega \omega'}(t) &= - \frac{i}{2} \left( B_{\o\o'}-B_{\o'\o}^*\right)\\
&= - \frac{i}{2} \int_0^t d s \int_0^s d s' \left( e^{i (\omega' s - \omega s') } \mathcal{B}(s,s') - e^{-i (\omega s - \omega' s') } \mathcal{B}(s',s) \right) \ .
\end{align}
\ees
Note that $\left(Q_{\omega \omega'} \right)^{\ast} =  Q_{\omega' \omega}$, so that $\mathcal{Q}^\dagger = \mathcal{Q}$, as required for the interpretation of $\mathcal{Q}$ as a Hamiltonian.

\subsection{Why the first order cumulant can be made to vanish}
\label{sec:cumulant1}

We argued [Eq.~\eqref{eq:K1=0}] that we can shift the bath operator $B$ such that $K^{(1)} (t) \rho(0) = 0$. Here we show why.

Let $\r_B(0) = \sum_\m \lambda_\m \ketbra{\m}$ and 
\beq
B_d(t) \equiv \text{diag}(B(t)) = \sum_\m B_{\m\m}(t)\ketb{\m}{\m}\ ,
\eeq 
i.e., the diagonal part of $B$ in the eigenbasis of $\r_B(0)$. Here $B_{\m\m}(t) = \bra{\m}B(t)\ket{\m}$.
Let us define a new bath operator 
\beq
B'(t) \equiv B(t) - B_d(t)\ .
\eeq
Then
\beq
\ave{B'(t)} = \ave{B(t)} - \ave{B_d(t)} =  \sum_\m \lambda_\m \bra{\m}B(t)\ket{\m} - \sum_\m \lambda_\m \bra{\m}\left[\sum_\n B_{\n\n}(t)\ketb{\n}{\n}\right] \ket{\m} = 0\ .
\label{eq:429}
\eeq

Let $H'_{SB} = A\ox B'$, so that $\tilde{H}'(t) = U_0^\dag (t) H'_{SB} U_0(t)$. Then
\beq
\Tr_B \left( \left[ \tilde{H}'(t) , \rho_{SB}(0) \right] \right) = \Tr_B \left( \left[ A(t)\ox B'(t) , \r_S(0)\ox \rho_{B}(0) \right] \right) = \ave{B'(t)}[A(t),\r_S(0)] = 0\ .
\eeq
Therefore, $K^{\prime(1)} (t) \rho(0) = 0$, with $K^{\prime(1)}$ defined with the modified system-bath interaction $H'_{SB}$. The price we have to pay for this is the shift of $B$ to $B'$. This shift manifests itself only through the bath correlation function
$\mathcal{B}(s,s')$ [Eq.~\eqref{eq:B-corr-func}]. The shifted correlation function becomes $\mathcal{B}'(s,s') =  \langle B^{'\dagger}(s) B'(s') \rangle$, and nothing else changes, since the bath operators only appear through the bath correlation function.

\subsection{Derivation of the Lindblad equation}

We will now prove that Eq.~\eqref{eqt:TimeIndependent1} can be transformed into Eq.~\eqref{eq:K2-gen}. It turns out that the unequal upper integration limits in $B_{\o\o'}$ [Eq.~\eqref{eq:uneq}] are problematic, while the equal upper integration limits in $b_{\o\o'}$ [Eq.~\eqref{eq:bomompr}] are what allows us to prove complete positivity, as we show in Sec.~\ref{sec:CP-bB} directly below. To replace the unequal upper limits by equal limits we note the following relations for the integral, where for notational simplicity we suppress the $t$-dependence for now:
\bes
\begin{align}
\mathcal{B}_{\omega \omega'} \equiv \int_0^t d s \int_0^s ds' e^{i ( \omega' s - \omega s' )}  \mathcal{B}(s,s') \ , =&\left[ \int_0^t d s \int_0^t ds' - \int_0^t d s \int_s^t ds'  \right] e^{i ( \omega' s - \omega s' )}  \mathcal{B}(s,s') \ ,  \\
=&  \left[\int_0^t d s \int_0^t ds' -  \int_0^t d s' \int_0^{s'} ds  \right] e^{i ( \omega' s - \omega s' )}  \mathcal{B}(s,s') \ , \\
=& \int_0^t d s \int_0^t ds'  e^{i ( \omega' s - \omega s' )}  \mathcal{B}(s,s')  -  \int_0^t d s \int_0^{s} ds'  e^{i ( \omega' s' - \omega s)}  \mathcal{B}(s',s) \ , \\
=& \ b_{\omega \omega'} - \mathcal{B}_{\omega' \omega}^{\ast} \ ,
\end{align}
\ees
where $b_{\omega \omega'}$ [Eq.~\eqref{eq:bomompr}] has the desired equal upper integration limits. It follows immediately that
\beq
\mathcal{B}_{\omega \omega'}^{\ast} = b_{\omega' \omega} - \mathcal{B}_{\omega' \omega} \ .
\label{eq:428}
\eeq
Therefore, the first summand $+\mathrm{h.c.}$ in Eq.~\eqref{eqt:TimeIndependent1} yields:
\bes
\begin{align}
\sum_{\omega, \omega'} \left[\mathcal{B}_{\omega \omega'}A_\omega \rho A_{\omega'}^\dagger  + \mathcal{B}_{\omega \omega'}^{\ast}A_{\omega'} \rho A_{\omega}^\dagger  \right] 
=& \sum_{\omega, \omega'} \left[ b_{\omega \omega'}A_{\omega} \rho A_{\omega'}^\dagger  + b_{\omega' \omega} A_{\omega'} \rho A_{\omega}^\dagger  - \left( B_{\omega' \omega}^{\ast}A_{\omega} \rho A_{\omega'}^\dagger  + B_{\omega' \omega} A_{\omega'} \rho A_{\omega}^\dagger \right) \right]  \ , \\
=&  \sum_{\omega, \omega'} \left[ b_{\omega \omega'} A_{\omega} \rho A_{\omega'}^\dagger  + b_{\omega' \omega} A_{\omega'} \rho A_{\omega}^\dagger  - \left( B_{\omega \omega'}^{\ast}  A_{\omega'} \rho A_{\omega}^\dagger + B_{\omega \omega'} A_{\omega} \rho A_{\omega'}^\dagger \right) \right] \ ,
\end{align}
\ees
where in the second term on the RHS we have switched $\omega \leftrightarrow \omega'$, which is permissible since we are summing over all $\o$ and $\o'$.  Furthermore, this second term is now exactly in the form of the original term, so we have the result:
\begin{eqnarray}
\sum_{\omega, \omega'} \left[ \mathcal{B}_{\omega \omega'}A_\omega \rho A_{\omega'}^\dagger  + \mathcal{B}_{\omega \omega'}^{\ast}A_{\omega'} \rho A_{\omega}^\dagger  \right] 
&=&\frac{1}{2} \sum_{\omega, \omega'} \left[ b_{\omega \omega'} A_{\omega} \rho A_{\omega'}^\dagger  + b_{\omega' \omega}  A_{\omega'} \rho A_{\omega}^\dagger  \right]  
= \sum_{\omega, \omega'} b_{\omega \omega'} A_{\omega} \rho A_{\omega'}^\dagger  \ .
\end{eqnarray}

The second summand $+\mathrm{h.c.}$ in Eq.~\eqref{eqt:TimeIndependent1} is of the form $-A^\dagger_{\omega'} A_\omega \rho(0)$, which reminds us of the anti-commutator term in the Lindblad equation, except that it doesn't have the factor of $1/2$. However, note that since $b^{\ast}_{\omega' \omega} = b_{\omega \omega'} = \mathcal{B}_{\omega \omega'}^{\ast} + \mathcal{B}_{\omega' \omega}$, where we used Eq.~\eqref{eq:428}. Therefore by writing $b_{\omega \omega'} = \frac{1}{2} (b_{\omega \omega'} + b^{\ast}_{\omega' \omega}) $, and again using Eq.~\eqref{eq:428}, we have:
\begin{eqnarray}
\mathcal{B}_{\omega \omega'}&=& \frac{1}{2} b_{\omega \omega'} + \frac{1}{2} \left(  \mathcal{B}_{\omega \omega'}- \mathcal{B}_{\omega' \omega}^{\ast} \right) \ .
\end{eqnarray}
This allows us to write the second summand $+\mathrm{h.c.}$ in Eq.~\eqref{eqt:TimeIndependent1} as:
\bes
\begin{align}
&-\sum_{\omega, \omega'} \left[ \mathcal{B}_{\omega \omega'}  A_{\omega'}^\dagger  A_\omega \rho +   \mathcal{B}_{\omega \omega'}^{\ast} \rho A_{\omega}^\dagger A_{\omega'}  \right] \\
&\quad  = -\frac{1}{2} \sum_{\omega, \omega'} \left( b_{\omega \omega'} A_{\omega'}^\dagger  A_\omega \rho  + b_{\omega' \omega} \rho A_{\omega}^\dagger A_{\omega'} \right)  - \frac{1}{2} \sum_{\omega, \omega'} \left[ \left( \mathcal{B}_{\omega \omega'} - \mathcal{B}_{\omega' \omega}^{\ast} \right) A_{\omega'}^\dagger  A_\omega \rho +  \left( \mathcal{B}_{\omega' \omega} - \mathcal{B}_{\omega' \omega}^{\ast} \right)  \rho A_{\omega}^\dagger  A_{\omega'}  \right]  \\
&\quad =  -\frac{1}{2} \sum_{\omega, \omega'}  b_{\omega \omega'} \left( A_{\omega'}^\dagger  A_\omega \rho  + \rho A_{\omega'}^\dagger A_{\omega}  \right) 
-\frac{1}{2} \sum_{\omega, \omega'} \left( \mathcal{B}_{\omega \omega'} - \mathcal{B}_{\omega' \omega}^{\ast} \right) \left[ A_{\omega'}^\dagger  A_\omega \rho  - \rho A_{\omega'}^\dagger A_{\omega}  \right] \\
&\quad = -\frac{1}{2} \sum_{\omega, \omega'}  b_{\omega \omega'} \left\{ A_{\omega'}^\dagger  A_\omega , \rho  \right\}
-\frac{1}{2} \sum_{\omega, \omega'}  \left(\mathcal{B}_{\omega \omega'} - \mathcal{B}_{\omega' \omega}^{\ast} \right) \left[ A_{\omega'}^\dagger  A_\omega , \rho   \right] \ .
\end{align}
\ees
We can now write the RHS of Eq.~\eqref{eqt:TimeIndependent1} as:
\bes
\begin{align}
& \sum_{\omega, \omega'} \mathcal{B}_{\omega \omega'}(t) \left( A_\omega \rho(0)  A^\dagger_{\omega'}  -A^\dagger_{\omega'} A_\omega \rho(0) \right)  + \mathrm{h.c.} \\
 \quad &= 
\sum_{\omega, \omega'} b_{\omega \omega'} A_{\omega} \rho(0) A_{\omega'}^\dagger 
-\frac{1}{2} \sum_{\omega, \omega'}  b_{\omega \omega'} \left\{ A_{\omega'}^\dagger  A_\omega , \rho(0)  \right\}
-\frac{1}{2} \sum_{\omega, \omega'}  \left(\mathcal{B}_{\omega \omega'} - \mathcal{B}_{\omega' \omega}^{\ast} \right) \left[ A_{\omega'}^\dagger  A_\omega , \rho(0)   \right] \\
\quad &=  -i \Big[ \sum_{\omega, \omega'} \frac{-i}{2} \left( B_{\o\o'}-B_{\o'\o}^*\right) [A_{\omega'}^{\dagger} A_\omega, {\rho}(0) ] \Big] + \sum_{\omega, \omega'} b_{\omega \omega'}(t) \Big[ A_{\omega} {\rho}(0) A_{\omega'}^\dagger - \frac{1}{2} \left\{A_{\omega'}^{\dagger} A_\omega , {\rho}(0) \right\} \Big]\ ,
\end{align}
\ees
which is Eq.~\eqref{eq:K2-gen}, together with the identification of the term in the commutator as the Lamb shift  $\mc{Q}(t)$ as defined in Eq.~\eqref{eq:Q-LS}.
 
\subsection{Complete positivity}
\label{sec:CP-bB}

Clearly, the dissipative (second) term on the RHS of Eq.~\eqref{eq:K2-gen} appears to be in Lindblad form, but we must still prove the positivity of the matrix $b(t)$. To this end we again expand the bath density matrix in its eigenbasis, and use this to write the correlation function $\mathcal{B}(s,s') = \< B^\dagger(s) B(s')\>_B$ explicitly. Let $\vec{v}$ be some arbitrary vector; then positivity amounts to showing that $\vec{v}b(t)\vec{v}^{\dagger} \geq 0$ for all $\vec{v}$. Indeed:
\bes
\label{eq:bpos}
\begin{align}
\vec{v}b(t)\vec{v}^{\dagger} &= \sum_{\omega\omega'}v_\omega b_{\omega \omega'}(t) v^*_{\omega'} = \int_0^t d s \int_0^t ds'  \sum_{\omega}(v_\omega e^{-i \omega s' })  \sum_{\omega'}(v_{\omega'}e^{-i\omega' s})^* {\rm Tr}[\sum_\mu \lambda_\mu \ket{\mu}\bra{\mu} B^\dagger(s) B(s')] \\
& = \sum_\mu \lambda_\mu \bra{\mu}  F^\dagger(t) F(t)\ket{\mu} =  \sum_\mu \lambda_\mu \| F(t)\ket{\mu}\|^2 \geq 0 ,
\end{align}
\ees
where $F(t) \equiv \int_0^t ds B(s) \sum_\omega v^*_\omega e^{-i\omega s}$. Note how it was crucial in this proof that the upper limits of the integrals are the same, since otherwise the factorization would have failed.

Therefore, our quantum map is given by:
\beq 
\label{eqt:CPmap}
\tilde{\rho}(t) = e^{\lambda^2 K^{(2)}(t) } \rho(0)\ .
\eeq
The only approximation we have introduced so far is the truncation at order $\lambda^2$, i.e., the Born approximation. The CP map \eqref{eqt:CPmap} is in principle already sufficient, and one can use it to compute Kraus operators. However, in order to find the time-dependent system state $\tilde{\rho}(t)$ one has to compute $e^{\lambda^2 K^{(2)}(t) }$ for each $t$, which is laborious. In order to arrive at a master equation, with the associated advantages (e.g., a quantum trajectories unravelling) we need to introduce an additional, Markovian approximation.

\subsection{LE from the cumulant expansion and coarse-graining}
%
Let us show how to obtain the LE from the results above. Expanding the exponential in Eq.~\eqref{eqt:CPmap} to second order in $\lambda$,
we have:
\beq
\tilde{\rho}(t)-\tilde{\rho}(0) = -i \left[\lambda^2 {\mathcal{Q}}(t), {\rho}(0) \right] + \sum_{\omega, \omega'} \lambda^2 {b}_{\omega\omega'}(t) \left[ A_{\omega} {\rho}(0) A_{\omega'}^{\dagger} - \frac{1}{2} \left\{ A_{\omega'}^{\dagger} A_{\omega}, {\rho}(0) \right\} \right] \ . 
\label{eq:482}
\eeq
It is straightforward to check that $\mathcal{Q}(0) = {b}_{\omega\omega'}(0) =0$ (due to the upper integration limit being $0$). Therefore, dividing both sides of Eq.~\eqref{eq:482} by $\tau$, and setting $t=\tau$, we have:
\beq
\ave{\dot{\tilde{\rho}}}_0 = -i \left[\lambda^2 \ave{\dot{\mathcal{Q}}}_0, {\rho}(0) \right] + \sum_{\omega, \omega'} \lambda^2 \ave{\dot{b}_{\omega\omega'}}_0 \left[ A_{\omega} {\rho}(0) A_{\omega'}^{\dagger} - \frac{1}{2} \left\{ A_{\omega'}^{\dagger} A_{\omega}, {\rho}(0) \right\} \right] \ ,
\label{eq:483}
\eeq
where we used the coarse-graining definition, Eq.~\eqref{eq:aveX}.

%
%
Similarly to Sec.~\ref{sec:LE-CG}, the path to the Lindblad equation is to now introduce a Markovian assumption in terms of the coarse-graining timescale $\tau$. The Markovian assumption amounts to assuming that both $\langle \dot{\mathcal{Q}} \rangle_0$ and $\ave{\dot{b}_{\omega\omega'}}_0$ are constant for all $t$, i.e., that $\langle \dot{\mathcal{Q}} \rangle_j = \langle \dot{\mathcal{Q}} \rangle_0$ and $\langle \dot{b}_{\omega\omega'} \rangle_j = \langle\dot{b}_{\omega\omega'} \rangle_0$ for all $j$. This can be rigorously justified by first assuming that the bath correlation function is translationally invariant, i.e., $\mathcal{B}(s,s') = \mathcal{B}(s-s')$. This is true for stationary baths. A bath is stationary if 
\beq
[H_B,\r_B(0)]=0\ ,
\label{eq:b-stat} 
\eeq
which implies that $\r_B(t) = U_B(t)\r_B(0)U_B^\dgr(t) = \r_B(0)$. This is the case, e.g., if $\r_B(0)=e^{-\b H_B}/Z$, i.e., is a Gibbs state. In addition we assume that the bath correlation function decays over a timescale $\tau_B$, i.e., 
\beq
\mathcal{B}(t) \sim e^{-t/\tau_B}\ , 
\eeq
while the coarse graining is done over a much longer timescale, so that the integrand in Eq.~\eqref{eq:bomompr} has already decayed. The RHS of Eq.~\eqref{eq:483} is then valid for all times, allowing us to also shift the time argument of $\r$ to arbitrary $j\tau$. Let us now define the Lamb-shift and the Lindblad rates as:
\bes
\begin{align}
H_{\mathrm{LS}} &\equiv \lambda^2 \langle \dot{\mathcal{Q}} \rangle_0 \ ,  \\
\gamma_{\omega \omega'} &\equiv \lambda^2 \langle \dot{b}_{\omega\omega'} \rangle_0   \ .
\end{align}
\ees
Moreover, we assume that $\tau$ is very small on the timescale $\tau_S$ over which $\r(t)$ changes, so that $\ave{\dot{\tilde{\rho}}}_j = [\tilde{\rho}((j+1)\tau)-\tilde{\rho}(j\tau)]/\tau$ can be replaced by $\dot{\r}(t)$. These assumptions can be summarized as
\beq
\tau_B \ll \tau \ll \tau_S \ .
\label{eq:tau-cond}
\eeq
We can thus write the interaction picture Lindblad equation in the final form:
\beq
\dot{\tilde{\rho}}(t) = -i \left[H_{\mathrm{LS}}, {\rho}(t) \right] + \sum_{\omega, \omega'} \g_{\omega \omega'} \left[ A_{\omega} {\rho}(t) A_{\omega'}^{\dagger} - \frac{1}{2} \left\{ A_{\omega'}^{\dagger} A_{\omega}, {\rho}(t) \right\} \right] \ . 
\label{eq:LE-CG-final}
\eeq
The RHS contains the free parameter $\tau$, which can be determined using Eq.~\eqref{eq:tau-cond}.  Everything else is determined in terms of the given specification of the Hamiltonian $H = H_S + H_B + H_{SB}$ and the initial state of the bath $\r_B(0)$. In particular, 
\begin{itemize}
\item The Bohr frequencies $\o$ are determined by $H_S$;
\item The Lindblad operators are determined by the system operator $A$ in $H_{SB}$ and the Bohr frequencies (i.e., $H_S$);
\item The bath correlation function $\mc{B}(s,s')$ is determined by the bath operator $B$ in $H_{SB}$, the bath Hamiltonian $H_B$ (which determines the time-dependence of $B(t)$), and the initial bath state $\r_B$;
\item The Lamb shift is determined by the bath correlation function and the Bohr frequencies.
\end{itemize}


\subsection{Illustration using the spin-boson model for phase damping} 
\label{app:2level}

Consider once more the spin-boson model defined in Sec.~\ref{sec:spin-boson-1q}. Let us denote the eigenvalues of $H_S = -(g/2)Z$ by $\varepsilon_{\pm} = \pm g / 2$ and their respective eigenvectors by $\ket{\varepsilon_-}=\ket{0}$ (ground state) and $\ket{\varepsilon_+}=\ket{1}$ (excited state). Using Eq.~\eqref{eq:A_om}, the Lindblad operators are then given by:
\begin{eqnarray}
A_{-g} &=& | \varepsilon_+ \rangle \langle \varepsilon_+ | Z | \varepsilon_- \rangle \langle \varepsilon_- | = 0 \\
A_0 &=&  | \varepsilon_+ \rangle \langle \varepsilon_+ | Z | \varepsilon_+ \rangle \langle \varepsilon_+ |  +  | \varepsilon_- \rangle \langle \varepsilon_- | Z | \varepsilon_- \rangle \langle \varepsilon_- | = Z \\
A_{g} &=& | \varepsilon_- \rangle \langle \varepsilon_- | Z | \varepsilon_+ \rangle \langle \varepsilon_+ | = 0 \ .
\end{eqnarray}
Thus, only the (elastic, or on-shell) $\omega= 0$ term contributes to the sums over $\omega$.  This means that the Lamb shift is given by:
\beq
H_{\mathrm{LS}}  = \frac{\lambda^2}{\tau} \mathcal{Q}(\tau) = \frac{\lambda^2}{\tau} Q_{00}(t) A_{0}^{\dagger} A_{0}  \propto I\ ,
\eeq 
so that $\left[ H_{\mathrm{LS}},\tilde{\rho}(t) \right] = 0$.  The dissipative part of the LE [Eq.~\eqref{eq:LE-CG-final}] is given by:
\beq
\sum_{\omega, \omega'} \g_{\omega, \omega'} \left[ A_{\omega} {\rho} A_{\omega'}^{\dagger} - \frac{1}{2} \left\{ A_{\omega'}^{\dagger} A_{\omega}, {\rho}\right\} \right]  = \g_{00} \left(Z\tilde{\r}Z-\frac{1}{2}\left\{ I, {\tilde{\rho}} \right\}\right) = \g\left(Z\tilde{\r}Z-\tilde{\r}\right)\ ,
\eeq
where 
\beq
\g \equiv \g_{00} = \frac{\lambda^2}{\tau} b_{00}(\tau) = \frac{\lambda^2}{\tau} \int_0^\tau d s \int_0^\tau ds'   \mathcal{B}(s,s') \ ,
\eeq
and where we used Eq.~\eqref{eq:bomompr}. We already computed this decay rate when we solved the spin-boson model analytically, and found it in Eq.~\eqref{eq:401}. The result after coarse graining is given in Eq.~\eqref{eq:g-CGLE}.

While we already saw the solution of the corresponding LE in Sec.~\ref{sec:PD-Lind}, let us solve it again using a nice and useful ``vectorization" trick. Let us define:
\begin{equation}
 \mathrm{vec}(\rho) \equiv \left(
\begin{array}{c}
\text{col}_{1}(\rho) \\
\vdots \\
\text{col}_j(\rho)\\
\vdots
\end{array}
\right)
\end{equation}
i.e., $\mathrm{vec}(\rho)$ corresponds to stacking the columns of $\rho$ (in some basis). We now use the identity \cite{ZAMM:ZAMM19920721231}:
\begin{equation}
\mathrm{vec}\left( A B C \right) = \left( C^{T} \otimes A \right) \mathrm{vec}\left(B \right)
\end{equation}
where $(A,B,C)$ are arbitrary matrices of appropriate dimensions allowing their multiplication. Using this, we can write the LE $\dot{\tilde{\rho}} = \gamma(Z\tilde{\rho} Z - I\tilde{\rho}I)$ as
\begin{eqnarray}
\mathrm{vec}\left( \dot{\tilde{\rho}}  \right)&=& \gamma \left( Z\ox Z - I\ox I\right) \mathrm{vec}({\tilde{\r}})  \equiv \mathcal{L}  \mathrm{vec}({\tilde{\rho}})\ .
\end{eqnarray}
Conveniently, $\mathcal{L}$ is diagonal with entries $(0, -2 \gamma, -2 \gamma, 0)$, so we can immediately write:
\begin{equation}
\mathrm{vec}({\tilde{\rho}}(t)) = \exp(\mathcal{L})\mathrm{vec}({\tilde{\rho}}(0)) = \left( \begin{array}{cccc}
1 &&& \\
& \exp(-2 \gamma t) && \\
&& \exp(-2 \gamma t )& \\
&&& 1
\end{array} \right)
\left(
\begin{array}{c}
{\r}_{00}(0) \\
{\r}_{10}(0) \\
{\r}_{01}(0)\\
{\r}_{11}(0)
\end{array}
\right)
\end{equation}
Therefore, we find as before:
\begin{equation}
\tilde{\rho}(t) = \left( \begin{array} {cc}
{\rho}_{00}(0) & \exp(-2 \gamma t) {\rho}_{01}(0) \\
\exp(-2 \gamma t) {\rho}_{10}(0) & {\rho}_{11}(0) 
\end{array} \right)
\end{equation}
Transforming back to the Schr\"odinger picture, the result is adjusted to
\begin{equation}
\tilde{\rho}(t) = \left( \begin{array} {cc}
{\rho}_{00}(0) & \exp(-2 \gamma t- i g t) {\rho}_{01}(0) \\
\exp(-2 \gamma t + i g t) {\rho}_{10}(0) & {\rho}_{11}(0) 
\end{array} \right)\ .
\end{equation}
%



\section{First-principles derivation of the Lindblad equation from the Born, Markov, and rotating wave approximations}
\label{sec:MicroLindblad}

We now present our last derivation of the Lindblad equation. This is the standard approach found in textbooks such as \cite{Breuer:book}, but we will add some clarifications concerning the limitations of the validity of this approach. We will also discuss the differences between this and the cumulant-based approach.

\subsection{Setting up}
Our starting point is identical to the one we used in the cumulant expansion approach (Sec.~\ref{sec:cumulant-exp}). The only difference is that we now consider the more general system-bath interaction
\beq
H_{SB}=g \sum_{\alpha}A_{\alpha}\otimes B_{\alpha}\ ,
\label{eq:gHSB}
\eeq
where $g$ has units of energy. Thus, in the interaction picture:
\beq
\tilde{H}(t)=g \sum_{\alpha}A_{\alpha}(t) \ox B_{\alpha}(t)\ ,
\label{Hint}
\eeq
and
\begin{subequations}
\label{eq:468A}
\begin{align}
A_{\alpha}(t)&=U_S^\dgr(t) A_\alpha U_S(t) \ , \qquad U_S(t) = e^{-iH_St} \label{At}\\
B_{\alpha}(t)&=U_B^\dgr(t) B_\alpha U_B(t) \ , \qquad U_B(t) = e^{-iH_Bt} \label{Bt}\ .
\end{align}
\end{subequations}
Formally integrating the Liouville-von Neumann equation 
\begin{equation}
\frac{d}{dt} \tilde{\rho}_{SB}(t) =  - i \left[ \tilde{H}(t), \tilde{\rho}_{SB}(t) \right] \ ,
\label{eq:LvN}
\end{equation}
we have:
\beq
\tilde{\rho}_{SB}(t) = \rho_{SB}(0) - i \int_0^t ds \left[ \tilde{H}(s), \tilde{\rho}_{SB}(s) \right] \ .
\eeq
Let us now substitute this solution back into Eq.~\eqref{eq:LvN} and take the partial trace:
\beq
\frac{d}{dt} \tilde{\rho}(t) = \Tr_B\left\{\frac{d}{dt} \tilde{\rho}_{SB}(t)\right\} =  - i  \Tr_B\left\{\left[ \tilde{H}(t), \rho_{SB}(0) \right]\right\} +(-i )^2 \Tr_B\left\{\left[\tilde{H}(t), \int_0^t d s   \left[ \tilde{H}(s) , \tilde{\rho}_{SB}(s) \right] \right]\right\} \ .
\eeq
Just as we argued in Sec.~\ref{sec:cumulant1}, the first order term can again be made to vanish provided we shift the bath operators. We are thus left with
\beq
\frac{d}{dt} \tilde{\rho}(t) = - \Tr_B\left\{\left[\tilde{H}(t), \int_0^t d s   \left[ \tilde{H}(s) , \tilde{\rho}_{SB}(s) \right] \right]\right\} \ .
\eeq
Let us change variables to $\tau = t-s$, so that $\int_0^t d s = -\int^0_t (-d \tau) = \int_0^t d \tau$, and:
\beq
\frac{d}{dt} \tilde{\rho}(t) = - \Tr_B\left\{\left[\tilde{H}(t), \int^t_0 d \tau   \left[ \tilde{H}(t-\tau) , \tilde{\rho}_{SB}(t-\tau) \right] \right]\right\} \ .
\eeq

\subsection{Born approximation}
\label{sec:Born-appr}

To proceed we now make our first approximation. For a sufficiently large bath that is in particular much larger than the system, it is reasonable to assume that while the system undergoes non-trivial evolution, the bath remains unaffected, and hence that the state of the composite system at time $t$ is
\begin{equation}
\label{Born}
\tilde{\rho}_{SB}(t)= \tilde{\rho}(t)\otimes \rho_{B} (0) + \chi(t) \approx \tilde{\rho}(t)\otimes \rho_{B}\ ,
\end{equation}
where $\r_B$ is the time-independent, stationary bath state, and the correlations $\chi(t)$ can be neglected. This is (again) called the \emph{Born approximation}. 

Using this and Eq.~\eqref{Hint}, we have:
\begin{equation}
\label{integEq1}
\frac{d\tilde{\rho}}{dt} =-g^2\sum_{\alpha,\beta} \Tr_B\left\{
\left[A_\alpha(t) \ox B_\alpha(t), \int_0^t d\tau [A_\beta(t-\tau) \ox B_\beta(t-\tau), \tilde{\rho}(t-\tau)\otimes \rho_{B}]\right]\right\}\ .
\end{equation}
Let's expand the double commutator:
\bes
\begin{align}
&\Tr_B \left[A_\alpha(t) \ox B_\alpha(t),  [A_\beta(t-\tau) \ox B_\beta(t-\tau), \tilde{\rho}(t-\tau)\otimes \rho_{B}]\right]  \\
\label{eq:476b}
&\quad = A_\alpha(t)A_\beta(t-\tau)\tilde{\rho}(t-\tau) \Tr[B_\alpha(t)B_\beta(t-\tau)\r_B] \\
\label{eq:476c}
&\quad - A_\beta(t-\tau)\tilde{\rho}(t-\tau)A_\alpha(t) \Tr[B_\beta(t-\tau)\r_B B_\alpha(t)] \\
\label{eq:476d}
&\quad - A_\alpha(t)\tilde{\rho}(t-\tau)A_\beta(t-\tau) \Tr[\r_B B_\beta(t-\tau)B_\alpha(t)] \\
\label{eq:476e}
&\quad + \tilde{\rho}(t-\tau)A_\beta(t-\tau)A_\alpha(t) \Tr[B_\alpha(t)\r_B B_\beta(t-\tau)] \ .
\end{align}
\ees
We now assume again that the bath is stationary (i.e., $[\r_B,H_B]=0$). 
As in Eq.~\eqref{eq:404}, let $\langle X \rangle_B \equiv {\rm Tr}[\rho_B X]$. 
Similarly to Eq.~\eqref{eq:421}, we define the bath two-point correlation function:
\bes
\begin{align}
\mc{B}_{\a\b}(t,t-\tau) &\equiv 
\langle B_{\alpha}(t)B_{\beta}(t-\tau)\rangle_{B}=\Tr\bigl(e^{iH_{B}t}B_{\alpha}e^{-iH_{B}t}e^{iH_{B}(t-\tau)}B_{\beta}e^{-iH_{B}(t-\tau)}\rho_{B}\bigr)\\
  &=\Tr\bigl(e^{-iH_{B}(t-\tau)}e^{iH_{B}t}B_{\alpha}e^{-iH_{B}t}e^{iH_{B}(t-\tau)}B_{\beta}\rho_{B}\bigr)=
  \Tr\bigl(e^{iH_{B}\tau}B_{\alpha}e^{-iH_{B}\tau}B_{\beta}\rho_{B}\bigr)=\langle B_{\alpha}(\tau) B_{\beta}\,\rangle_{B} \\ 
  & = \mc{B}_{\a\b}(\tau,0) \equiv \mc{B}_{\a\b}(\tau)\ ,
  \label{eq:477c}
\end{align}
\ees
where we used the bath stationarity assumption to go the second line, and in the third line we denoted $\mc{B}_{\a\b}(\tau,0)$ by $\mc{B}_{\a\b}(\tau)$ for simplicity, since only the time shift $\tau$ matters, so we can measure everything from $t=0$. Thus, $\mc{B}_{\a\b}(\tau)$ measures the autocorrelation of the bath after time $\tau$. Note that in the $\mc{B}_{\a\b}(\tau)$ notation we implicitly associate $t=\tau$ with the first index (in this case $\a$), whereas the second index is associated with $t=0$. Also,
\bes
\label{eq:511}
\begin{align}
\mc{B}_{\b\a}^*(\tau) &= \Tr[(\r_B B_\b(\tau)B_\a)^\dag] = \Tr[B_\a (U_B^\dag(\tau) B_\b U_B(\tau))^\dag\r_B] = \Tr[\r_B B_\a U_B^\dag(\tau) B_\b U_B(\tau) ]= \mc{B}_{\a\b}(0,\tau) \\
&=  \Tr[ \r_B U_B(\tau) B_\a U_B^\dag(\tau) B_\b ] = \mc{B}_{\a\b}(-\tau)\ .
\label{eq:511b}
\end{align}
\ees

Then, noting that the terms in lines \eqref{eq:476b} and \eqref{eq:476e} are Hermitian conjugates, as are the terms in lines \eqref{eq:476c} and \eqref{eq:476d}, we have:
\begin{equation}
\label{integEq2}
\frac{d\tilde{\rho}}{dt}=-g^2\sum_{\alpha\beta}\int_0^t d\tau\bigl\{\mc{B}_{\a\b}(\tau)\,[A_{\alpha}(t),A_{\beta}(t-\tau)\tilde{\rho}(t-\tau)]+\text{h.c.}\bigr\}\ .
\end{equation}
%

\subsection{Markov approximation and Redfield equation}

Note that the RHS of Eq.~\eqref{integEq2} depends on the entire history of the system state, since the argument of $\tilde{\rho}(t-\tau)$ ranges from $t$ to $0$ as $\tau$ increases from the lower to the upper limit of the integral. Thus, Eq.~\eqref{integEq2} is time-nonlocal. We would like to arrive at a time-local differential equation for the system state, which depends only on $t$, but not on the state's history.

To attain this, at this point we need to introduce our second approximation, the \emph{Markov approximation}. Informally, it states that the bath has a very short correlation time $\tau_B$, i.e., that the correlation function $\mc{B}_{\a\b}(\tau)$ decays rapidly with some characteristic timescale $\tau_B$, e.g., $|\mc{B}_{\a\b}(\tau)| \sim e^{-\tau/\tau_B}$. We also assume that 
\beq
g \ll 1/\tau_B \ , \qquad t \gg \tau_B\ .
\label{eq:479}
\eeq
The first of these is a weak-coupling limit ($g$ is small), and the second states that we do not expect our approximation to be accurate for times $t$ that are comparable to the bath correlation time (instead, we only consider times much larger than the latter).
Now, since the correlation function $\mc{B}_{\a\b}(\tau)$ is essentially zero for $\tau \gg \tau_B$, and since we assume that $t \gg \tau_B$, we can replace $\tilde{\rho}(t-\tau)$ by $\tilde{\rho}(t)$, since  the short ``memory" of the bath correlation function causes it to keep track of events only within the short period $[0,\tau_B]$. Under this approximations, Eq.~\eqref{integEq2} becomes:
\begin{equation}
\label{eq:Redfield}
\frac{d\tilde{\rho}}{dt}=-g^2\sum_{\alpha,\beta}\int_0^t d\tau\bigl\{\mc{B}_{\a\b}(\tau) [A_{\alpha}(t),A_{\beta}(t-\tau)\tilde{\rho}(t)]+\text{h.c.}\bigr\}  \ ,
\end{equation}
which is known as the \emph{Redfield equation}. It is notoriously non-CP, which means that the density matrix can be become non-positive (though various fixes have been proposed \cite{Gaspard:1999aa,Whitney:2008aa}).

Moreover, for the same reason (correlation function negligible for $\tau \gg \tau_B$) we can extend the upper limit of the integral to infinity without changing the value of the integral. 
\begin{equation}
\label{integEq3}
\frac{d\tilde{\rho}}{dt}=-g^2\sum_{\alpha,\beta}\int_0^\infty d\tau\bigl\{\mc{B}_{\a\b}(\tau) [A_{\alpha}(t),A_{\beta}(t-\tau)\tilde{\rho}(t)]+\text{h.c.}\bigr\} + O(g^4 \tau_B^3) \ ,
\end{equation}

That Eq.~\eqref{integEq2} can be replaced by Eq.~\eqref{integEq3} can be proven rigorously under the following sufficient condition \cite{ABLZ:12-SI}, as we will show in Sec.~\ref{sec:Markov-approx-err-bound}:
\beq
\int_0^\infty \tau^n |\mc{B}_{\a\b}(\tau)|  d\tau \sim \tau_B^{n+1}\ , \qquad n\in\{0,1,2\} \ .
\label{eq:481}
\eeq
This is satisfied, e.g., by an exponentially decaying correlation function. Indeed:
\beq
\int_0^\infty \tau^n e^{-\tau/\tau_B}  d\tau = \frac{d^n}{d(-1/\tau_B)^n} \int_0^\infty e^{-\tau/\tau_B} d\tau = \frac{d^n}{d(-1/\tau_B)^n} \left( \left. -\tau_B e^{-\tau/\tau_B} \right|_0^\infty \right) = \frac{d^n}{d(-1/\tau_B)^n} \tau_B = n! \tau_B^{n+1} \ .
\eeq
More generally, if $|\mc{B}_{\a\b}(\tau)| \sim e^{-(\tau/\tau_B)^k}$ where $k>0$, we have:
\beq
\int_0^\infty \tau^n |\mc{B}_{\a\b}(\tau)|  d\tau= \frac{1}{k}\Gamma\left(\frac{n+1}{k}\right) \tau_B^{n+1} \ ,
\label{eq:483}
\eeq
where $\Gamma(x)$ is the gamma function [recall that $\Gamma(n+1)=n!$ for $n\in\mathbb{N}$]. Thus, in fact even a subexponential ($k<1$) decay will suffice.

Note that thanks to Eq.~\eqref{eq:481}, the integral in Eq.~\eqref{integEq3} is of order $\tau_B$. Thus the ratio between the leading order correction and the integral is $(g^4 \tau_B^3)/(g^2\tau_B) = (g\tau_B)^2 \ll 1$, by our assumption that $g\tau_B \ll 1$.

\subsection{Going to the frequency domain}

After dropping the correction term, Eq.~\eqref{integEq3} is now a differential equation for $\tilde{\rho}(t)$, but is not yet in Lindblad form. To convert it into this form we once again convert the system operators $A(t)$ to the frequency domain. The procedure is essentially the same as in Sec.~\ref{sec:cumulant-2}, except that we need to keep track of the system operator index as well. Thus, after expanding $H_S$ in its eigenbasis as $H_S = \sum_a \varepsilon_a \ketb{\varepsilon_a}{\varepsilon_a}$, 
we have
\begin{equation}
A_\a(t) = U_S^\dagger(t) A_\a U_S(t) = \sum_{a, b} e^{-i \left( \varepsilon_b - \varepsilon_a \right) t } |\varepsilon_a \rangle \! \langle \varepsilon_a | A_\a | \varepsilon_b \rangle \! \langle \varepsilon_b| =  \sum_{\omega} A_\a({\omega}) e^{-i \omega t } \ ,
\label{eq:417Aa}
\end{equation}
where $\o \equiv \varepsilon_b - \varepsilon_a$ is a Bohr frequency, and  
\beq
A_\a({\omega}) \equiv \sum_{\varepsilon_b - \varepsilon_a=\o} \langle \varepsilon_a | A_\a | \varepsilon_b \rangle |\varepsilon_a \rangle \! \langle \varepsilon_b|\ = A^\dgr_\a(-{\omega}) \ ,
\label{eq:A_oma}
\eeq
where the last equality follows since Hermitian conjugation interchanges $\varepsilon_a$ and $\varepsilon_b$. Also, note that since $A_\a(t)$ is Hermitian,
\beq
\sum_{\omega} A_\a({\omega}) e^{-i \omega t } = \sum_{\omega} A^\dgr_\a({\omega}) e^{i \omega t }\ .
\eeq
Returning to Eq.~\eqref{integEq3}, consider the two terms in the commutator $[A_{\alpha}(t),A_{\beta}(t-\tau)\tilde{\rho}(t)]$:
\bes
\begin{align}
A_{\alpha}(t)A_{\beta}(t-\tau)\tilde{\rho}(t) &= \sum_{\o\o'} e^{i\o't}e^{-i\o(t-\tau)}A^\dgr_\a({\omega'})A_\b(\o)\tilde{\rho}(t) = \sum_{\o\o'} e^{i\o\tau}e^{i(\o'-\o)t}A^\dgr_\a({\omega'})A_\b(\o)\tilde{\rho}(t) \\
A_{\beta}(t-\tau)\tilde{\rho}(t) A_{\alpha}(t) &= \sum_{\o\o'} e^{-i\o(t-\tau)}e^{i\o't}A_\b(\o)\tilde{\rho}(t)A^\dgr_\a({\omega'}) = \sum_{\o\o'} e^{i\o\tau}e^{i(\o'-\o)t}A_\b(\o)\tilde{\rho}(t)A^\dgr_\a({\omega'}) \ .
\end{align}
\ees 
The entire $\tau$-dependence is thus in the factor $e^{i\o\tau}$, which motivates collecting everything that is $\tau$-dependent in Eq.~\eqref{integEq3} into one function:
\begin{equation}
\label{Gamma}
\Gamma_{\alpha\beta}(\omega)\equiv \int_0^\infty d\tau e^{i\omega\tau}\mc{B}_{\alpha\beta}(\tau)\ ,
\end{equation}
which is the one-sided Fourier transform of the bath correlation function. 
This allows us to rewrite Eq.~\eqref{integEq3} as
\begin{equation}
\label{BornMarkov2}
\frac{d\tilde{\rho}}{dt}=-g^2\sum_{\alpha,\beta}\sum_{\omega,\omega'}\bigl\{\Gamma_{\alpha\beta}(\omega)e^{i(\omega'-\omega)t}[A_\alpha^\dagger(\omega'),A_\beta(\omega)\tilde{\rho}(t)]\bigr\} + \text{h.c.}
\end{equation}
Note that $\Gamma$ as defined here has dimensions of \emph{time}, and $g^2\Gamma$ has units of \emph{frequency}.

\subsection{Rotating Wave Approximation}
\label{sec:RWA}

Alas, Eq.~\eqref{BornMarkov2} is still not in Lindblad form. The problem is the ``non-secular" (off-diagonal) terms with $\o\neq \o'$. While these did not present a problem in the cumulant derivation (recall that we proved complete positivity in Sec.~\ref{sec:CP-bB}), they do now. Therefore we next introduce the final approximation, known as the \emph{rotating wave approximation} (RWA), sometimes also called the secular approximation. This approximation is based on the idea that the terms with $\o\neq \o'$ in Eq.~\eqref{BornMarkov2} are rapidly oscillating if $t\gg |\omega-\omega'|^{-1}$, which thus (roughly) average to zero. Since we already assumed that $t\gg \tau_B$, the former assumption is consistent provided we also assume that the Bohr frequency differences satisfy
\beq
\min_{\o\neq\o'} |\omega-\omega'| > 1/\tau_B  \ .
\eeq
Note that this means that also the Bohr frequencies themselves (by setting $\o'=0$) must be large compared to the inverse of the bath correlation time, and this therefore excludes the treatment of systems with gaps that are small relative to $1/\tau_B$ (this has implications for the applicability to systems that are typically of interest in adiabatic quantum computing, for example). Also note that, combining this with the previous assumption [Eq.~\eqref{eq:479}], we get:
\beq
\boxed{
 g \ll 1/\tau_B <\min_{\o\neq\o'} |\omega-\omega'| 
 }
\eeq
This shows that the coupling also lower bounds the Bohr frequencies.

Let 
\beq
\gamma_{\alpha\beta}(\omega) = \int_{-\infty}^\infty e^{i\o \tau} \mc{B}_{\a\b}(\tau) d\tau \ ,
\label{eq:492}
\eeq
i.e., the full Fourier transform of the bath correlation function. 
Using Eq.~\eqref{eq:511}:
\beq
\gamma^*_{\alpha\beta}(\omega) = \int_{-\infty}^\infty e^{-i\o \tau} \mc{B}_{\b\a}(-\tau) d\tau = \int_{-\infty}^\infty e^{i\o \tau} \mc{B}_{\b\a}(\tau) d\tau = \gamma_{\b\a}(\omega) \ ,
\label{eq:493}
\eeq
i.e., $\g(\o)$ is a Hermitian matrix.
The inverse Fourier transform is 
\beq
\mc{B}_{\a\b}(\tau) = \frac{1}{2\pi} \int_{-\infty}^\infty e^{-i\o' \tau}  \gamma_{\alpha\beta}(\omega') d\o' \ .
\eeq
Then
\begin{align}
\Gamma_{\alpha\beta}(\omega) = \int_{0}^\infty e^{i\o \tau} \mc{B}_{\a\b}(\tau) d\tau 
= \int_{0}^\infty e^{i\o \tau} d\tau \frac{1}{2\pi} \int_{-\infty}^\infty e^{-i\o' \tau}  \gamma_{\alpha\beta}(\omega') d\o' = \frac{1}{2\pi} \int_{-\infty}^\infty d\o' \gamma_{\alpha\beta}(\omega') \int_{0}^\infty d\tau e^{i(\o-\o') \tau} \ .
\label{eq:494}
\end{align}
Now recall that the Dirac $\delta$ function can be represented as $\delta(x) = \frac{1}{2\pi} \int_{-\infty}^\infty d\tau  e^{ix\tau}$. When the integration lower limit is $0$ instead of $-\infty$, we have the identity
\beq
\int_{0}^\infty d\tau e^{ix\tau} = \pi \delta(x) + i \mc{P}\left(\frac{1}{x}\right)\ ,
\label{eq:Cauchy-P}
\eeq
where the Cauchy principal value is defined as
\beq
\mc{P}\left(\frac{1}{x}\right)[f] = \lim_{\epsilon\to 0} \int_{-\epsilon}^\epsilon \frac{f(x)}{x}dx \ ,
\eeq
for smooth functions $f$ with compact support on the real line $\mathbb{R}$. Substituting Eq.~\eqref{eq:Cauchy-P} into Eq.~\eqref{eq:494}, we can thus write
\begin{equation}
\label{gammaS}
\Gamma_{\alpha\beta}(\omega)=\frac{1}{2}\gamma_{\alpha\beta}(\omega)+iS_{\alpha\beta}(\omega) \ ,
\end{equation}
where
\beq
S_{\alpha\beta}(\omega) = \frac{1}{2\pi} \int_{-\infty}^\infty \gamma_{\alpha\beta}(\omega') \mc{P}\left(\frac{1}{\o-\o'}\right)d \omega' = S^*_{\beta\alpha}(\omega) \ ,
\label{eq:534}
\eeq
and we used the fact that $\g$ is Hermitian in the last equality. Therefore:
\beq
\gamma_{\alpha\beta}(\omega) = \Gamma_{\alpha\beta}(\omega) + \Gamma^{\ast}_{\beta\alpha}(\omega)\ , \quad S_{\alpha\beta}(\omega) = \frac{1}{2i}\left(\Gamma_{\alpha\beta}(\omega) - \Gamma^\ast_{\beta\alpha}(\omega)\right) \ .
\label{eq:499}
\eeq

Finally, we will show in Sec.~\ref{sec:BornMarkov2-Lindbladw} that by introducing Eq.~\eqref{gammaS} and the RWA into Eq.~\eqref{BornMarkov2}, we arrive at the interaction picture Lindblad equation:
\begin{equation}
\label{Lindbladw}
\frac{d\tilde{\rho}}{dt}=-i[H_{\mathrm{LS}},\tilde{\rho}(t)] + g^2\sum_\omega\sum_{\alpha\beta}\gamma_{\alpha\beta}(\omega)\Bigl(A_\beta(\omega)\tilde{\rho}(t)A_\alpha^\dagger(\omega)-\frac{1}{2}\{A_\alpha^\dagger(\omega)A_\beta(\omega),\tilde{\rho}(t)\}\Bigr)\ ,
\end{equation}
where the Lamb shift Hamiltonian is given by
\begin{equation}
\label{HLS}
H_{\mathrm{LS}}\equiv g^2\sum_\omega\sum_{\alpha\beta}S_{\alpha\beta}(\omega)A_\alpha^\dagger(\omega)A_\beta(\omega)\ .
\end{equation}
To justify calling $H_{\mathrm{LS}}$ a Hamiltonian we should show that it is Hermitian:
\begin{align}
H_\tx{LS}^\dagger &= g^2\sum_{\alpha\beta\omega}S_{\alpha\beta}^*(\omega)A_\beta^\dagger(\omega)A_\alpha(\omega) = g^2\sum_{\alpha\beta\omega}S_{\beta\alpha}(\omega)A_\beta^\dagger(\omega)A_\alpha(\omega) = g^2\sum_{\alpha\beta\omega}S_{\alpha\beta}(\omega)A_\alpha^\dagger(\omega)A_\beta(\omega) = H_\tx{LS}\ .
\end{align}
We will show in Sec.~\ref{proofofHLSHS=0} that 
\beq
[H_{\mathrm{LS}},H_S]=0 \ .
\label{eq:HLSHS=0}
\eeq
Note that $\gamma_{\a\b}(\o)$ and $S_{\a\b}(\o)$ as defined in Eqs.~\eqref{eq:492} and~\eqref{eq:534} have dimensions of \emph{time}, while $g^2\gamma_{\a\b}(\o)$ and $g^2 S_{\a\b}(\o)$ have units of \emph{frequency}. The factor $g^2$ can always be reabsorbed into the definition of $\gamma_{\alpha\beta}(\omega)$ and $S_{\a\b}(\o)$.%
\footnote{  
Also note that in our derivation of the LE using coarse graining (Sec.~\ref{sec:LE-CCG}) we did not include the coupling strength $g$. Instead we used a dimensionless parameter $\lambda$ when we wrote down the system-bath interaction as $\lambda H_{SB}$, where $H_{SB}$ has dimensions of energy. As a result, $\gamma_{\omega \omega'}$ in the CG-LE has units of frequency, while as noted above, in the RWA-LE $\g(\o)$ has units of time, and $g^2\g(\o)$ has units of frequency (or energy, since we're using units where $\hbar=1$).}

We will show in Sec.~\ref{eq:transtoSP} that we can transform back to Schr\"odinger picture via $\rho(t)=U_S(t)\tilde{\rho}(t)U_S^\dgr(t)$ and thus finally obtain the RWA-LE:
\begin{equation}
\label{LindbladSch}
\boxed{
\frac{d\rho}{dt}
=-i[H_S+H_{\mathrm{LS}},{\rho}]+g^2\sum_\omega\sum_{\alpha\beta}\gamma_{\alpha\beta}(\omega)\Bigl(A_\beta(\omega){\rho}A_\alpha^\dagger(\omega)-\frac{1}{2}\{A_\alpha^\dagger(\omega)A_\beta(\omega),{\rho}\}\Bigr)\ .
}
\end{equation}

We will show in Sec.~\ref{sec:gposproof} that:
\beq
\g(\o) > 0 \ ,
\eeq
as required for complete positivity.

Let us now provide all the missing steps indicated above.

\subsection{The missing steps}

\subsubsection{From Born-Markov [Eq.~\eqref{BornMarkov2}] to the RWA-LE [Eq.~\eqref{Lindbladw}]}
\label{sec:BornMarkov2-Lindbladw}

Let us start by expanding the commutator and the Hermitian conjugate term in Eq.~\eqref{BornMarkov2}, relabelling indices, and combining terms. This gives us for the summands:
\beq
e^{i(\omega'-\omega)t}\pen{\Gamma_{\alpha\beta}(\omega)A_\alpha^\dagger(\omega')A_\beta(\omega)\tilde{\r}(t)+\Gamma_{\beta\alpha}^*(\omega')\tilde{\r}(t)A_\alpha^\dagger(\omega')A_\beta(\omega)}  -e^{i(\omega'-\omega)t}\pen{\Gamma_{\alpha\beta}(\omega)+\Gamma_{\beta\alpha}^*(\omega')}A_\beta(\omega)\tilde{\r}(t)A_\alpha^\dagger(\omega')\ .
\eeq
	Applying the RWA (i.e., setting $\o=\o'$) and substituting $ \Gamma_{\alpha\beta}(\omega) = \half\gamma_{\alpha\beta}(\omega)+iS_{\alpha\beta}(\omega) $, this becomes:
\bes
\begin{align}
&\half\gamma_{\alpha\beta}(\omega)A_\alpha^\dagger(\omega)A_\beta(\omega)\tilde{\r}(t)+\half\gamma_{\beta\alpha}^*(\omega)\tilde{\r}(t)A_\alpha^\dagger(\omega)A_\beta(\omega) \\
& \qquad +iS_{\alpha\beta}(\omega)A_\alpha^\dagger(\omega)A_\beta(\omega)\tilde{\r}(t)-iS_{\beta\alpha}^*(\omega)\tilde{\r}(t)A_\alpha^\dagger(\omega)A_\beta(\omega)-\gamma_{\alpha\beta}(\omega)A_\beta(\omega)\tilde{\r}(t)A_\alpha^\dagger(\omega) \ .
\end{align}
\ees
Since $\gamma(\omega) $ and $ S(\omega) $ are Hermitian this becomes:
\bes
\begin{align}
&\half\gamma_{\alpha\beta}(\omega)\pen{A_\alpha^\dagger(\omega)A_\beta(\omega)\tilde{\r}(t)+\tilde{\r}(t)A_\alpha^\dagger(\omega)A_\beta(\omega)} \\
&\qquad +iS_{\alpha\beta}(\omega)\pen{A_\alpha^\dagger(\omega)A_\beta(\omega)\tilde{\r}(t)-\tilde{\r}(t)A_\alpha^\dagger(\omega)A_\beta(\omega)}-\gamma_{\alpha\beta}(\omega)A_\beta(\omega)\tilde{\r}(t)A_\alpha^\dagger(\omega)\\
&= \half\gamma_{\alpha\beta}(\omega)\cen{A_\alpha^\dagger(\omega)A_\beta(\omega),\tilde{\r}(t)}+iS_{\alpha\beta}(\omega)\ben{A_\alpha^\dagger(\omega)A_\beta(\omega),\tilde{\r}(t)}-\gamma_{\alpha\beta}(\omega)A_\beta(\omega)\tilde{\r}(t)A_\alpha^\dagger(\omega) \ .
\end{align}
\ees
Putting this back into the original sum in Eq.~\eqref{BornMarkov2} then gives us our desired result:
\bes
\begin{align}
\derv{\tilde{\r}}{t}{} &= -i{g^2\sum_{\alpha,\beta,\omega}S_{\alpha\beta}(\omega)\ben{A_\alpha^\dagger(\omega)A_\beta(\omega),\tilde{\r}(t)}} +g^2\sum_{\alpha,\beta,\omega}\gamma_{\alpha\beta}(\omega)\pen{A_\beta(\omega)\tilde{\r}(t)A_\alpha^\dagger(\omega)-\half\cen{A_\alpha^\dagger(\omega)A_\beta(\omega),\tilde{\r}(t)}} \\
& = -i\ben{H_\tx{LS},\tilde{\r}(t)}+g^2\sum_{\alpha,\beta,\omega}\gamma_{\alpha\beta}(\omega)\pen{A_\beta(\omega)\tilde{\r}(t)A_\alpha^\dagger(\omega)-\half\cen{A_\alpha^\dagger(\omega)A_\beta(\omega),\tilde{\r}(t)}} \ .
\end{align}
\ees

\subsubsection{Proof of Eq.~\eqref{eq:HLSHS=0}}
\label{proofofHLSHS=0}

Let us write the system operators [Eq.~\eqref{eq:A_oma}] as 
\begin{align}
{{A}_{\a }}(\omega )=\sum_{\varepsilon_b-\varepsilon_a=\omega }\Pi(\varepsilon_a)A_{\a }\Pi(\varepsilon_b)={{A}^\dgr_{\alpha }}(-\omega )\ ,
\label{eq:552a}
\end{align}
where the projectors $\Pi(\varepsilon_a) = \ketb{\varepsilon_a}{\varepsilon_a}$ are the eigenprojectors of $H_S$, i.e., 
\beq
H_S = \sum_a \varepsilon_a \Pi(\varepsilon_a) \ ,
\label{eq:HSedecomp}
\eeq
and hence $H_S\Pi(\varepsilon_a)=\Pi(\varepsilon_a) H_S=\varepsilon_a\Pi(\varepsilon_a)$. 
Then:
\bes
\begin{align}
H_SA_\alpha^\dagger(\omega)A_\beta(\omega) &= \sum_a \varepsilon_a \Pi(\varepsilon_a)\sum_{\varepsilon_i-\varepsilon_j=\omega}\Pi(\varepsilon_i)A_\alpha^\dagger\Pi(\varepsilon_j)\sum_{\varepsilon_k-\varepsilon_l=\omega}\Pi(\varepsilon_l)A_\beta\Pi(\varepsilon_k) \\
&= \sum_{\varepsilon_i-\varepsilon_j=\omega}\varepsilon_i\Pi(\varepsilon_i)A_\alpha^\dagger\Pi(\varepsilon_j)\sum_{\varepsilon_k-\varepsilon_l=\omega}\Pi(\varepsilon_l)A_\beta\Pi(\varepsilon_k)\ ,
\end{align}
\ees
and similarly:
\begin{align}
A_\alpha^\dagger(\omega)A_\beta(\omega)H_S = \sum_{\varepsilon_i-\varepsilon_j=\omega}\Pi(\varepsilon_i)A_\alpha^\dagger\Pi(\varepsilon_j)\sum_{\varepsilon_k-\varepsilon_l=\omega}\varepsilon_k\Pi(\varepsilon_l)A_\beta\Pi(\varepsilon_k) \ .
\end{align}
It follows that 
\bes
\begin{align}
\ben{H_S,A_\alpha^\dagger(\omega)A_\beta(\omega)} &= \sum_{\substack{\varepsilon_i-\varepsilon_j=\omega\\\varepsilon_k-\varepsilon_l=\omega}}\pen{\varepsilon_i-\varepsilon_k}\Pi(\varepsilon_i)A_\alpha^\dagger\Pi(\varepsilon_j)\Pi(\varepsilon_l)A_\beta\Pi(\varepsilon_k)\\
&= \sum_{\substack{\varepsilon_i-\varepsilon_j=\omega\\\varepsilon_k-\varepsilon_j=\omega}}\pen{\varepsilon_i-\varepsilon_k}\Pi(\varepsilon_i)A_\alpha^\dagger\Pi(\varepsilon_j)A_\beta\Pi(\varepsilon_k)\\
&=0\ ,
\end{align}
\ees
where the second line follows from the product of the two inner projection operators, and the third line from the summation conditions, which set $ \varepsilon_i=\varepsilon_k $. Consequently:
\beq
\ben{H_S,H_\tx{LS}} = g^2\sum_{\alpha\beta\omega}S_{\alpha\beta}(\omega)\ben{H_S,A_\alpha^\dagger(\omega)A_\beta(\omega)} = 0\ .
\eeq

\subsubsection{Transformation back to the Schr\"{o}dinger picture}
\label{eq:transtoSP}

Recall that $ \tilde{\r}(t) = e^{iH_St}\rho(t)e^{-iH_St} $, so $\derv{\tilde{\r}}{t}{} = i\ben{H_S,\rho}+e^{iH_St}\derv{\rho}{t}{}e^{-iH_St}$, and hence:
\beq
 \derv{\rho}{t}{} = -i\ben{H_S,\rho}+e^{-iH_St}\derv{\tilde{\r}}{t}{}e^{iH_St}\ .
 \label{eq:513}
\eeq
Also, using Eq.~\eqref{eq:HSedecomp} again:
\beq
A_\alpha(\omega)e^{iH_St} = \sum_{\varepsilon'-\varepsilon=\omega}\Pi(\varepsilon)A_\alpha\Pi(\varepsilon')e^{i\varepsilon't}\ .
\eeq
Thus:
\bes
\begin{align}
e^{-iH_St} A_\beta(\omega)\tilde{\r}A_\alpha^\dagger(\omega) e^{iH_St}&= \sum_{\substack{\varepsilon_i-\varepsilon_j=\omega\\\varepsilon_k-\varepsilon_l=\omega}}e^{-iH_St}\Pi(\varepsilon_j)A_\beta\Pi(\varepsilon_i)e^{iH_St}\rho(t)e^{-iH_St}\Pi(\varepsilon_k)A_\alpha^\dagger\Pi(\varepsilon_l)e^{iH_St} \\
& = \sum_{\substack{\varepsilon_i-\varepsilon_j=\omega\\\varepsilon_k-\varepsilon_l=\omega}}e^{i\pen{-\varepsilon_j+\varepsilon_i-\varepsilon_k+\varepsilon_l}t}\Pi(\varepsilon_j)A_\beta\Pi(\varepsilon_i)\rho(t)\Pi(\varepsilon_k)A_\alpha^\dagger\Pi(\varepsilon_l) \\
& = \sum_{\substack{\varepsilon_i-\varepsilon_j=\omega\\\varepsilon_k-\varepsilon_l=\omega}}\Pi(\varepsilon_j)A_\beta\Pi(\varepsilon_i)\rho(t)\Pi(\varepsilon_k)A_\alpha^\dagger\Pi(\varepsilon_l) \\
&= A_\beta(\omega)\rho A_\alpha^\dagger(\omega) \ ,
\end{align}
\ees
	and
\bes
\begin{align}
e^{-iH_St} A_\alpha^\dagger(\omega)A_\beta(\omega)\tilde{\r}(t)e^{iH_St}  &= \sum_{\substack{\varepsilon_i-\varepsilon_j=\omega\\\varepsilon_k-\varepsilon_l=\omega}}e^{-iH_St} \Pi(\varepsilon_i)A_\alpha^\dagger\Pi(\varepsilon_j) \Pi(\varepsilon_l)A_\beta\Pi(\varepsilon_k) e^{iH_St}\rho(t)e^{-iH_St} e^{iH_St}\\
&=\sum_{\substack{\varepsilon_i-\varepsilon_j=\omega\\\varepsilon_k-\varepsilon_l=\omega}}e^{i\pen{-\varepsilon_i+\varepsilon_k}t}\Pi(\varepsilon_i)A_\alpha^\dagger\Pi(\varepsilon_j)\Pi(\varepsilon_l)A_\beta\Pi(\varepsilon_k)\rho(t) \\
& = \sum_{\substack{\varepsilon_i-\varepsilon_j=\omega\\\varepsilon_k-\varepsilon_l=\omega}}e^{i\pen{-\varepsilon_j+\varepsilon_l}t}\Pi(\varepsilon_i)A_\alpha^\dagger\Pi(\varepsilon_j)\Pi(\varepsilon_l)A_\beta\Pi(\varepsilon_k)\rho(t) = A_\alpha^\dagger(\omega)A_\beta(\omega)\rho(t) \ ,
\end{align}
\ees
	and similarly for the second term in the anti-commutator. This shows that
\beq
e^{-iH_St} \Bigl(A_\beta(\omega)\tilde{\rho}(t)A_\alpha^\dagger(\omega)-\frac{1}{2}\{A_\alpha^\dagger(\omega),A_\beta(\omega),\tilde{\rho}(t)\}\Bigr) e^{iH_St} = \Bigl(A_\beta(\omega){\rho}(t)A_\alpha^\dagger(\omega)-\frac{1}{2}\{A_\alpha^\dagger(\omega),A_\beta(\omega),{\rho}(t)\}\Bigr)\ .
\eeq
Now, since we showed that $ H_S $ and $ H_\tx{LS} $ commute:
\bes
\begin{align}
e^{-iH_St} [H_{\mathrm{LS}},\tilde{\rho}(t)]e^{iH_St} &= e^{-iH_St} H_{\mathrm{LS}} e^{iH_St}e^{-iH_St}\tilde{\rho}(t)e^{iH_St} - e^{-iH_St} \tilde{\rho}(t) e^{iH_St}e^{-iH_St}H_{\mathrm{LS}}e^{iH_St} \\
& = [H_{\mathrm{LS}},{\rho}(t)]\ .
\end{align}
\ees
Hence, using Eqs.~\eqref{Lindbladw} and \eqref{eq:513} we obtain Eq.~\eqref{LindbladSch} as required.

\subsubsection{Proof that $\g(\o) > 0$}
\label{sec:gposproof}
We'll give two different proofs. 

\paragraph{First proof}
The idea is to establish the following identity:
\begin{mylemma}
\beq
\gamma_{\alpha\beta}(\omega)=\int_{-\infty}^{+\infty}e^{i\omega u}\mathcal{B}_{\alpha\beta}(u)du=\lim_{T\to\infty}\frac{1}{T}\int_{0}^{T}dt\int_{0}^{T}e^{i\omega (t-s)}\mathcal{B}_{\alpha\beta}(t-s)ds\ .
\eeq
\end{mylemma}

\begin{proof}
Consider the following integral:
\beq
I(\o,T)\equiv \frac{1}{T}\int_{0}^{T}dt \int_{0}^{T}e^{i\omega (t-s)}\mathcal{B}_{\alpha\beta}(t-s)ds\ .
\eeq

\begin{figure}[t]
	\includegraphics[width=.7\linewidth]{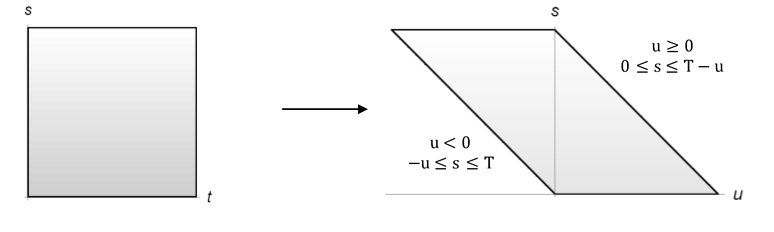}
	\caption{Left: original integration region. Right: new integration region.}
	\label{fig:intregion}
\end{figure}

First, we change the variables from $(t,s)$ to $(u,s)$ with $u=t-s$. For every value of $s$, sweeping $t$ from $0$ to $T$ will yield a horizontal line of length $T$ in the $(u,s)$ plane. The new integration region is therefore a parallelogram in the variables $(u,s)$, as illustrated in Fig.~\ref{fig:intregion}. We can split this region into $u\in[-T,0]$ and $u\in[0,T]$, and perform the integration over $s$ first. As is clear from the figure, $s$ varies from $-u$ to $T$ in the $u\in[-T,0]$ region, and from $0$ to $T-u$ in the $u\in[0,T]$ region. The area is preserved so the Jacobian yields $1$. 
Consequently, 
\beq
\int_{0}^{T}ds\int_{0}^{T}dt=\int_{-T}^{0}du\int_{-u}^{T}ds+\int_{0}^{T}du\int_{0}^{T-u}ds\ .
\eeq
When we integrate over some function independent of $s$,
\beq
\int_{-T}^{0}du\int_{-u}^{T}dsf(u)+\int_{0}^{T}du\int_{0}^{T-u}dsf(u)
=\int_{-T}^{0}duf(u)(T+u)+\int_{0}^{T}duf(u)(T-u)=\int_{-T}^{T}f(u)(T-|u|)\ .
\eeq
Therefore, after the change of variables we get
\bes
\begin{align}
I(\o,T)&=\frac{1}{T}\int_{-T}^{T}e^{i\omega u}\mathcal{B}_{\alpha\beta}(u)(T-|u|)du\\
&=\int_{-T}^{T}e^{i\omega u}\mathcal{B}_{\alpha\beta}(u)du-\frac{1}{T}\int_{-T}^{T}e^{i\omega u}\mathcal{B}_{\alpha\beta}(u)|u|du\ .
\end{align}
\ees
Now recall that in the Markov approximation we assumed [Eq.~\eqref{eq:481}] that $\int_{0}^{\infty}u^n|\mathcal{B}_{\alpha\beta}(u)|du\sim\tau_B^{n+1}$, where $\tau_B<\infty$ is the bath correlation time.
Therefore, using Eq.~\eqref{eq:477c}:
\bes
\begin{align}
\int_{-T}^{T}e^{i\omega u}|u|\mathcal{B}_{\alpha\beta}(u)du 
&=\int_{0}^{T}e^{i\omega u}u\mathcal{B}_{\alpha\beta}(u)du-\int_{-T}^{0}e^{i\omega u}u\mathcal{B}_{\alpha\beta}(u)du\\
&=\int_{0}^{T}e^{i\omega u}u\mathcal{B}_{\alpha\beta}(u)du+\int_{0}^{T}e^{-i\omega u}u\mathcal{B}^*_{\beta\alpha}(u)du \\
&\leq \int_{0}^{\infty}u|\mathcal{B}_{\alpha\beta}(u)|du+\int_{0}^{\infty}u|\mathcal{B}^*_{\beta\alpha}(u)|du\sim2\tau^2_B\ .
\end{align}
\ees
Consequently $\lim_{T\to\infty}\frac{1}{T}\int_{-T}^{T}e^{i\omega u}|u|\mathcal{B}_{\alpha\beta}(u)du=0$, and
\beq
\lim_{T\to\infty} I(\o,T) = \g_{\alpha\beta}(\o)
\eeq
as claimed.
\end{proof}

Now, for any vector 
$v=\left(v_1,v_2,\dots \right)^t$ we have 
\bes
\begin{align}
v^\dgr\gamma(\o) v &= \sum_{\alpha\beta}v^*_\alpha\gamma_{\alpha\beta}(\o) v_\beta=\sum_{\alpha\beta}v^*_\alpha v_\beta\int_{-\infty}^{+\infty}e^{i\omega u}\mathcal{B}_{\alpha\beta}(u)du \\
&= \lim_{T\to\infty}\frac{1}{T}\sum_{\alpha\beta}v_\alpha^*v_\beta\int_{0}^{T}dt\int_{0}^{T}e^{i\omega (t-s)}\mathcal{B}_{\alpha\beta}(t-s)ds\\
&=\lim_{T\to\infty}\frac{1}{T}\sum_{\alpha\beta,\mu}\lambda_\mu\langle\mu|\int_{0}^{T}v_\alpha^*e^{i\omega t}B_\alpha(t)dt\int_{0}^{T}v_\beta e^{-i\omega s}B_\beta(s)ds|\mu\rangle\\
& = \lim_{T\to\infty}\frac{1}{T}\sum_{\mu}\lambda_\mu\left|\sum_\a\int_{0}^{T}v_\a B_\a(s)e^{-i\omega s}ds|\mu\rangle\right|^2\ge 0\ .
\end{align}
\ees
Therefore $\gamma(\o) \ge 0$.

\paragraph{Second proof}
The following proof 
uses Bochner's theorem as suggested, e.g., in the textbook~\cite{Breuer:book}.

Since $\g(\o)$ is Hermitian [Eq.~\eqref{eq:493}] we can diagonalize it using a unitary transformation:
\beq
 D \equiv U\gamma U^\dagger \Rightarrow D_{\alpha\beta} = \sum_{i,j}U_{\alpha i}\gamma_{ij}U_{\beta j}^* \ .
 \eeq
$ D $ is diagonal so we need only consider the diagonal elements (i.e., the eigenvalues of $ \gamma $). Plugging in $ \gamma_{ij} = \int_{-\infty}^\infty e^{i\omega s}\mathcal{B}_{ij}(s)ds$ gives
\beq
D_\alpha = \int_{-\infty}^\infty e^{i\omega s}\pen{\sum_{i,j}U_{\alpha i}\mathcal{B}_{ij}(s)U_{\alpha j}^*}ds\ .
\eeq
	We wish to show that $ D_\alpha $ is non-negative for each $ \alpha $. To do this we must consider the function in parenthesis. $ D_\alpha $ is the Fourier transform of this function so if we can show that it is of positive type then $ D_\alpha $ must be positive by Bochner's theorem \cite{Reed:1975aa}. Define the following function with $ \cen{t_i} $ an arbitrary time partition:
\beq
f_{mn}^\alpha \equiv \sum_{i,j}U_{\alpha i}\mathcal{B}_{ij}(t_m-t_n)U_{\alpha j}^* \ .
\eeq
	Now use the property $ \expv{B_\alpha(s)B_\beta(0)} = \expv{B_\alpha(t)B_\beta(t-s)} $ [Eq.~\eqref{eq:477c}] to write $ f_{mn}^\alpha $ as
\beq
 f_{mn}^\alpha = \sum_{i,j}U_{\alpha i}\Tr\ben{\rho_BB_i(t_m)B_j(t_n)}U_{\alpha j}^* = \Tr\pen{\rho_B\sum_iU_{\alpha i}B_i(t_m)\sum_jU_{\alpha j}^*B_j(t_n)} \ .
 \eeq
	We need to show that $ f^\alpha $ is a positive matrix. For arbitrary $ \ket{v} $ we have
\bes
\begin{align}
\bra{v}f^\alpha\ket{v} &= \sum_{m,n}v_m^*v_nf_{mn}^\alpha = \Tr\ben{\pen{\sum_{i,m}v_m^*U_{\alpha i}\sqrt{\rho_B}B_i(t_m)}\pen{\sum_{j,n}v_nU_{\alpha j}^*B_j(t_n)\sqrt{\rho_B}}} \\
& = \Tr\ben{\pen{\sum_{i,m}v_m^*U_{\alpha i}\sqrt{\rho_B}B_i(t_m)}\pen{\sum_{i,m}v_mU_{\alpha i}^*B_i(t_m)\sqrt{\rho_B}}} \\
&\equiv \Tr\pen{M_\alpha^\dagger M_\alpha} \geq 0 \ ,
\end{align}
\ees
	where the final inequality follows from the fact that $ M_\alpha^\dagger M_\alpha $ is non-negative which follows immediately from right polar decomposing $ M_\alpha $ (then $ M_\alpha^\dagger M_\alpha = RU^\dagger UR = R^2 \geq 0 $).
	
We have established that $ \bra{v}f^\alpha\ket{v} \geq 0 $ for any time partition $ \cen{t_i} $. Therefore $ D_\alpha $ is positive by Bochner's theorem. Consequently, $ \gamma $ is a positive matrix since all its eigenvalues are non-negative.



\section{The Kubo-Martin-Schwinger (KMS) condition and the Gibbs state as a stationary state Lindblad equation}

In this section we formalize the folklore notion that ``systems like to relax into lower energy states", and that systems ``tend to equilibrate".

\subsection{The KMS condition}

Consider a general system-bath Hamiltonian of the form $H_{SB} = \sum_{a} A_a\ox B_a$ (we're using $a$ and $b$ since we'll reserve $\b$ for the inverse temperature in this subsection).
Let us assume again that the bath state is {stationary} [Eq.~\eqref{eq:b-stat}], which as we saw implies that $\r_B(t) = U_B(t)\r_B(0)U_B^\dgr(t) = \r_B(0) \equiv \r_B$. We also saw that this means that the bath correlation function is time-translation-invariant:
\beq
\ave{B_a(t+\tau)B_b(t)} = 
\ave{B_a(\tau) B_b(0)}\ ,
\eeq
where for notational simplicity we dropped the $B$ subscript we used before in $\ave{X}_B = \Tr[\r_B X]$.

%

If we assume not only that the bath state is stationary, but that it is also in thermal equilibrium at inverse temperature $\beta$, i.e., $\rho_B=e^{-\beta H_B}/\mathcal{Z}$, then it follows that the correlation function satisfies the \emph{Kubo-Martin-Schwinger (KMS) condition} \cite{Breuer:book}:
\beq
\label{eq:KMSt}
\langle B_a(\tau)B_b(0)\rangle = \langle B_b(0)B_a(\tau+i\beta)\rangle  \ .
\eeq
The proof is the following calculation: 
\bes
\bea
\langle B_a(\tau)B_b \rangle &=& \textrm{Tr}[\rho_B U^\dag_B(\tau) B_a U_B(\tau) B_b] = \frac{1}{\mathcal{Z}}\textrm{Tr}[B_b e^{-(\beta-i\tau)H_B} B_a e^{-i\tau H_B}] \\
&=& \frac{1}{\mathcal{Z}}\textrm{Tr}[B_b e^{i(\tau+i\beta)H_B} B_a e^{-i(\tau+i\beta) H_B}e^{-\beta H_B}]
= \textrm{Tr}[\rho_B B_b U^\dag_B(\tau+i\beta) B_a U_B(\tau+i\beta) ] \\
&=& \langle B_b B_a(\tau+i\beta)\rangle\ .
\eea 
\ees
Note that using the same technique it also follows that 
\beq
\langle B_a(\tau)B_b\rangle = \langle B_{b}(-\tau - i \beta) B_{a} \rangle\ .
\label{eq:KMS-t}
\eeq 

\begin{figure}[b]
   \includegraphics[width=2in]{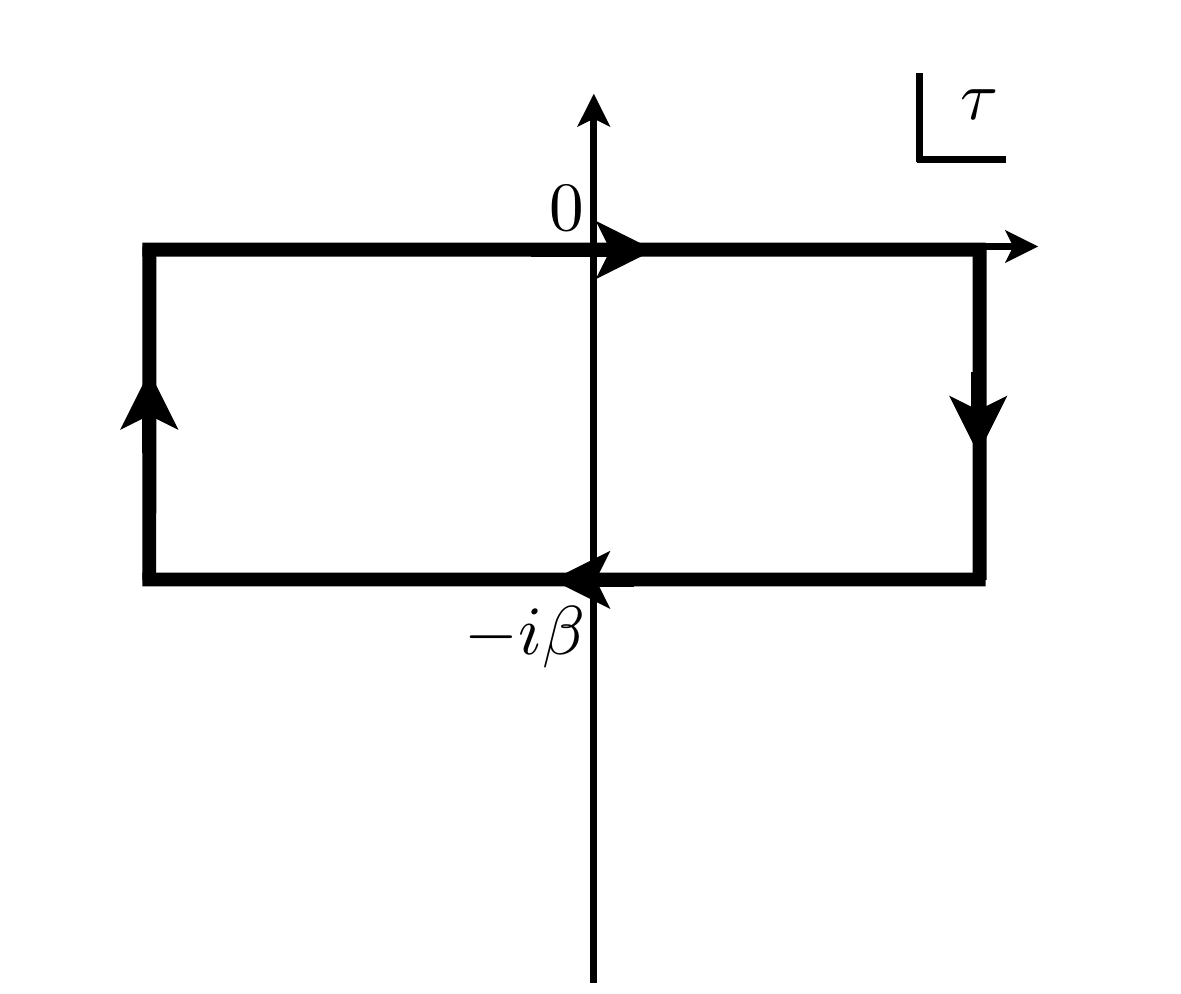} 
   \caption{Contour used in proof of the KMS condition.}    
   \label{fig:Contour}
\end{figure}

If in addition the correlation function is analytic in the strip between $\tau=-i\b$ and $\tau=0$, then it follows that the Fourier transform of the bath correlation function satisfies the \emph{frequency domain KMS condition}:
\beq
\label{eq:KMS}
\gamma_{ab}(-\omega) = e^{-\beta\omega}\gamma_{ba}(\omega)\ .
\eeq
This is an extremely important condition, which is used in proving ``detailed balance", as we shall see when we discuss the Pauli master equation, in Sec.~\ref{sec:PauliME}.

To prove this let us use the time-domain KMS condition, Eq.~\eqref{eq:KMS-t}:
\bea
\gamma_{ab} (\omega) = \int_{-\infty}^{\infty} d \tau e^{i \omega \tau} \langle B_a(\tau)B_b(0)\rangle 
= \int_{-\infty}^{\infty} d \tau e^{i \omega \tau} \langle B_{b}(-\tau - i \beta) B_{a}(0) \rangle
\label{eqt:integrand1}
\eea
To perform this integral we replace it with a contour integral in the complex $\tau$ plane, 
$\oint_C d \tau e^{i \omega \tau} \langle B_{b}(-\tau - i \beta) B_{a}(0) \rangle $,
with the contour $C$ as shown in Fig.~\ref{fig:Contour}. This contour integral vanishes by the Cauchy-Goursat theorem \cite{complex:book} since the closed contour encloses no poles (by assumption, the correlation function $\langle B_{b}(\tau) B_{a}(0)\rangle$ is analytic in the open strip $(0, -i \beta)$ and is continuous at the boundary of the strip \cite{KMS}), so that
\begin{equation}
\oint_{C} \left( \dots \right) = 0 = \int_{\uparrow} \left( \dots \right) + \int_{\mathrm{\downarrow}}  \left( \dots \right)  + \int_{\rightarrow}  \left( \dots \right)  +  \int_{\leftarrow}  \left( \dots \right) 
\end{equation}
where $\left( \dots \right)$ is the integrand of Eq.~\eqref{eqt:integrand1}, and the integral $ \int_{\rightarrow}$ is the same as in Eq.~\eqref{eqt:integrand1}.  After making the variable transformation $\tau = -x - i \beta$, where $x$ is real, we have
\begin{equation}
\int_{\leftarrow}  \left( \dots \right)  = -  e^{\beta \omega}  \int_{-\infty}^{\infty} dx\ e^{-i \omega x}  \langle B_{b}(x) B_{a} \rangle = -  e^{\beta \omega}\gamma_{ba}(-\omega) \ .
\end{equation}
Assuming that $\langle B_a(\pm\infty-i\b)B_b(0)\rangle = 0$ (i.e., the correlation function vanishes at infinite time), we further have $\int_{\uparrow} \left( \dots \right) = \int_{\mathrm{\downarrow}}  \left( \dots \right) =0$, and hence we find the result:
\begin{equation}
0 = \gamma_{ab}(\omega) + 0 + 0 -  e^{\beta \omega}\gamma_{ba}(-\omega)
\end{equation}
which proves Eq.~\eqref{eq:KMS}. 

The KMS condition \eqref{eq:KMS} is important, since it tells us that transitions involving negative Bohr frequencies are exponentially suppressed, as $e^{-\b\o}$, compared to the opposite transitions involving positive Bohr frequencies. I.e., when a system is coupled to a thermal bath, an excitation in the system is exponentially suppressed relative to a relaxation event at the same frequency.%
\footnote{Recall Eq.~\eqref{eq:A_om}: $\o=\varepsilon_b-\varepsilon_a<0$ corresponds to a transition from $\ket{\varepsilon_b}$ to $\ket{\varepsilon_a}$, i.e., from energy $\varepsilon_b$ to a \emph{higher} energy $\varepsilon_a$.}


\subsection{The Gibbs state is a stationary state of the RWA-LE}
\label{sec:Gibbs-stat}

Consider a bath at inverse temperature $\b$.
We would like to show that the system Gibbs state
\beq 
\r_G=\frac{1}{Z}{{e}^{-\beta {{H}_{S}}}} 
=\frac{1}{Z}\left(
               \begin{array}{ccc}
                 e^{-\beta\varepsilon_0} &  &  \\
                  & e^{-\beta\varepsilon_1} &  \\
                  &  & \ddots\\
               \end{array}
             \right)
\ , \quad Z=\Tr[{{e}^{-\beta {{H}_{S}}}}]
\label{eq:Gibbs_state}
\eeq
is always a stationary state, in the sense that $\dot{\r}_G = 0$. Here the energies are listed in increasing order, starting from ground state energy $\epsilon_0$. We will show this here directly from the RWA-LE, and given an alternative derivation from the Pauli master equation in Sec.~\ref{sec:PauliME}.

In the Schr\"{o}dinger picture the RWA-LE has the form:
\beq
\dot{\rho }=-i\left[ H_S+H_{\mathrm{LS}},\rho  \right]+D(\rho )\ ,
\eeq
where the dissipator is
\beq
D(\rho )=g^2 \sum_{\alpha \beta }{\sum_{\omega }{{{\gamma }_{\alpha \beta }}(\omega )\left( {{A}_{\beta }}(\omega )\rho A_{\alpha }^{\dagger }(\omega )-\frac{1}{2}\left\{ A_{\alpha }^{\dagger }(\omega ){{A}_{\beta }}(\omega ),\rho  \right\} \right)}}\ .
\label{eq:dissi}
\eeq

To show that $\dot{\r}_G = 0$, consider first the Hamiltonian part. That $[H_S,\r_G]=0$ follow immediately from Eq.~\eqref{eq:Gibbs_state}. Now recall that $[H_S,H_{\mathrm{LS}}]=0$ [Eq.~\eqref{eq:HLSHS=0}]. Thus $H_S$ and $H_{\mathrm{LS}}$ are diagonalizable in the same basis, i.e., there exists a unitary $V$ such that $VH_S V^\dgr=D_1$ and $VH_{\mathrm{LS}} V^\dgr=D_2$, where $D_1$ and $D_2$ are both diagonal (and of course commute). Therefore 
\beq
V[H_{\mathrm{LS}},\r_G]V^\dgr = VH_{\mathrm{LS}}V^\dgr V\r_GV^\dgr- V\r_GV^\dgr VH_{\mathrm{LS}}V^\dgr = \frac{1}{Z}\left[D_2, e^{-\b D_1}\right] = 0 \ ,
\eeq
which means that $[H_{\mathrm{LS}},\r_G]=0$.

Next let us consider the dissipative part. This requires us to calculate $\r_G A^\dgr_\a(\o)$ and $A_\b(\o)\r_G$. Now, for any pair of operators $A$ and $B$ it is easy to prove (e.g., by Taylor expansion) that:
\beq 
\label{eq:math_help}
{{e}^{-\alpha A}}B{{e}^{\alpha A}}=\sum_{n=0}^{\infty }{\frac{{{(-\alpha) }^{n}}}{n!}\left[ A,B \right]_n}\ ,
\eeq
where the nested commutator is defined recursively via
\beq
\left[A,B \right]_n=\left[ A,\left[ A,B \right]_{n-1} \right]\ , \left[ A,B \right]_0=B \ .
\eeq
Simplifying our notation via $\ket{a}\equiv\ket{\varepsilon_a}$, let us write the system operators [Eq.~\eqref{eq:A_oma}] as 
\begin{align}
{{A}_{\a }}(\omega )=\sum_{b-a=\omega }{{{\Pi }_{a}}{{A}_{\a }}{{\Pi }_{b}}}={{A}^\dgr_{\alpha }}(-\omega )\ ,
\label{eq:552}
\end{align}
where the projectors $\Pi_a = \ketb{a}{a}$ are in the energy basis, i.e., $H_S\Pi_a=\Pi_a H_S=a\Pi_a$, where $H_S = \sum_a a \Pi_a$.
Using the property $\Pi_a \Pi_b = \d_{ab} \Pi_a$, note that:
\bes 
\label{eq:H_s-A_a-commutator}
\begin{align}
\left[ {{H}_{S}},{{A}_{\alpha }}(\omega ) \right] &= \sum_a a \Pi_a \sum_{b-a'=\o}\Pi_{a'}A_\a\Pi_b - \sum_{b-a=\o}\Pi_a A_\a \Pi_b \sum_{a'} a' \Pi_{a'} \\
&= \sum_{b-a=\o} a \Pi_a A_\a \Pi_b-\sum_{b-a=\o} b\Pi_a A_\a \Pi_b=\sum_{b-a=\o} (a-b) \Pi_a A_\a \Pi_b\\
\label{eq:553c}
&= -\omega {{A}_{\alpha }}(\omega )\\
\label{eq:553d}
\left[ {{H}_{S}},A_{\alpha }^{\dagger }(\omega ) \right]&=\omega A_{\alpha }^{\dagger }(\omega )\ ,
\end{align}
\ees
where Eq.~\eqref{eq:553d} follows by taking the Hermitian conjugate of Eq.~\eqref{eq:553c}.

Therefore:
\bes
\begin{align}
[H_S,A_\a(\o)]_n &= (-\o)^n A_\a(\o) \\
[H_S,A^\dgr_\a(\o)]_n &= \o^n A_\a(\o)\ .
\end{align}
\ees
Hence, using Eq.~\eqref{eq:math_help}:
\beq
{{e}^{-\beta {{H}_{S}}}}{{A}_{\alpha }}(\omega ){{e}^{\beta {{H}_{S}}}}=
\sum_{n=0}^\infty \frac{(-\b)^n(-\o)^n}{n!}A_\a(\o)=
{{e}^{\beta \omega }}{{A}_{\alpha }}(\omega )\ ,
\eeq
which tells us that 
\beq 
\label{eq:A_ro}
{{A}_{\alpha }}(\omega ){{\rho }_{G}}={{e}^{-\beta \omega }}{{\rho }_{G}}{{A}_{\alpha }}(\omega ) \ .
\eeq
It follows by Hermitian conjugation that:
\beq 
\label{eq:A-dagger_ro}
{{\rho }_{G}}A_{\alpha }^{\dagger }(\omega )={{e}^{-\beta \omega }}A_{\alpha }^{\dagger }(\omega ){{\rho }_{G}}\ .
\eeq
We are now ready to consider the terms in the dissipator, Eq.~\eqref{eq:dissi}. Commuting $\r_G$ to the right we find:
\bes
\begin{align}
{{A}_{\beta }}(\omega ){{\rho }_{G}}A_{\alpha }^{\dagger }(\omega )&={{e}^{-\beta \omega }}{{A}_{\beta }}(\omega )A_{\alpha }^{\dagger }(\omega ){{\rho }_{G}}\\
{{\rho }_{G}}A_{\alpha }^{\dagger }(\omega ){{A}_{\beta }}(\omega )&=e^{-\b\o}A_\a^\dgr(\o)\r_GA_\a(\o)=A_{\alpha }^{\dagger }(\omega ){{A}_{\beta }}(\omega ){{\rho }_{G}}\ ,
\end{align}
\ees
The action of the dissipator thus becomes:
\beq
D({{\rho }_{G}})=g^2\sum_{\alpha \beta }{\sum_{\omega }{{{\gamma }_{\alpha \beta }}(\omega )\left( {{e}^{-\beta \omega }}{{A}_{\beta }}(\omega )A_{\alpha }^{\dagger }(\omega )-A_{\alpha }^{\dagger }(\omega ){{A}_{\beta }}(\omega ) \right){{\rho }_{G}}}}\ .
\eeq
Let us now separate the sum over $\o$ as $\sum_{\o<0} + (\o=0) + \sum_{\o>0}$. Recall that KMS result [Eq.~\eqref{eq:KMS}]: $\gamma_{\a\b}(-\omega) = e^{-\beta\omega}\gamma_{\b\a}(\omega)$.  
We know from Eq.~\eqref{eq:552} that $A_\a(0) = A_\a^\dgr(0)$, so that the $\o=0$ cancels since the remaining sum is over all $\a$ and $\b$, and by KMS, $\gamma_{\a\b}(0)=\g_{\b\a}(0)$. As for the sum over negative frequencies, using KMS and Eq.~\eqref{eq:552} again we have:
\bes
\begin{align}
\sum_{\o<0} &= \sum_{\o'=-\o>0} \g_{\a\b}(-\o')\left(e^{\b\o'}A_\b(-\o')A_\a^\dgr(-\o')-A_\a^\dgr(-\o')A_\b(-\o')\right)\r_G\\
&=\sum_{\o'>0} \g_{\b\a}(\o')e^{-\b\o'}\left(e^{\b\o'}A_\b^\dgr(\o')A_\a(\o')-A_\a(\o')A_\b^\dgr(\o')\right)\r_G\ ,
\end{align}
\ees
so that
\bes
\begin{align}
\sum_{\alpha \beta }\sum_{\o<0} &= \g_{\b\a}(\o)\left(A_\b^\dgr(\o)A_\a(\o)-e^{-\b\o}A_\a(\o)A_\b^\dgr(\o)\right)\r_G\\
&=-\sum_{\alpha \beta }\sum_{\o>0}\ ,
\end{align}
\ees
and hence $\sum_{\o<0}  + \sum_{\o>0}=0$.

So, the dissipator is also zero, and the Gibbs state is indeed stationary:
\beq
{{\dot{\rho }}_{G}}=0\ .
\eeq




\subsection{Return to equilibrium, quantum detailed balance, and ergodicity under the RWA-LE}
\label{sec:return}

A natural next question is under which conditions the Gibbs state is actually reached. To answer this we need to define the concept of ergodicity. A system is ergodic if it holds that for any arbitrary system operator $X$
\beq
\left[ X,{{A}_{\alpha }}(\omega ) \right]=\left[ X,A_{\alpha }^{\dagger }(\omega ) \right]=0,\quad \forall \alpha ,\omega
\label{eq:ergodi}
\eeq
if and only if $X$ is proportional to the identity operator.
 
It is possible to prove that if a system is ergodic and in addition $\mc{L} = -i[H,\cdot ] + \mc{L}_D$ satisfied the \emph{quantum detailed balance} condition with respect to the stationary state $\tilde{\r}$ (the state for which $\mc{L}\tilde{\r}=0$)
\begin{subequations}
\begin{align}
& [H,\tilde{\r}] = 0 \\
& (\mc{L}_D^\dagger A,B) = (A, \mc{L}_D^\dag B)
\end{align}
\end{subequations}
for $(A,B) \equiv \Tr[\tilde{\r}A^\dag B]$ and all $A,B\in \mathrm{domain}(\mc{L}^\dag)$,
then for \emph{any} initial state $\r(0)$ the stationary state is the Gibbs state. I.e., the Gibbs state is an attractor for the dynamics: $\rho (t) = {{e}^{\mathcal{L}t}}\rho (0)\xrightarrow{t\to \infty }{}{{\rho }_{G}}$. This is a fundamental result, as it tells us the conditions under which a system is guaranteed to become thermally equilibrated. The proof is given in Sec.~1.3.4 of Ref.~\cite{alicki_quantum_2007} (see also Ref.~\cite{Majewski:1998aa}).

However, not all systems are ergodic \cite{Alicki:88}. For example, consider a system of $N$ qubits coupled to a bath such that 
\beq
{{A}_{\a}}=\sum_{j=1}^{N}{\sigma_{j}^{\a}} , \quad \a\in\{x,y,z\}\ .
\label{eq:collect-dec}
\eeq
Clearly, all $A_\a$ are invariant under permutations, so that they commute with the elements of the permutation group. This means that Eq.~\eqref{eq:ergodi} is satisfied for operators $X$ that are not proportional to the identity (e.g., the SWAP operator between any pair of qubits), and hence such a system is not ergodic. Indeed, Eq.~\eqref{eq:collect-dec} describes ``collective decoherence", under which there exist subspaces that are invariant under the action of the $A_\a$ operators, and undergo unitary dynamics \cite{Zanardi:97c,Lidar:1998fk}. Initial states in such subspaces do not converge to the Gibbs state, and do not equilibrate.

More generally, if the system-bath interaction possesses some symmetry (e.g., a permutational symmetry as above), then ergodicity does not hold and the system need not equilibrate \cite{Lidar:2003fk}.

\section{Pauli Master Equation}
\label{sec:PauliME}

Sometimes we are particularly interested in finding out the evolution of just the populations (diagonal elements) in the energy eigenbasis. For example, this is the case in adiabatic quantum computing and quantum annealing, where the answer to a computation is encoded in the ground state~\cite{Albash-Lidar:RMP}. In other applications we are interested in finding out the Gibbs distribution $\r_G$ [Eq.~\eqref{eq:Gibbs_state}] 
in order to compute various thermodynamic averages $\ave{X} = \Tr(X \r_G)$, where $X$ could be any observable of interest; the Gibbs state is an example of a state that is diagonal in the energy eigenbasis, i.e., the eigenbasis $\{\ket{\e_a}\}$ of $H_S=\sum_a\epsilon_a|\epsilon_a\>\<\epsilon_a| = \sum_a \epsilon_\a \Pi_a$. 

Recall that the RWA-LE in the Schr\"{o}dinger picture is
\[
\frac{d\rho}{dt}=-i\left[H_S+H_{\mathrm{LS}},\rho\right]+g^2\sum_{\alpha\beta}\sum_{\omega}
\gamma_{\alpha\beta}(\omega)\left[A_\beta(\omega)\rho A_\alpha^\dag(\omega)-\frac{1}{2}\left\{A_\alpha^\dag(\omega) A_\beta(\omega),\rho\right\}\right]\ .
\]
The population in the $a$th energy eigenbasis state is:
\beq
p_a(t) = \bra{\e_a}\r(t)\ket{\e_a} = \r_{aa}(t) = \Tr[\Pi_a \r]\ .
\eeq
Our goal is to derive a master equation for the evolution of these populations, known as the \emph{Pauli master equation}. We will see that the populations in the energy eigenbasis are decoupled from the coherences (off diagonal elements) in the same basis. Consider then, the time-derivative of the populations, while using the fact that $H_S$ is time-independent (and hence so are its eigenvalues and eigenvectors):
\bes
\label{eq:pauli_master_1}
\begin{align}
\label{eq:pauli_master_1a}
\dot{p}_a &=\<\epsilon_a|\dot{\rho}|\epsilon_a\>\ = \Tr[\Pi_a \dot{\r}] \\
\label{eq:pauli_master_1b}
&=-i\<\e_a|[H_S,\rho]|\e_a\>-i\<\e_a|[H_{\mathrm{LS}},\rho]|\e_a\>\\
\label{eq:pauli_master_1c}
&+g^2\sum_{\alpha\beta}\sum_{\omega}
\gamma_{\alpha\beta}(\omega)\<\e_a|\left[A_\beta(\omega)\rho A_\alpha^\dag(\omega)-\frac{1}{2}\left\{A_\alpha^\dag(\omega) A_\beta(\omega),\rho\right\}\right]|\e_a\>\ .
\end{align}
\ees

The first term in Eq.~\eqref{eq:pauli_master_1b} is:
\beq
\<\e_a|[H_S,\rho]|\e_a\>=\<\e_a|H_S\rho|\e_a\>-\<\e_a|\rho H_S|\e_a\>=\e_a\<\e_a|\rho|\e_a\>-\e_a\<\e_a|\rho|\e_a\>=0\ .
\label{eq:619}
\eeq
As for $\<\e_a|[H_\text{LS},\rho]|\e_a\>$, recall that $[H_S,H_{\mathrm{LS}}]=0$, which means that  $H_S$ and $H_{\mathrm{LS}}$ share a common eigenbasis, i.e., the energy eigenbasis $\{\ket{\epsilon_a}\}$; hence $H_{\mathrm{LS}}$ is diagonal in the same basis and the same calculation as in Eq.~\eqref{eq:619} also implies that $\<\e_a|[H_{\mathrm{LS}},\rho]|\e_a\>=0$. Therefore there is no contribution from the unitary part to the evolution of the populations in the energy eigenbasis.

Now consider the dissipative part, i.e., line~\eqref{eq:pauli_master_1c}. Recall that
\beq
A_\beta(\omega)=\sum_{\epsilon_b-\epsilon_a=\omega}\ketb{\epsilon_a}{\e_a}{A_{\beta}}\ketb{\epsilon_b}{\e_b}=\sum_{b-a=\omega}|a\>A_{a b,\beta}\<b|\ , \qquad A_\alpha^\dag(\omega)=\sum_{b-a=\omega}|b\>
A_{ba,\alpha}\<a|
\eeq
where we again used the simplified notation $\epsilon_a \mapsto a$.
We have for the first term in line~\eqref{eq:pauli_master_1c}:
\bes
\label{eq:621}
\begin{align}
\bra{\e_a}A_\b(\o)\r A_\a^\dgr(\o)\ket{\e_a} &= \bra{a}\sum_{\o={b'}-{a'}}A_{a'b',\b}\ketb{a'}{b'}\r  \sum_{\o={b''}-{a''}}A_{b''a'',\a}\ketb{b''}{a''}a\> \\
&= \sum_{\substack{\o=b'-a\\ \o=b''-a}}A_{ab',\b}\r_{b'b''}A_{b''a,\a}\\
&=\sum_{\o=b'-a}A_{ab',\b}p_{b'}A_{b'a,\a} = \sum_{\o=a'-a}A_{aa',\b}p_{a'}A_{a'a,\a} \ ,
\end{align}
\ees
where to go the second line we used $\bk{a}{a'}=\d_{aa'}$ and $\bk{a''}{a}=\d_{a''a}$, and to go to the third line we used the 
fact that $b'$ must equal $b''$ due to the summation constraints.

Similarly,
\bes
\label{eq:622}
\begin{align}
\bra{a}A_\a^\dgr(\o)A_\b(\o)\r\ket{a} &= \sum_{\o=b'-a'}\bra{a}A_{b'a',\a }\ketb{b'}{a'}\sum_{\o=b''-a''}A_{a''b'',\b}\ketb{a''}{b''}\r\ket{a}\\
&= \sum_{\substack{\o=a-a'\\ \o=b''-a'}}A_{aa',\a}A_{a'b'',\b}\r_{b''a}\\
&=\sum_{\o=a-a'}A_{aa',\a}A_{a'a,\b}p_a\ ,
\label{eq:622c}
\end{align}
\ees
and
\bes
\label{eq:623}
\begin{align}
\bra{a}\r A_\a^\dgr(\o)A_\b(\o)\ket{a} &= \sum_{\o=b'-a'}\bra{a}\r A_{b'a',\a}\ketb{b'}{a'}\sum_{\o=b''-a''}A_{a''b'',\b}\ketb{a''}{b''}a\rangle\\
&= \sum_{\substack{\o=b'-a'\\ \o=a-a'}}\r_{ab'}A_{b'a',\a}A_{a'a,\b}\\
&=\sum_{\o=a-a'}p_aA_{aa',\a}A_{a'a,\b}\ ,
\end{align}
\ees
which is the same result as in Eq.~\eqref{eq:622}.

Combining Eqs.~\eqref{eq:621}-\eqref{eq:623}, we have:
\beq
\dot{p}_a = \sum_{\a\b} \left(\sum_{\o=a'-a} A_{a'a,\a}A_{aa',\b}p_{a'} -  \sum_{\o=a-a'} A_{aa',\a}A_{a'a,\b}p_a\right)\g_{\a\b}(\o) \ .
\eeq
Since the index $a$ is fixed, the sum over $\o$ really only involves varying $a'$. Thus:
\beq
\label{eq:pauli_master_2}
\dot{p}_a=\sum_{\alpha\beta}\sum_{a^\prime}\gamma_{\alpha\beta}(a^\prime-a)A_{a^\prime a,\a}A_{a a^\prime ,\beta}p_{a^\prime}-\gamma_{\alpha\beta}(a-a^\prime)A_{aa^\prime,\alpha}A_{a^\prime a,\beta}p_a\ .
\eeq
Now define a \emph{transition matrix} $W$ via
\beq
W(a|a^\prime)\equiv\sum_{\alpha\beta}\gamma_{\alpha\beta}(a^\prime-a)A_{a^\prime a,\alpha}A_{aa^\prime,\beta}\ .
\eeq
Note that $W(a|a^\prime)\geq 0$.
To prove this, let $u$ be the unitary matrix that diagonalizes $\gamma$: $\g_{\a\b} = \sum_{\a'}u_{\a\a'}\g_{\a'} u^*_{\b\a'}$. Then:
\bes
\begin{align}
W(a|a^\prime) &= \sum_{\a'\a\b} u_{\a\a'}\g_{\a'}(a'-a) u^*_{\b\a'}A_{a'a,\a}A_{aa',\b} = \sum_{\a'} \g_{\a'}(a'-a) \left(\sum_{\a} u_{\a\a'}A_{a'a,\a}\right) \left(\sum_{\b} u^*_{\b\a'}A_{aa',\b}\right) \\
&= \sum_{\a'} \g_{\a'}(a'-a) \abs{\tilde{A}_{a'a,\a'}}^2 \geq 0 \ ,
\end{align}
\ees
where $\tilde{A}_{a'a,\a'} =  \sum_{\a} u_{\a\a'}A_{a'a,\a}$, and we used the Hermiticity of $A_{\b}$ to write $A_{aa',\b} = A^*_{a'a,\b}$.
Eq.~\eqref{eq:pauli_master_2} can thus be simplified as:
\beq
\dot{p}_a=\sum_{a^\prime}W(a|a^\prime)p_{a^\prime}-W(a^\prime|a)p_a \ .
\label{eq:628}
\eeq
This represents a closed set of rate equations for the populations $\{p_a\}$.

If we assume that the KMS condition $\gamma_{\alpha\beta}(-\omega)=e^{-\beta\omega}\gamma_{\beta\alpha}(\omega)$ (for $\o>0$) holds, then this allows us to write, for $a>a'$:
\beq
\gamma_{\alpha\beta}(a^\prime-a)=e^{-\beta(a-a')}\gamma_{\beta\alpha}(a-a^\prime).
\eeq
Then $W(a|a^\prime)$ can be rewritten as:
\beq
W(a|a^\prime)=\sum_{\alpha\beta}e^{-\beta(a-a')}\gamma_{\beta\alpha}(a-a^\prime)A_{a^\prime a,\alpha}A_{aa^\prime,\beta} = e^{-\beta(a-a')}{\sum_{\alpha\beta}\gamma_{\alpha\beta}(a-a^\prime)A_{a^\prime a,\beta}A_{aa^\prime,\alpha}} = e^{-\beta(a-a')}W(a^\prime|a) \ .
\eeq
This is the \emph{detailed balance condition}:
\beq
\frac{``\uparrow "}{``\downarrow "} = \frac{W(a|a^\prime)}{W(a^\prime|a)}= e^{-\beta(a-a^\prime)}\ .
\eeq
It says that the rate for an ``up" transition, from the low energy state $\ket{a'}$ to the high energy state $\ket{a}$, is exponentially less likely than the reverse, ``down" transition, with the exponent given by the energy difference in units of the bath temperature. This is an extremely important result, since it establishes rigorously the intuition that at very low temperatures (relative to the smallest energy gap) systems tend to relax towards their ground states. This is a special case of the quantum detailed balance condition we mentioned in Sec.~\ref{sec:return}.

Finally, we can also reestablish that the Gibbs state is the stationary state (recall that we showed this in Sec.~\ref{sec:Gibbs-stat}). For a stationary state $\dot{p}_a=0$. It follows from Eq.~\eqref{eq:628} that in this case:
\bes
\begin{align}
\frac{W(a|a^\prime)}{W(a^\prime|a)}&=\frac{p_a}{p_a^\prime}=e^{-\beta({a}-a')}=\frac{e^{-\beta a}}{e^{-\beta {a^\prime}}}\\
&\Rightarrow p_a=\frac{e^{-\beta a}}{Z};\quad Z=\sum_a e^{-\beta a}\ ,
\end{align}
\ees
which is the Gibbs distribution.

\section{Lindblad Equation in the Singular Coupling Limit (SCL)}

All our derivations of the LE so far have assumed the weak coupling limit of system-bath coupling. Somewhat surprisingly, the opposite limit of strong coupling also allows us to derive the Lindblad equation, while avoiding the use of the RWA. 

\subsection{Derivation}

Assume that the Hamiltonian takes the form
\begin{equation}
 H=H_{S}+{\frac{1}{\epsilon}}H_{SB}+{\frac{1}{\epsilon^2}}H_{B}\ ,
\end{equation}
where $H_{SB} = g \sum_\a A_\a\ox B_\a$ as in Eq.~\eqref{eq:gHSB},
so that the $A_\a, B_\a$ operators are dimensionless.
Since we are interested in the limit of small $ \epsilon $, this is called the \emph{singular coupling limit} (SCL). In this limit the bath Hamiltonian dominates over the system and system-bath Hamiltonians. 

Note that in order for the Gibbs state of the bath to remain invariant ($\r_B = e^{-\b H_B}/Z$), the bath must be in
thermal equilibrium with respect to $H_B/\e^2$ at the temperature $T/\e^2\to\infty$. Thus, we can also interpret the SCL as a high temperature limit. For a more detailed discussion see
Ref.~\cite{PhysRevA.73.052311}.

Our starting point is the interaction picture Born approximation [Eq.~\eqref{integEq2}], which we write here with $\epsilon$ included:
\begin{equation}
\frac{d\tilde{\rho}}{dt}=-g^2\frac{1}{\epsilon^2} \sum_{\alpha\beta}\int_0^t d\tau\bigl\{\mc{B}_{\a\b}(\tau)\,[A_{\alpha}(t),A_{\beta}(t-\tau)\tilde{\rho}(t-\tau)]+\text{h.c.}\bigr\}\ .
\end{equation}

Let us transform this to the Schr\"{o}dinger picture via Eq.~\eqref{eq:513}:
\begin{align}
\label{eq:599}
&\frac{d\r}{dt}=-i[H_S,\r(t)]+g^2\sum_{\alpha\beta}\frac{1}{\epsilon^2} {\int_0^t}{d\tau}U_{S}(t)\left(\left[A_{\beta}(t-\tau)U_S^\dgr(t-\tau)\r(t-\tau)U_S(t-\tau)A_\a(t)-\right.\right.\notag \\
&\quad \left.\left. A_\a(t)A_{\beta}(t-\tau)U_S^\dgr(t-\tau)\r(t-\tau)U_S(t-\tau)\right]\mc{B}_{\a\b}(\tau)+\text{h.c.}\right)U_S^\dgr(t)\ .
\end{align}
We can perform a change of variables to $ \tau={\epsilon^{2}}{{\tau}^\prime} $, and take the limit $ \epsilon \rightarrow 0$, so that $\tau\to 0$. Then, recalling Eq.~\eqref{eq:468A}, the various terms in Eq.~\eqref{eq:599} transform as follows:
\bes
\label{eq:600}
\begin{align}
& U_S(t)A_{\beta}(t-\tau)U_S^\dgr(t-\tau)\r(t-\tau)U_S(t-\tau)A_\a(t)U_S^\dgr(t) = U^\dgr_S(-\tau)A_{\beta}\r(t-\tau)U^\dgr_S(\tau)A_\a  \notag \\
&\qquad \to A_{\beta}\r(t)A_\a  = A_\b \r(t) A_\a^\dgr \\
& U_{S}(t)A_\a(t)A_{\beta}(t-\tau)U_S^\dgr(t-\tau)\r(t-\tau)U_S(t-\tau)U_S^\dgr(t) = A_\a U_S^\dgr(-\tau)A_{\beta}\r(t-\tau)U_S^\dgr(\tau) \notag \\
&\qquad \to A_\a A_{\beta}\r(t) = A_\a^\dgr A_\b \r(t) \\
& \quad \mc{B}_{\a\b}(\tau) = \Tr\left( e^{i\e^2\tau'H_B/\e^2}B_\a e^{-i\e^2\tau'H_B/\e^2}B_\b \r_B \right) = \Tr\left(U_B^\dgr(\tau') B_\a U_B(\tau') B_\b\r_B\right) = \mc{B}_{\a\b}(\tau')\\
&\quad \frac{1}{\epsilon^2}\int_0^t d\tau= \int_0^{t\epsilon^{-2}}d\tau' \to \int_0^\infty d\tau' \ .
\end{align}
\ees
Thus the $\e\to\infty$ strong coupling and bath limit,  is essentially a Markovian limit, as it allows us to extend the integration limit to $\infty$ and make $\r $ time-local. It also removes the time dependence from the $A_\a$ system operators.

Applying the transformations in Eq.~\eqref{eq:600} to Eq.~\eqref{eq:599} gives:
\begin{align}
\label{eq:601}
\frac{d\r}{dt}=-i[H_S,\r(t)]+g^2\sum_{\a\b}(A_\b\r(t)A_\a^\dgr-A_\a^\dgr A_\b\r(t))\int_0^\infty d\tau \mc{B}_{\a\b}(\tau )+\text{h.c.}\ .
\end{align}
Now recall Eqs.~\eqref{Gamma} and \eqref{gammaS}, which tell us that 
\beq
\int_0^\infty d\tau \mc{B}_{\alpha\beta}(\tau) = \Gamma_{\alpha\beta}(0)= \frac{1}{2}\gamma_{\alpha\beta}(0)+iS_{\alpha\beta}(0)\ .
\eeq
Thus
\bes
\label{eqt:SCL}
\begin{align}
\frac{d\r}{dt}&=-i[H_S+H_\text{LS}, \rho(t)]+g^2\sum_{\a\b}{\gamma_{\alpha\beta}}(0)\left(A_\b{\rho(t)}A_\a^\dgr-\frac{1}{2} \{A_\a^\dgr{A_\beta},{\rho}(t)\}\right)\\
H_{\text{LS}} &=\sum_{\a\b} S_{\a\b}(0)A_\a^\dgr A_\b
 \ , \qquad \g_{\a\b}(0) = \int_{-\infty}^\infty d\tau \mc{B}_{\alpha\beta}(\tau)\ .
\end{align}
\ees
Note that the SCL keeps only the $\o=0$ component out of all the Bohr frequencies, so it is clearly a more ``extreme" limit than the WCL. We can understand this as a consequence of the fact that the SCL is designed to accelerate the internal evolution of the bath by rescaling the bath Hamiltonian via $H_B \mapsto H_B/\e^2$; this means that all system frequencies are effectively zero relative to the very high effective bath evolution frequency, and only the static component $\o=0$ survives.



\subsection{Examples contrasting the WCL and SCL}

Let us consider a single qubit. 

\subsubsection{Phase damping when $[H_S,H_{SB}]= 0$}
We assume that
\begin{align} 
\label{eq:Hqubit}
H_S = -\frac{1}{2}\omega_z \sigma^z \ , \qquad
H_{SB} =  g\sigma^z \otimes B .
\end{align}
For the interaction Hamiltonian in Eq.~\eqref{eq:Hqubit}, there is only a single system operator $A_z  = \sigma^z = \ketb{0}{0}-\ketb{1}{1}$.  The eigenstates are $\ket{\eps_0}=\ket{0}$ and $\ket{\eps_1}=\ket{1}$. Considering the RWA-LE (the weak coupling limit case) Eq.~\eqref{Lindbladw} 
and $\bra{\eps_a} A_z \ket{\eps_b} \propto \delta_{ab}$, there is only a single Lindblad operator that is non-zero:
\beq 
\label{eqt:LindbladOp_Z}
A_z({0}) = \sigma^z\ ,
\eeq
as given by Eq.~\eqref{eq:A_oma}.  This follows since $[H_S,H_{SB}]=0$.  Therefore, the RWA-LE [Eq.~\eqref{LindbladSch}] takes the simple form
\begin{align}
\frac{d}{d t} \rho(t) &= -i \left[ H_S, \rho(t) \right] + g^2\gamma(0) \left( \sigma^z \rho(t) (\sigma^z)^\dgr - \frac{1}{2} \left\{ (\sigma^z)^{\dagger} \sigma^z , \rho(t) \right\} \right) \ ,
\label{eq:624}
\end{align}
where we have also used the fact that $H_{\textrm{LS}} \propto I$. This form is the same as what is predicted in the SCL, since only the $\o=0$ component appears. We have encountered this equation several times before [e.g., Eq.~\eqref{eq:267}]. After expanding $\rho(t) = \sum_{i,j\in\{0,1\}}\rho_{ij}\ket{i}\bra{j}$, and taking matrix elements in the computational basis (which here is equivalent to the energy eigenbasis) we obtain:
\bes
\begin{align}
\rho_{0 0}(t) & = \rho_{0 0}(0) = 1-  \rho_{1 1}(t)\ , \\
\rho_{0 1}(t) & = \exp(- t/T_2^{(c)} + i \omega_z t) \rho_{0 1}(0) = \rho_{1 0}^*(t) \ ,
\end{align}
\ees
%
%
%
where 
\beq 
\label{eqt:T2Z}
T_2^{(c)} = \frac{1}{2 g^2\gamma(0)} \ ,
\eeq
where the `$c$' superscript denotes the computational basis (we shall shortly see a second $T_2$ associated with the energy eigenbasis). This is the familiar phase damping channel, where only the off-diagonals elements (transverse magnetization) decay with a characteristic timescale $T_2^{(c)}$. The stronger the coupling to the bath $g$, the shorter the qubit coherence time.  Note that the qubit energy gap $\omega_z$ plays no role in the result for $T_2^{(c)}$, and $T_2^{(c)}$ here is entirely determined by the spectrum of the bath correlation function at zero frequency.  In this example there is no thermal relaxation (the $T_1$ time is infinite), since the population of the energy states remains fixed, as a consequence of $[H_S,H_{SB}]=0$.

\subsubsection{Phase damping when $[H_S,H_{SB}]\neq 0$}
Let us now replace the system Hamiltonian so that $[H_S,H_{SB}]\neq 0$. Specifically, consider 
\begin{align} 
\label{eq:Hqubit2}
H_S = -\frac{1}{2}\omega_x \sigma^x\ , \qquad 
H_{SB} =  g\sigma^z \otimes B\ .
\end{align}
We shall see that there is a sharp contrast between the WCL and SCL, with the WCL resulting in decoherence in the energy eigenbasis, while the SCL results in decoherence in the computational basis, just as in the previous subsection, when $H_S$ and $H_{SB}$ were commuting.

\paragraph{WCL}
\label{sec:WCL1}
The energy eigenstates of $H_S$ are $\ket{\eps_0} = \ket{+}$ with eigenvalue $-\frac{1}{2} \omega_x$(ground state) and $\ket{\eps_1}=\ket{-}$ with eigenvalues $\frac{1}{2} \omega_x$ (excited state), where $\ket{\pm} = \frac{1}{\sqrt{2}} \left( \ket{0} \pm \ket{1} \right)$. Therefore the possible Bohr frequencies are $\o\in\{0,\pm\o_x\}$. Since $\sigma^z \ket{\pm} =  \ket{\mp}$, we find $A_{z}(0)=0$, and the non-zero Lindblad operators are:
\beq \label{eqt:LindbladOp_X}
A_{z}(\omega_x) = \ketb{+}{+}\s^z\ketb{-}{-} = \ketb{+}{-} \ , \quad A_{z}(-\omega_x) = \ketb{-}{-}\s^z\ketb{+}{+} = \ketb{-}{+} \ .
\eeq
Note that we now have a non-trivial Lamb shift term:
\beq
H_{\textrm{LS}} = S(\omega_x) \ketb{-}{-} + S(-\omega_x) \ketb{+}{+} \ .
\eeq
Now we need to compute the terms in the RWA-LE [Eq.~\eqref{LindbladSch}]. It is most convenient to do so in the \emph{energy eigenbasis}, i.e., the basis that diagonalizes $H_S$, namely the $\{ \ket{\pm}\}$ basis we used above.
Note that:
\bes
\label{eq:630}
\begin{align}
\label{eq:630a}
H_S+H_{\textrm{LS}} &= \Omega_+\ketb{+}{+} + \Omega_-\ketb{-}{-} \ , \qquad \Omega_{\pm} = \frac{1}{2}\o_x+S(\pm \o_x)\\
\label{eq:630b}
g^2\sum_\omega\sum_{\alpha\beta}\gamma_{\alpha\beta}(\omega) \cdots  &=g^2\left[\g(\o_x) \left(  \ketb{+}{-}\r\ketb{-}{+}-\frac{1}{2}(\ketb{-}{-}\r+\r\ketb{-}{-}) \right)
 + \g(-\o_x) \left( \ketb{-}{+}\r\ketb{+}{-}-\frac{1}{2}(\ketb{+}{+}\r+\r\ketb{+}{+})\right)\right]\ .
\end{align}
\ees
Writing $\rho(t) = \sum_{i,j\in\{+,-\}}\rho_{ij}\ketb{i}{j}$, and taking matrix elements of 
Eq.~\eqref{eq:630}, we find:
\beq
\bra{-}\dot{\r}\ket{-} = \dot{\r}_{--} = -i\bra{-}(H_S+H_{\textrm{LS}})\r-\r(H_S+H_{\textrm{LS}})\ket{-}
-g^2\g(\o_x)\r_{--}+g^2\g(-\o_x)\r_{++}\ ,
\eeq
and the first (Hamiltonian) term is easily seen to vanish. Also, note that $ \Tr[\r(t)]=\rho_{++}(t)+ \rho_{--}=1$ implies that $\dot{\r}_{--} = -\dot{\r}_{++}$. After a similar calculation for the off-diagonal components, we find that the Lindblad equation for the density matrix components is:
\bes 
\begin{align}
\label{eq:632a}
-\frac{d}{dt} \rho_{++}&= \frac{d}{dt} \rho_{--}  = - g^2\gamma(\omega_x) \rho_{--}(t) + g^2\gamma(-\omega_x) \rho_{++}(t)   \\
\label{eq:632b}
\frac{d}{dt} \rho_{+-}^*(t)  = \frac{d}{dt} \rho_{-+}(t)  & = \Omega \rho_{-+}(t)\ , \qquad \Omega\equiv   - i \left[ \O(-\omega_x) - \O(\o_x) \right]  
- \frac{1}{2} g^2\left[\gamma(\omega_x) +\g(-\o_x) \right]  \ .
\end{align}
\ees
The solution for the off-diagonal elements [Eq.~\eqref{eq:632b}] is immediate: $ \rho_{-+}(t)  =  \rho_{ - +}(0) e^{-i\Omega t}$, i.e.:
\beq
\label{eq:633}
\rho^*_{+-}(t) = \rho_{-+}(t)  =  \rho_{ - +}(0) e^{-i \omega'_x t} e^{-t / T_2^{(e)}} \ ,
\eeq
where
\beq 
\label{eqt:T2X}
T_2^{(e)} = \frac{2}{g^2\gamma(\omega_x) \left( 1 + e^{-\b \omega_x} \right)}\ , \quad \o_x'=\o_x+ S(\omega_x) - S(-\omega_x)\ ,
\eeq
where the `$e$' superscript denotes the energy eigenbasis (as opposed to the computational basis) , and where we used the KMS condition [Eq.~\eqref{eq:KMS}] to write $\gamma(\omega_x) +\g(-\o_x) = \gamma(\omega_x)(1+e^{-\b\o_x})$.  Contrast this result with Eq.~\eqref{eqt:T2Z}, where the dephasing rate depended only on $\g(0)$ and did not exhibit a temperature dependence.

To solve for the populations, let us substitute $\rho_{++}=1- \rho_{--}$ into Eq.~\eqref{eq:632a}, so that we can write $\dot{\r}_{--} = a-b\r_{--}$, where $a=g^2\g(-\o_x)$ and $b=g^2[\g(-\o_x)+\g(\o_x)]$. As a solution let us try the ansatz $\r_{--}(t) = c e^{-t/T_1^{(e)}}+d$, so that the initial condition yields $c=\r_{--}(0)-d$. Then
\beq
\dot{\r}_{--} = -\frac{c}{T_1^{(e)}} e^{-t/T_1^{(e)}}= a-b\left(c e^{-t/T_1^{(e)}}+d\right) = a-bc e^{-t/T_1^{(e)}}-bd\ ,
\eeq
which tells us that $d=a/b = \g(-\o_x)/[\g(-\o_x)+\g(\o_x)]$ and $T_1^{(e)}=1/b$, i.e.:
\beq
T_1^{(e)} = \frac{1}{2}T_2^{(e)} \ .
\eeq
Moreover, recall that the Gibbs state is 
\beq
\r_G = \frac{1}{Z}e^{-\b H_S} = \frac{1}{Z}e^{\frac{1}{2}\b \o_x \s^x} = p_G(-) \ketb{+}{+} + p_G(+) \ketb{-}{-}\ ,
\eeq
where
\beq
p_G(\pm) = \frac{e^{\pm\b \omega_x/ 2}}{Z}\ , \qquad Z = \Tr(\r_G) = p_G(-)+p_G(+) = 2\cosh(\beta \omega_x/2)\ .
\eeq
Using this and the KMS condition, we have
\beq
d = \frac{\g(-\o_x)}{\g(-\o_x)(1+e^{\b\o_x})} = P_G(-)\ .
\eeq
Using our ansatz we thus find for the populations, finally: 
\beq
\label{eqt:decoherence2}
1-\rho_{++}(t) = \rho_{--}(t)  = p_G(-) + \left[ \rho_{--}(0) -  p_G(-)  \right] e^{- t / T_1^{(e)}} \ .
\eeq

We note several important facts about these results:
\begin{itemize}
\item The decoherence occurs in the energy eigenbasis, i.e., the off-diagonal components \emph{in the energy eigenbasis} (not in the computational basis) decay exponentially to zero with a timescale determined by $T_2^{(e)}$.
\item The entire contribution of the Lamb shift is in shifting the rotation rate of the off-diagonal elements from $\o_x$ to $\omega_x+ S(\omega_x) - S(-\omega_x)$ [Eq.~\eqref{eqt:T2X}].
\item The populations ($\rho_{++},\rho_{--}$) approach the Gibbs state associated with the Hamiltonian $H_S$ within a timescale determined by $T_1^{(e)}$ [Eq.~\eqref{eqt:decoherence2}]. In particular, for the ground state population: $\rho_{++} \to p_G(+) = \frac{e^{\b \omega_x/ 2}}{Z}$.
\item The two timescales ($T_1^{(e)},T_2^{(e)}$) are strictly related (relaxation is twice as fast as dephasing) and have a non-trivial dependence on the energy gap $\omega_x$.
\item Even in the zero temperature limit ($\b \to \infty$), the dephasing and relaxation times can be non-vanishing: $T_1^{(e)} =T_2^{(e)}/2= \frac{1}{g^2\gamma(\omega_x)} >0$.
\end{itemize}

\paragraph{SCL}
\label{sec:SCL1}
Let us contrast this with what happens in the SCL case, Eq.~\eqref{eqt:SCL}. This simply becomes Eq.~\eqref{eq:624}, with $H_S = -\frac{1}{2}\o_x \s^x$, i.e.: 
\begin{align}
\dot{\rho} &= i\frac{\o_x}{2}\left[ \s^x, \r \right] + g^2\gamma(0) \left( \sigma^z \rho \sigma^z - \rho  \right) \ ,
\label{eq:641}
\end{align}
In this case the evolution of the density matrix elements is most conveniently solved for in the computational basis. Taking matrix elements  in this basis yields:
\bes
\begin{align}
\frac{d}{dt} \rho_{0 0} & = -i \frac{1}{2} \omega_x \left(  \rho_{1 0} - \rho_{0 1} \right) \ , \\
\frac{d}{dt} \rho_{1 1} & = -i \frac{1}{2} \omega_x \left(  \rho_{0 1} -  \rho_{1 0} \right) \ , \\
\frac{d}{dt} \rho_{0 1} &=  i \frac{1}{2} \omega_x \left( \rho_{1 1} - \rho_{0 0} \right) - 2 g^2 \gamma(0) \rho_{0 1} \ , \\
\frac{d}{dt} \rho_{1 0} &=  i \frac{1}{2} \omega_x \left(\rho_{0 0} -  \rho_{1 1}  \right) - 2 g^2 \gamma(0) \rho_{1 0 } \ .
\end{align}
\ees
This set of equations can be solved analytically for arbitrary initial conditions, but for brevity, let us consider the case where the density matrix is initially in a uniform computational basis superposition (the ground state of the previous WCL case), i.e., $\rho(0) = \ketb{+}{+}$.  The solution is then given by:
\beq \label{eqt:SCL_sol}
\rho_{0 0} = \rho_{1 1}= \frac{1}{2} \ , \quad \rho_{0 1} = \rho_{1 0} =  \frac{1}{2} e^{- t/T_2^{(c)} } \ .
\eeq
In this case, the off-diagonal elements \emph{in the computational basis} decay exponentially with a timescale determined by $T_2^{(c)}$ [Eq.~\eqref{eqt:T2Z}], so we have decoherence in the computational basis regardless of the fact that the system Hamiltonian does not commute with $H_{SB}$. The predictions made under the WCL and SCL assumptions are thus starkly different.

\begin{figure}[t]
	\includegraphics[width=0.4\linewidth]{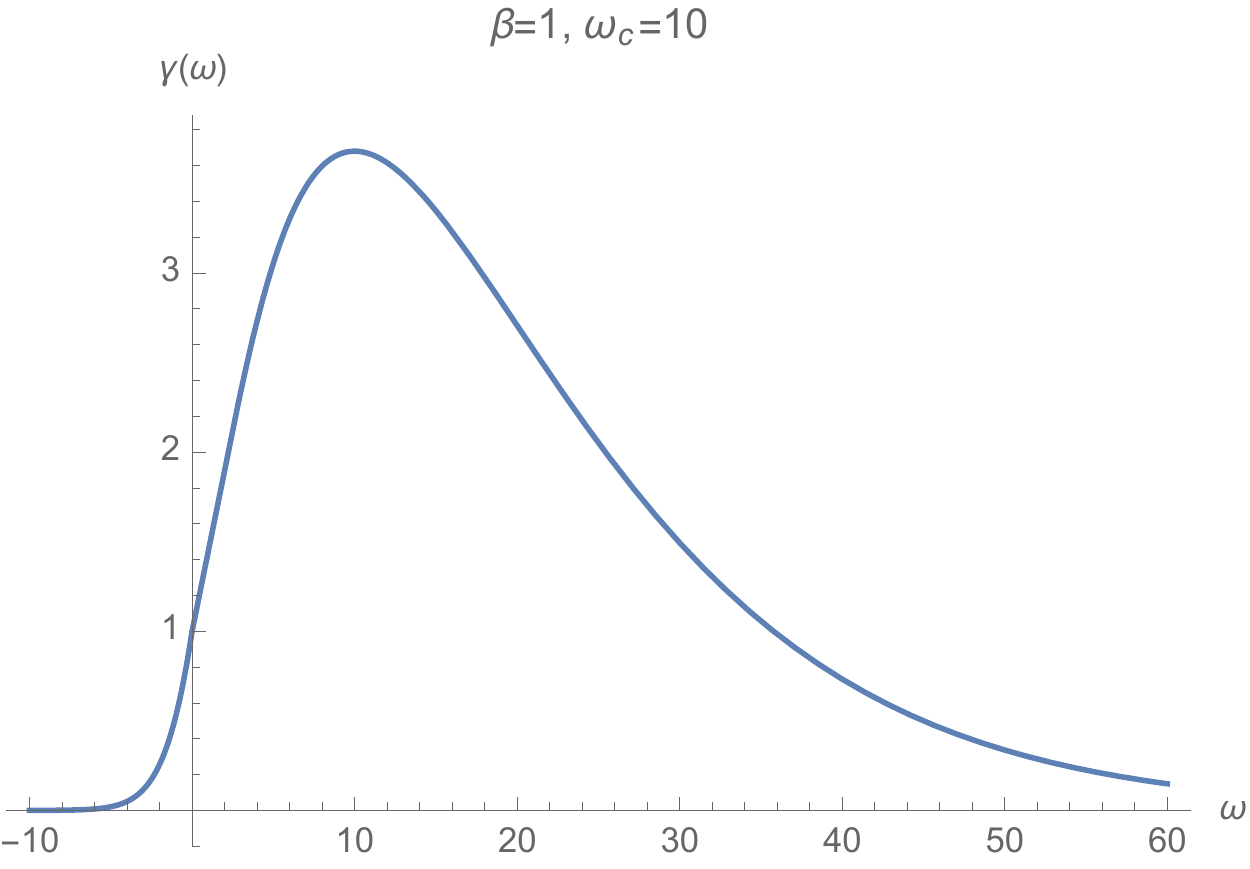}
	\includegraphics[width=0.4\linewidth]{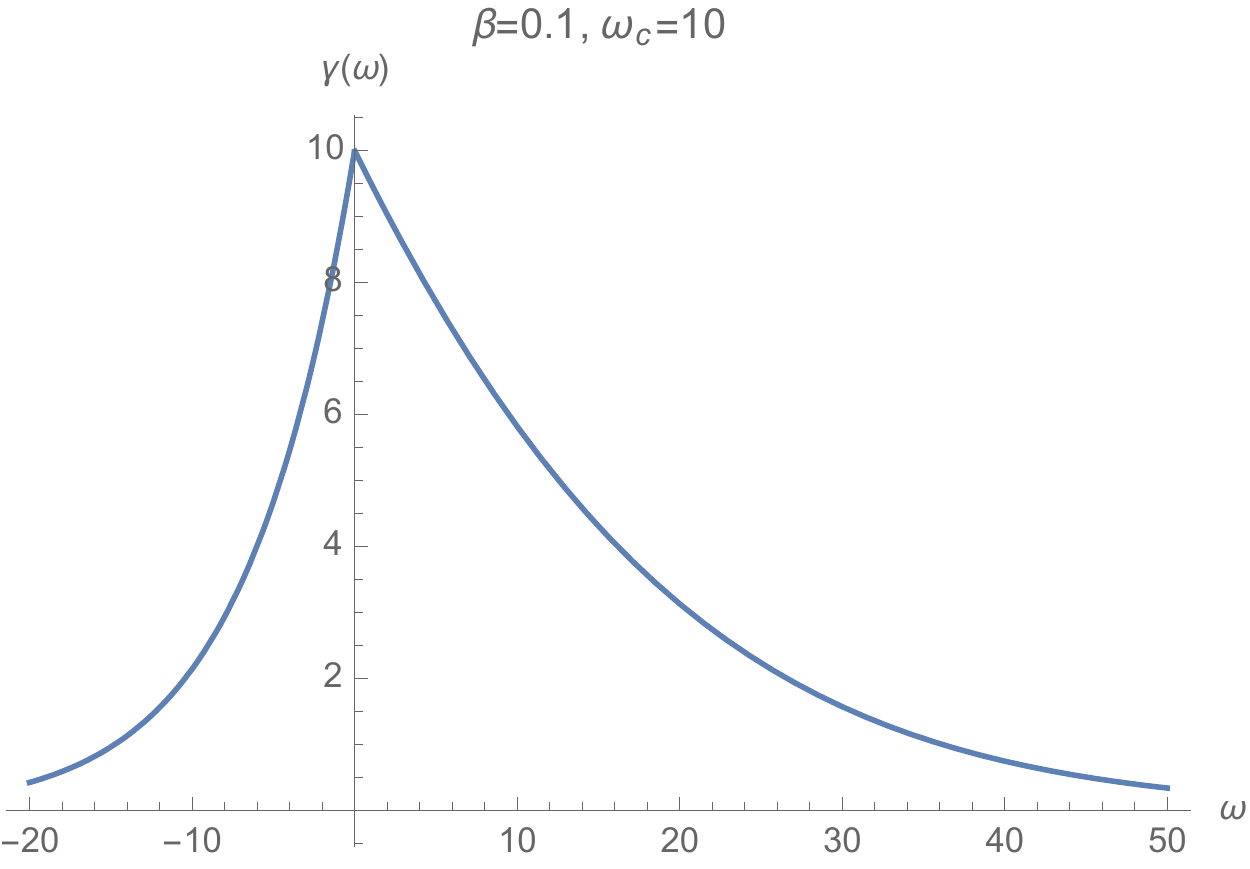}
	\caption{The Lindblad rate $\g(\o)$ for an Ohmic spectral density [Eq.~\eqref{eq:618c}], for $\eta=1/(2\pi),\o_c=20$, and low temperature $\b=10$ (left) or high temperature $\b=0.1$ (right). It can be checked numerically that the peak is always at $\o\approx\o_c$ for sufficiently large $\b$, or at $\o=0$ for sufficiently small $\b$. Note that $\g(0)=1/\b$. }
	\label{fig:Ohmic}
\end{figure}

\paragraph{Results for a bosonic bath}
So far we didn't specify the bath, and hence $\g(\o)$ was left unspecified as well. Let us now assume 
that the bath is bosonic:
\beq
H_B = \sum_k\o_{k} b^\dgr_{k} b_{k}\ ,
\eeq 
where $b_{k}$ is the annihilation operator associated with bosonic mode $k$, and the system-bath interaction is
\beq
H_{SB} = g A\ox B\ , \qquad A = \s^z \ , \quad B = \sum_{k} (g_{k}/g) (b_{k}+b_{k}^\dgr)\ .
\eeq

There is only a single bath correlation function, because there is only a single bath operator $B$. For a bath in a Gibbs state at inverse temperature $\b$ it can be shown that the bath correlation function in this case is~\cite[Appendix H]{ABLZ:12-SI}:
\beq
\langle B(t)B \rangle_B = \sum_{k} \frac{(g_k/g)^2}{1-e^{-\beta\omega_k}} \left(e^{-i \omega_k t} + e^{i\omega_k t-\beta \omega_k}\right) \ .
\label{eq:646}
\eeq
Let us introduce a spectral density $J(\omega) = \sum_k (g_k/g)^2 \d(\o-\o_k)$ via
\beq
\sum_{k} (g_k/g)^2 \mapsto \int_0^\infty d\omega J(\omega)\ ,
\eeq
and let us further assume that it is Ohmic:
\beq
J(\omega) = \eta \omega e^{-\omega/\omega_c}\ ,
\eeq
where $\omega_c$ is a cut-off frequency and $\eta$ is a dimensionless parameter.\footnote{If $J(\omega) = \eta \omega^{\zeta} e^{-\omega/\omega_c}$ then the $\zeta>1$ case is called {super-Ohmic}, and the $0<\zeta<1$ case is called {sub-Ohmic}.}

With this model of the bath spectral density function, we can compute the rate $\gamma(\omega)$ as the Fourier transform of the bath correlation function,
\bes
\begin{align}
\gamma(\omega) &= 
\int_{-\infty}^\infty dt e^{i\omega t}\langle B(t)B(0) \rangle = \int_{-\infty}^\infty dt e^{i\omega t} \int_0^\infty d\omega^\prime \frac{J(\omega^\prime)}{1-e^{-\beta \omega^\prime}} \left(e^{-i\omega^\prime t}+e^{i\omega^\prime t-\beta \omega^\prime}\right) \\
&= \frac{2\pi\eta \abs{\omega} e^{-\abs{\omega}/\omega_c}}{1-e^{\beta\abs{\omega}}}\left(\Theta(\omega)+e^{-\beta\abs{\omega}}\Theta(-\omega)\right) \\
\label{eq:618c}
&= 
\frac{2\pi\eta \omega e^{-\abs{\omega}/\omega_c}}{1-e^{-\beta\omega}}\ .
\end{align}
\ees
where $\Theta(x)$ is the Heaviside step function ($0$ if $x<0$ or $1$ if $x>0$).
Note that the KMS condition is satisfied. The result is shown in Fig.~\ref{fig:Ohmic}.
\beq
\gamma(-\omega) =  \frac{2\pi\eta (-\omega) e^{-\abs{\omega}/\omega_c}}{e^{\beta \omega} (e^{-\beta \omega}-1)} = e^{-\beta\omega}\gamma(\omega)\ .
\eeq
In the limit of large $\b\o$ we can neglect $e^{-\beta\omega}$ in the denominator; differentiating we then get $2\pi\eta e^{-\abs{\omega}/\omega_c}(1-\o/\o_c)$, so that the maximum is at $\o=\o_c$.
Also note that
\beq
\lim_{\omega\to 0} \gamma(\omega) = \frac{2\pi\eta}{\beta} = 2\pi\eta k_B T\ ,
\eeq
which tells that the transition rate in the limit of small gaps is linear in the temperature. This means that the SCL result for the dephasing rate becomes 
\beq
1/T_2^{(c)} = 2 g^2\gamma(0) = 4\pi g^2\eta k_B T\ ,
\eeq
meaning that the dephasing rate increasing in proportion to the temperature and the square of the coupling strength.

For the WCL, recall that we found that the dephasing and relaxation rates in the energy eigenbasis are $1/T_2^{(e)} = 1/(2T_1^{(e)}) = \frac{1}{2}[\gamma(\omega_x) + \gamma(-\omega_x)]$ [Eq.~\eqref{eqt:T2X}]. Considering Fig.~\ref{fig:Ohmic}, we see that $\gamma(-\omega_x)\ll\gamma(\omega_x)$, and that both rates are highly suppressed when $\o_x \gg \o_c$. For large $\b$ they are maximized when $\o_x\approx\o_c$ and become small for $\o_x < \o_c$, but are lower-bounded by $\g(0)=1/\b$.

\subsection{Example: collective \textit{vs} independent phase damping}
\label{sec:collec-indep-PD}

To close our discussion of the RWA-LE, let us revisit the phase damping model we considered in Sec.~\ref{sec:spin-boson-1q}, but this time for $n$ qubits. Thus the system Hamiltonian is
\begin{equation}
H_S = \sum_{\a=1}^n \varepsilon_\a Z_\a\ .
\label{eq:612}
\end{equation}
The eigenstates $\{\ket{\e_a}\}_{a=0}^{2^n-1}$ are just the computational basis states, i.e., all length-$n$ bit strings.

We will consider two cases: collective and independent phase damping.

\subsubsection{The collective case}

In the collective phase damping case there is a qubit permutation symmetry and the qubits are all coupled to the same bosonic modes. 
Thus the system-bath interaction is
\beq
H_{SB} = \sum_{k,\a} g_{k} Z_\a \ox(b_{k}+b_{k}^\dgr) = gA\ox B\ , \qquad A = \sum_{\a=1}^n Z_\a \ , \quad B = \sum_{k} (g_{k}/g) (b_{k}+b_{k}^\dgr)\ .
\eeq
Since there is only one bath operator, the analysis starting from Eq.~\eqref{eq:646} holds without any change.

\subsubsection{The independent case}

Here each qubit is coupled to a separate bosonic bath. Thus the bath Hamiltonian is
\beq
H_B = \sum_{\a=1}^n H_{B,\a}\ , \qquad H_{B,\a} = \sum_k\o_{k,\a} b^\dgr_{k,\a} b_{k,\a}\ ,
\eeq 
where $b_{k,\a}$ is the annihilation operator associated with bosonic mode $k$ and qubit $\a$, and the system-bath interaction is
\beq
H_{SB} = \sum_{k,\a} g_{k,\a} Z_\a \ox(b_{k,\a}+b_{k,\a}^\dgr) = g\sum_{\a=1}^n A_\a \ox B_\a\ , \qquad A_\a = Z_\a \ , \quad B_\a = \sum_k (g_{k,\a}/g) (b_{k,\a}+b_{k,\a}^\dgr)\ .
\eeq
The bath Gibbs state factors since operators belonging to different qubit indices commute:
\beq
\r_B = \frac{1}{Z}e^{-\b H_B} = \bigotimes_{\a} \r_{B,\a} \ , \qquad \r_{B,\a} = \frac{1}{Z_\a}e^{-\b H_{B,\a}} \ ,
\eeq
where $Z_\a = \Tr e^{-\b H_{B,\a}}$.
In light of this case, consider the bath correlation functions, and recall that $\Tr(A\ox B)=\Tr A \times \Tr B$ for any pair of operators $A$ and $B$:
\bes
\begin{align}
\langle B_\alpha(t) B_\beta \rangle_B &= \Tr \left(\rho_{B}e^{iH_{B,\alpha}t}B_\alpha e^{-iH_{B,\alpha}t} B_\beta \right) \\
&\stackrel{\alpha\neq\beta}{=}\Tr\left(\rho_{B,\alpha}e^{iH_{B,\alpha}t}B_\alpha e^{-iH_{B,\alpha}t}\right)\Tr\left(\rho_{B,\beta}B_\beta\right) \\
&=0 \ ,
\end{align}
\ees
where the last equality follows since (as in Sec.~\ref{sec:cumulant1}) we can always ensure that $\Tr(\rho_B B)=0$. If $\alpha=\beta$, we recover the expression we obtained in the collective case but with the bath parameters corresponding to the $\alpha$-th bath. Thus,
\beq
\langle B_\alpha(t) B_\beta \rangle_B = \delta_{\alpha\beta} \langle B_\alpha(t) B_\alpha \rangle_B\ .
\eeq
This, in turn, implies that 
\beq
\gamma_{\alpha\beta}(\omega) =  \int_{-\infty}^\infty dt e^{i\omega t}\langle B_\a(t)B_\b \rangle_B = \delta_{\alpha\beta} \gamma_{\alpha\alpha}(\omega)\ .
\eeq
If we again assume an Ohmic spectral density, now of the form 
\beq
J_\a(\omega) = \eta_\a \omega e^{-\omega/\omega_{c,\a}}\ ,
\eeq
then the same calculation as in the collective case yields
\beq
\g_{\a\a}(\o) = \frac{2\pi\eta_\a \omega e^{-\abs{\omega}/\omega_{c,\a}}}{1-e^{-\beta\omega}}\ ,
\eeq
where we have assumed that all baths are thermally equilibrated at the same inverse temperature $\b$.

\subsubsection{Contrasting the dephasing rates in the collective and independent cases}

We can now compare the predictions of the collective and independent dephasing models. Consider the time evolution of the density matrix elements in the energy eigenbasis, i.e., $\dot{\rho}_{ab}$. Using the RWA-LE we have:
\beq
\dot{\rho}_{ab} = \braket{\epsilon_a|\dot{\rho}|\epsilon_b} = \bra{\e_a}\sum_{\a\b,\o}\g_{\a\b}(\o)\left(A_\b(\o)\r A_\a^\dgr(\o)-\frac{1}{2}\left\{A_\a^\dgr(\o)A_\b(\o),\r\right\}\right)\ket{\e_b} \ .
\eeq
Evaluating this yields, after some algebra:
\bes
\begin{align}
\label{eq:independentCoherenceME}
\text{independent:} \qquad \dot{\rho}_{ab}&= -  \rho_{ab}/\tau_{ab}^{\text{ind}}\ , \quad 1/\tau_{ab}^{\text{ind}}=\frac{1}{2}g^2\sum_{\alpha=1}^n \gamma_{\alpha\alpha}(0) (A_{aa,\alpha} - A_{bb,\alpha})^2\\
\label{eq:collectiveCoherenceME}
\text{collective:} \qquad \dot{\rho}_{ab} &= - \rho_{ab}/\tau_{ab}^{\text{col}}\ , \quad 1/\tau_{ab}^{\text{col}} = \frac{1}{2} g^2\gamma(0) (A_{aa} - A_{bb})^2 \ ,
\end{align}
\ees
where we used the explicit form of the eigenstates of the system Hamiltonian in Eq.~\eqref{eq:612}.
We see that, as expected from single-qubit dephasing case (recall, e.g., Sec.~\ref{app:2level}) that there is no change in the populations, i.e., $\dot{\rho}_{aa}=0$. The solution to these decoupled equations for the off-diagonal elements is of the form $\rho_{ab}(t)=\rho_{ab}(0)e^{-t/\tau_{ab}}$, where $\tau_{ab}$ is the dephasing time. 

Let us compare the scaling of this time with the number of qubits $n$ in the independent and collective dephasing settings.

\begin{itemize}

\item Independent-dephasing:
\beq
A_{aa,\alpha} = \braket{\epsilon_a|Z_\alpha |\epsilon_a} =\pm 1\ .
\eeq
Thus, $(A_{aa,\alpha} - A_{bb,\alpha})^2 = 4$ for $a\neq b$.

\item Collective dephasing:
\beq
A_{aa} = \braket{\epsilon_a|\sum_{\alpha=1}^n Z_\alpha |\epsilon_a} \in \{-n,-n+2,\dots ,n-2,n\}
\eeq
Thus $\max (A_{aa} - A_{bb})^2 = 4n^2$  and $\min(A_{aa} - A_{bb})^2=0$ for even $n$, or $\min(A_{aa} - A_{bb})^2=4$ for odd $n$. 

\end{itemize}
There is thus a substantial difference between the two models. In the independent case, using Eq.~\eqref{eq:independentCoherenceME}, we find $1/\tau_{ab}^\text{ind} = O(n)$, or simply $1/\tau_{ab}^\text{ind} =2n\g(0)$ if all rates $\g_{\a\a}(0)$ are equal [to $\g(0)$]. In the collective case, using Eq.~\eqref{eq:collectiveCoherenceME}, we have a range of dephasing rates, varying from ``superdecoherent" $1/\tau_{ab}^{\text{col}} = 2n^2\g(0)$, to ``decoherence-free" $1/\tau_{ab}^{\text{col}} = 0$ for even $n$ or to ``subdecoherent" $1/\tau_{ab}^{\text{col}} = 2\g(0)$ for odd $n$. The decoherence-free case is of particular interest in quantum computing, and arises for the zero-eigenvalue system eigenstates of the collective dephasing operator $\sum_{\a=1}^n$, i.e., states $\ket{\e_a}$ that have an equal number of $0$'s and $1$'s in the computational basis. Such states form a conserved subspace under the action of the RWA-LE, and hence are called a decoherence-free subspace \cite{Zanardi:97c,Lidar:1998fk,Lidar:2003fk} (recall also our discussion of non-equilibration in Sec.~\ref{sec:return}). At the other extreme, the states in the superdecoherent subspace dephase quadratically faster than in the independent dephasing case.



\subsection{Bounding the Markov approximation error}
\label{sec:Markov-approx-err-bound}

Earlier we asserted that it is permissible to go from Eq.~\eqref{integEq2} to Eq.~\eqref{integEq3}. Our goal is now to prove this, and in particular to derive the associated error estimate, $O(g^4 \tau_B^3)$.

Consider just one of the four (two due to the commutator, times two due to the h.c.) terms in Eqs.~\eqref{integEq2}, and its Markov approximation [as in Eq.~\eqref{integEq3}]:
\bes
\label{integEq2-a}
\begin{align}
\text{true} &\equiv g^2 \sum_{\a\b} \int_0^t d\tau \mc{B}_{\a\b}(\tau) A_{\alpha}(t)A_{\beta}(t-\tau)\tilde{\rho}(t-\tau)  \\ 
&\approx g^2 \sum_{\a\b} \int_0^\infty d\tau \mc{B}_{\a\b}(\tau) A_{\alpha}(t)A_{\beta}(t-\tau)\tilde{\rho}(t) \equiv \text{approx} \\
&=g^2 \sum_{\a\b} \underbrace{\int_0^\infty d\tau \mc{B}_{\a\b}(\tau) A_{\alpha}(t)A_{\beta}(t-\tau)\left(\tilde{\rho}(t)-\tilde{\rho}(t-\tau)\right)}_{\Delta_1} + g^2 \sum_{\a\b} \int_0^\infty d\tau \mc{B}_{\a\b}(\tau) A_{\alpha}(t)A_{\beta}(t-\tau)\tilde{\rho}(t-\tau) \\
&= g^2 \sum_{\a\b} \Delta_1 + \text{true}  + g^2 \sum_{\a\b} \underbrace{\int_t^\infty d\tau \mc{B}_{\a\b}(\tau) A_{\alpha}(t)A_{\beta}(t-\tau)\tilde{\rho}(t-\tau)}_{\Delta_2} \ .
\end{align}
\ees
Thus, $\text{approx} = \text{true} + \Delta_1 + \Delta_2$, or
\beq
\text{error} = \| \text{true} - \text{approx}\| = \|\Delta_1 + \Delta_2\| \leq \|\Delta_1\| + \|\Delta_2\| \ .
\label{eq:515}
\eeq
This shows that in order to bound the error it suffices to bound $\|\Delta_1\|$ and $\|\Delta_2\|$ in a convenient norm, which we will take to be the operator norm (see Appendix~\ref{app:norms} for a discussion of the various norms we use here and their properties). The other three terms in Eqs.~\eqref{integEq2} will obey exactly the same bound, since they are different from Eq.~\eqref{integEq2-a} only in the operator order, which is removed once we take the norm. Thus, it suffices to concern ourselves with the term in Eq.~\eqref{integEq2-a}.

\subsubsection{Bound on $\|\Delta_1\|$}

Using the triangle inequality and submultiplicativity of the operator norm $\|\cdot\|_\infty$:
\bes
\begin{align}
\|\Delta_1\|_\infty &\leq \int_0^\infty d\tau |\mc{B}_{\a\b}(\tau)| \on{A_{\alpha}(t)} \on{A_{\beta}(t-\tau)}\on{\tilde{\rho}(t)-\tilde{\rho}(t-\tau)} \\
&= \int_0^\infty d\tau |\mc{B}_{\a\b}(\tau)| \on{A_{\alpha}} \on{A_{\beta}}\on{\tilde{\rho}(t)-\tilde{\rho}(t-\tau)} \\
&\leq \eta^2 \int_0^\infty d\tau |\mc{B}_{\a\b}(\tau)| \on{\tilde{\rho}(t)-\tilde{\rho}(t-\tau)}  \ ,
\end{align}
\ees
where in the second line we used unitary invariance, and where
\beq
\eta \equiv \max_\a \on{A_{\alpha}} \ .
\eeq 
Now, by the mean value theorem of elementary calculus, there exists a point $t'\in [t-\tau,t]$ such that
\beq
\frac{ \tilde{\rho}(t)-\tilde{\rho}(t-\tau)}{\tau} = \dot{\tilde{\rho}}(t')\ .
\eeq
Therefore
\beq
\on{\tilde{\rho}(t)-\tilde{\rho}(t-\tau)} \leq \tau \sup_{t'\in [t-\tau,t]} \on{\dot{\tilde{\rho}}(t')} \ ,
\eeq
and
\bes
\begin{align}
\|\Delta_1\|_\infty \leq \eta^2 \int_0^\infty d\tau \ \tau |\mc{B}_{\a\b}(\tau)| \sup_{t'\in [t-\tau,t]} \on{\dot{\tilde{\rho}}(t')}  \ .
\end{align}
\ees
To bound $\on{\dot{\tilde{\rho}}(t')}$ we can return to Eq.~\eqref{integEq2}:
\bes
\begin{align}
\on{\dot{\tilde{\rho}}(t')} &\leq g^2\sum_{\alpha,\beta}\int_0^{t'} d\tau |\mc{B}_{\a\b}(\tau)|\on{[A_{\alpha}(t),A_{\beta}(t-\tau)\tilde{\rho}(t-\tau)]+\text{h.c.}} \\
& \leq 4 g^2\sum_{\alpha,\beta}\int_0^{t'} d\tau |\mc{B}_{\a\b}(\tau)|\on{A_{\alpha}(t)A_{\beta}(t-\tau)\tilde{\rho}(t-\tau)} \\
&\leq 4g^2\sum_{\alpha,\beta}\int_0^{t'} d\tau |\mc{B}_{\a\b}(\tau)|\on{A_{\alpha}} \on{A_{\beta}}\|\tilde{\rho}(t-\tau)\|_1 \\
& \leq 4(\eta g)^2 M \int_0^{t'} d\tau |\mc{B}_{\a\b}(\tau)|  \ ,
\end{align}
\ees
where in the second line we used the fact that all four terms in the first line (again, after the commutator and h.c.) have the same operator norm, and where $M \equiv \sum_{\a\b}1$ is the square of the number of summands in $H_{SB} = \sum_\a A_\a \ox B_\a$. Now, since 
\beq
 \sup_{t'\in [t-\tau,t]}  \int_0^{t'} d\tau |\mc{B}_{\a\b}(\tau)| \leq  \int_0^{\infty} d\tau |\mc{B}_{\a\b}(\tau)| \ ,
\eeq
we have
\beq
\|\Delta_1\|_\infty \leq 4M \eta^4 g^2 \int_0^\infty d\tau \ \tau |\mc{B}_{\a\b}(\tau)| \int_0^{\infty} d\tau |\mc{B}_{\a\b}(\tau)| \sim 4M \eta^4 g^2 \tau_B^3\ ,
\eeq
where we used Eq.~\eqref{eq:481} once with $n=1$, and once with $n=0$.

\subsubsection{Bound on $\|\Delta_2\|$}
Similarly, 
\bes
\begin{align}
\on{\Delta_2} &\leq \int_t^\infty d\tau |\mc{B}_{\a\b}(\tau)| \on{A_{\alpha}(t)}\on{A_{\beta}(t-\tau)}\|\tilde{\rho}(t-\tau)\|_1 \\
&\leq \eta^2 \int_t^\infty d\tau |\mc{B}_{\a\b}(\tau)| \ .
\end{align}
\ees
Intuitively, we know that $\int_t^\infty d\tau |\mc{B}_{\a\b}(\tau)|$ should be arbitrarily small as long as $t \gg \tau_B$, as we assumed in Eq.~\eqref{eq:479}, since the correlation function decays over a timescale of $\tau_B$. To formalize this, note that convergence of $\int_t^\infty d\tau |\mc{B}_{\a\b}(\tau)|$ is guaranteed if 
\beq
|\mc{B}_{\a\b}(\tau)| \sim (\tau_B/\tau)^x\ , \quad x>1\ .
\eeq
Thus, we will assume that the correlation function decays no more slowly than this power-law dependence [this is even slower than the subexponential decay we assumed to get Eq.~\eqref{eq:483}]. Under this assumption, we have 
\beq
\int_t^\infty d\tau |\mc{B}_{\a\b}(\tau)| \sim \int_t^\infty d\tau \left(\frac{\tau_B}{\tau}\right)^x = \left.  \frac{\tau_B^x}{(1-x) \tau^{x-1}}\right|_t^{\infty} = \frac{1}{x-1}\frac{\tau_B^x}{t^{x-1}}\ .
\eeq
Now, to use the assumption that $t \gg \tau_B$, let us write $t = c\tau_B$, where $c\gg 1$. Then:
\beq
\int_t^\infty d\tau |\mc{B}_{\a\b}(\tau)| \sim \frac{\tau_B}{(x-1)c^{x-1}} \ .
\eeq
Therefore, even with a power-law decaying correlation function, we have
\begin{align}
\on{\Delta_2} \lesssim \eta^2 \frac{\tau_B}{(x-1)c^{x-1}} \ ,
\end{align}
which can be made arbitrarily small by making $c = t/\tau_B$ large enough. 

\subsubsection{Putting the bounds together}
We have seen that $\|\Delta_1\|_\infty \lesssim 4M \eta^4 g^2 \tau_B^3 = O(g^2 \tau_B^3)$ and $\on{\Delta_2}$ can be made arbitrarily small. Thus the dominant contribution to the error comes from $\|\Delta_1\|_\infty$, which is the error due to replacing all the intermediate-time states (at $t-\tau$) by the state at the single time $t$. Moreover, we need $t\gg \tau_B$ in order to ensure that $\on{\Delta_2}$ can be neglected.

When accounting for the additional $g^2$ prefactor in Eq.~\eqref{integEq2-a} (as well as $\sum_{\a\b}$, which just gives rise to another factor of $M$), we finally have from Eq.~\eqref{eq:515}:
\beq
\text{error} = O(g^4 \tau_B^3) \ ,
\end{equation}
as claimed.


\subsection{The RWA-LE is the infinite coarse-graining time limit of the cumulant-LE}

The RWA we used in Sec.~\ref{sec:RWA} in order to derive the Lindblad equation leaves something to be desired. We simply dropped terms with different Bohr frequencies, without a rigorous mathematical justification. We will now show that \emph{the RWA-LE can be rigorously derived from the cumulant-LE, in the limit of an infinite coarse-graining timescale}. This shows that the cumulant-LE is truly more general than the (standard) RWA-LE.

\subsubsection{Quick summary}
For convenience, let us collect the main results of each of the two approaches. For simplicity we'll set $\lambda=g=1$ and also assume that $H_{SB} = A\ox B$ (not a sum), so that we can drop the $\alpha$ index from Eq.~\eqref{Lindbladw}. The RWA-LE is then:
\beq
\dot{\tilde{\rho}}(t) = -i \left[ H_{\mathrm{LS}}, \tilde{\rho}(t) \right] + \sum_{\omega} \gamma({\omega}) \left( A_{\omega} \tilde{\rho}(t) A_{\omega}^\dagger - \frac{1}{2} \left\{A_{\omega}^{\dagger} A_{\omega}, \tilde{\rho}(t) \right\} \right) 
\label{eqt:SME_RWA}
\eeq
with
\begin{align} 
\label{eq:g531}
\gamma(\omega) = \int_{-\infty}^{\infty} d s e^{i \omega s} \mathcal{B}(s,0)\ .
\end{align}

The cumulant-LE is
\beq \label{eqt:SME_Average}
\dot{\tilde{\rho}}(t) = -i \left[ H'_{\mathrm{LS}}, \tilde{\rho}(t) \right] + \sum_{\omega, \omega'} \gamma_{\omega \omega'}(\tau) \left( A_{\omega} \tilde{\rho}(t) A_{\omega'}^\dagger - \frac{1}{2} \left\{A_{\omega'}^{\dagger} A_{\omega}, \tilde{\rho}(t) \right\} \right) \ ,
\eeq
where the rates $\gamma$ keep a dependence on two different Bohr frequencies $\omega$ and $\o'$:
\begin{align}
\label{eq:gammaww'}
\gamma_{\omega \omega'} (\tau) & =  \frac{1}{\tau} b_{\omega \omega'} (\tau) \ , \quad 
b_{\omega \omega'} (\tau) = \int_0^\tau d s \int_0^\tau ds'  e^{i ( \omega' s - \omega s' )}  \mathcal{B}(s,s')  \ .
 \end{align}

Our goal is to show that in an appropriate sense the cumulant-LE tends to the RWA-LE in the limit as $\tau\to\infty$, where $\tau$ is the coarse-graining timescale. More specifically, we will show that $\lim_{\tau\to\infty} \gamma_{\omega \omega'} (\tau) = \gamma(\omega)\d_{\o\o'}$~\cite{Majenz:2013qw}. 
We will assume stationarity, i.e., $\mathcal{B}(s,s') = \mathcal{B}(s-s',0)$. 

\subsubsection{A useful lemma}
\begin{mylemma}
The following equivalent form holds for $\g_{\omega \omega'} (\tau)$:
\beq
\g_{\omega \omega'} (\tau) = \frac{1}{\tau}e^{i\frac{\omega'-\omega}{2}\tau}\intop_0^{\tau}dv\ \cos\left(\frac{\omega'-\omega}{2}(v-\tau)\right)\intop_{-v}^{v}du\ e^{i\frac{\omega+\omega'}{2}u}\mathcal{B}(u,0) \ .
\label{eq:g-simplified}
\eeq
\end{mylemma}

\begin{proof}
In the RWA we dropped terms with $\o\neq\o'$, so it makes sense to rewrite $\omega' s - \omega s'$ in terms of a sum and difference of Bohr frequencies: 
\beq
\omega' s - \omega s' = \frac{1}{2}(\o'-\o)v+\frac{1}{2}(\o'+\o)u\ ,
\eeq
where $u=s-s'$ and $v=s+s'$. After this change of variables $\mathcal{B}(s-s',0) = \mathcal{B}(u,0)$, and since $s=(v+u)/2$ and $s'=(v-u)/2$, the Jacobian of the transformation is $\left| 
\left(
\begin{array}{cc}
1/2  &  1/2    \\
1/2  &  -1/2      
\end{array}
\right)
\right| = 1/2$. In terms of the new variables the integration region is diamond shaped (a square rotated by $\pi/4$), bounded between the lines $u=v$ and $u=-v$ for $v\in [0,\tau]$ and the lines $u=2\tau-v$ and $v-2\tau$ for $v\in [\tau,2\tau]$.
Thus:
\begin{align}
b_{\omega \omega'}(\tau)= \frac{1}{2}\intop_0^{\tau}dv\ e^{i\frac{\omega'-\omega}{2} v}\intop_{-v}^{v}du\ e^{i\frac{\omega+\omega'}{2}u}\mathcal{B}(u,0)+\frac{1}{2} \intop_{\tau}^{2 \tau}dv\ e^{i\frac{\omega'-\omega}{2} v}\intop_{-(2\tau-v)}^{2 \tau-v}du\ e^{i\frac{\omega+\omega'}{2}u}\mathcal{B}(u,0) \ .
\end{align}
To get the integration limits to be the same we make a change of variables from $v$ to $2\tau-v$ in the second double integral:
\bes
\begin{align}
 b_{\omega \omega'}(\tau)&=\frac{1}{2}\intop_0^{\tau}dv\ e^{i\frac{\omega'-\omega}{2} [(v-\tau)+\tau]}\intop_{-v}^{v}du\  e^{i\frac{\omega+\omega'}{2}u}\mathcal{B}(u,0)+ \frac{1}{2} \intop_{0}^{\tau}dv\ e^{-i\frac{\omega'-\omega}{2} [(v-\tau)-\tau]}\intop_{-v}^{v}du\ e^{i\frac{\omega+\omega'}{2}u}\mathcal{B}(u,0)\\
&=e^{i\frac{\omega'-\omega}{2}\tau}\intop_0^{\tau}dv\ \cos\left(\frac{\omega'-\omega}{2}(v-\tau)\right)\intop_{-v}^{v}du\ e^{i\frac{\omega+\omega'}{2}u}\mathcal{B}(u,0) \ .
\end{align}
\ees
The claim now follows from Eq.~\eqref{eq:gammaww'}.
\end{proof}

\subsubsection{The $\o=\o'$ case}

For $\omega=\omega'$ we now have:

\begin{equation}
 \g_{\omega \omega}(\tau)=\frac{1}{\tau}\intop_0^\tau dv\ \intop_{-v}^{v}du\ e^{i\omega u}\mathcal{B}(u,0) \ .
\end{equation}
Let $U = \int_{-v}^{v}du\ e^{i\omega u}\mathcal{B}(u,0)$. Recall the Leibnitz rule for differentiating a definite integral:
\beq
\partial_z \int_{a(z)}^{b(z)} f(x,z)dx = \int_{a(z)}^{b(z)}\partial_z f(x,z)dx + f(b(z),z)b' - f(a(z),z)a'\ .
\label{eq:Leibnitz}
\eeq
Therefore $dU = \left(e^{i\o v} \mc{B}(v,0)+e^{-i\o v} \mc{B}(-v,0)\right)dv$. 
Then, integrating by parts ($\int_0^\tau Udv = \left[Uv\right]_0^\tau - \int_0^\tau vdU$) gives:
\begin{align}
\g_{\omega\omega}(\tau)&=\frac{1}{\tau}\left[v \intop_{-v}^{v}du\ e^{i\omega u}\mathcal{B}(u,0)\right]_0^\tau-\frac{1}{\tau}\intop_0^{\tau}dv\ v\left( e^{i\omega v}\mathcal{B}(v,0)+e^{-i\omega v}\mathcal{B}(-v,0)\right)\ .
\label{eq:539}
\end{align}
Consider the second integral:
\bes
\begin{align}
\left| \frac{1}{\tau} \intop_0^{\tau}dv\ v e^{i\omega v}\mathcal{B}(v,0) \right| &\leq \frac{1}{\tau} \intop_0^{\tau}dv\ v \left| \mathcal{B}(v,0) \right| \leq \frac{1}{\tau} \intop_0^{\infty}dv\ v \left| \mathcal{B}(v,0) \right| \\
&\sim \frac{1}{\tau} \tau_B^2 \stackrel{\tau\to\infty}{\longrightarrow} 0\ ,
\end{align}
\ees
where in the last step we used the assumption~\eqref{eq:481} that the bath correlation function decays with a finite timescale $\tau_B$. Since $\mc{B}(v,0) =\mc{B}^*(-v,0)$ [recall Eq.~\eqref{eq:511b}], the third integral in Eq.~\eqref{eq:539} satisfies the same bound and limit. We are thus left with
\beq
\lim_{\tau \to \infty} \g_{\o\o}(\tau) =\intop_{-\infty}^{\infty}du\ e^{i\omega u}\mathcal{B}(u,0)=\gamma(\omega) \ ,
\eeq
where the last equality is due to Eq.~\eqref{eq:g531}.

\subsubsection{The $\o\neq\o'$ case}
For $\omega\neq\omega'$ we also perform integration by parts of Eq.~\eqref{eq:g-simplified}, but we shall see that this time the boundary terms vanish. 
We write 
$\g_{\omega \omega'} (\tau) = \frac{1}{\tau}e^{i\frac{\omega'-\omega}{2}\tau} \intop_0^{\tau}dV\  U(v)$,
where now $dV = \cos\left(\frac{\omega'-\omega}{2}(v-\tau)\right)dv$ and $U(v) = \intop_{-v}^{v}du\ e^{i\frac{\omega+\omega'}{2}u}\mathcal{B}(u,0)$. Then 
\bes
\begin{align}
V(v) &= \frac{2}{\omega'-\omega}\sin\left(\frac{\omega'-\omega}{2}(v-\tau)\right) \\
dU/dv &= e^{i\frac{\omega+\omega'}{2}v} \mc{B}(v,0)+e^{-i\frac{\omega+\omega'}{2}v} \mc{B}(-v,0)\\
\left[U(v) V(v) \right]_0^\tau &= U(\tau) V(\tau) - U(0)V(0) = 0 \ .
\label{eq:695c}
\end{align}
\ees
Therefore:
\begin{equation}
\g_{\omega \omega'}(\tau)=-\int_{0}^{\tau}V dU = -\frac{2 e^{i\frac{\omega'-\omega}{2}\tau}}{(\omega'-\omega)\tau}\intop_0^{\tau}dv \sin\left( \frac{(\omega'-\omega)}{2}(v-\tau)\right)\left[e^{i\frac{\omega+\omega'}{2}v}\mathcal{B}(v,0)+e^{-i\frac{\omega+\omega'}{2}v}\mathcal{B}(-v,0)\right] \ .
\end{equation}
Changing from $v$ to $-v$ in the second term we get
\bes
\begin{align}
\g_{\omega \omega'}(\tau)=&-\frac{2 e^{i\frac{\omega'-\omega}{2}\tau}}{(\omega'-\omega)\tau}\Bigg[\intop_0^{\tau}dv\sin\left(\frac{(\omega'-\omega)}{2}(v-\tau)\right)e^{i\frac{\omega+\omega'}{2}v}\mathcal{B}(v,0)+\intop_{-\tau}^{0}dv\sin\left(\frac{(\omega'-\omega)}{2} (-v-\tau)\right)e^{i\frac{\omega+\omega'}{2}v}\mathcal{B}(v,0)\Bigg]\\
=&\frac{e^{i\frac{\omega'-\omega}{2}\tau}}{(\omega'-\omega)\tau}\intop_{-\tau}^{\tau}dv\left[\sin\left(\frac{\omega'-\omega}{2}\tau\right)\left(e^{i \omega v}+e^{i \omega' v}\right)+\frac{\mathrm{sgn(v)}}{i}\cos\left(\frac{\omega'-\omega}{2}\tau\right)\left(e^{i \omega v}-e^{i \omega' v}\right)\right]\mathcal{B}(v,0) \ ,
\end{align}
\ees
where we used the angle sum identity for the sine in the last equality. Thus:
\beq
\lim_{\tau \to \infty} \g_{\omega \omega'}(\tau) = \lim_{\tau \to \infty} \frac{e^{i\frac{\omega'-\omega}{2}\tau}}{(\omega'-\omega)\tau}\left[\sin\left(\frac{\omega'-\omega}{2}\tau\right)(\gamma(\omega)+\gamma(\omega'))+2\cos\left(\frac{\omega'-\omega}{2}\tau\right)(S(\omega)-S(\omega'))\right] \ .
\eeq
where we have used that for $\Gamma(\omega) = \int_0^{\infty} ds\ e^{i \omega s} \mathcal{B}(s,0)$ [recall Eq.~\eqref{Gamma}], we have $\gamma(\omega) = \Gamma(\omega) + \Gamma^{\ast}(\omega)$ and $2 i S(\omega) = \Gamma(\omega) - \Gamma^\ast(\omega)$ [recall Eq.~\eqref{eq:499}].  Since nothing cancels with the overall $\tau^{-1}$, we find that the $\omega \neq \omega'$ term vanishes.  

A similar calculation could be done for the Lamb shift term~\eqref{eq:Lambshift}.  Therefore, the RWA results can be understood as the $\tau\to\infty$ limit of the coarse-graining timescale.


\section{The Nakajima-Zwanzig Equation}

The master equations we have developed so far are approximations to the true dynamics. In this section we take a step back and derive an exact master equation. Since it is exact, it will naturally be non-Markovian. 

Consider the total Hamiltonian 
\beq
H=H_0 + \alpha H_{SB}\ , \qquad H_0 = H_S+H_B\ ,
\eeq 
where $0<\a<1$ is a dimensionless parameter.
Let us work in the interaction picture, so that the total system-bath state $\tr$ satisfies
\beq
\partial_t \tr = -i \a [\tilde{H}(t),\tr(t)] \equiv \a \mc{L}\tr(t)\ ,
\label{eq:655}
\eeq
where as usual $\tilde{H}(t) = U_0^\dagger (t) H_{SB} U_0(t)$, with $U_0(t) = e^{-iH_0 t}$, and $H_{SB} = \sum S_{\alpha} \otimes B_{\alpha}$. We abbreviate $\partial_t \equiv \frac{\partial}{\partial t}$. For the rest of this section we drop the tilde decoration on states to simplify the notation, so that, e.g., $\r$ denotes the interaction-picture system-bath state.

\subsection{Feshbach $\mc{P}$-$\mc{Q}$ partitioning}
Consider a \emph{fixed} bath state $\r_B$. As usual, $\r_S = \Tr_B\r$ is the system state of interest. Consider the projection superoperator $\mc{P}$ defined via
\beq
\mc{P}\r = \Tr_B(\r)\ox \r_B \ .
\eeq
That $\mc{P}$ is a projection follows from applying it twice:
\beq
\mc{P}^2\r = \mc{P}[\Tr_B(\r)\ox \r_B] = \Tr_B[\Tr_B(\r)\ox \r_B]\ox \r_B = \Tr_B(\r)\ox \r_B = \mc{P}\r\ .
\eeq
Define the orthogonal projection $\mc{Q}$ via 
\beq
\mc{Q} = I-\mc{P} \ .
\eeq
We call $\mc{P}\r$ the ``relevant" part, and $\mc{Q}\r$ the ``irrelevant part". This procedure is sometimes called Feshbach $\mc{P}$-$\mc{Q}$ partitioning, after a method introduced in nuclear scattering theory \cite{Feshbach:1958aa}.

We are interested in deriving a master equation for $\partial_t (\mc{P}\r)$. Now, note that
\beq
\partial_t (\mc{P}\r) = \partial_t [\Tr_B(\r)\ox \r_B] = \Tr_B(\partial_t\r)\ox \r_B = \mc{P}(\partial_t\r)\ ,
\eeq
i.e., $[\mc{P},\partial_t]=0$. Therefore, using Eq.~\eqref{eq:655}:
\beq
\pt (\mP\r)=\a \mP \mL\r\ .
\eeq
Likewise:
\beq
\pt(\mQ\r) = \pt[(I-\mP)\r]=\a\mL\r-\a\mP\mL\r=\a(I-\mP)\mL\r=\a\mQ\mL\r \ .
\eeq
Let us now insert $I=\mP+\mQ$ into the last two equations:
\bes
\label{eq:662}
\begin{align}
\pt(\mP\r) &= \a\mP\mL(\mP+\mQ)\r = \a \mP \mL \mP\r + \a\mP\mL\mQ\r\\
\pt(\mQ\r) &= \a\mQ\mL(\mP+\mQ)\r = \a \mQ \mL \mP\r + \a\mQ\mL\mQ\r\ .
\end{align}
\ees
These are coupled differential equations for the relevant ($\mP\r$) and irrelevant ($\mQ\r$) parts. To solve them, let us eliminate the irrelevant part. 

Define 
\beq 
\hat{X} \equiv \mP X\ , \qquad  \bar{X} \equiv \mQ X
\eeq
for any operator $X$. 
then Eq.~\eqref{eq:662} can be rewritten more compactly as:
\bes
\label{eq:664}
\begin{align}
\label{eq:664a}
\pt\hat\r &= \a \hat \mL \hat\r + \a\hat\mL\bar\r\\
\label{eq:664b}
\pt\bar\r &=  \a \bar \mL \hat\r + \a\bar\mL\bar\r\ .
\end{align}
\ees

\subsection{Derivation}

We can formally solve the second of these equations and substitute the solution into the first.  Consider first 
$\pt\bar\r =  \a \bar \mL \bar\r$. This has the immediate solution $\bar\r(t) = T_+ \exp\left(\a\int_{t_0}^t\bar\mL(t')dt'\right)\bar\r(t_0)$, where $T_+$ denotes the usual forward Dyson time-ordering. We thus define
\beq
\mG(t,t_0) \equiv T_+ e^{\a\int_{t_0}^t\bar\mL(t')dt'}\ .
\label{eq:665}
\eeq
Eq.~\eqref{eq:664b} contains another term, and we can easily guess that the solution integrates over this term, but first applies $\mG$, i.e.:
\beq
\bar\r(t) = \mG(t,t_0)\bar\r(t_0) + \underbrace{\a\int_{t_0}^t\mG(t,t')\bar\mL(t')\hat\r(t')dt'}_{\circledast} \ .
\label{eq:666}
\eeq
To verify that this is the formal solution of Eq.~\eqref{eq:664b}, we apply the Leibnitz rule~\eqref{eq:Leibnitz} to get $\pt\int_{t_0}^tf(t,t')dt' = f(t,t) + \int_{t_0}^t \pt f(t,t')dt'$, and also note that $\mG(t,t_0)$ has the property $\mG(t,t)=I$, $\pt\mG(t,t')=\bar\mL(t)\mG(t,t')$. Using all of the above we have:
\beq
\pt \circledast = \a\mG(t,t)\bar\mL(t)\hat\r(t) + \a\int_{t_0}^t\pt\mG(t,t')\bar\mL(t')\hat\r(t')dt' =\a\bar\mL(t)\hat\r(t) +\a\bar\mL(t)\circledast = \a\bar\mL(t)\left(\hat\r(t)+\circledast\right)\ . 
\eeq
Therefore, if we differentiate Eq.~\eqref{eq:666} we find:
\beq
\pt\bar\r(t) = \a \bar\mL(t)\mG(t,t_0)\bar\r(t_0) + \a \bar{\mL}(t)(\hat\r(t)+\circledast) = \a\bar\mL(t)\hat\r(t) + \a \bar\mL(t)\underbrace{\left(\mG(t,t_0)\bar\r(t_0)+\circledast\right)}_{\bar\r(t)}\ ,
\eeq
which agrees with Eq.~\eqref{eq:664b} as required.

Substituting the solution for $\bar\r(t)$ into Eq.~\eqref{eq:664a}, we have:
\beq
\pt\hat\r(t) = \underbrace{\a \hat \mL(t) \hat\r(t)}_{(a)} + \underbrace{\a\hat\mL(t) \mG(t,t_0)\bar\r(t_0)}_{(b)} + \underbrace{\a^2{\int_{t_0}^t\hat\mL(t)\mG(t,t')\bar\mL(t')\hat\r(t')dt'}}_{(c)}\ .
\label{eq:669}
\eeq
\begin{itemize}
\item We can show that term (a) can always be made to vanish in a similar way to what we did in Sec.~\ref{sec:cumulant1}. To see this, note that
\bes
\label{eq:670}
\begin{align}
\hat \mL(t) \hat\r(t)&=\mP\mL(t)\mP\r(t) = \mP\mL(t)\Tr_B[\r(t)]\ox\r_B = -i\mP\left[\tilde{H}(t),\r_S(t)\ox\r_B\right] \\
&=-i\sum_\a \Tr_B\left({A}_\a(t)\r_S(t)\ox{B}_\a(t)\r_B\right)\ox\r_B - \Tr_B\left(\r_S(t){A}_\a(t)\ox\r_B{B}_\a(t)\right)\ox\r_B \\
&=-i\sum_\a \left[{A}(t),\r_S(t)\right]\ave{{B}_\a(t)}\ox\r_B = 0
\end{align}
\ees
since $\ave{{B}_\a(t)}$ can be made zero in the same way as in Eq.~\eqref{eq:429}, i.e., $[H_B,\r_B(0)]=0$.

\item Term (b) is an inhomogeneity that depends on the initial condition and measures how much correlation there is in the initial state: 
\beq
\bar\r(0) = (I-\mP)\r(0) =  \r(0) - \Tr_B[\r(0)]\ox\r_B .
\eeq
It vanishes for a factorized initial state, i.e., if $\r(0) = \r_S(0)\ox\r_B$ (the same fixed initial state we chose for the bath at the beginning of the derivation).
\end{itemize}
Thus, assuming a factorized initial state Eq.~\eqref{eq:669} becomes:
\bes
\label{eq:NZeq}
\begin{align}
\pt\hat\r(t) &= \int_{t_0}^t\mK(t,t')\hat\r(t')dt'\\
\label{eq:NZker}
\mK(t,t') &\equiv \a^2 \hat\mL(t)\mG(t,t')\bar\mL(t')\mP\ .
\end{align}
\ees
Equation~\eqref{eq:NZeq} is called the (homogeneous) \emph{Nakajima-Zwanzig master equation} (NZ-ME), and the superoperator $\mK$ is called the \emph{memory kernel} (note that we multiplied it from the right by $\mP$, which we can do since it acts on $\hat\r = \mP\r$). If we include the (b) term $\a\hat\mL(t) \mG(t,t_0)\bar\r(t_0)$ from Eq.~\eqref{eq:669} on the RHS we have the inhomogeneous NZ-ME.

The NZ-ME is exact, non-perturbative, and in the inhomogeneous case it can even describe non-factorized initial conditions. It is clearly non-local in time, in the sense that the RHS retains a memory of the entire history of the state evolution, weighted via the memory kernel. The Nakajima-Zwanzig equation is an integro-differential equation, and solving it is essentially as hard as solving the original Liouville-von Neumann equation~\eqref{eq:655}. Nevertheless, it provides an important and convenient starting point for perturbative expansions, as we shall see shortly.

\subsection{From the Nakajima-Zwanzig equation to the Born master equation}
Consider a perturbative expansion in $\a$. To lowest order we have from Eq.~\eqref{eq:665}:
\beq
\mG(t,t_0) = I + O(\a)\ ,
\eeq
so at the same order the memory kernel becomes
\beq
\mK(t,t') = \a^2\hat\mL(t)[I+O(\a)]\bar\mL(t')\mP = \a^2\mP\mL(t)\mQ\mL(t')\mP+O(\a^3)\ ,
\eeq
and hence:
\bes
\begin{align}
\pt[\mP\r(t)] &= \pt\r_S(t)\ox\r_B = \a^2\int_{t_0}^t\mP\mL(t)\mQ\mL(t')\mP\r(t')dt' \\
&= \a^2 \int_{t_0}^t\mP\mL(t)\mL(t')\mP\r(t')dt'\\
&=-\a^2\int_{t_0}^t\mP\left[\tilde{H}(t),\left[\tilde{H}(t'),\r_S(t')\ox\r_B\right]\right]dt'\\
&=-\a^2\int_{t_0}^t\Tr_B\left[\tilde{H}(t),\left[\tilde{H}(t'),\r_S(t')\ox\r_B\right]\right]\ox\r_B dt' \,
\end{align}
\ees
where in the second line we used $\mP\mL(t)\mP=0$ [Eq.~\eqref{eq:670}]. Applying one final $\Tr_B$ to both sides finally gives
\beq
\pt\r_S(t) = -\a^2\int_{t_0}^t\Tr_B\left[\tilde{H}(t),\left[\tilde{H}(t'),\r_S(t')\ox\r_B\right]\right] dt'\ ,
\label{eq:NZ2}
\eeq
which we recognize as the Born master equation [Eq.~\eqref{integEq2}] discussed in Sec.~\ref{sec:Born-appr}.

\subsection{The $O(\a^3)$ term of the Nakajima-Zwanzig master equation}
The O($\alpha^3$) term comes from the $\alpha^1$ term in the propagator 
\beq
\mG(t,t_0) = I + g_1(t,t_0) + O(\a^2)\ ,
\eeq
where
\beq
g_1(t,t_0)=\alpha\int_{t_0}^{t}\bar{\mathcal{L}}(s)ds\ .
\eeq
The O($\alpha^3$) term is 
\bes
\begin{align}
 \alpha^2\int_{t_0}^{t}dt'\mP\mL(t)g_1(t,t')\mQ\mL(t')\mP\r(t') &= \alpha^3\int_{t_0}^{t}dt'\mP\mL(t)\int_{t'}^{t}ds\mQ\mL(s) \mQ\mL(t')P\r(t')\\
&= \alpha^3\int_{t_0}^{t}\int_{t'}^{t}dt'ds\mP\mL(t)\mQ\mL(s)\mQ\mL(t')\mP\r(t')\ ,
\end{align}
\ees
where
\bes
\begin{align}
\mP\mL(t)\mQ\mL(s)\mQ\mL(t')\mP&=\mP\mL(t)(I-\mP)\mL(s)(I-\mP)\mL(t')\mP \\
&= \mP\mL(t)\mL(s)\mL(t')\mP-\mP\mL(t)\mL(s)\mP\mL(t')\mP-\mP\mL(t)\mP\mL(s)\mL(t')\mP+\mP\mL(t)\mP\mL(s)\mP\mL(t')\mP\ .
\end{align}
\ees
It turns out that we can always ensure that
\beq
\mP \mL(t_1)\cdots\mL(t_n) \mP = 0
\label{eq:726}
\eeq
for any odd $n$ and any ordering of the time argument, by appropriately shifting the bath operators. Therefore the order $O(\a^3)$ term vanishes, and the Nakajima-Zwanzig master equation is unchanged at this order, namely:
\beq
\pt\r_S(t) = -\a^2\int_{t_0}^t\Tr_B\left[\tilde{H}(t),\left[\tilde{H}(t'),\r_S(t')\ox\r_B\right]\right] dt'+O(\a^4)\ .
\eeq


\section{The Time Convolutionless (TCL) Master Equation}

The Nakajima-Zwanzig equation~\eqref{eq:NZeq} contains a convolution with a complicated memory kernel [Eq.~\eqref{eq:NZker}]: $\pt\hat\r(t) = \int_{t_0}^t\mK(t,t')\hat\r(t')dt'$. It seems that this is an unavoidable feature of an exact, non-Markovian master equation. In this section we will see that it is possible to remove the memory kernel by making a type of short-time approximation, and arrive at a fully time-local, convolutionless master equation. The main insight we'll need to achieve this, is that the memory kernel can be removed by formally back-propagating the system state.

\subsection{Derivation}

\subsubsection{Back-propagation}

Let us start again from the Liouville-von-Neumann equation $\partial_t \tr = -i \a [\tilde{H}(t),\tr(t)] \equiv \a \mc{L}\tr(t)$ [Eq.~\eqref{eq:655}]. Its formal solution is
\beq
\tr(t) = T_+ e^{\a \int_{t'}^t \mc{L}(s)ds} \tr(t') = 
\mU_+(t,t')\tr(t')\ ,
\label{eq:728}
\eeq
where $\mU_+(t,t')$ is a forward time-ordered superoperator. This can be inverted so that
\beq
 \tr(t') = 
 T_- e^{-\a \int_{t'}^t \mc{L}(s)ds} \tr(t)= 
 \mU_-(t,t')\tr(t)\ ,
 \label{eq:677}
 \eeq
which defines the backward time-ordered superoperator $\mU_-(t,t')$. To get an explicitly representation note first that by substituting $\tr(t')$ from Eq.~\eqref{eq:677} into Eq.~\eqref{eq:728} we get $\mU_+(t,t') \mU_-(t,t')=I$. Now, since 
\beq
\mU_+(t,t') = T_+ e^{\a \int_{t'}^t\mL(s)ds} = \lim_{\Delta t\to 0} e^{\a\Delta t \mL(t-\Delta t)} \cdots e^{\a\Delta t\mL(t'+\Delta t)}e^{\a\Delta t\mL(t')}\qquad (\Delta t=\lim_{N\to\infty}\frac{t-t'}{N})\ ,
\eeq
in order to have $\mU_+(t,t')$ and $\mU_-(t,t')$ be each other's inverse, it must be that $\mU_-(t,t')$ has the opposite order and $\a$ is replaced by $-\a$, so that when multiplied the two products cancel equal and opposite terms. I.e.,
\beq
\mU_-(t,t') = T_- e^{-\a \int_{t'}^t\mL(s)ds} = \lim_{\Delta t\to 0} e^{-\a\Delta t\mL(t')} e^{-\a\Delta t\mL(t'+\Delta t)} \cdots e^{-\a\Delta t \mL(t-\Delta t)}  \qquad (\Delta t=\lim_{N\to\infty}\frac{t-t'}{N})\ .
\label{eq:679}
\eeq

Applying $\mP$ to both sides of Eq.~\eqref{eq:677}, and again dropping the tilde decoration to simplify the notation (though we continue to work in the interaction picture) we have $\hat\r(t') = \hat\mU_-(t,t')\r(t)$, so that Eq.~\eqref{eq:666} becomes:
\bes
\begin{align}
\label{eq:678a}
\bar\r(t) &= \mG(t,t_0)\bar\r(t_0) + \Sigma(t) \r(t)\\
\label{eq:678b}
\Sigma(t) &\equiv\a\int_{t_0}^t\mG(t,t')\bar\mL(t')\hat\mU_-(t,t')dt'\ .
\end{align}
\ees
Note that the superoperator $\Sigma(t)$ is not chronologically ordered since it contains both forward [via 
$\mG(t,t')$; recall Eq.~\eqref{eq:665}] and backward time propagation. For this reason we do not write $\Sigma(t,t_0)$, despite the dependence of $\Sigma(t)$ on $t_0$, since that notation is reserved for propagation from $t_0$ to $t$.\footnote{We could write $\Sigma_{t_0}(t)$ without danger of confusion, but this more cumbersome notation won't turn out to be particularly helpful.} Equation~\eqref{eq:678a} has removed the memory kernel and replaced it by (the even more complicated object) $\Sigma(t)$. However, in terms of the time-dependence of $\r$, it is time-local, i.e., depends only on $t$ (apart from the initial condition $t_0$). Next we solve this equation.

\subsubsection{Solving for the relevant part}
\label{sec:relevantpart}
Let us insert $I=\mP+\mQ$ into Eq.~\eqref{eq:678a}:
\bes
\begin{align}
\bar\r(t) &= \mG(t,t_0)\bar\r(t_0) + \Sigma(t) (\mP+\mQ)\r(t)\\
&\implies \bar\r(t) = \mG(t,t_0)\bar\r(t_0) + \Sigma(t)\hat\r(t) +\Sigma(t)\bar\r(t)\\
&\implies \left(I-\Sigma(t)\right)\bar\r(t) = \mG(t,t_0)\bar\r(t_0) + \Sigma(t)\hat\r(t) \ .
\end{align}
\ees
We can solve this for $\bar\r(t)$ provided $I-\Sigma(t)$ is invertible, i.e., provided $\Sigma(t)$ is not too close from identity. Since $\Sigma(t_0)=0$, we can conclude that $I-\Sigma(t)$ is invertible for sufficiently short evolution times. In addition, $\Sigma(t) = O(\a)$, so invertibility should also hold provided the system-bath coupling is sufficiently weak. Thus, from now we shall assume that $I-\Sigma(t)$ is indeed invertible, which is the only assumption we shall make to arrive at the TCL master equation. Then:
\beq
\bar\r(t) = \left(I-\Sigma(t)\right)^{-1}\mG(t,t_0)\bar\r(t_0) + \left(I-\Sigma(t)\right)^{-1}\Sigma(t)\hat\r(t) \ ,
\eeq
and substituting this into Eq.~\eqref{eq:664a} we find:
\bes
\begin{align}
\pt\hat\r(t) &= \a \hat \mL(t) \hat\r(t) + \a\hat\mL\left(I-\Sigma(t)\right)^{-1}\mG(t,t_0)\bar\r(t_0) + \a\hat\mL\left(I-\Sigma(t)\right)^{-1}\Sigma(t)\hat\r(t)\\
&= \a\hat\mL\left(I-\Sigma(t)\right)^{-1}\mG(t,t_0)\mQ\bar\r(t_0)+\a \hat \mL(t)\left[I+\left(I-\Sigma(t)\right)^{-1}\Sigma(t)\right]\mP\hat\r(t) \ ,
\end{align}
\ees
where in the second line we used the freedom to insert a $\mP$ and $\mQ$ in front of $\hat\r$ and $\bar\r$, respectively. Note that 
\beq
I+\left(I-\Sigma\right)^{-1}\Sigma = (I-\Sigma)^{-1}(I-\Sigma)+(I-\Sigma)^{-1}\Sigma =  (I-\Sigma)^{-1}(I-\Sigma+\Sigma) = (I-\Sigma)^{-1}\ .
\eeq
We have thus arrived at the \emph{time-convolutionless master equation} (TCL-ME): 
\beq
\pt\hat\r(t) =\mJ(t)\bar\r(t_0)+\mK(t)\hat\r(t)\ ,
\label{eq:TCL-ME}
\eeq
where
\bes
\label{eq:684}
\begin{align}
\label{eq:684a}
\mJ(t) &\equiv \a\hat\mL\left(I-\Sigma(t)\right)^{-1}\mG(t,t_0)\mQ\qquad \text{inhomogeneity}\ ,\\
\label{eq:684b}
\mK(t) &\equiv \a \hat \mL(t)(I-\Sigma(t))^{-1}\mP \qquad \text{TCL generator} \ .
\end{align}
\ees
The most salient feature of the TCL-ME is that (when the inhomogeneity vanishes, e.g., for factorized initial conditions) it is purely time-local, in stark contrast to the NZ-ME [Eq.~\eqref{eq:NZeq}].

\subsection{Perturbation theory}
Despite the formal appearance of the result we have found so far, it is a convenient starting point for perturbation theory.

\subsubsection{Matching powers of $\a$}
Let us write $(I-\Sigma(t))^{-1} = \sum_{n=0}^\infty\Sigma^n(t)$, i.e., as a geometric series. It follows from Eq.~\eqref{eq:678b} that $\Sigma^n(t) = \a^n \left(\int_{t_0}^t \cdots\right)^n$, so that after substitution into $\mK(t)$ [Eq.~\eqref{eq:684b}] we have a series expansion in powers of $\a$:
\bes
\label{eq:685}
\begin{align}
\label{eq:685a}
\mK(t) &= \a\hat\mL(t)\left(\sum_{n=0}^\infty \Sigma^n(t)\right)\mP \\
\label{eq:685b}
&= \sum_{n=1}^\infty\a^n\mK_n(t) \ ,
\end{align}
\ees
where we need to determine the operators $\mK_n(t)$. To do so we need to first expand $\Sigma(t)$ in powers of $\a$. It also follows from Eq.~\eqref{eq:678b} that the expansion must start from $\a^1$, since $\mG(t,t_0) = T_+ e^{\a\int_{t_0}^t\bar\mL(t')dt'}=I+O(\a)$ [Eq.~\eqref{eq:665}]:
\beq
\Sigma(t) = \sum_{m=1}^\infty\a^m\Sigma_m(t)\ .
\eeq
Substituting this expansion into Eq.~\eqref{eq:685a} yields a cumulant expansion:
\beq
\mK(t) = \a\hat\mL(t)\left[\sum_{n=0}^\infty \left(\sum_{m=1}^\infty\a^m\Sigma_m(t)\right)^n\right]\mP 
= \hat\mL(t)\left[ \a I + \sum_{m=1}^\infty \a^{m+1} \Sigma_m(t) + \sum_{m,m'=1}^\infty \a^{m+m'+1} \Sigma_m(t) \Sigma_{m'}(t) + \cdots \right]\mP
\eeq
Matching terms of equal power of $\a$ with Eq.~\eqref{eq:685b} yields, for the lowest four orders:
\bes
\begin{align}
\a^1:&\quad \mK_1(t) = \a\hat\mL(t)\mP = 0 \\
\label{eq:690b}
\a^2:&\quad \mK_2(t) = \a\hat\mL(t)\Sigma_1(t)\mP \qquad \text{Redfield equation}\\
\a^3:&\quad \mK_3(t) = \a\hat\mL(t)\left(\Sigma^2_1(t)+\Sigma_2(t)\right)\mP =0\\
\a^4:&\quad \mK_4(t) = \a\hat\mL(t)\left(\Sigma^3_1(t)+\{\Sigma_1(t),\Sigma_2(t)\}+\Sigma_3(t)\right)\mP \qquad \text{lowest order non-Markovian}
\end{align}
\ees
The vanishing of $\mK_1(t)$ is for the same reason as in Eq.~\eqref{eq:670}; that of $\mK_3(t)$ is explained below. First we need to explicitly find the lowest order $\Sigma_m(t)$'s. The expansions of $\mG(t,s) = T_+ e^{\a\int_{s}^t\bar\mL(t')dt'}$ and $\mU_-(t,t')=T_- e^{-\a \int_{t'}^t\mL(s)ds}$ [Eq.~\eqref{eq:679}] yield:
\bes
\begin{align}
\mG(t,t') &= I + \a\int_{t'}^t\bar\mL(s)ds + \frac{\a^2}{2!}T_+ \left(\a\int_{t'}^t\bar\mL(s)ds\right)^2+\cdots\\ 
\hat\mU_-(t,t') &= \mP \left[ I-\a\int_{t'}^t\mL(s)ds + \frac{\a^2}{2!}T_- \left(\int_{t'}^t\mL(s)ds\right)^2+\cdots\right]\ .
\end{align}
\ees
We can now collect equal powers of $\a$ in $\Sigma(t) = \a\int_{t_0}^t\mG(t,t')\bar\mL(t')\hat\mU_-(t,t')dt' = \sum_{m=1}^\infty\a^m\Sigma_m(t)$:
\bes
\begin{align}
\a^1:\quad \Sigma_1(t) &= \int_{t_0}^t \bar\mL(t')\mP dt' \\
\a^2:\quad \Sigma_2(t) &= -\int_{t_0}^t dt'\bar\mL(t')\mP\int_{t'}^t\mL(s)ds + \int_{t_0}^t dt'\left(\int_{t'}^t\bar\mL(s)ds\right) \bar\mL(t')\mP\\
&= \int_{t_0}^t ds\int_{t_0}^{s}dt'\left[\bar\mL(s)\bar\mL(t')\mP - \bar\mL(t')\mP\mL(s)\right] \ ,
\end{align}
\ees
where in the last line we switched the order of integration via $\int_{t_0}^t dt'\int_{t'}^t ds = \int_{t_0}^tds\int_{t_0}^{s}dt'$.

Therefore, using Eq.~\eqref{eq:690b}: 
\beq
\mK_2(t) = \hat\mL(t)\int_{t_0}^t \bar\mL(t') dt' \mP = \hat\mL(t)\int_{t_0}^t (I-\mP)\mL(t')dt' \mP=\hat\mL(t)\int_{t_0}^t \mL(t') dt' \mP\ ,
\label{eq:mK_2}
\eeq
where we again used $\mP\mL(t)\mP=0$, which we also use repeatedly below.

To calculate $\mK_3(t)$, first note that $\Sigma_1^2(t) = \int_{t_0}^t \int_{t_0}^t dt' dt'' \mQ\mL(t')\mP \mQ \mL(t'')\mP = 0$, since $\mP \mQ=0$. Second, note that $\mK_3(t)$ contains the term $\hat\mL(t)\bar\mL(t')[\mP\mL(s)\mP] = 0$. The final term it contains is $\hat\mL(t)\bar\mL(s)\bar\mL(t')\mP = \mP\mL(t)(I-\mP)\mL(s)(I-\mP)\mL(t')\mP = \mP\mL(t)\mL(s)\mL(t')\mP=0$, by Eq.~\eqref{eq:726}.
Therefore $\mK_3(t)=0$.


%

\subsubsection{The TCL-ME at second order yields the Redfield equation}
\label{sec:TCL2-Redfield}

Let us consider the lowest non-vanishing order of the TCL-ME, Eq.~\eqref{eq:TCL-ME}. At this order:
\beq
\pt\hat\r(t) =\a^2\mK_2(t)\hat\r(t)\ ,
\label{eq:TCL2}
\eeq
where we have assumed a factorized initial condition, so that the inhomogeneity vanishes. We already found $\mK_2(t)$ in Eq.~\eqref{eq:mK_2}, so what remains is to make it explicit using the definition of the projection to the relevant part:
\bes
\begin{align}
\mK_2(t)\hat\r(t) &= \int_{t_0}^t \mP\mL(t)\mL(t') dt' \mP \r(t) \\
&= \int_{t_0}^t dt'\Tr_B\left[-i\tilde{H}(t),\left[-i\tilde{H}(t'),\left(\Tr_B\r(t)\right)\ox\r_B\right]\right]\ox\r_B \ .
\end{align}
\ees
Thus, after applying $\Tr_B$ to both sides:
\beq
\pt\r_S(t) =-\a^2\int_{t_0}^t dt'\Tr_B\left[\tilde{H}(t),\left[\tilde{H}(t'),\r_S(t)\ox\r_B\right]\right]\ .
\label{eq:TCL2-final}
\eeq
This is the Redfield equation, Eq.~\eqref{eq:Redfield}. It is identical to the Born-Markov approximation [Eq.~\eqref{integEq2}], except for the finite upper limit of the integral. It is also nearly identical to the second order NZ-ME [Eq.~\eqref{eq:NZ2}], the only difference being the fact that, by construction, Eq.~\eqref{eq:TCL2-final} is time-local, in the sense that the argument of $\r_S$ is $t$ rather than $t'$. This is an important difference: whereas when we derived the RWA-LE we had to just assume that we can replace $t'$ by $t$ [in going from Eq.~\eqref{integEq2} to Eq.~\eqref{eq:Redfield}], here this is a systematic result of our derivation.


\subsection{Example: spin-boson model of a qubit in a cavity}

As an application of the TCL-ME we now consider a qubit in a cavity. This is an analytically solvable model subject to a simplifying assumption about the initial condition. As such, it will allow us to compare the predictions of the TCL to an exact result.

Consider as usual the total Hamiltonian $H = H_0 + H_{SB}$, where $H_0 = H_S + H_B$, with
\bes
\begin{align}
H_S &= \o_0 \ketb{1}{1} = \omega_0\sigma_+\sigma_- \ , \quad H_B = \sum_k \omega_k b_k^\dagger b_k =  \sum_k \omega_k n_k\ ,
\label{eq:defH0}\\
H_{SB} &= \sigma_+ \otimes B + \sigma_- \otimes B^\dagger \ , \quad B = \sum_k g_k b_k\ .
\end{align}
\ees
Here $\sigma_+ = \ketb{1}{0}$ and $\sigma_- = \ketb{0}{1}$ are the qubit raising and lowering operators, while $b_k$ and $b_k^\dgr$ are the bosonic lowering and raising operators for mode $k$, satisfying the canonical bosonic commutation relations $[b_{k},b_{k'}^{\dagger}]=\d_{kk'}$. The $g_k$ are coupling constants with dimensions of energy, and $n_k$ is the number operator for mode $k$. This Hamiltonian describes a qubit (the system) with ground state $\ket{0}$ of energy $0$ and excited state $\ket{1}$ with energy $\o_0$ coupled to a QHO bath. The coupling either excites the qubit and removes excitations from the bath, or \textit{v.v.} It will be useful to think of the bath in this case as electromagnetic modes of cavity. 

As usual, let us transform to the interaction picture wrt $H_0$, so that 
\bes
\begin{align}
\tilde{H}(t) &= U_0^\dgr(t)H_{SB}U_0(t) = \s_+(t)\ox B(t)+\text{h.c.}\\
\s_+(t) &= e^{i\o t}\s_+ \ , \qquad B(t) = \sum_k e^{-i\o_k t} g_k b_k\ .
\end{align}
\ees
Then the joint system-bath state $\ket{\phi(t)}$ (assume it is pure) in the interaction picture is given by $\ket{\phi(t)} = U(t)\ket{\phi(0)}$, where $U(t) = T_+ \exp\left(-i \int_0^t \tilde{H}(t')dt'\right)$.

This model is not analytically solvable in general. However, we shall assume that the cavity  supports at most one photon. Under this assumption the model becomes analytically solvable, as we shall see.

\subsubsection{Analytical solution in the $1$-excitation subspace}

\paragraph{The $1$-excitation subspace is conserved}

Let $\ket{0}_B$ denote the vacuum state of
the bath and consider the following joint system-bath states:
\bes
\begin{align}
\ket{\psi_0} &= \ket{0}_S \otimes \ket{v}_B,\\
\ket{\psi_1} &= \ket{1}_S \otimes \ket{v}_B,\\
\ket{\psi_k} &= \ket{0}_S \otimes \ket{k}_B,
\end{align}
\ees
where $\ket{k}=b_k^\dagger\ket{v}_B = \ket{0_1,\dots,0_{k-1},1_k,0_{k+1},\dots}$ denotes the state with one photon in mode $k$ ($\ket{k}$ is not to be confused with the usual labels for the computational basis of a qubit). Assume that the initial joint system-bath state contains at most a single excitation, i.e.:
\beq
\ket{\phi(0)} = c_{0}\ket{\psi_{0}} + c_1 (0) \ket{\psi_{1}} + \sum_{k} c_{k} (0) \ket{\psi _{k}}\ .
\label{eq:752}
\eeq
We wish to show that under the Hamiltonian above this remains true for all times, i.e., for all $t$:
\beq
\ket{\phi(t)} = c_0(t)\ket{\psi_{0}} + c_1(t) \ket{\psi_{1}} + \sum_{k} c_{k}(t) \ket{\psi _{k}}\ .
\label{eq:753}
\eeq
This is intuitively clear, since the system-bath coupling either excites the qubit while removing a photon, or \textit{v.v.}, and $H_0$ creates no new excitations. Nevertheless, let us give a formal argument for completeness.

Define the \emph{excitation number operator} by
\begin{align}
N = \sigma_+ \sigma_- \otimes \mathbb{I} + \mathbb{I}\otimes\sum_k b_k^\dagger b_k.
\end{align}
The name is well deserved since:
\bes
\begin{align}
\label{eq:755a}
N \ket{\psi_0} &= (\s^+\s_-\ket{0}) \ox \ket{v} + \ket{0}\ox \sum_k b_k^\dagger b_k \ket{v} = 0 \cdot \ket{\psi_0} \\
\label{eq:755b}
N \ket{\psi_1} &= (\s^+\s_-\ket{1}) \ox \ket{v} + \ket{1}\ox \sum_k b_k^\dagger b_k \ket{v} = \ket{1}\ox\ket{v} = 1  \cdot \ket{\psi_1} \\
\label{eq:755c}
N \ket{\psi_k} &= (\s^+\s_-\ket{0}) \ox \ket{v} + \ket{0}\ox \sum_{k'}b_{k'}^\dagger b_{k'} \ket{k} = \ket{0}\ox \sum_{k'}\d_{kk'}\ket{k} =  1 \cdot \ket{\psi_k} \ ,
\end{align}
\ees
where we used $\s_-\ket{0} = b_k \ket{v} = 0$. I.e., $N$ counts the number of excitations.

Next, note that the excitation number operator commutes with the total Hamiltonian $H$. That $[N,H_0]=0$ is obvious. As for $H_{SB}$, note first that $[\sigma_+ \sigma_-,\sigma_{\pm}] = \pm \s_\pm$, and $[n_{k'},b_k] = -b_k\d_{kk'}$, $[n_{k'},b^\dgr_k] = b^\dgr_k\d_{kk'}$. Therefore:
\bes
\begin{align}
[N,H_{SB}] &= [\sigma_+ \sigma_-,\sigma_{+}]\otimes B + [\sigma_+ \sigma_-,\sigma_{-}]\otimes B^{\dagger} + \sigma_+ \otimes \left[ \sum_k b_k^\dagger b_k, B \right] + \sigma_{-}\otimes\left[\sum_{k} b_{k}^{\dagger}b_{k},B^{\dagger}\right]\\
&= \sigma_{+}\otimes B - \sigma_{-}\otimes B^{\dagger} + \sigma_+ \otimes \left[ \sum_k n_k, \sum_{k'} g_{k'} b_{k'} \right] + \sigma_{-}\otimes\left[\sum_{k} n_{k},\sum_{k'} g_{k'} b_{k'} ^{\dagger}\right]\\
&= \sigma_{+}\otimes B - \sigma_{-}\otimes B^{\dagger} - \sigma_{+} \otimes B + \sigma_{-}\otimes B^{\dagger}= 0\ .
\end{align}
\ees
This means that $N$ is a conserved quantity, i.e., its eigenvalues are conserved under the evolution generated by $H$, or by $\tilde{H}(t)$ in the interaction picture. It also means that $H$ and $N$ share a common set of eigenvectors, which can be indexed using the eigenvalues of both $H$ and $N$. Eigenvectors with different eigenvalues of $N$ don't mix under the dynamics generated by $H$ or $\tilde{H}(t)$. This explains why, assuming the initial state is Eq.~\eqref{eq:752}, the state subsequently must be as in Eq.~\eqref{eq:753}: the state $\ket{\psi_{0}}$ has eigenvalue $0$ under $N$ [Eq.~\eqref{eq:755a}] and evolves as a separate one-dimensional subspace, and the states $\ket{\psi_{1}}$ and $\ket{\psi_{k}}$ have eigenvalue $1$ under $N$ [Eqs.~\eqref{eq:755b}, \eqref{eq:755c}], and also evolve as a separate two-dimensional subspace. $U(t)$ evolves each subspace separately and does not couple different subspaces labeled by different eigenvalues of $N$. 

Note that $i\partial_t\ket{\psi_{0}}=\tilde{H}(t)\ket{\psi_{0}} = 0$, which means, since $\ket{\psi_{0}}$ evolves separately, that $\ket{\psi_{0}(t)}=\ket{\psi_{0}(0)}$. Therefore $c_0(t)=c_0(0)$.

Even though the subspace spanned by $\{\ket{\psi_{0}},\ket{\psi_{1}},\ket{\psi_{k}}\}$ contains both $0$ and $1$ excitations, we loosely refer to it as the $1$-excitation subspace.

\paragraph{Schr\"{o}dinger dynamics in the $1$-excitation subspace}
Substituting Eq.~\eqref{eq:753} into the Schr\"{o}dinger equation, we have:
\bes
\begin{align}
i \partial_t{\ket{\phi(t)}} &= \dot{c}_1(t) \ket{\psi_{1}} + \sum_{k} \dot{c}_{k}(t) \ket{\psi _{k}} \\
&= \tilde{H}(t) \ket{\phi(t)} =  \left( \sigma_{+}(t)\ox B(t) + \sigma_{-}(t)\ox B^{\dagger}(t) \right)\left( c_{0}(0)\ket{\psi_{0} }+ c_{1}(t)\ket{\psi_{1}}  + \sum_{k} c_{k}(t) \ket{\psi_{k}}\right)\\
&= [\s_+(t)\ox B(t)]\sum_k c_k(t)\ket{0}\ox\ket{k}+c_1(t) [\s_-(t)\ox B^\dgr(t)](\ket{1}\ox\ket{v})\\
&= \ket{1} \otimes \sum_{k} g_{k} \ket{v} c_{k}(t)\ee^{\ii\omega_{0}t - \ii\omega_{k}t}+ c_{1}(t)\ket{0} \otimes \sum_{k} g^*_{k} \ket{k} \ee^{-\ii\omega_{0}t + \ii\omega_{k}t}  \\
&= \sum_{k} g_{k} c_{k}(t)\ee^{\ii(\omega_{0} -\omega_{k})t} \ket{\psi_{1}} + \sum_{k} g_{k}^{*} c_{1}(t)\ee^{-i(\omega_{0}-\omega_{k})t} \ket{\psi_{k}}  \ .
\end{align}
\ees
Multiplying by $\bra{\psi_1}$ and $\bra{\psi_k}$ gives us two coupled differential equations for the amplitudes $c_{1}$ and $c_{k}$:
\bes
\begin{align}
\dot{c_1}(t) &= -\ii \sum_k g_k c_{k}(t)\ee^{\ii(\omega_{0} -\omega_{k})t} \label{eq:c1t}\\
\dot{c_k}(t) &= -\ii g_{k}^{*} c_{1}(t)\ee^{-\ii(\omega_{0} -\omega_{k})t} \label{eq:ckt}\ .
\end{align}
\ees
Integrating Eq.~\eqref{eq:ckt} gives:
\begin{align}
c_k (t) - c_k(0)= -\ii\int_0^t dt' g_{k}^{*} c_{1}(t')\ee^{-\ii(\omega_{0} -\omega_{k})t'}\ .
\label{eq:759}
\end{align}
For simplicity, let us assume that the cavity starts in the vacuum state, i.e., $c_k(0)=0$. 
Then, after substituting the above into Eq.~\eqref{eq:c1t} we obtain:
\begin{align}
\label{eq:diff_c1}
\dot{c_1}(t) = - \int_{0}^{t}dt' f(t-t') c_1(t')\ ,
\end{align}
where the ``memory function" $f$ is:
\begin{align}
f(t) &= \sum_k \left|  g_k \right|^2 \ee^{\ii(\omega_{0}-\omega_k)t}= \int_0^\infty d\omega J(\omega )\ee^{\ii(\omega_0-\omega)t}\ ,
\label{eq:761}
\end{align}
where $J(\omega)$ is the bath spectral density, formally given as usual by $J(\omega )= \sum_k |g_k|^2 \d(\o-\o_k)$.
 
Since Eq.~\eqref{eq:diff_c1} is a convolution, it can be solved by means of a Laplace transform,
\beq
\textrm{Lap}[f] \equiv \hat{f}(s) \equiv \int_0^{\infty} dt\ e^{-st} f(t)\ ,
\label{eq:Lap}
\eeq 
since the Laplace transform of a convolution of two functions is the product of their Laplace transforms:
\beq
\textrm{Lap} [\int_{0}^{t}dt' f(t-t')c_1(t') ] = \hat{f}(s) \hat{c}_1(s)\ .
\label{eq:Lap-conv}
\eeq
Also, the Laplace transform of a derivative of a function $g(t)$ is
\beq
L[\frac{\partial g}{\partial t}] = s\tilde g(s) -g(0)\ .
\label{eq:Lap-deriv}
\eeq
Therefore
\begin{align}
\hat{c}_1(s) &= \frac{c_1(0)}{s + \hat{f}(s)}\ .
\label{eq:765}
\end{align}
This completes the analytical solution, since given the spectral density we can compute the excited state amplitude $c_1(t)$ by inverse Laplace transform of $\hat{c}_1(s)$, and from there the $c_k(t)$ amplitudes via Eq.~\eqref{eq:759}. Finally, recall that $c_0(t) = c_0(0)$. Eq.~\eqref{eq:753} then gives us the joint system-bath state in the $1$-excitation subspace.

\paragraph{System-only state}
With the analytical solution in hand for the joint system-bath state $\ket{\phi(t)}$, we can find the system-only state:
\begin{align}
\rho(t) = \Tr_B(\ketb{\phi(t)}{\phi(t)}) = 
\begin{pmatrix}
		\rho_{00}(t) & \rho_{01}(t) \\
		\rho_{01}^*(t) & \rho_{11}(t)
	\end{pmatrix} =
	\begin{pmatrix}
		1-\abs{c_1}^2 & c_0 c_1^* (t) \\
		c_0^*c_1(t) & \abs{c_1}^2
	\end{pmatrix}\ .
	\label{eq:766}
\end{align}
Note that normalization implies that $\abs{c_0}^2 + \abs{c_1(t)}^2 + \sum_k \abs{c_k (t)}^2=1$, so that $1-\abs{c_1}^2 \neq |c_0|^2$ (indeed, $c_0$ is constant), which is why  $\rho_{00}(t)\neq |c_0|^2$. To verify Eq.~\eqref{eq:766}, let us explicitly calculate the partial trace, recalling that $\ket{\phi(t)} = [c_0\ket{0} + c_1(t) \ket{1}]\ket{v} + \ket{0} \sum_{k} c_{k}(t) \ket{k}$:
\bes
\begin{align}
 \Tr_B(\ketb{\phi(t)}{\phi(t)}) &= \ave{v\ketb{\phi(t)}{\phi(t)}v} + \sum_k \ave{k\ketb{\phi(t)}{\phi(t)}k} \\
 &=[c_0\ket{0} + c_1(t) \ket{1}][c^*_0\bra{0} + c^*_1(t) \bra{1}] + \ketb{0}{0} \sum_{k',k''} c_{k'}\d_{kk'}(t)c^*_{k''}(t)\d_{k''k}\\
 &= [|c_0|^2 + \sum_k |c_k(t)|^2]  \ketb{0}{0} + c_0 c_1^* (t)\ketb{0}{1} + c_0^*c_1(t)\ketb{1}{0} + |c_1(t)|^2\ketb{1}{1}\ .
\end{align}
\ees

\paragraph{Exact master equation}
To connect the analytical solution to the master equation framework, let us now find the exact master equation satisfied by $\r(t)$. To do so, we differentiate Eq.~\eqref{eq:766}, to find:
\begin{align}
\dot{\rho} = 
	\begin{pmatrix}
		-\partial_t\abs{c_1}^2 & c_0\dot{c}_1^*(t) \\ 
		c_0^*\dot{c}_1 (t)  & \partial_t\abs{c_1}^2
	\end{pmatrix}\ . 
	\label{eq:dif_rhos}
\end{align}
The system-bath Hamiltonian describes an excitation and relaxation process. Therefore, recalling Eq.~\eqref{eq:282}, a reasonable ansatz for the exact master equation in the interaction picture is of the form
\begin{align}
\label{eq:master_eq_two_level}
\dot{\rho} &= -\frac{i}{2}S(t)[\sigma_+\sigma_-, \rho (t)] + \gamma (t) \left( \sigma_-\rho(t) \sigma_+ - \frac{1}{2}\{\sigma_+\sigma-,\rho (t)\} \right)\ ,
\end{align}
where the first term represents the Lamb shift and the second term represents relaxation. We will shortly verify this ansatz. Meanwhile, note that unlike Eq.~\eqref{eq:282}, the relaxation rate $\gamma$ is now time-dependent. This is an important difference, since there is now no guarantee that the rate is always positive and finite. 

Let us now check and confirm the ansatz. 
Note that
\bes
\label{eq:771}
\begin{align}
\sigma_- \rho \sigma_+ &= \begin{pmatrix} \r_{11} &0 \\ 0 & 0 \end{pmatrix} = \begin{pmatrix} |c_1(t)|^2 &0 \\ 0 & 0 \end{pmatrix}  \\
[\sigma_+ \sigma_- ,\rho] &= \begin{pmatrix} 0 & -\r_{01} \\ \r_{10} & 0 \end{pmatrix} = \begin{pmatrix} 0 & -c_0 c_1^*(t) \\ c_0^* c_1(t) & 0 \end{pmatrix}\\
\{\sigma_+ \sigma_-,\r\} &= \begin{pmatrix} 0 & \r_{01} \\ \r_{10} & 2\r_{11} \end{pmatrix} = \begin{pmatrix} 0 & c_0 c_1^*(t) \\ c_0^* c_1(t) & 2|c_1(t)|^2 \end{pmatrix}\ ,
\end{align}
\ees
where we used Eq.~\eqref{eq:766} for the second equality in each line. If Eq.~\eqref{eq:master_eq_two_level} holds then it must be true, using the first equality in each line of Eq.~\eqref{eq:771}, that:
\beq
\label{eq:772}
\dot{\r} = \begin{pmatrix} \g(t)|c_1|^2  & \left(\frac{i}{2}S(t)-\frac{1}{2}\g(t)\right)c_0c_1^*(t) \\ \left(-\frac{i}{2}S(t)-\frac{1}{2}\g(t)\right)c_0^*c_1(t) & -\g(t)|c_1|^2 \end{pmatrix}\ .
\eeq
Comparing the off-diagonal elements of Eqs.~\eqref{eq:dif_rhos} and Eq.~\eqref{eq:772} we find that they agree provided $\dot{c}_1 = -\frac{1}{2}c_1(t)[\g(t)+iS(t)]$, i.e.:
\bes
\label{eq:773}
\begin{align}
\label{eq:773a}
S(t) &= -2\Im\left(\frac{\dot{c}_1(t)}{c_1(t)}\right)\\
\gamma (t) &= -2\Re\left(\frac{\dot{c}_1(t)}{c_1(t)}\right)\ .
\label{eq:773b}
\end{align}
\ees
We have thus identified the Lamb shift rate and relaxation rate from the exact master equation~\eqref{eq:master_eq_two_level}.

But, to ensure that the ansatz is correct we still need to confirm that this identification also works for the diagonal elements. Let $c_1(t) = r(t)e^{i\t(t)}$. Then $\partial_t |c_1|^2 = 2\dot{r}r$, and also $\dot{c}_1 = \dot{r}e^{i\t(t)}+i\dot{\t}c_1$, which implies $\dot{c}_1/c_1 = \dot{r}/r + i\dot{\t}$, i.e., $\Re(\dot{c}_1/c_1) = \dot{r}/r$. Therefore, if Eq.~\eqref{eq:773b} holds then:
\beq
\g(t)|c_1(t)|^2 = -2(\dot{r}/r)r^2 = -2\dot{r}r = -\partial_t|c_1(t)|^2\ ,
\eeq
as required if Eqs.~\eqref{eq:dif_rhos} and Eq.~\eqref{eq:772} are to agree.

\paragraph{Connection with the TCL formalism}

Note that Eq.~\eqref{eq:master_eq_two_level} is in the form of the TCL-ME, since it is \emph{time-local}. Namely, we can introduce a time-local generator and rewrite it as
\begin{align}
\dot{\rho} &= \mathcal{K}_S (t) \rho (t) = \Tr_B\left[ \mathcal{K}(t)\rho(t)\otimes\rho_B \right] \ ,
\end{align}
where $\mathcal{K}(t)$ is the TCL generator [Eq.~\eqref{eq:TCL-ME}], which can be computed directly from the time-local generator $\mathcal{K}_S (t)$, which we identify here as 
$\mathcal{K}_S (t) = -\frac{i}{2}S(t)[\sigma_+\sigma_-, \cdot] + \gamma (t) \left( \sigma_\cdot \sigma_+ - \frac{1}{2}\{\sigma_+\sigma-,\cdot\} \right)$. Next, recall that $\mK(t) = \sum_{n=1}^\infty\a^{2n}\mK_n(t)$ [Eq.~\eqref{eq:685b}, where we have shifted the bath operators so all odd orders vanish]. Correspondingly, $\mK_S(t) = \sum_{n=1}^\infty\a^{2n}\mK_n(t)$, and therefore also
\beq
\g(t) = \sum_{n=1}^\infty\a^{2n}\g_{2n}(t)\ , \quad S(t) = \sum_{n=1}^\infty\a^{2n}S_{2n}(t)\ .
\label{eq:776}
\eeq

To make the connection between the exact solution of the qubit-in-cavity model and this perturbative expansion of the TCL-ME, recall that we started from the Liouville-von-Neumann equation in the form $\partial_t \tr = -i \a [\tilde{H}(t),\tr(t)] \equiv \a \mc{L}\tr(t)$ [Eq.~\eqref{eq:728}]. This means that if we were to introduce the dimensionless parameter $\alpha$ into the formulation of the qubit-in-cavity model, it would multiply the coupling constants $g_k$, and hence we would need to replace $f(t)$ with $\alpha^2 f(t)$ in Eq.~\eqref{eq:761}. Then Eq.~\eqref{eq:diff_c1} is replaced by
\begin{align}
\label{eq:diff_c1-2}
\dot{c_1}(t) = - \alpha^2 \int_{0}^{t}dt' f(t-t') c_1(t')\ .
\end{align} 

If we consider the Laplace transform solution for $c_1(t)$, given by the inverse Laplace transform of Eq.~\eqref{eq:765}, then to lowest order in $\alpha$ we simply have $c_1(t)=c_1(0)$. The reason is that the inverse Laplace transform of $c_1(0)/s$ [where have taken $\alpha\to 0$ in Eq.~\eqref{eq:765}] is $c_1(0)$. Therefore to lowest order in $\alpha$, Eq.~\eqref{eq:diff_c1-2} yields $\dot{c_1}(t) = - \alpha^2 c_1(0) \int_{0}^{t}dt' f(t-t') + O(\alpha^3)$, and it follows from Eq.~\eqref{eq:773} that
\bes
\label{eq:777}
\begin{align}
\label{eq:777a}
S_2(t) &= 2\Im\left(\int_{0}^{t}dt' f(t-t')\right)\\
\gamma_2 (t) &= 2\Re\left(\int_{0}^{t}dt' f(t-t')\right)\ .
\label{eq:777b}
\end{align}
\ees

\subsection{Jaynes-Cummings model on resonance}

Having derived the exact master equation for a qubit in a cavity, we can now apply it to compare the predictions of various master equations to the exact solution. To do so we need to specify the bath spectral density $J(\o)$. We will consider the Jaynes-Cummings model on resonance, a model in which the cavity supports a single mode with a frequency $\o_0$ equal to that of the qubit. First we consider the case where the cavity is completely isolated from the external world, then we consider the case where the cavity is coupled to the external electromagnetic field.

\subsubsection{Isolated cavity}
Assume that the cavity has opaque walls that act as infinitely tall barriers, so that no radiation can leak into or out of the cavity. In this case, with $\omega_0$ being the qubit transition frequency, since the cavity only has a single mode, at this frequency, the  spectral density becomes
\beq
J(\omega )= \sum_k |g_k|^2 \d(\o-\o_k) \mapsto |g|^2 \delta(\omega-\omega_0)\ .
\eeq
Therefore the memory function $f(t)$ [Eq.~\eqref{eq:761}] is
\beq
f(t)=\int_{0}^\infty d\omega J(\omega)e^{i(\omega_0-\omega)t}=|g|^2\ ,
\eeq
and the amplitude of the qubit excited state, $c_1(t)$, then satisfies
\beq
\dot{c_1}(t)=-\int_0^t ds f(t-t')c_1(s)=-|g|^2\int_0^t dt' c_1(t')\ .
\eeq
Rather than using the Laplace transform solution, it is simpler to differentiate both sides to get
\beq
\ddot{c_1}(t)=-|g|^2 c_1(t)\ .
\eeq
The solution of this differential equation is 
\beq
c_1(t)=A \cos(|g|t)+B \sin(|g|t)\ ,
\eeq
where $A$ and $B$ are constants. Thus, the population of the excited state is $\rho_{11}(t)=|c_1(t)|^2$, which oscillates with a period given by $\pi/|g|$, as expected from a qubit coupled to an oscillator resonant with it.

\subsubsection{Leaky cavity}
Next we consider the case where, instead of opaque walls, the cavity allows photons to leak out or in. It can be shown that in this case 
the memory function is 
\beq
f(t) = \frac{1}{2\tau_M \tau_B}e^{-t/\tau_B}
\eeq
where $\tau_M$ is a Markovian timescale whose exact meaning will become apparent below, and $\tau_B$ is the usual bath correlation time (decay time of $\ave{B(t)B(0)}_B$). Moreover, it can be shown that $\alpha^2 = \tau_B/\tau_M$, where $\alpha$ is the dimensionless system-bath coupling strength we have used as a dimensionless prefactor for $H_{SB}$ in the TCL-ME.

The excited state amplitude $c_1(t)$ then satisfies
\begin{align}
\dot{c_1}&=-\int_0^t dt' f(t-t')c_1(t') = -\frac{1}{2\tau_M \tau_B}\int_0^t dt' e^{-(t-t')/\tau_B}c_1(t') \ .
\end{align}
It is again simpler to differentiate once more rather than use the Laplace transform:
\beq
\ddot{c}_1 +\frac{1}{\tau_B} \dot{c}_1 +\frac{1}{2\tau_M \tau_B} c_1 = 0\ ,
\eeq
a simple second order differential equation. Its solution is:
\beq
c_1(t)=c_1(0) e^{-\frac{t}{2\tau_B}}
\left[
\cosh\left(
\frac{t\d}{2}\right)
+\frac{1}{\tau_B\d} \sinh\left(
\frac{t\d}{2}\right)
\right]\ ,
\label{eq:c1-TCL}
\eeq
where 
\beq
\d=\sqrt{\frac{1}{\tau_B^2}-\frac{2}{\tau_M\tau_B}} = \frac{1}{\tau_B}\sqrt{1-2\alpha^2} \ .
\eeq
The excited state population is $\rho_{11}(t)=|c_1(t)|^2$. 
We thus have two distinct cases:

\paragraph{Weak coupling}
This is the case when $\alpha^2 = \tau_B/\tau_M \leq 1/2$, so that  $\d \in \Re$. Then $S(t) = 0$ [Eq.~\eqref{eq:773a}] and Eqs.~\eqref{eq:773b} and~\eqref{eq:c1-TCL} yield:
\bes
\begin{align}
\g(t) &= \frac{\frac{2}{\tau_M\tau_B}\cosh\left(\frac{t\d}{2}\right)}{\d \cosh\left(\frac{t\d}{2}\right)+\frac{1}{\tau_B} \sinh\left(\frac{t\d}{2}\right)}\\
\rho_{11}(t)&=\r_{11}(0) e^{-\frac{t}{\tau_B}}
\left|\left[
\cosh\left(
\frac{t\d}{2}\right)
+\frac{1}{\tau_B\d} \sinh\left(
\frac{t\d}{2}\right)
\right]\right|^2\ .
\label{eq:788}
\end{align}
\ees
In this case the population decays, i.e., the dynamics is Markovian-like.

\paragraph{Strong coupling}
This is the case when $\alpha^2 = \tau_B/\tau_M > 1/2$, so that  $\d \in \Im$. Then $S(t) \neq 0$ [Eq.~\eqref{eq:773a}] and Eq.~\eqref{eq:773b}  and~\eqref{eq:c1-TCL} yield:
\bes
\label{eq:789}
\begin{align}
\label{eq:789a}
\g(t) &= \frac{\frac{2}{\tau_M\tau_B}\cos\left(\frac{t|\d|}{2}\right)}{\d \cos\left(\frac{t|\d|}{2}\right)+\frac{1}{\tau_B} \sin\left(\frac{t|\d|}{2}\right)}\\
\label{eq:789b}
\rho_{11}(t)&=\r_{11}(0) e^{-\frac{t}{\tau_B}}
\left|\left[
\cos\left(
\frac{t|\d|}{2}\right)
+\frac{1}{\tau_B|\d|} \sin\left(
\frac{t|\d|}{2}\right)
\right]\right|^2\ .
\end{align}
\ees
In this case the population exhibits damped oscillations, i.e., the dynamics is non-Markovian. 

With this analytical solution in hand, we are ready to compare to the predictions of the TCL-ME.

\subsubsection{Comparison to TCL-ME, Markov limit, and NZ-ME}

Recall that the TCL-ME expansion is, in the present case, equivalent to an expansion of $\g(t)$ and $S(t)$ in powers of $\alpha$, as in Eq~\eqref{eq:776}. We can thus obtain the $\g_{2n}(t)$ terms for the weak coupling case by expanding $\g(t)$ of Eq.~\eqref{eq:788} in powers of $\alpha$, and similarly for the strong coupling case. 

We can also use Eq.~\eqref{eq:777b}, so that:
\beq
\gamma_2 (t) = 2\Re\left(\int_{0}^{t}dt' f(t-t')\right) = \Re\left(\int_{0}^{t}dt'\frac{1}{\tau_M \tau_B}e^{-(t-t')/\tau_B}\right) =
\frac{1}{\tau_M}\left(
1-e^{-t/\tau_B}\right)\ ,
\label{eq:791}
\eeq
which is clearly an example of the weak coupling case (as expected for a low-order-in-$\alpha$ expansion) since the rate exhibits no oscillations. Note that $\gamma_2 (t)$ has a rise time of $\tau_B$ to its asymptotic value of $1/\tau_M$.

Recall that the TCL-2 result is exactly the Redfield equation, as we showed in Sec.~\ref{sec:TCL2-Redfield}. Moreover, if we take the upper limit of the integral to infinity we have the Markov limit. Therefore:
\beq
\g_2(\infty) = 1/\tau_M \equiv \g_0 \ ,
\eeq
which explains the subscript $M$ notation we used all along in this example. We already know the solution in the Markovian limit: $\r_{11}(t) = \r_{11}(0)e^{-t/\tau_M}$.

By doing the expansion to fourth order in $\alpha$ we find:\footnote{Note that the result given in the book~\cite{Breuer:book} differs from Ref.~\cite{Breuer:99}[Eq.~(69)]; the latter is the correct one.}
\beq
\gamma_4(t)=\frac{1}{\tau_M} \left(1-e^{-t/\tau_B}+\frac{\tau_M}{\tau_B}[
\sinh(t/\tau_B) -t/\tau_B]e^{-t/\tau_B}\right)\ ,
\eeq
which has the limiting behavior $\gamma_4(\infty) = \frac{1}{\tau_M}+\frac{1}{2\tau_B} > \g_2(\infty)$.

What about the NZ-ME? It can be shown that to second order in $\alpha$, the NZ-ME yields exactly the same result as TCL-2, except that two changes are needed: (1) $e^{-t/\tau_B}$ is replaced by $e^{-t/(2\tau_B)}$ in the results for $\r_{11}(t)$, and (2) $\d$ is replaced by 
\beq
\d' = \sqrt{\frac{1}{\tau_B^2}-\frac{4}{\tau_M\tau_B}} = \frac{1}{\tau_B}\sqrt{1-4\alpha^2} \ .
\eeq

\begin{figure}[hb]
	\includegraphics[scale=0.5]{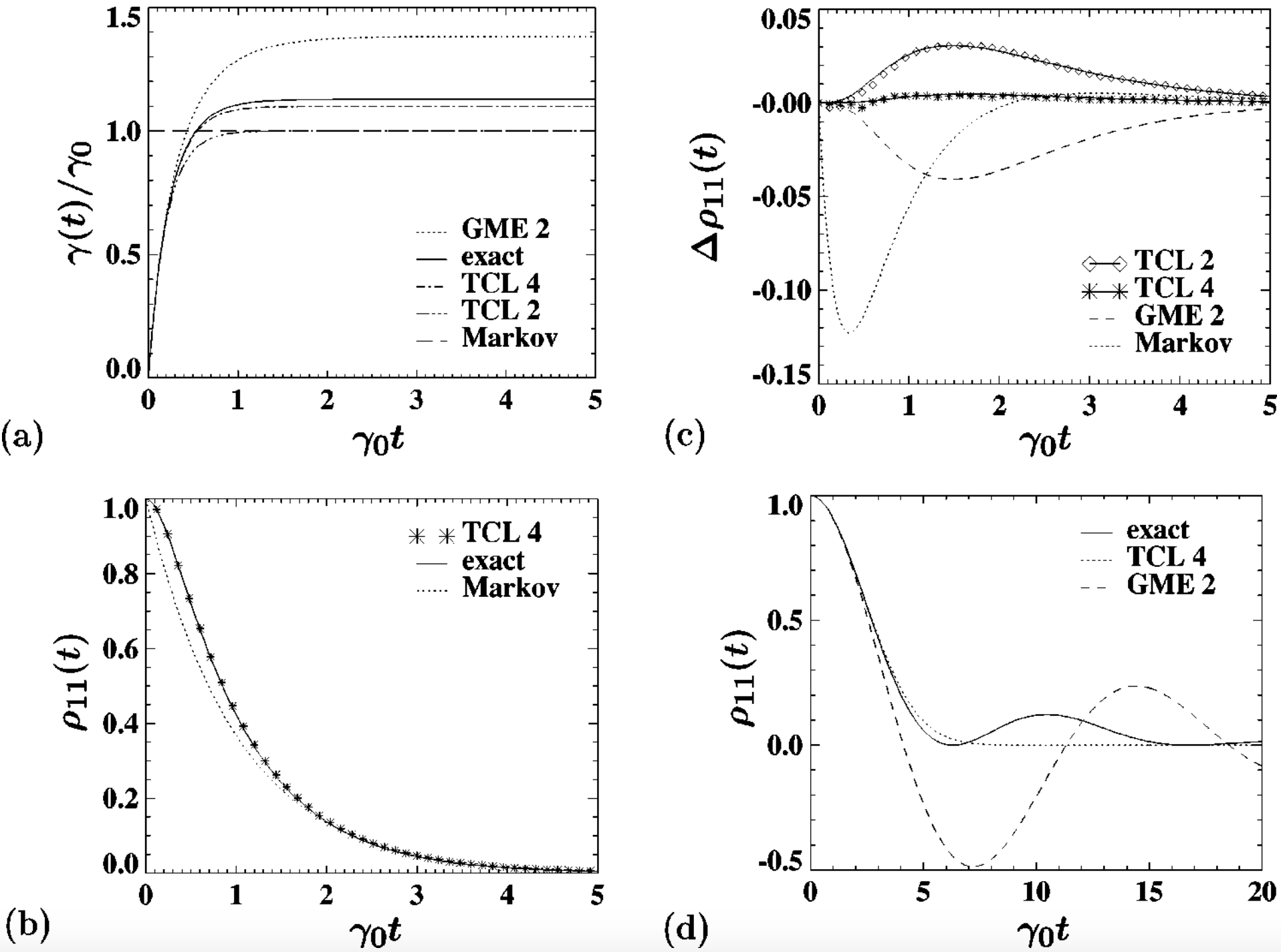}
	\caption{Damped Jaynes-Cummings model on resonance. Exact solution (exact), TCL-ME to second (TCL 2)
and fourth order (TCL 4), NZ-ME to second order (GME 2), and the RWA-LE (Markov).
(a) Decay rate of the excited state population, (b) the population of the excited state, including a stochastic simulation of the TCL-ME with $10^5$ realizations (diamonds for TCL 2and stars for TCL 4), and (c) deviation of the approximate solutions from the exact result, for
$1/\gamma_0 \equiv \tau_M=5\tau_B$ (weak coupling). (d) Population of the excited state for $1/\gamma_0 \equiv \tau_M=0.2\tau_B$ (strong coupling).
Source: Ref.~\cite{Breuer:99}.}
\label{fig:JCcomparison}
\end{figure}

Figure~\ref{fig:JCcomparison} shows these various results in terms of the deviation of the excited state population from the exact result. Focusing on panels (a)-(c), which report results for the weak coupling case, it illustrates a number of points:
\begin{itemize}
\item All approximations, except Markov, are good for very short times (shorter than $\tau_B$).
\item The Markov approximation initially overestimates the depopulation of the excited state, the underestimates it for longer times. It is a particularly poor approximation for times shorter than $\tau_B$, which is the rise-time of the curves in panel (a). 
\item TCL-2 (Redfield) underestimates the depopulation of the excited state for intermediate times.
\item TCL-2 converges to Markov in the long-time limit.
\item NZ-2 overestimates the depopulation of the excited state for intermediate times.
\item TCL-4 is a better approximation than both the TCL-2 and the Markov approximation. Its rate $\g_4(t)$ goes above the Markov rate, as expected since $\gamma_4(\infty) - \frac{1} = \frac{1}{2\tau_B}$.
\end{itemize}

\subsubsection{Breakdown of the NZ-ME and TCL-ME expansions for strong coupling}

What about the strong coupling case? The exact result is shown in Fig.~\ref{fig:JCcomparison}(d), and exhibits damped oscillations. The second order NZ-ME also exhibits damped oscillations, but the excited state population becomes negative! This result is physically non-sensical and is a clear example of violation of complete positivity of the evolution map. The TCL-4 approximation is good for short times but fails to capture the oscillations. To understand this let us take a step back and recall that the TCL-ME requires the invertibility of the operator $I-\Sigma$. The present example serves to illustrate how this invertibility condition can be violated, and how therefore the TCL can break down.

Assume that for different initial conditions $\{\r_{11}^{(1)}(0),\r_{11}^{(2)}(0),\r_{11}^{(3)}(0),\dots\}$ there is a common time $t_0$ at which the exact solution gives $\r_{11}^{\textrm{exact}}(t_0)=0$. This is indeed the case shown in Fig.~\ref{fig:JCcomparison}(d), as is easy to verify from Eq.~\eqref{eq:789b}: solving for its roots we have:
\beq
\tan(|\d|t_n/2) = -|\d|\tau_B\ \Longrightarrow \ t_n = \frac{2}{|\d|}(\arctan(|\d|\tau_B)+n\pi)\ \Longrightarrow \ t_0 = \min_n t_n \ ,
\eeq
where $n$ runs over the integers. Now, since the TCL-ME is time-\emph{local}, i.e., it only ``knows" about the current time $t$, this means that for $t\geq t_0$ it is impossible to invert the evolution back to the initial condition, as this information is lost in a time-local description. We therefore expect the TCL-ME to give unreliable results when the exact solution predicts a vanishing population. This is precisely what is seen in Fig.~\ref{fig:JCcomparison}(d).

Mathematically, we can see this another way. Eq.~\eqref{eq:789a} tells us that $\g(t)$ diverges at the same times $t=t_n$ where $\r_{11}^{\textrm{exact}}(t)=0$. More fundamentally, this is because $c_1(t)=0$ implies via Eq.~\eqref{eq:773b} that $\g(t)$ diverges (unless $\dot{c}_1(t)=0$ at the same time). But if $\g(t)$ diverges then it does not have a Taylor series, so the various $\g_{2n}(t)$ are undefined, and the TCL-ME expansion does not exist.


\section{Post Markovian Master Equation}

We have seen a variety of approaches to describing the reduced system dynamics via master equations, ranging from the exact Nakajima-Zwanzig equation, via the time-convolutionless, to the Markovian limit. In this section we will review a master equation approach that naturally interpolates between the Markovian limit and the limit of exact dynamics, as expressed in terms of CP maps via the Kraus OSR~\cite{ShabaniLidar:05}. The key idea will be to understand both limits as arising from a non-selective measurement process of the bath state. The exact dynamics corresponds to a single measurement at the final time, whereas Markovian dynamics corresponds to the limit of infinitely many measurements. The interpolation will thus limit the number of measurements in order to arrive at an non-Markovian approximation.

\subsection{Measurement interpretation of the Kraus OSR and the Lindblad equation}

Consider the usual setup of open system evolution, with the initial state $\r(0)=\r_S(0)\ox\r_B$ evolving under a joint unitary $U$ to the final state $\r(t)=U(t)\r(0)U^\dgr(t)$. The reduced system state at the final time is $\r_S(t) = \Tr_B[\r(t)]$. We wish to show that this can be understood equivalently as a projective measurement of the bath at the final time, as depicted schematically in Fig.~\ref{Fig:time_line_1}.

Suppose that we measure the bath at the final time $t$ via the complete set of projection operators $\{P_k = \ketb{k}{k}\}$. Thus, if outcome $k$ was observed, then the joint state transforms as 
\beq
\r(t) \overset{P_k}{\longmapsto} \frac{(I_S\ox P_k)\r(t)(I_S\ox P_k)}{p_k} \equiv \r^{(k)}(t)
\eeq
with probability $p_k = \Tr[(I_S\ox P_k)\r(t)]$. The reduced system state for this outcome is
\beq
\r_S^{(k)}(t) = \Tr_B[\r^{(k)}(t)] = \sum_{k'}\bra{k'}\r^{(k)}(t)\ket{k'} = \frac{\bra{k}\r(t)\ket{k}}{p_k} \ .
\eeq
Assuming we do not keep track of the measurement outcome, i.e., the measurement is non-selective, the final system state is the mixed state ensemble [recall Eq.~\eqref{eq:mixedens}] $\{p_k,\r_S^{(k)}(t)\}$, i.e.,
\beq
\r_S(t) = \sum_k p_k \r_S^{(k)}(t) = \sum_k \bra{k}\r(t)\ket{k} = \Tr_B[\r(t)]\ ,
\eeq
i.e., exactly the Kraus OSR result. Thus we can indeed understand the Kraus OSR as joint unitary evolution followed by a single non-selective measurement of the bath at the final time $t$.

In other words, we have shown that the following two evolutions are equivalent:
\bes
\begin{align}
&\r(0) \overset{U(t)}{\longmapsto} \r(t) \overset{\Tr_B}{\longmapsto} \r_S(t) \\
&\r(0) \overset{U(t)}{\longmapsto} \r(t) \overset{P_B}{\longmapsto} \r^{(k)}(t) \overset{\Tr_B}{\longmapsto} \r_S^{(k)}(t) \overset{\text{non-selective}}{\longmapsto} \r_S(t) \ ,
\end{align}
\ees
where $P_B$ denotes a projective measurement of the bath with projectors $\{P_k\}$.

\begin{figure}[!ht]
\includegraphics[width=80mm,height=30mm]{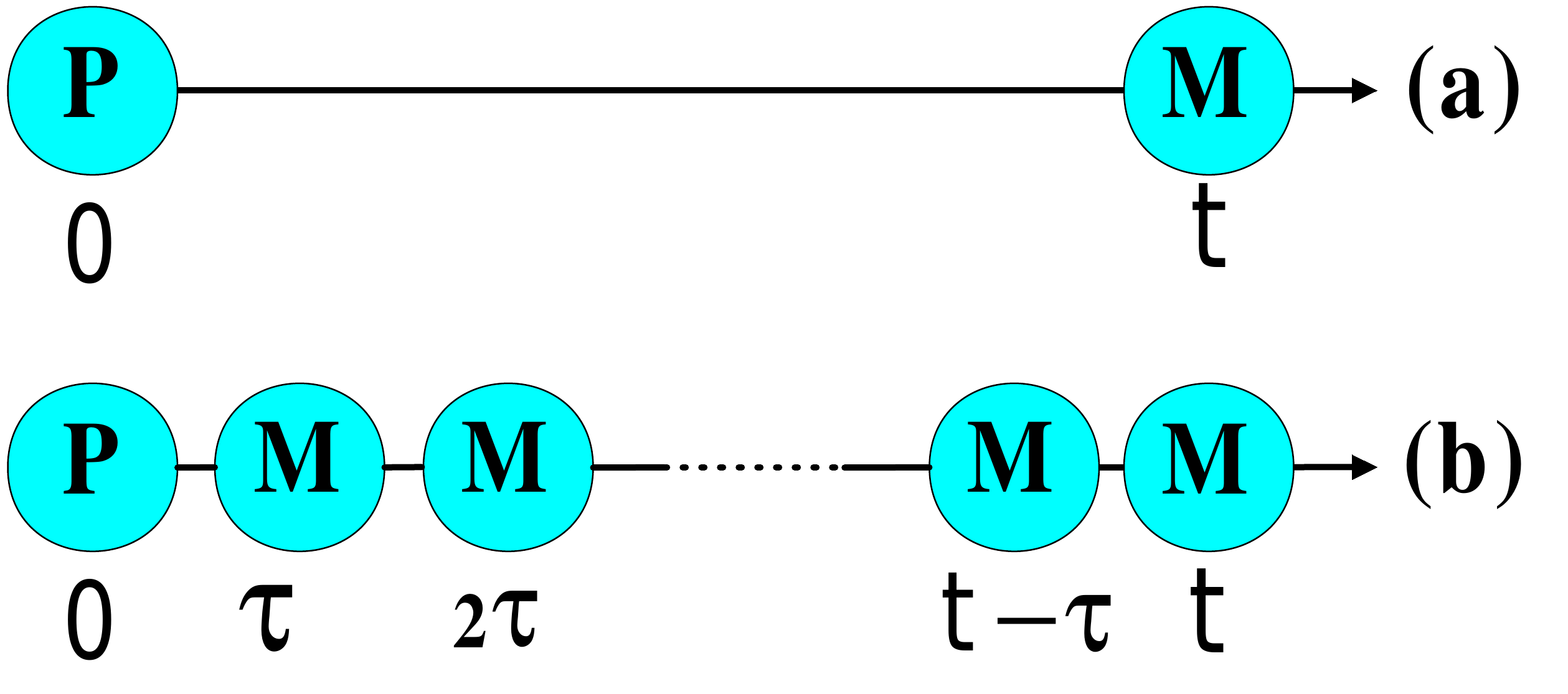}
\caption{Measurement approach to open system dynamics. P=preparation, M= measurement, time proceeds from left to right. (a) Exact Kraus operator sum representation. (b) Markovian approximation.}
\label{Fig:time_line_1} 
\end{figure}

For the Lindblad equation, we have already shown in Sec.~\ref{sec:LE-derivation} that it can be understood as arising from a sequence of infinitesimal CP maps. More specifically, we showed that the LE
\beq
\dot{\rho}_S=-i[H,\r_S]+\sum_{\alpha\geq 1} L_\alpha\rho L_\alpha^\dag - \frac{1}{2}\{L_\alpha^\dag L_\alpha ,\rho_S \}
\eeq
is equivalent to the sequence of CP maps
\beq
\r_S(t+\tau) = \sum_{\a\geq 0} K_\a\r_S(t) K_\a^\dgr \ ,
\eeq
where $\tau\to 0$ and 
\bes
\begin{align}
K_0 &= I+(-iH-\frac{1}{2}\sum_{\a\geq 1}L_\a^\dgr L_\a)\tau \qquad \text{conditional evolution}\\
K_\a &= L_\a \sqrt{\tau}\ , \quad (\a\geq 1) \qquad \text{jumps} \ .
\end{align}
\ees
Since we have just shown that each CP map can be understood as a projective measurement of the bath, we see that the LE can also be understood as representing an infinite sequence of such measurements, taking place in intervals of length $\tau$. Since each such measurement disentangles the system and bath state, it can be viewed as a preparation step of a new product state between the system and bath; see Fig.~\ref{Fig:time_line_1}.

\subsection{Interpolating between the two limits: derivation of the PMME}

Having seen that the exact Kraus OSR and the fully Markovian LE are two measurement limits, it is natural to consider an intermediate scenario, of a finite number of intermediate measurements between the initial and final times. Consider the simplest case, of a single projective measure of the bath at a random time $t^\prime\in(0,t)$, and note that the more measurements we introduce, the more Markovian the evolution becomes. We assume that the same CPTP map $\Lambda$ governs the evolution in the period $[0,t^\prime)$ and $(t^\prime,t]$, as shown in Fig.~\ref{Fig:time_line_2}.
\begin{figure}[b]
\includegraphics[width=80mm]{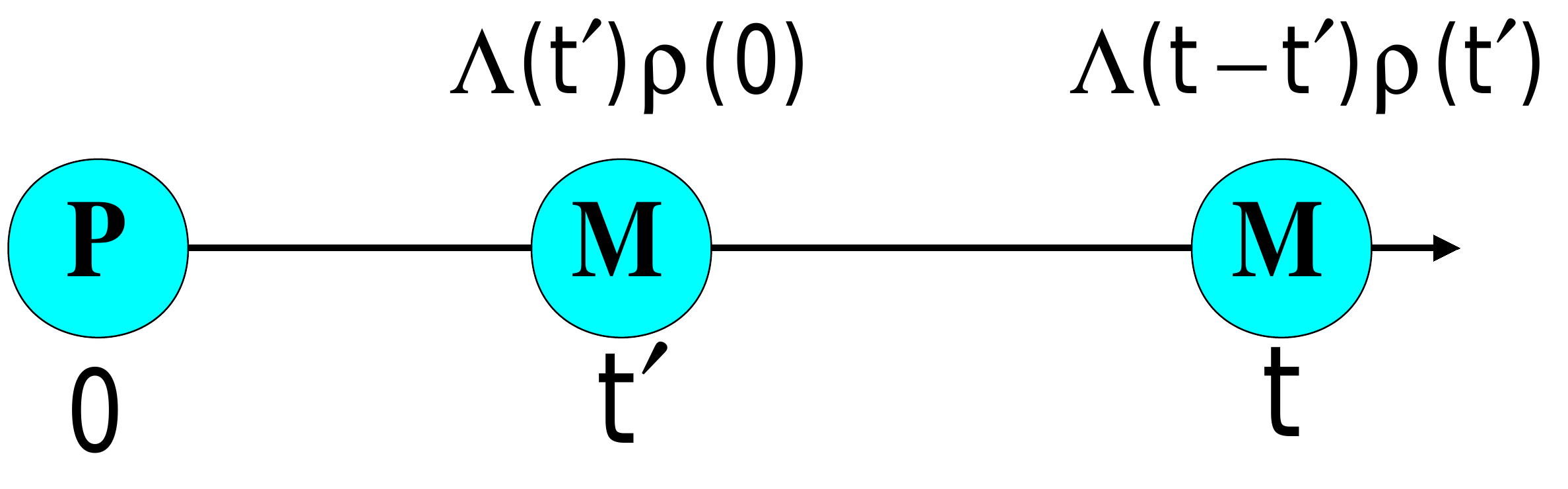}
\caption{\label{Fig:time_line_2} A single projective measurement of the bath is preceded and followed by a CPTP map $\Lambda$. For that specific outcome $\rho(t')=\Lambda(t^\prime)\rho(0)$ and $\rho(t)=\Lambda(t-t^\prime)\rho(t^\prime)$. To account for all possible outcomes each such trajectory is weighted as in Eq.~\eqref{eq:PMME_first_form}.}
\end{figure}
The measurement produces a random system state $\r(t')$ (where we from here on we drop the subscript $S$ since we are interested only in the system dynamics), which is then propagated to $\r(t)$, i.e., $\r(t) = \Lambda(t-t')\r(t')$. But since we do not know the outcome, nor the time $t'$, we introduce a weighting function $k(t-t',t)$ (the choice to make the argument depend on the remaining time interval $t-t'$ rather than $t'$ is for later convenience). The final state $\rho(t)$ can then be represented in the following form:
\beq
\label{eq:PMME_first_form}
\rho(t)=\int_0^t\underbrace{k(t-t^\prime,t)}_{\text{weight (kernel)}}\Lambda(t-t^\prime)\rho(t^\prime)dt^\prime
\eeq
It is convenient to change variables to $s=t-t^\prime$, so that:
\beq
\rho(t)=\int_0^t k(s,t)\Lambda(s)\rho(t-s)ds \ .
\label{eq:789}
\eeq
%
Our purpose is to arrive at a master equation, so let us differentiate Eq.~\eqref{eq:789} with respect to $t$:
\bes
\begin{align}
\frac{\partial{\rho}}{\partial t}&=\frac{\partial}{\partial t}\int_0^t k(s,t)\Lambda(s)\rho(t-s)ds \\
&= k(t,t)\Lambda(t)\rho(0) + \int_0^t \left(\frac{\partial k(s,t)}{\partial t}\Lambda(s)\rho(t-s)+k(s,t)\Lambda(s)\frac{\partial \rho(t-s)}{\partial t}\right)ds
\ .
\label{eq:804b}
\end{align}
\ees
The first term corresponds to performing the bath measurement at $t=0$ and then evolving from $\r(0)$ via $\Lambda(t)$. This term can thus be dropped [formally, by setting $k(t'=0,t)=k(s=t,t)=0$] since we assumed that the intermediate measurement weighted by $k$ occurs in the open interval $(0,t)$. 
%
%
To make further progress let us specify the form of the CP map $\Lambda$. 
For simplicity, let us assume that the intermediate evolutions are themselves Markovian:
\beq
\Lambda(t)=e^{\mathcal{L}t} \ ,
\label{eq:806}
\eeq
where $\mL$ is a Lindbladian, since this is the unique way to ensure that $\Lambda$ is CPTP in the Markovian case. Then 
\beq
\frac{\partial \rho(t-s)}{\partial t} = \frac{\partial e^{\mL(t-s)} }{\partial t}\rho(0) = \mL e^{\mL(t-s)}\rho(0) = \mL \rho(t-s)\ ,
\eeq
so that Eq.~\eqref{eq:804b} simplifies to:
\beq
\label{eq:PMME_third_form}
\frac{\partial \rho}{\partial t}=\int_{0}^{t}\left(\frac{\partial k(s, t)}{\partial t}+k(s, t)\mathcal{L}\right)e^{\mathcal{L}s}\rho(t-s)ds\ .
\eeq
We now seek to ensure that this evolution is trace-preserving. This requires the RHS to be traceless, since then $0 = \Tr\partial_t \rho = \partial_t \Tr\rho = 0$, so that $\Tr\r(t)=\textrm{const}$. It is sufficient to this end to demand that $\partial_t k(s,t) = 0$, since the second term is already traceless:
\beq
\Tr\left[ \int_{0}^{t} ds\ k(s, t)\mathcal{L}e^{\mathcal{L}s}\rho(t-s)\right] =  \int_{0}^{t} ds\ k(s, t) \Tr\left[\mathcal{L}e^{\mathcal{L}s}\rho(t-s) \right]=0\ ,
\eeq
since for a Lindbladian $\mL$ acting on any operator $X$
\beq
\Tr[\mathcal{L}X]=\Tr[\sum_\alpha L_\alpha X L_\alpha^\dag-\frac{1}{2}L_\alpha^\dag L_\alpha X-\frac{1}{2}X L_{\alpha}^\dag L_\alpha]= \sum_\alpha \Tr[ X L_\alpha^\dag L_\alpha]-\frac{1}{2}\Tr[  X L_\alpha^\dag L_\alpha]-\Tr[X L_{\alpha}^\dag L_\alpha ] = 0\ .
\eeq
Now, since $\partial_t k(s,t) = 0$, it follows that $k(s,t) = ck(s)$, where $c$ is a constant we can choose to be $1$. Therefore
\beq
k(s,t)\equiv k(s)\ .
\label{eq:k(s)}
\eeq
Then Eq.~\eqref{eq:PMME_third_form} reduces to:
\bes
\label{eq:PMME-final}
\begin{align}
\label{eq:PMME-final-a}
\frac{\partial \rho}{\partial t}&=\mathcal{L}\int_0^t k(s)e^{\mathcal{L}s}\rho(t-s)ds\\
&=\mathcal{L}k(t)e^{\mathcal{L}t}\ast\rho(t)\ ,
\end{align}
\ees
where in the second line $\ast$ denotes a convolution. \emph{Equation~\eqref{eq:PMME-final} is the PMME}.

Now consider two special cases of Eq.~\eqref{eq:PMME-final}:
\begin{itemize}
\item $k(s)=\delta(s)$: In this case the PMME reduces to $\frac{\partial \rho}{\partial t}=\mathcal{L}\rho(t)$, which is the standard Lindblad equation. Therefore the PMME includes the LE as a special case.
\item Expanding the exponential to zeroth order in $\mL$ (assuming $\|\mathcal{L}t\|\ll1$), the PMME reduces to $\frac{\partial \rho}{\partial t}=\mathcal{L}\int_0^t k(s)\rho(t-s)ds$, which is a form that has been proposed heuristically in the literature on non-Markovian master equations.
\item Since the PMME involves a convolution, it can be viewed as a special case of the NZ-ME. Namely, we can write the PMME in the NZ-ME form $\pt\hat\r(t) = \int_{0}^t\mK(t,t')\hat\r(t')dt'$ [recall Eq.~\eqref{eq:NZker}], where $\mK(t,t')$ is directly obtainable from Eq.~\eqref{eq:PMME-final-a}.

\end{itemize}


\subsection{Solution of the PMME}

To solve the PMME~\eqref{eq:PMME-final} we can use the Laplace transform~\eqref{eq:Lap}. Recall that the Laplace transform of the convolution of two functions is the product of their Laplace transforms: Eq.~\eqref{eq:Lap-conv}, and also recall the result for the Laplace transform of a derivative in Eq.~\eqref{eq:Lap-deriv}. 
Therefore, upon taking the Laplace transform of the PMME we find:
\begin{align}
s \tilde \rho(s) - \rho(0) &= \mL \textrm{Lap}[k(t)e^{\mL t}] \tilde \rho(s) .
\label{eq:819}
\end{align}
The Laplace transform satisfies the following shifting property:
\begin{align}
\textrm{Lap}[f(t)e^{at}] &= \tilde f(s-a) ,
\label{eq:Lap-shift}
\end{align}
but to use it requires a few extra steps, since it is not immediately clear how to deal with $e^{\mL t}$ in this context. Thus, we find that
it is most convenient to 
work in the eigenbasis of $\mL$. 
Since $\mL$ is not normal ($[\mL,\mL^\dag]\neq 0$ in general), it can have distinct right and left eigenvectors, i.e., we can find a set of operators $\{R_i\}$ and $\{L_i\}$ such that $\mL R_i = \lambda_i R_i$ and $L_i \mL = \lambda_i L_i$. Both sets are complete, and they are mutually orthonormal in the sense that after normalization $\Tr[L_iR_j] = \delta_{ij}$. 

We can therefore expand $\r$ in this so-called ``damping basis" (the basis of right eigenvectors of $\mL$), to get:
\beq
\rho(t) = \sum_i \mu_i(t) R_i \ ,
\label{eq.expandD}
\eeq
where the expansion functions are given by
\beq
\mu_j(t) = \sum_i \mu_i(t) \Tr(L_jR_i) = \Tr[L_j \rho(t)] \ .
\eeq
Substituting into the PMME Eq.~\eqref{eq:PMME-final-a} we obtain
\bes
\label{eq:PMME-final2}
\begin{align}
\label{eq:PMME-final-a2}
\frac{\partial \mu_i}{\partial t}R_i&=\sum_i \mathcal{L}\int_0^t k(s)e^{\mathcal{L}s}\mu_i(t-s)R_i ds\\
&=\sum_i \lambda_i \int_0^t k(s)e^{\lambda_i s}\mu_i(t-s)R_i ds\ ,
\end{align}
\ees
where we used $e^{\mathcal{L}s}R_i = e^{\lambda_i s} R_i$. Multiplying both sides by $L_i$ and taking the trace yields:
\beq
\frac{\partial \mu_i}{\partial t} = \lambda_i \int_0^t k(s) e^{\lambda_i s}\mu_i(t-s) ds .
\eeq
At this point we can take the Laplace transform of both sides and use the shifting property~\eqref{eq:Lap-shift}, to get:
\bes
\begin{align}
s\tilde{\mu}_i(s)-\mu_i(0) &= \lambda_i \textrm{Lap} \bigg[ k(t) e^{\lambda_i}t \bigg] \tilde{\mu}_i(s) \\
&= \lambda_i\tilde{k}(s-\lambda_i)\tilde{\mu}_i(s)
\end{align}
\ees
Therefore:
\begin{align}
\tilde \mu_i(s) = \frac{1}{s - \lambda_i \tilde k(s-\lambda_i)} \mu_i(0) .
\end{align}
Finally, taking the inverse Laplace transform:
\begin{align}
\mu_i(t) = \xi_i(t) \mu_i(0) \ ,
\end{align}
where:
\begin{align}
\xi_i(t) &\equiv \mathrm{Lap}^{-1} \ls \frac{1}{s - \lambda_i \tilde k(s-\lambda_i)}\rs\ , \quad \mu_i(0) = \Tr[L_i\rho(0)]\ .
\label{eq:823}
\end{align}
This completes the exact solution of the PMME. 

To summarize, given $\mL$ we need to compute its eigenvalues $\lambda_i$ and associated left and right eigenvectors, and given the kernel $k(t)$ we need to compute its Laplace transform. Using the initial condition $\r(0)$, we can then compute $\xi_i(t)$ and $\mu_i(t)$, from which we obtain $\r(t)$ using Eq.~\eqref{eq.expandD}. The kernel $k(t)$ was assumed to satisfy the condition $k(0)=0$.

\subsection{The PMME as a map, and its relation to the TCL-ME}

The solution of the PMME can be viewed as a map $\Phi$:
\beq
\rho(t) =  \sum_i \mu_i(t)R_i = \sum_i \xi_i(t) \mu_i(0) R_i = \sum_i \xi_i(t) \Tr[L_i\rho(0)]R_i = \Phi [ \rho(0)] \ ,
\eeq
where
\beq
\Phi [ X ] \equiv \sum_i \xi_i(t) \Tr[L_i X] R_i\ .
\label{eq:Phi-PMME}
\eeq

Let us assume that $\xi_i(t) \neq 0$ $\forall t$. If this is the case then $\Phi$ is invertible, i.e., if we let
\beq
\Phi^{-1} [ X ] = \sum_i \xi^{-1}_i(t) \Tr[L_i X] R_i
\eeq
then
\beq
\Phi^{-1}\circ \Phi [ X ] = \sum_i \xi^{-1}_i(t) \Tr[L_i \Phi(X)] R_i = \sum_{ij} \xi^{-1}_i(t) \xi_j(t) \Tr[L_j X]\Tr[   L_iR_j] R_i = \sum_i \Tr[L_i X] R_i = X\ ,
\eeq
as required. Therefore, using $\r(t) = \Phi [\r(0)]$  we can write $\rho(t-t') = \Phi(t-t')[\rho(0)] = \Phi(t-t')\Phi^{-1}(t) [\rho(t)]$, and so we have:
\begin{align}
\frac{\partial \rho}{\partial t}= \ls\mL \int_0^t k(t') e^{\mL t'} \Phi(t-t')\Phi^{-1}(t)dt'\rs \rho(t) \equiv \mK(t) \rho(t)\ ,
\label{eq:827}
\end{align}
where $\mK(t)$ is now a convolutionless generator, and Eq.~\eqref{eq:827} is time-local. Therefore, despite the appearance of the convolution in the PMME~\eqref{eq:PMME-final}, it can be written in TCL-ME form. This is similar to what we did to transform the NZ-ME into the TCL-ME, where an invertibility assumptions was likewise assumed (recall Sec.~\ref{sec:relevantpart}). It is an interesting open problem to identify the conditions under which the TCL-ME reduces to the PMME.


\subsection{Complete Positivity of the PMME}

Due to the freedom in choosing the kernel $k(t)$, complete positivity is not a guaranteed feature of the PMME. The following theorem provides us with a way to construct a complete positivity test. 

Consider a linear map $\Phi: \mathbf{C}^{d\times d} \mapsto \mathbf{C}^{d\times d}$, i.e., $\Phi$ acts on operators represented by $d\times d$ matrices, acting on the Hilbert space $\mathcal{H} = \textrm{span}\{\ket{i}\}_{i=1}^d$. Let us pick $\ket{i}$ as a column vector of zeroes, except for a single $1$ in position $i$. Let $C = \{\Phi[ \ketb{i}{j}]\}_{ij} = \sum_{ij} \ketb{i}{j} \otimes \Phi[ \ketb{i}{j}]$. I.e., $C$, known as the Choi matrix, is a $d\times d$ matrix of the $d\times d$ matrices $\Phi[ \ketb{i}{j}]$, meaning that $C$ is $d^2\times d^2$. 

\begin{thm}[Choi's theorem~\cite{Choi:75}]
$\Phi$ is completely positive if and only if $C>0$.
\end{thm}

Constructing the Choi matrix $C$ for the PMME using Eq.~\eqref{eq:Phi-PMME} we have:
\bes
\begin{align}
C = \sum_{ij} \ketb{i}{j} \otimes \sum_k \xi_k(t)\Tr[L_k  \ketb{i}{j}] R_k = \sum_k \xi_k(t) \sum_{ij} \ketb{i}{j} \otimes  \bra j L_k \ket i R_k = \sum_k \xi_k(t) \sum_{ij} \ketb{i}{j} (L_k^T)_{ij} \otimes R_k\ .
\end{align}
\ees
Hence:
\beq
C= \sum_k \xi_k(t) L^T_k\otimes R_k >0 
\label{eq:PMMECP}
\eeq
Eq.~\eqref{eq:PMMECP} is the complete positivity for the kernel $k(t)$, for a given Lindbladian $\mL$ and its set of left and right eigenvectors.

\subsection{Example of the PMME: phase damping Lindbladian with an exponential kernel}
To illustrate the solution of PMME, consider the phase damping Lindbladian:
\beq
\mL \rho = \frac \g 2 (Z\rho Z- \rho) 
\eeq
To find the left and right eigenvectors of $\mL$, consider its action on the Pauli matrices $\{I,X,Y,Z\}$: 
\bes
\begin{align}
\mL I &= \frac \g 2 (ZIZ-I)=0\ , \quad \mL Z = \frac \g 2 (Z^3-Z)=0 \\
\mL X &= \frac \g 2 (ZXZ-X) = -\g X\ , \quad \mL Y = \frac \g 2 (ZYZ-Y) = -\gamma Y\ .
\end{align}
\ees 
Thus the Pauli matrices $\{R_i\} = \{I,X,Y,Z\}$ are $\mL$'s right eigenvectors, with corresponding eigenvalues $\{\lambda_i\} = \{0,-\g,-\g,0\}$. Representing the Pauli matrices as vectors, i.e., as $I = (1,0,0,0)^T, X=(0,1,0,0)^T$, etc., we can write $\mL$ as a diagonal matrix with diagonal entries $\{0,-\g,-\g,0\}$. It is then clear that the left eigenvectors are again the Pauli matrices, i.e., in this example $L_i=R_i$ for $i\in\{I,X,Y,Z\}$, and the condition $\Tr(L_i R_j) = \d_{ij}$ is automatically satisfied. 

Let us express the density matrix in terms of the Bloch vector: $\r(t) = \frac{1}{2}(I+\vec{v}(t )\cdot\vec{\s})$. The initial condition can then be written as
\beq
\mu_i(0) = \Tr[L_i \r(0)] = \frac{1}{2}v_i(0) \, 
\eeq
where $v_I(0) = 1$.

Let us now assume that the kernel $k(t)$ is:
\beq
k(t)=Ae^{-a t}\ .
\eeq
Recall that $\textrm{Lap}(e^{at}) = 1/(s-a)$, so that after the Laplace transformation we have
\beq
\tilde{k}(s)=\frac{A}{s+a}\ .
\eeq
Using Eq.~\eqref{eq:823} we thus find:
\beq
\xi_i(t) = \textrm{Lap}^{-1} \ls \frac{1}{s - \lambda_i \frac{A}{s-\lambda_i+a}}\rs\ .
\eeq
The $\mL$ eigenvectors $I$ and $Z$ have the eigenvalue $\lambda=0$, so that:
\beq
\xi_{I,Z}(t)=\textrm{Lap}^{-1}\ls \frac{1}{s} \rs= e^{0t} = 1\ .
\eeq
The $\mL$ eigenvectors $X$ and $Y$ have the eigenvalue $\lambda=-\g$, so that:
\begin{align}
\xi_{X,Y}(t)=\textrm{Lap}^{-1} \ls \frac{1}{s +\g \frac{A}{s+\g+a}}\rs=e^{-\frac 1 2 (a+\gamma)t}\left(\cos{\omega t}+\frac{a+\gamma}{2\omega}\sin{\omega t} \right)\ ,
\end{align}
where $\omega=\frac 1 2 \sqrt{4\g A-(\gamma+a)^2}$.
Thus the density matrix is
\beq
\r(t) = \sum_i\mu_i(0) \xi_i(t) R_i =  \frac{1}{2}\left[ I + (v_X(0)X+v_Y(0)Y)\xi_{X,Y}(t) + v_Z(0)Z\right]\ .
\eeq
This describes a Bloch vector with fixed $Z$-component but with $X$ and $Y$ components undergoing damped oscillations with frequency $\o$. This is clearly non-Markovian dynamics. The condition for oscillation is $4\g A> (\gamma+a)^2$; otherwise the oscillations become exponential decay.

Finally, we can use the complete positivity criterion we found above. The Choi matrix is:
\bes
\begin{align}
C= \sum_k \xi_k(t) L^T_k\otimes R_k&=\xi_II^T\otimes I+\xi_XX^T\otimes X+\xi_YY^T\otimes Y+\xi_ZZ^T\otimes Z\\
&=2\left(\begin{array}{cccc}
1 &0 &0 &\xi_X\\
0 &0 &0 &0\\
0 &0 &0 &0\\
\xi_X &0 &0 &1 \end{array}\right)\ .
\end{align}
\ees
Its eigenvalues are easily found to be $\{0,0,2(1+\xi_X),2(1-\xi_X)$. Therefore the PMME in this case corresponds to a CP map iff
\beq
|\xi_X|=|\xi_Y|<1\ ,
\eeq
which is a condition on the problem parameters $A,a,\g$.

\subsection{Experimental determination of the Lindbladian $\mL$ and kernel $k(t)$}
Since both $\mL$ and $k(t)$ are phenomenological in the PMME, is there a way we can determine them experimentally? To do so, we need to express the kernel in terms of  measurable quantities. Let us assume that we $\r(t)$ can be determined via quantum state tomography, let us guess $\mL$ (perhaps based on physical intuition as to the prevalent noise). Then we know $\rho(t)$, the initial condition $\r(0)$, and the left and right eigenvector sets $\{L_i,R_i\}$, so that we can compute $\xi_i(t)$: 
\bes
\begin{align}
\rho(t) &= \sum_i \mu_i(t) R_i = \sum \xi_i(t) \mu_i(0) R_i = \sum \xi_i(t) \Tr[L_i \rho(0)] R_i \\
\Longrightarrow & \Tr[L_j\rho(t)] =\Tr[L_j \rho(0)] \xi_j(t) \\
\Longrightarrow & \xi_i(t) = \frac{\Tr[ L_i \rho(t)]}{\Tr[L_i \rho(0)]}\ ,
\end{align}
\ees
which gives us way to compute $\xi_i(t)$ from purely experimentally measurable quantities.
But at the same time $\xi_i(t)$ is related to the kernel via Eq.~\eqref{eq:823}. We can invert the latter for $k(t)$ as follows:
\beq
\tilde{\xi}(s)  = \frac{1}{s-\lambda_i\tilde{k}(s-\lambda_i)} \quad \Longrightarrow \quad \tilde{k}(s-\lambda_i) = \frac{1}{\lambda_i}\left(s-\frac{1}{\tilde{\xi}_i(s)}\right)\ ,
\eeq
where we used the identity $\textrm{Lap}^{-1} [\tilde{k}(s-\lambda)] = k(t)e^{\lambda t}$, 
so that
\beq
k(t)=\frac{e^{-\lambda_it}}{\lambda_i}\textrm{Lap}^{-1} \ls s-\frac{1}{\tilde{\xi}_i(s)}\rs\ .
\eeq
Note that in this expression only the RHS depends on the eigenvalue index $i$. This gives us an opportunity to optimize the choice of the Lindbladian by minimizing the deviation for different $i$ values, since they must all agree in order to give a unique result for $k(t)$. The experimental determination of $\mL$ and $k(t)$ is thus an iterative process involving this minimization.

\newpage

\appendix

\section{Linear algebra background and Dirac notation}
\label{app:A}

Everything in this Appendix is about the finite dimensional case, unless explicitly noted otherwise.

\subsection{Inner Product}
The inner product of two vectors is a function operating on two copies of a vector space $V$ that outputs a complex number, $f:V\times V\mapsto \C$. By definition it must satisfy the following conditions:
\bes
\begin{align}	
		& f\lp |v\>, \sum_i \lambda_i |w_i\>\rp=\sum_i \lambda_i f\lp |v\>,|w_i\>\rp \\
		& f\lp |v\>,|w\>\rp^*= f\lp |w\>,|v\>\rp \\
		& f\lp |v\>,|v\>\rp \geq 0 .
\end{align}
\ees
It is easy to show that an immediate consequence is
\beq
		f\lp |v\>, \sum_i \lambda_i |w_i\>\rp^*=\sum_i \lambda_i^* f\lp |w_i\>,|v\>\rp.
\eeq
We define the inner product between two Dirac kets as follows:
	\beq
		f\lp |v\>,|w\>\rp \equiv \sum_{i=1}^nv_i^*w_i = \lp v_i^*, ...,v_n^* \rp \lp \begin{array}{c} w_1\\ ...\\ w_n \end{array}\rp = \<v|w\> .
	\eeq
	
\subsection{Orthonormal Bases}
Two vectors $|v\>$ and $|w\>$ are orthogonal if and only if their inner product is zero: $\<v|w\>=0 \iff |v\> \perp |w\>$. The norm of a vector is
\beq
\| |v\> \| \equiv \sqrt{\<v|v\>}.
\eeq
A unit vector is normalized: $\| |v\> \|=1$. A set of vectors forms a basis if it spans the vector space and is linearly independent. Using the previous definitions, we can then say that an orthonormal basis is a set of normalized orthogonal vectors that span the vector space $V$ and are linearly independent:

	\begin{description}
		\item[Orthonormal set]$\{|v_i\>\}_{i=1}^n,\ \<v_i|v_j\>=\delta_{ij},\ \delta_{ij}=\begin{cases} 1, & \text{if}\ i=j \\0, & \text{if}\ i\neq j \end{cases}$
	\end{description}

\subsection{Linear Operators}
Another concept important to our formulation of quantum mechanics is that of linear operators. Consider an operator $A$ that maps one vector space to another:

	\beq
		A:V\mapsto W
	\eeq
For $A$ to be linear, it must be true that for $a,\ b \ \epsilon \ \C$ and $|v\>,\ |w\> \ \epsilon \ V$

	\beq
		A(a|v\>+b|w\>)= aA|v\>+bA|w\> \ \epsilon \ W
	\eeq
In words, the operator $A$ acting on a linear combination of vectors in the space $V$ produces a linear combination of the operator acting on each vector individually, and this sum is an element of space $W$. A good example of a linear operator is the outer product.

\subsection{Outer Product}
If we consider vectors $|v\>\ ,\ |z\> \ \epsilon \ V$ and $|w\> \ \epsilon \ W$, the outer product of $|v\>$ and $|w\>$ is defined as follows:

	\beq
		A=\underbrace{|w\>\<v|}_{\text{outer product}}\ :\ \lp |w\>\<v| \rp |z\> \equiv |w\> \underbrace{\lp \<v|z\> \rp}_{\epsilon \ \C}= \<v|z\> |w\>
	\eeq
One important use of the outer product is in the case of expansion in an orthonormal basis. Consider a vector $|v\> \ \epsilon \ V$ and a set of vectors $\{ |i\> \}_{i=1}^n$ which forms an orthonormal basis set for $V$. We can equivalently write $|v\>=\sum_{i=1}^nc_i|i\>$, in which $c_i$ is an arbitrary constant. The inner product of some vector $|j\>$ with $|v\>$ produces the coefficient of $\ket{v}$ in the given basis:

	\beq
		\<j|v\>=\sum_ic_i\<j|i\>=\sum_ic_i\delta_{ij}=c_j
	\eeq

If we take the outer product of $|v\>$ with itself, we generate an $n\times n$ identity matrix:

	\beq
		\sum_{i=1}^n|i\>\<i|=I= \lp \begin{array}{cccc} 1 & 0& ... & 0  \\  0 & 1& ... & 0  \\... & ... & ... & ...  \\0 & 0& ... & 1  \\ \end{array} \rp
	\eeq

 We can confirm this is true by applying this inner product as an operator on a vector $|v\>$:

 	\beq
		\lp \sum_{i=1}^n|i\>\<i| \rp |v\>= \sum_{i=1}^n |i\> \underbrace{\<i|v\>}_{c_i}= \sum_{i=1}^n c_i|i\>=|v\>
	\eeq
 The operator acting on the vector returned the vector, and is known as the ``resolution of the identity".
 This special case of the outer product is used to generate a matrix representation of an operator in the appropriate basis. If we consider an operator $A$ that preserves the space, $A:\ V \mapsto V $, multiplication of the operator by the identity matrix produces a matrix with elements that perform the operation $A$ in the following way:

 	\begin{align}
		A&=IAI \\
		&=\lp \sum_{i=1}^n|i\>\<i| \rp A \lp \sum_{j=1}^n|j\>\<j| \rp \\
		&=\sum_{i,j}|i\> \underbrace{\<i|A|j\>}_{a_{ij}} \<j| \\
		&=\sum_{i,j}a_{ij}|i\>\<j|
	\end{align}

The scalar $a_{ij}$ is known as a matrix element of the operator $A$. Recall that since the vectors $|i\>$ and $|j\>$ are members of an orthonormal basis, $a_{ij}|i\>\<j|$ is actually a matrix with all but the $ijth$ element equal to zero and the $(i,j)$th element equal to $a_{ij}$:

	\beq
		a_{ij}|i\>\<j|=\lp \begin{array}{ccc} 0 & ... & 0 \\ .. & a_{ij} & ... \\ 0 & ... & 0  \end{array} \rp
	\eeq

The sum over all combinations of $i$ and $j$ therefore produces a matrix with elements $a_{ij}$:

	\beq
		\sum_{i,j}a_{ij}|i\>\<j|= \lp \begin{array}{ccc} a_{11} & ... & a_{1n} \\ .. & ... & ... \\a_{n1} & ... & a_{nn}  \end{array} \rp	
	\eeq

\subsection{The Cauchy-Schwartz Inequality}

The Cauchy-Schwartz inequality is
\beq
|\<v|w\>|^2 \leq \<v|v\>\<w|w\> .
\eeq
It helps us make powerful statements about the properties of vectors in Hilbert space that define the domain of quantum mechanics. In its elementary form it states that, from the definition of the inner product $\vec{a} \cdot \vec{b}=\|\vec{a}\|\| \vec{b}\|\cos\theta$,
it follows that the magnitude of the inner product of those vectors is less than or equal to the product of their norms: $|\vec{a} \cdot \vec{b}| \leq \|\vec{a}\|\|\vec{b}\|$.

We can prove this for Hilbert spaces while demonstrating the power of Dirac notation.
\begin{proof}
Pick an orthonormal basis whose first element is $\ket{1} = \ket{w}/\|\ket{w}\|$ (we can always do this using the Gram-Schmidt process to complete the basis). Then, using the resolution of identity we have
\begin{align}
\<v|v\>\<w|w\> &= \<v|I|v\>\<w|w\> = \sum_{i=1}^n\<v|i\>\<i|v\>\<w|w\>=\underbrace{\frac{\<v|w\>}{\||w\>\|}\frac{\<w|v\>}{\||w\>\|}}_{i=1}\<w|w\>+\sum_{i=2}^n \underbrace{|\<v|i\>|^2\||w\>\|^2}_{\geq 0} \notag \\
&= \<v|w\>\<w|v\>+ \text{positive number}
\end{align}
Therefore, since $\<v|w\>\<w|v\>=|\<v|w\>|^2$, we see that $|\<v|w\>|^2 \leq \<v|v\>\<w|w\>$.
\end{proof}

\subsection{Trace equalities}
\label{app:trace}
The following are some useful equalities satisfied by the trace operation. They are easily provable by the rules of matrix multiplication. $A$ and $B$ are arbitrary matrices of matching dimensions.
\bes
\label{eq:trace-equalities}
\begin{align}
\Tr(AB) &= \Tr(BA) \\ 
\Tr(A\ox B) &= \Tr(A)\Tr(B) \\
[\Tr(AB)]^* &= \Tr[B^\dag A^\dag]\ .
\end{align}
\ees

\subsection{Positive operators}
\label{app:pos-ops}

An operator is positive definite (or positive, for short) if all its eigenvalues are positive. An operator is positive semi-definite if all its eigenvalues are non-negative. To test this for a given operator $A$, it suffices to prove that for all vectors $\ket{v}$, the diagonal matrix elements $\bra{v}A\ket{v}$ are positive or non-negative, respectively. The reason is that this will obviously include the eigenvectors of $A$.

\subsection{Pauli matrices}
\label{app:Pauli}

The four Pauli matrices are:
\bea
\sigma_0 = I =\begin{pmatrix} 1 & 0 \\ 0 & 1 \end{pmatrix}, \quad\sigma_1 =\sigma_x=X = \begin{pmatrix} 0 & 1 \\ 1 & 0
\end{pmatrix},\quad
\sigma_2 = \sigma_y=Y = \begin{pmatrix} 0 & -\ii \\ \ii & 0 \end{pmatrix},\quad \sigma_3 =\sigma_z=Z = \begin{pmatrix} 1 & 0 \\ 0 & -1
\end{pmatrix}.
\eea
The last three are traceless by inspection. The Pauli matrices satisfy the identity
\beq
\sigma_k \sigma_l = \delta_{kl} I + i \sum_m \varepsilon_{klm} \sigma_m
\label{eq:Pauli-mult}
\eeq
where $\delta_{kl}$ is the Kronecker symbol (it is $1$ if $k=l$, otherwise it is $0$), and $\varepsilon_{klm}$ is the completely anti-symmetric Levi-Civita
symbol [it is $1$ if $(klm)$ is an even permutation of $(123)$, $-1$ if it is an odd permutation, and $0$ if any index is repeated]. 

Since the Pauli matrices are traceless we also have the useful identity
\beq
\label{eq:tracePaulis}
\Tr(\s_k\s_l) =2\d_{kl} .
\eeq

\section{Unitarily invariant norms}
\label{app:norms}

Let $\mathcal{V}$ an inner product space equipped
with the Euclidean norm $\Vert x\Vert \equiv \sqrt{\sum_{i}|x_{i}|^{2}%
\langle e_{i},e_{i}\rangle }$, where $x=\sum_{i}x_{i}e_{i}\in \mathcal{V}$
and $\mathcal{V}=\mathrm{Span}\{e_{i}\}$. Let $A :\mathcal{V}\mapsto 
\mathcal{V}$. Define
\beq
|A|\equiv \sqrt{A^{\dagger }A}\ .
\eeq
Unitarily invariant norms\index{norm!unitarily invariant} are norms that satisfy, for all unitary $U,V$ \cite{Bhatia:book}:
\begin{equation}
\Vert UAV\Vert _{\mathrm{ui}}=\Vert A\Vert _{\mathrm{ui}}\ .
\end{equation}%
We list some important examples.
\begin{enumerate}
\item The trace norm: 
\begin{equation}
\Vert A\Vert _{1}\equiv \mathrm{Tr}|A|=\sum_{i}s_{i}(A)\ ,
\label{eq:A1-trace-norm}
\end{equation}%
where $s_{i}(A)$ are the singular values of $A$ (i.e., the eigenvalues of $|A|$). If $A=\rho 
$ is a normalized quantum state, then $\Vert \rho \Vert _{1}=\mathrm{Tr}{\rho }=1$.

\item The operator norm: 
\begin{equation}
\Vert A \Vert _{\infty }\equiv \sup_{x\in \mathcal{V}}\frac{\Vert A x\Vert 
}{\Vert x\Vert }=\max_{i}s_{i}(A)\ .
\end{equation}
Therefore $\Vert A x\Vert \leq \Vert A \Vert _{\infty }\Vert x\Vert $. Also note that, by definition $\Vert A \Vert _{\infty } \leq \Vert A \Vert _{1}$, since the largest singular value is one of the summands in $\Vert A \Vert _{1}$. 

\item The Hilbert-Schmidt norm: 
\begin{equation}
\Vert A\Vert _{2}\equiv \sqrt{\mathrm{Tr}A^{\dagger }A}=\sqrt{\sum_{i}s_{i}^2(A)}\ .
\end{equation}
Again, by definition $\Vert A \Vert _{\infty } \leq \Vert A \Vert _{2}$, since $\sqrt{\sum_{i}s_{i}^2(A)} \geq \sqrt{\max_{i}s_{i}^2(A)} = \Vert A \Vert _{\infty }$. In addition, $\Vert A\Vert _{1}^2=\sum_{i,j}s_{i}(A)s_{j}(A) \geq \sum_{i}s^2_{i}(A) = \Vert A\Vert _{2}^2$.
\end{enumerate}
We have thus established the ordering
\beq
\Vert A\Vert _{\infty } \leq \Vert A\Vert _{2}\leq \Vert A\Vert _{1}\ .
\eeq
All unitarily invariant norms satisfy the important property of submultiplicativity:
\begin{equation}
\Vert AB\Vert _{\mathrm{ui}}\leq \Vert A\Vert _{\mathrm{ui}}\Vert B\Vert _{\mathrm{ui}}.
\end{equation}
It follows that
\beq
\Vert AB\Vert_{\infty} \leq \Vert A\Vert _{\infty }\Vert B_{i}\Vert\ , \ \Vert B\Vert _{\infty }\Vert A\Vert _{i}\quad i=1,2,\infty\ .
\eeq
The norms of interest to us are also multiplicative over tensor products: 
\begin{eqnarray}
\Vert A\otimes B\Vert _{i} =\Vert A\Vert _{i}\Vert B\Vert _{i}\quad
i=1,2,\infty \ .  
\label{eq:ui}
\end{eqnarray}

As an application of unitarily invariant norms, let us revisit the convergence of the iterative expansion we saw in Eq.~\eqref{eqt:formal}. We have, for the $n$th order term:
\bes
\begin{align}
\label{eq:513a}
&\Vert (-i \lambda)^n \int_0^t dt_1 \int_0^{t_1} dt_2 \cdots \int_0^{t_{n-1}}dt_n  \left[\tilde{H}(t_1), \left[ \tilde{H}(t_2) , \dots \left[ \tilde{H}(t_n),\rho_{SB}(0) \right] \right] \dots \right] \Vert_\infty  \\
\label{eq:513b}
&\leq \lambda^n \int_0^t dt_1 \int_0^{t_1} dt_2 \cdots \int_0^{t_{n-1}}dt_n \| \left[\tilde{H}(t_1), \left[ \tilde{H}(t_2) , \dots \left[ \tilde{H}(t_n),\rho_{SB}(0) \right] \right] \dots \right] \|_\infty \\
\label{eq:513c}
&\leq \lambda^n 2^n \int_0^t dt_1 \int_0^{t_1} dt_2 \cdots \int_0^{t_{n-1}}dt_n \|\tilde{H}(t_1)\|_\infty \|\tilde{H}(t_2)\|_\infty \cdots \|\tilde{H}(t_n)\|_\infty \|\rho_{SB}(0)\|_1\\
\label{eq:513d}
&= (2\lambda)^n \int_0^t dt_1 \int_0^{t_1} dt_2 \cdots \int_0^{t_{n-1}}dt_n \|H_{SB}\|^n \\
\label{eq:513e}
&=(2\lambda \|H_{SB}\|)^n \frac{t^n}{n!} \ .
\end{align}
\ees
To go from Eq.~\eqref{eq:513a} to Eq.~\eqref{eq:513b} we used the triangle inequality; to go from Eq.~\eqref{eq:513b} to Eq.~\eqref{eq:513c} we used the fact that $\| [A,B]\| = \|AB-BA\| 
\leq \|AB\|+\|BA\| \leq 2\|A\|\|B\|$ for any unitarily invariant norm; to go from Eq.~\eqref{eq:513c} to Eq.~\eqref{eq:513d} we used the fact that $\| \tilde{H}(t_j)\| = \|H_{SB}\|$, since $\tilde{H}(t_j) = U_0^\dgr(t)H_{SB} U_0(t)$ and $U_0$ is unitary. Thus, the norm of the $n$th order term is $O[(\|H_{SB}\|t)^n]$.


\section{Distance and Fidelity between quantum states}

Consider two quantum states represented by their density matrices $\r$ and $\sigma$. Suppose we perform a POVM measurement with operators ${E_i}$, and obtain measurement outcome $i$ with probability ${p_i}$ for state $\r$, and ${q_i}$ for state $\sigma$: 
\bea
p_i = \Tr ( E_i \r )     \\
q_i = \Tr ( E_i \sigma )
\eea
How close are the two outcomes, or equivalently, how close are the two distributions? We address this next.

\subsection{Total variation distance and quantum distance}

The total variation distance between two classical probability distributions $p=\{p_i\}$ and $q=\{q_i\}$ is defined as
\beq
D(p,q) \equiv \frac{1}{2} \sum_i |p_i -q_i| .
\eeq
The total variation distance measure forms a metric on the space of classical probability distributions, as it satisfies all the three properties of a metric, viz. the distance between the same variables is zero,  it is symmetric, and it satisfies the triangle inequality:
\bes
\bea
D(x,x)&=&0 \\
D(x,y)&=&D(y,x)\\
D(x,y)&\leq& D(x,z) + D(z,y)
\eea
\ees
The trace-norm distance can then be realized as a quantum analogue of the total variation distance.
\beq
D(\r,\sigma) \equiv \frac{1}{2}\norm{\r - \sigma}_1
\label{eq:tnD}
\eeq

Here we have introduced the one-norm, also called the trace norm, which we define for an arbitrary matrix $A$:
\beq
\norm{A}_1 = \sum_i \sigma_i(A)
\eeq
where $\sigma_i(A)$ are the singular values of $A$, i.e., the eigenvalues of $|A|\equiv = \sqrt{A^\dagger A}$. The name trace norm comes from
\beq
\norm{A}_1  \equiv  \Tr |A|.
\eeq
While we're at it, there is a useful inequality relating the trace norm and the operator norm \cite{Bhatia:book}:
\beq
\label{eq:norm1inf}
\|AB\|_1 \leq \|A\|_1 \|B\|
\eeq
for any pair of operators $A$ and $B$.

Some useful properties of the trace-norm distance are:
\begin{enumerate}
\item Bounded between $0$ and $1$: Clearly $D(\r,\r)=0$ and $D(\r,\sigma)$ cannot be negative since it is the sum of non-negative quantities (the singular values are the absolute values of the eigenvalues). Also, by letting $\r=\ketbra{\psi}$ and $\s = \ketbra{\phi}$ such that $\bra{\psi}{\phi}\rangle=0$, we have $\Tr\sqrt{(\r-\s)^\dagger (\r-\s)} = \Tr\sqrt{\r+\s} = \Tr(\r+\s) = 2$, where we used $\r+\s = (\r+\s)^2$ and positivity. Thus $D(\r,\sigma)=1$ in this case, and it's not hard to see that $D$ can't be larger.
\item Invariance under a simultaneous unitary transformation of both $\r$ and $\s$: 
\beq
D(U\r U^\dagger,U\sigma U^\dagger) = \frac{1}{2}\norm{U\r U^\dagger - U\sigma U^\dagger}_1 = \frac{1}{2}\norm{U(\r  - \sigma) U^\dagger}_1 = \frac{1}{2}\norm{\r - \sigma}_1 = D(\r,\sigma)\ ,
\eeq
where we've used the fact that the trace norm is unitarily invariant \cite{Bhatia:book}.
\item If $\r$ and $\sigma$ commute, the trace-norm distance reduces to the total variation distance between the set of paired eigenvalues of $\r$ and $\sigma$. The pairing is done by their common eigenvectors (which they have by virtue of being commuting Hermitian operators).
\item It can be shown that if $p$ and $q$ are the probability distributions of $\r$ and $\sigma$ for some POVM, $D(\r,\sigma) \ge D(p,q)$. In other words, the trace-norm distance is always an upper bound on the corresponding total variation distance. Moreover, there always exists a POVM which saturates the bound.
\end{enumerate}
Hence, an equivalent definition of the quantum distance measure is
\beq
D(\r,\sigma)=\sup_{ \{\textrm{POVM}\} } D(p,q)
\label{eq:D-sup}
\eeq
This is very useful since we'd like to find a measurement which makes the two states as distinguishable as possible. The trace-norm distance automatically tells us how far apart the states would be if we could find such a measurement.

\subsection{Fidelity Measures}

A fidelity measure can be thought of as an overlap of two states, or the inner product between them. The classical fidelity is defined as
\beq
F(p,q) = \sum_i \sqrt{p_i}\sqrt{q_i} = (\vec{\sqrt{p}},\vec{\sqrt{q}})\ ,
\label{eq:F-classical}
\eeq
i.e., it is the inner product between two vectors $\vec{\sqrt{p}} = (\sqrt{p_1},\sqrt{p_2},\ldots )$ and $\vec{\sqrt{q}} = (\sqrt{q_1},\sqrt{q_2},\ldots )$, whose elements are given by square roots of the elements of classical probability distribution.
The fidelity is not a metric since it doesn't satisfy the triangle inequality. However, $\arccos(F)$ is a distance, also known as the Bures angle, or Bures length (related to the Bures or Fubini-Study metric). \\

A quantum fidelity measure was first introduced by Uhlmann. The Uhlmann's fidelity between two distribution $\r$ and $\sigma$ is clearly inspired by the classical fidelity, and is given by
\beq
F(\r,\sigma)\equiv\norm{\sqrt{\r}\sqrt{\sigma}}_1 .
\eeq

\subsection{The distance and fidelity inequality}

Fidelity and distance both give us a sense of how close two states are. While the distance gives us the separation between two states, fidelity measure the amount of overlap, or similarity of two states. We use two such measures, as while the distance measure has a nice interpretation as resulting from the optimal POVM, the fidelity measure is often easier to calculate. The two measures are related by the following inequality \cite{Fuchs:99}:
\beq
1 - F \leq D \leq \sqrt{1 - F^2} \iff 1 - D \leq F \leq \sqrt{1 - D^2}.
\label{eq:D-F}
\eeq

\subsection{Uhlman's Theorem}

Uhlman's theorem gives a nice operational interpretation for the fidelity. Consider two states $\r$ and $\sigma$, acting on the same Hilbert space $\mcal{H}_1$. Next consider the ``doubled" Hilbert space given by $\mcal{H}_1 \otimes \mcal{H}_2$, where $\mcal{H}_2 = \mcal{H}_1$.

One can always find two pure states $\ket{\Psi},\ket{\Phi}\in \mcal{H}_1 \otimes \mcal{H}_2$ such that
\bes
\bea
\r     &=& \Tr_{\mcal{H}_2} \ketbra{\Psi} \\
\sigma &=& \Tr_{\mcal{H}_2} \ketbra{\Phi} .
\eea
\ees
Indeed, if the spectral decomposition of $\r$ is $\sum_i r_i \ket{i}\bra{i}$, then $\ket{\Psi} = \sum_i \sqrt{r_i} \ket{i}\otimes\ket{i}$ yields $\Tr_{\mcal{H}_2} \ketbra{\Psi} = \Tr_{\mcal{H}_2} \sum_{ij} \sqrt{r_i r_j} \ket{i}\bra{j}\otimes \ket{i}\bra{j} = \sum_{ij} \sqrt{r_i r_j} \ket{i}\bra{j}\Tr (\ket{i}\bra{j}) = \r$, and similarly for $\s$.

This procedure is called ``purification", and $\ket{\Psi}$ is called a purification of $\r$. While the purification of a state is not unique (e.g., we could have picked $\ket{\Psi} = \sum_i e^{i \theta_i} \sqrt{r_i} \ket{i}\otimes\ket{i}$ instead), it can clearly always be found. Uhlman's theorem states that
\beq
F(\r,\sigma)=\sup_{\{\ket{\Psi},\ket{\Phi}\}} |\braket{\Psi | \Phi}|\ ,
\eeq
i.e., the fidelity has the appealing interpretation of being the largest possible overlap among the purifications of the two states. Thus it is also an inner product, just like the classical fidelity in Eq.~\eqref{eq:F-classical}. Moreover, since $|\braket{\Psi | \Phi}| = |\braket{\Phi | \Psi}|$, clearly 
\beq
F(\r,\sigma) = F(\s,\r)\ .
 \label{eq:F-sym}
\eeq

Using the definition of the trace norm and the positivity of $\r$ and $\s$, we have
\bes
\bea
\norm{\sqrt{\r}\sqrt{\sigma}}_1 &=& \Tr{\sqrt{(\sqrt{\r}\sqrt{\sigma})^{\dagger}(\sqrt{\r}\sqrt{\sigma})}} \\
&=& \Tr{\sqrt{\sqrt{\sigma}\sqrt{\r}\sqrt{\r}\sqrt{\sigma}}} \\
&=& \Tr{\sqrt{\sqrt{\sigma} \r \sqrt{\sigma}}} \\
&=& \Tr{\sqrt{\sqrt{\r} \sigma \sqrt{\r}}} = \norm{\sqrt{\s}\sqrt{\r}}_1,
\eea
\ees
where the last line follows from Eq.~\eqref{eq:F-sym}.

It turns out that, just like the trace distance is the maximum of the classical distance of the probability distributions from arbitrary POVMs [Eq.~\eqref{eq:D-sup}], the quantum fidelity is the
minimum of the classical fidelity of the probability distributions from arbitrary POVMs \cite{nielsen2010quantum}[p.412]:
\beq
F(\r,\s) = \inf_{\{\textrm{POVM}\}} F(p,q).
\eeq

\subsection{Fidelity for a pure state passing through a noise channel}

Suppose a pure state $\ket{\psi}$ passes through a noise channel $\mc{N}$, as depicted below, and we wish to compare the resultant mixed state $\r = \mcal{N}(\ketbra{\psi})$ with the original.
\begin{center} \includegraphics[scale=0.5]{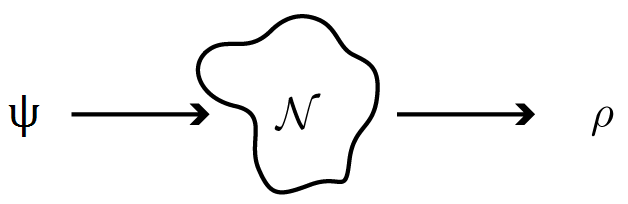} \end{center}
In this case we can simplify the expression for the fidelity (note that $\ketbra{\psi}>0,(\ketbra{\psi})^2 = \ketbra{\psi} \Rightarrow \ketbra{\psi} = \sqrt{\ketbra{\psi}}$):

\bes
\label{eq:F-pure}
\bea
F(\r,\ketbra{\psi}) &=& \Tr{\sqrt{\sqrt{\ketbra{\psi}} \r \sqrt{\ketbra{\psi}}}} \\
                    &=& \Tr{\sqrt{\ketbra{\psi} \r \ketbra{\psi}}}  \\
                                       &=& \sqrt{\braket{\psi|\r|\psi}} \Tr({\ketbra{\psi}})\\
                    &=& \sqrt{\braket{\psi|\r|\psi}} .
\eea
\ees
It turns out that in this case we can also obtain a tighter inequality than \eqref{eq:D-F},
\beq
1 - F^2(\r,\ketbra{\psi}) \le D(\r,\ketbra{\psi}) .
\eeq

\subsection{Fidelity is invariant under a joint unitary transformation}
If we rotate $\r$ and $\sigma$ by the same unitary transformation $U$, the Fidelity measure doesn't change, i.e.
\beq
F(\r,\sigma)=F( U \r U^\dagger,U \sigma U^\dagger)
\eeq

To prove this, we note that the trace norm is a unitarily invariant norm, and hence is submultiplicative [recall Eq.~\eqref{eq:ui-norm-sub}]. Also, if $A$ is positive, $U \sqrt{A} U^\dagger=\sqrt{(U \sqrt{A} U^\dagger)^2} = \sqrt{U \sqrt{A} \sqrt{A} U^\dagger}$, so that
\bea
    U \sqrt{A} U^\dagger = \sqrt{U A U^\dagger}.
\eea
Consequently,
\bes
\bea
F(U\r U^\dagger,U\sigma U^\dagger)  &=& \norm{\sqrt{U\r U^\dagger} \sqrt{U\sigma U^\dagger}}_1 \\
  &=& \norm{ U \sqrt{\r} U^\dagger U \sqrt{\sigma} U^\dagger}_1\\
  &=& \norm{ U \sqrt{\r} \sqrt{\sigma} U^\dagger}_1\\
  &=& \norm{ \sqrt{\r}  \sqrt{\sigma} }_1\\
  &=& F(\r,\sigma).
\eea
\ees

\subsection{Fidelity of Noise channels}

Consider a noise channel $\mc{N}$ that is completely positive and trace preserving (CPTP). Such maps can be represented by a set of Kraus operators $\{K_i\}$, such that $\mcal{N}(\r)=\sum_i K_i \r K_i^\dagger$ and $\sum_i K_i^\dagger K_i = I$. CPTP maps are \textit{contractive}, i.e., they can only make states become less distinguishable:
\bes
\bea
D(\mcal{N}(\r),\mcal{N}(\sigma)) \le D(\r,\sigma) \\
F(\mcal{N}(\r),\mcal{N}(\sigma)) \ge F(\r,\sigma)
\eea
\ees
As a heuristic justification of these inequalities, consider a completely depolarizing noise channel which maps all states to identity: $\mcal{N}(\r) = I$ $\forall \r$. Then $D(\mcal{N}(\r),\mcal{N}(\sigma)) = 0$ and $F(\mcal{N}(\r),\mcal{N}(\sigma))=1$. At the other extreme, if $\mc{N}$ is a unitary rotation (no decoherence), i.e., $\mc{N}(\r) = U\r U^\dagger$, then $D(\mcal{N}(\r),\mcal{N}(\sigma)) = D(\r,\s)$ and $F(\mcal{N}(\r),\mcal{N}(\sigma)) = F(\r,\s)$. Other CPTP maps lie in between these two extremes.

Since the fidelity can only increase under a CPTP map it makes sense to define the \textit{fidelity of a noise channel} by taking the minimum over all input states $\r$:
\beq
F(\mcal{N})\equiv \inf_\r F(\r,\mcal{N}(\r)).
\label{eq:F-channel}
\eeq
Actually we can simplify this somewhat: we can show that the minimization doesn't require general mixed states, but instead pure states suffice. The reason that the fidelity satisfies ``strong-concavity'', i.e., for any two convex combinations of mixed states defined over the same index set,
\beq
F(\sum_i p_i \r_i,\sum_i q_i \sigma_i) \ge \sum_i \sqrt{p_i q_i } F(\r_i,\sigma_i)\ .
\eeq
With this result, and the spectral decomposition $\r = \sum_i \lambda_i \ketbra{i}$, we have from Eq.~\eqref{eq:F-channel}
\bes
\bea
F(\mcal{N}) &=& \inf_\r F(\sum_i \lambda_i \ketbra{i},\mcal{N}(\sum_i \lambda_i \ketbra{i})) \\
            &\ge& \inf_\r \sum_i \sqrt{\lambda_i \lambda_i} F(\ketbra{i},\mcal{N}(\ketbra{i})) \\
            &\ge& \inf_{\ket{i}} F(\ketbra{i},\mcal{N}(\ketbra{i}) \left(\sum_i \lambda_i\right)  \\
            &=&  \inf_{\ket{i}} F(\ketbra{i},\mcal{N}(\ketbra{i})\ ,
\eea
\ees
where in the penultimate line we used the fact that all terms of the form $F(\ketbra{i},\mcal{N}(\ketbra{i}) $ are non-negative, so eliminating all but the smallest among them certainly makes the expression smaller.

Since every mixed state has a spectral decomposition, the infimum will be achieved for some pure state belonging to the spectral decomposition of some mixed state. Hence the fidelity of a CPTP noise channel can be redefined as ($\ket{\psi}$ is a pure state)
\beq
F(\mcal{N}) = \inf_{\ket{\psi}} F(\ketbra{\psi},\mcal{N}(\ketbra{\psi})) = \inf_{\ket{\psi}} \sqrt{\braket{\psi|\mcal{N}(\ketbra{\psi})|\psi}} \ .
\label{eq:F-channel-pure}
\eeq

\subsection{Examples: fidelities of various noise channels}

\subsubsection{The pure-dephasing channel}

Consider a channel that flips the phase of a qubit with probability $p$, and acts as identity otherwise.
\beq
\mcal{N}_{\textrm{PD}}(\r)=(1-p)\r+ p Z \r Z
\eeq

The fidelity of this channel can be calculated using Eq.~\eqref{eq:F-channel-pure} as
\bes
\bea
F(\mcal{N}_{\textrm{PD}}) &=& \inf_{\ket{\psi}} F(\ketbra{\psi},\mcal{N}_{\textrm{PD}}(\ketbra{\psi}) \\
       &=& \inf_{\ket{\psi}} \sqrt{ \braket{\psi|\,\mcal{N}_{\textrm{PD}}(\ketbra{\psi})\,|\psi} } \\
       &=& \inf_{\ket{\psi}} \sqrt{ (1-p) \braket{\psi|\psi}\braket{\psi|\psi}+ p \braket{\psi|Z|\psi}\braket{\psi|Z|\psi} } \\
       &=& \inf_{\ket{\psi}} \sqrt{(1-p) + p\, \braket{Z}^2}
\eea
\ees
In this case the minimization is trivial, since, e.g., $\braket{+|Z|+}=0$. Therefore we have
\beq
F(\mcal{N}_{\textrm{PD}}) = \sqrt{1-p}=1-p/2+\mc{O}(p^2)
\eeq
We see that the fidelity has been degraded by a term of order $p$. In other words, the pure-dephasing channel introduces an error of order $O(p)$ on the system. \\

\subsubsection{The depolarizing channel}

The depolarizing channel is represented by
\beq
\mcal{N}_{\textrm{Dep}}(\r)=(1-p)\r+ \frac{p}{3} \sum_{\alpha \in \{x,y,z\}} \sigma^\alpha \r \sigma^\alpha
\eeq
Proceeding as in in the previous example,
\bes
\bea
F(\mcal{N}_{\textrm{Dep}}) &=&  \inf_{\ket{\psi}} \sqrt{ \braket{\psi|\,\mcal{N}(\ketbra{\psi})\,|\psi}} \\
       &=& \inf_{\ket{\psi}} \sqrt{ (1-p) + \frac{p}{3} \sum_{\alpha \in \{x,y,z\}}{\braket{\psi|\sigma^\alpha|\psi}^2} }
\eea
\ees
If $\ket{\psi}=a \ket{0}+ b\ket{1}$, we obtain $\braket{\sigma^z}=|a|^2-|b|^2$, $\braket{\sigma^x}=2\,\Re(a^*b)$ and $\braket{\sigma^y}=2\, \Im (a^*b)$. The minimization over all $a$ and $b$, subject to $|a|^2+|b|^2=1$, yields $a=1$ and $b=0$ as one possible solution (the easiest way to see this is to realize that the depolarizing channel is completely symmetric, so any state will do, e.g., $\ket{0}$). Thus,
\beq
F(\mcal{N}_{\textrm{Dep}}) = \sqrt{1-p + \frac{p}{3}} = 1-\frac{p}{3}+\mc{O}(p^2)
\label{eq:Fdep}
\eeq
Thus, the error is again $\mc{O}(p)$.

\newpage



%

\end{document}